\documentclass[12pt]{article}

\usepackage{amsfonts,amssymb,amsmath,exscale,relsize,enumitem,amsthm,mathtools, physics}
\usepackage{graphicx}
\usepackage{float}
\usepackage[thinlines]{easytable}
\usepackage{csquotes}
\usepackage{setspace}
\usepackage{booktabs}
\usepackage{enumitem}

\usepackage[authoryear,round]{natbib}

\usepackage[
colorlinks=true,
citecolor=blue,
linkcolor=blue,
urlcolor=blue,
pagebackref=true
]{hyperref}

\usepackage{varioref}
\usepackage{subcaption}
\usepackage{xcolor}
\usepackage{setspace}
\usepackage{appendix}
\usepackage[margin=1.35in]{geometry}
\usepackage{setspace} %\onehalfspacing        % 1.5 spacing
\DeclareMathOperator*{\argmax}{argmax}

\newtheorem{theorem}{Theorem}

\newtheorem{corollary}{Corollary}

\newtheorem{lemma}{Lemma}

\newtheorem{proposition}{Proposition}

\mathchardef\mhyphen="2D

\numberwithin{equation}{section}
\title{Covariance-Adaptive Residualization and Stagewise Calibration for Dependent Multiple Testing}

\author{%
	Prasenjit Ghosh\thanks{Department of Statistics, Texas A\&M University, College Station, TX 77843, USA. Email: \texttt{prasenjit@stat.tamu.edu}}%
	\hspace{1em} % reduce the gap
	Arijit Chakrabarti\thanks{Applied Statistics Unit, Indian Statistical Institute, Kolkata - 700108, India. Email: \texttt{arc@isical.ac.in}}%
}

\date{} % suppress date

\begin{document}
	
	\maketitle

\begin{abstract}
	
In this paper, we study the problem of simultaneous hypothesis testing for a large collection of multivariate Gaussian means under arbitrary covariance dependence. Building upon the Maximum Residual Down (MRD) procedure of \citet{COHEN_SACK_XU_2009}, we investigate a systematic stagewise calibration strategy for covariance-adaptive residual-based multiple testing based on the generalized step-down critical constants of \citet{GBS2009}. The proposed methodology retains the covariance-adaptive residualization mechanism of MRD while replacing the original model-dependent threshold specification with a simple, principled, and systematic stagewise calibration rule. Since the resulting procedure belongs to the general class of monotone residual-based step-down procedures studied by \citet{GC_ADMISSIBILITY_2026}, its admissibility under the vector loss formulation follows immediately from their general theory.

We further derive alternative representations of the MRD residual statistics that express all active residuals at each stage in terms of a single active precision matrix. Besides substantially reducing the computational burden of implementation, these representations reveal a direct connection between covariance-adaptive residualization and the evolving geometry of the active precision matrix.

An extensive simulation study under a broad range of dependence structures, including equicorrelated, Toeplitz, heterogeneous block, factor, fractional Gaussian noise, and sparse precision-matrix models, demonstrates that the proposed methodology frequently achieves substantially lower normalized misclassification risk than several widely used marginal testing procedures. Under many structured dependence models, the proposed calibration also exhibits remarkably strong signal-recovery performance, simultaneously maintaining false discovery rates close to the nominal level, extremely small false non-discovery rates, powers approaching one, and average numbers of rejections close to the expected number of true signals.

Taken together, the theoretical and empirical findings suggest that covariance-adaptive residualization and stagewise calibration play fundamentally different but complementary roles in dependent multiple testing. Beyond providing a practical calibration strategy for the MRD framework, the proposed methodology offers new insights into the complementary roles of covariance-adaptive residualization and stagewise calibration in exploiting dependence information for large-scale multiple testing under arbitrary covariance structures.

\end{abstract}

\bigskip


\medskip

	
\section{Introduction}
\label{sec:introduction}

Multiple hypothesis testing has become one of the central problems in modern statistical
inference, particularly in large-scale applications arising in genomics, bioinformatics,
neuroimaging, astronomy, economics, finance, medicine, and related scientific disciplines.
In such settings, a large number of hypotheses must be tested simultaneously in order to
identify a relatively small number of meaningful signals hidden among a large collection
of null effects. The resulting multiplicity problem has led to a rich literature on procedures
controlling global measures of type I error, including the family-wise error rate (FWER)
and the false discovery rate (FDR); see, for example, \citet{BH1995}, \citet{BY2001},
\citet{STO_2002}, \citet{STS_2004}, \citet{SKS2006}, \citet{BKY2006},  \citet{SKS2008}, \citet{GBS2009},
\citet{BLROQ2009}, \citet{SUN_CAI_2009}, \citet{NR2012}, and \citet{GHS2014}, only to name a few.

A substantial portion of the classical multiple testing literature relies on marginal test statistics or their associated \(p\)-values and was originally developed under the assumption that the underlying test statistics are independent. In many contemporary applications, however, substantial dependence arises naturally through spatial, temporal, biological, financial, network-based, and other scientific mechanisms. Such dependence can fundamentally influence the behavior of multiple testing procedures. Several authors have demonstrated that strong dependence may induce highly unstable performance in traditional marginal \(p\)-value based methods, resulting in considerable variability in both false discoveries and missed discoveries; see, for example, \citet{QBKY2005}, \citet{QKY2005}, \citet{GGQY2007}, \citet{KY2007} and \citet{QXGY2007}. In addition, \citet{Efron_2007} argued that ignoring dependence can lead to substantially distorted inferential conclusions in highly correlated settings.

These challenges have motivated extensive research on multiple testing procedures that remain valid under various forms of dependence; see, for example, \citet{BY2001}, \citet{ROM_SHK_WOLF}, \citet{SKS2008}, \citet{LEEK_STO_2008}, \citet{BLROQ2009}, \citet{FRI_KLO_CAU_2009}, \citet{SUN_CAI_2009}, and the references therein. While such developments have substantially expanded the scope of multiple testing methodology, many of these developments primarily focus on preserving validity under dependence rather than exploiting dependence as a source of inferential gain. This naturally raises a more fundamental question: can the dependence structure itself be exploited to improve multiple-testing performance?

Dependence often contains valuable information about the underlying signal configuration, and exploiting such information may substantially improve inferential efficiency. This viewpoint motivates procedures that incorporate covariance information directly into the construction of the test statistics, thereby using dependence as an inferential resource rather than merely as a nuisance to be adjusted for. From this perspective, dependence may be viewed as a \enquote{blessing} rather than a \enquote{curse}; see \citet{GRW2006}, and \citet{HJ2010}. The present paper is motivated by this latter viewpoint and investigates whether covariance information can be exploited systematically through a covariance-adaptive multiple-testing procedure that directly incorporates dependence into the testing mechanism.

%
%While error-rate control is an important objective, it does not by itself guarantee
%desirable decision-theoretic performance. In large-scale testing problems, one is often
%interested not only in controlling false discoveries, but also in accurately classifying
%null and non-null hypotheses. This motivates the study of loss functions that account
%for both false rejections and false non-rejections. In this direction, \citet{COHEN_SACK_XU_2009}
%showed that several commonly used multiple testing procedures based on marginal
%statistics, including procedures in the Benjamini--Hochberg tradition, may be
%inadmissible under dependence with respect to natural vector-valued loss functions.
%Such results demonstrate that classical error-control properties and decision-theoretic
%admissibility capture distinct aspects of multiple testing performance.

Among the relatively few procedures that explicitly incorporate covariance information, the Maximum Residual Down (MRD) procedure of \citet{COHEN_SACK_XU_2009} occupies a particularly prominent position. Rather than relying solely on marginal
statistics, MRD constructs adaptive residual statistics obtained from the conditional distribution of each coordinate given the remaining active coordinates. By sequentially
removing hypotheses exhibiting the strongest residual evidence against the null, MRD incorporates dependence information directly into the testing mechanism. The resulting statistics depend explicitly on the covariance matrix and evolve adaptively as hypotheses are removed during the step-down process. Such covariance-adaptive procedures naturally motivate performance criteria beyond traditional error-rate control, since their ability to exploit dependence may influence not only type I errors but also overall classification accuracy.

While error-rate control remains an important objective, it does not by itself provide a complete assessment of performance. In large-scale multiple testing problems, one is often interested not only in controlling some overall measure of false discoveries but also in accurately identifying the underlying signal configuration. This naturally motivates decision-theoretic criteria that account simultaneously for false rejections and false non-rejections. In this direction, \citet{COHEN_SACK_2005_B}, \citet{COHEN_KOL_SACK_2007},  \citet{COHEN_SACK_2007}, and \citet{COHEN_SACK_2008} showed that several widely used procedures based on marginal statistics, including the celebrated Benjamini--Hochberg (BH) method and its variants, may be inadmissible under dependence with respect to natural vector-valued loss functions. The inadmissibility of a multiple testing procedure under the vector loss formulation implies that it will be inadmissible with respect to the total loss function. These findings demonstrate that classical error-control properties and decision-theoretic admissibility capture distinct aspects of multiple-testing performance and highlight the importance of evaluating procedures using criteria that extend beyond error-rate control alone.

These considerations underscore a broader point. Procedures possessing similar error-rate control properties may nevertheless differ substantially in their ability to recover true signals, avoid missed discoveries, and minimize overall classification errors. Consequently, evaluating multiple testing procedures solely through a particular type I error criterion may provide an incomplete assessment of their practical performance.

%
%More broadly, control of a particular type I error criterion, while undoubtedly important, does not by itself guarantee desirable overall performance in large-scale testing problems. Procedures possessing similar error-rate control properties may differ substantially in their ability to recover true signals, avoid missed discoveries, or minimize overall classification errors. This observation motivates the study of multiple testing procedures from a broader decision-theoretic perspective and highlights the importance of evaluating criteria beyond error-rate control alone.

The MRD framework also has a strong decision-theoretic foundation. The original MRD procedure proposed 
by \citet{COHEN_SACK_XU_2009} was shown to be admissible under a natural vector-valued loss formulation. More recently, \citet{GC_ADMISSIBILITY_2026} established admissibility for a very broad class of monotone residual-based step-down
procedures under arbitrary covariance dependence. Thus, admissibility is not merely a property of one particular choice of critical constants, but rather a structural feature of
a large class of residual-based testing procedures. This observation is particularly important for the present work, since the GBS-calibrated MRD procedure proposed in this paper remains admissible under the vector loss formulation and therefore inherits the same decision-theoretic guarantees.

The decision-theoretic perspective adopted by \citet{COHEN_SACK_XU_2009} also provides a natural motivation for evaluating procedures through misclassification risk. Admissibility is defined with respect to classification-type loss functions that penalize both false rejections and false non-rejections, making the total number of classification errors a fundamental measure of performance. Moreover, the inadmissibility of several widely used procedures under the vector loss formulation implies the existence of alternative procedures with uniformly better risk performance under the corresponding decision-theoretic criterion. These considerations motivate our emphasis on misclassification risk, alongside traditional measures such as FDR, power, and false non-discovery rate, throughout the simulation study.

%The decision-theoretic perspective adopted by \citet{COHEN_SACK_XU_2009}
%also provides a natural motivation for evaluating procedures in terms of
%misclassification risk. Admissibility is defined with respect to
%classification-type loss functions that penalize both false rejections
%and false non-rejections. Consequently, the total number of
%classification errors emerges as a fundamental measure of overall
%performance. Moreover, the inadmissibility of several widely used
%procedures under the vector loss formulation implies the existence of
%alternative procedures that uniformly improve their risk performance
%under the corresponding loss function. Furthermore, inadmissibility
%under the vector loss formulation immediately implies inadmissibility
%under the corresponding total loss obtained by summing the component
%losses. Consequently, admissibility has direct implications for overall
%classification accuracy. These considerations motivate our emphasis on
%misclassification risk, alongside traditional measures such as FDR,
%power, and false non-discovery rate, throughout the simulation study.

Despite its attractive theoretical properties, practical implementation of the original MRD procedure requires specification of a sequence of stagewise critical constants. In the original work of \citet{COHEN_SACK_XU_2009}, these constants were chosen largely through numerical considerations and tailored to the particular dependence settings under investigation. The primary focus of the original MRD methodology was the construction and decision-theoretic analysis of covariance-adjusted residual statistics rather than the development of a broadly applicable calibration framework. Consequently, the original MRD procedure did not provide a unified principle for selecting critical constants across arbitrary covariance structures, leaving practical implementation tied to problem-specific choices. This naturally raises the question of whether one can develop a systematic calibration strategy while retaining the covariance-adaptive and admissible nature of the MRD framework.

To address this issue, the present paper investigates a covariance-adaptive calibration of the MRD procedure based on the generalized step-down critical constants proposed by \citet{GBS2009}. Although these critical constants were originally introduced as part of an adaptive step-down procedure possessing false discovery rate control under independence, subsequent theoretical developments have revealed considerably stronger optimality properties. In particular, recent work has established that the GBS procedure attains asymptotic Bayes optimality under sparsity \citep{GhoshChakrabarti2026GBSABOS} and, under suitable sparsity conditions, sharp asymptotic minimaxity in sparse Gaussian multiple testing problems \citep{Ghosh2026GBSMinimaxity}. Taken together, these results suggest that the GBS stagewise calibration provides a remarkably effective mechanism for balancing false discoveries and missed discoveries in high-dimensional sparse inference. The central question addressed in the present work is whether these attractive calibration properties can be successfully integrated with covariance-adaptive residualization to develop a systematic, decision-theoretically valid multiple testing procedure under arbitrary covariance dependence.

Viewed more broadly, the present paper investigates the interaction between two complementary components of covariance-aware multiple testing procedures. Covariance-adaptive residualization determines how dependence information is accumulated and propagated through the sequential testing process, whereas stagewise calibration determines how the accumulated evidence is translated into sequential testing decisions. Although these two components play complementary roles in covariance-aware multiple testing, they are typically developed independently. The central premise of the present work is that both play fundamental roles in determining the overall operating characteristics of a covariance-aware multiple testing procedure. The proposed MRD--GBS methodology combines these complementary ideas by integrating covariance-adaptive residualization with the stagewise calibration strategy of \citet{GBS2009}, thereby yielding a systematic covariance-aware multiple testing framework that integrates covariance-adaptive residualization with principled stagewise calibration under arbitrary covariance dependence.

From a methodological perspective, the underlying philosophy of the GBS calibration is that hypotheses surviving a long sequence of rejections have already undergone substantial screening and therefore need not be penalized as severely as hypotheses encountered at the beginning of the testing process. Consequently, the stagewise critical values become progressively less conservative as the step-down sequence evolves, allowing the evidence accumulated during sequential testing to be translated into progressively less conservative testing decisions. Since the MRD framework itself proceeds through a sequential elimination mechanism based on covariance-adaptive residual statistics, the GBS calibration provides a particularly natural and principled mechanism for allocating statistical evidence across successive stages of the residual-based testing procedure. The resulting methodology therefore combines two complementary ingredients: covariance-adaptive residualization for exploiting dependence information and stagewise GBS calibration for systematically translating the accumulated evidence into sequential testing decisions.

Our contributions in this paper are threefold. First, we introduce a systematic calibration strategy for the MRD procedure based on the generalized step-down critical constants of \citet{GBS2009}. The proposed calibration replaces the model-dependent threshold specification required by the original MRD formulation with a unified stagewise rule that is applicable across a broad range of dependence structures. As a result, the procedure retains the covariance-adaptive residualization mechanism of MRD while eliminating the need for covariance-specific threshold tuning. Moreover, since the proposed method remains within the general class of monotone residual-based step-down procedures considered by \citet{GC_ADMISSIBILITY_2026}, admissibility follows immediately from their general theory. Consequently, the resulting methodology is not merely a computational or empirical modification of MRD, but rather a decision-theoretically valid residual-based multiple testing procedure under arbitrary covariance dependence.

%First, we formulate a GBS-calibrated version of the MRD procedure for
%simultaneous testing under arbitrary known covariance dependence. The
%resulting methodology replaces the model-specific critical-value
%selection employed in the original MRD framework by a systematic
%stagewise calibration strategy, while preserving the covariance-aware
%residual structure of the procedure. 

%Second, we derive alternative computational representations of the
%stagewise residual statistics. Under the original formulation, direct
%implementation of MRD requires computing a separate leave-one-out
%submatrix inverse for each active coordinate at each stage, leading to
%\(m+(m-1)+\cdots+1=n(n+1)/2\) matrix inversions in the worst case, $n$ being the total number of multiple testing problems under consideration. The
%alternative representation developed here shows that all residual
%statistics at a given stage can be obtained from a single active
%covariance inverse. Consequently, the number of required matrix
%inversions is reduced from $m(m+1)/2$ to at most \(m\), yielding a
%substantial computational simplification. This converts the original
%leave-one-out implementation into a stagewise precision-matrix
%representation in which all active residual statistics are obtained
%simultaneously from a single matrix inverse. The resulting
%simplification substantially improves computational scalability
%while preserving the exact MRD residual statistics.

Second, we derive alternative computational representations of the stagewise residual statistics. Under the original formulation, direct implementation of the MRD procedure requires computing a separate leave-one-out sub-matrix inverse for each active coordinate at each stage, leading to $m(m+1)/2$ matrix inversions in the worst case, $m$ being the total number of null hypotheses to be tested. The alternative representation developed here shows that all residual statistics at a given stage can be obtained from a single active covariance-matrix inverse. Consequently, the number of required matrix inversions is reduced from $m(m+1)/2$ to at most \(m\), yielding a substantial computational simplification. The resulting
simplification substantially improves computational scalability while preserving the exact MRD residual statistics. More importantly, the new representation reveals that the residual-based testing mechanism may be viewed directly through the geometry of the active precision matrix. Thus, beyond providing a computationally efficient implementation, the resulting formulation offers a more transparent structural interpretation of the MRD framework itself.

Third, we conduct an extensive simulation study under a broad collection of dependence structures, including equicorrelation, factor, Toeplitz, fractional Gaussian noise, sparse precision-matrix, and heterogeneous block covariance models. The numerical results demonstrate that the proposed procedure frequently achieves the smallest normalized misclassification risks in sparse and moderately sparse regimes while simultaneously exhibiting remarkably strong signal-recovery behavior. Across several dependence structures, the procedure attains false discovery rates near the nominal level, false non-discovery rates approaching zero, powers close to one, and average numbers of rejections that closely track the expected number of true signals, providing compelling empirical evidence of near-support-recovery behavior. The simulations further reveal two notable phenomena. First, the relative performance of the original MRD procedure and its GBS-calibrated counterpart depends strongly on the underlying sparsity regime. While the proposed calibration frequently achieves the smallest normalized misclassification risks in sparse settings, the original MRD procedure often becomes increasingly competitive and may eventually outperform the calibrated version as the signal configuration becomes denser. Second, the magnitude of the gains varies substantially across covariance structures, suggesting that covariance geometry may strongly influence how effectively dependence information can be exploited through residual-based testing procedures. Collectively, these findings provide new empirical insights into the interaction among dependence, calibration, information propagation, and signal recovery in covariance-adaptive multiple-testing procedures.

We emphasize that the present paper does not claim formal FDR control or asymptotic
support recovery for the calibrated MRD procedure. Rather, our goal is methodological:
to study how GBS-type calibration interacts with residual-based covariance adaptation,
to provide efficient computational tools for implementation, and to examine the resulting
operating characteristics under diverse dependence structures. The numerical findings
suggest that the structure of the covariance matrix may play an important role in
determining how information propagates through sequential residual-based testing
procedures. More broadly, the numerical findings raise the possibility that the effective difficulty of a large-scale multiple testing problem under dependence may depend not only on the nominal number of hypotheses but also on structural features of the underlying covariance matrix and on the extent to which dependence information can be exploited through the testing mechanism itself. The methodology developed here therefore serves both as a practical
calibration strategy for residual-based multiple testing under dependence and as an empirical investigation of how covariance geometry influences large-scale testing performance. Taken together, these findings suggest that covariance-adaptive residualization and stagewise calibration interact in ways that are not captured by traditional marginal testing procedures, thereby motivating a broader theoretical investigation of optimal multiple testing under dependence.

The remainder of the paper is organized as follows.
Section~\ref{SECTION_PROBLEM_MRD_DEFNITION} introduces the testing framework and reviews the residual
statistics underlying the MRD procedure.
Section~\ref{SECTION_MRD_GBS_DEFINITION} presents the proposed GBS-calibrated MRD methodology and
establishes its admissibility properties.
Section~\ref{SECTION_MRD_STAT_ALTERNATIVE} develops alternative representations of the MRD residual
statistics together with substantial computational simplifications.
Section~\ref{SECTION_SIMULATIONS} contains an extensive simulation study under a variety of
dependence structures.
Section~\ref{SECTION_DISCUSSION} concludes with a discussion of the findings and several
directions for future research.
Appendix A contains proofs of the technical results presented in the
main text, while Appendix B and Appendix C report additional simulation results and supplementary numerical summaries.

\section{Problem Formulation and the MRD Procedure}
\label{SECTION_PROBLEM_MRD_DEFNITION}

In this paper, we consider the problem of simultaneous testing for means of a set of jointly normal variables. Towards that, let us assume that we observe a random vector $\boldsymbol{X}=(X_1,\dots,X_{m})$ (obtained through some suitable transformation, if necessary) such that $\boldsymbol{X} \sim N_{m}(\boldsymbol{\theta},\boldsymbol{\Sigma})$ where  $\boldsymbol{\theta}=(\theta_1,\dots,\theta_{m})$ is the vector of unknown means and $\boldsymbol{\Sigma}=((\sigma_{ij}))$ is an $m \times m$ known positive definite matrix with an arbitrary covariance structure. We are interested in testing simultaneously 
\begin{equation}\label{TESTING_PROBLEM}
	H_{0i}:\theta_{i}=0 \mbox{ against } H_{Ai}:\theta_{i}\neq 0, \mbox{ for } i=1,\dots,m.
\end{equation}

Note that since $\boldsymbol{\Sigma}$ is known, without loss of generality, one may assume $\boldsymbol{\Sigma}$ to be the correlation matrix of the $X_i$'s so that $X_i \sim N(\theta_{i},1)$ for each $i=1,\dots,m$. This is so since if
\begin{eqnarray}\label{DIAG_MATRIX}
	\boldsymbol{D}
	&=& \begin{pmatrix} \sigma_{11} & 0 & 0 & \dots & 0 \\ 0 & \sigma_{22} & 0 & \dots & 0 \\ \vdots \\ 0 & 0 & 0 & \dots & \sigma_{mm}\end{pmatrix},
\end{eqnarray}
then letting $\boldsymbol{U}=\boldsymbol{D}^{-1/2}\boldsymbol{X}$ we have $\boldsymbol{U}\sim N_{m}(\boldsymbol{\mu}, \boldsymbol{\Lambda})$, where $\boldsymbol{\mu} = \boldsymbol{D}^{-1/2} \boldsymbol{\theta}$ and $\boldsymbol{\Lambda}=\boldsymbol{D}^{-1/2}\boldsymbol{\Sigma} \boldsymbol{D}^{-1/2}$ is simply the correlation matrix of $\boldsymbol{X}$. Therefore, testing $H^{'}_{0i}:\mu_i=0\mbox{ against } H^{'}_{Ai}:\mu_i\neq 0$ simultaneously for $i=1,\dots,m$, is equivalent to the original testing problem (\ref{TESTING_PROBLEM}).

We now introduce the Maximum Residual Down (MRD) procedure proposed by \citet{COHEN_SACK_XU_2009}, which serves as the canonical residual-based
step-down procedure under arbitrary covariance dependence. For that, we adopt here similar convention of notations used in \citet{COHEN_SACK_XU_2009}. 

Let $\boldsymbol{X}^{(i_{1},\dots,i_{t})}$ be an $(m-t)\times1$ vector consisting of those components of $\boldsymbol{X}=(X_1,\dots,X_{m})$ with $(X_{i_{1}},\dots,X_{i_{t}})$ left out. Suppose $\boldsymbol{\Sigma}_{(i_{1},\dots,i_{t})} $ is the $(m-t)\times(m-t)$ sub-matrix obtained after eliminating the $i_{1},\dots,i_{t}$-th rows and the corresponding columns of $\boldsymbol{\Sigma}$. Let $\boldsymbol{\sigma}_{(j)}^{(i_{1},\dots,i_{t})} $ be the $(m-t-1)\times1$ vector obtained by eliminating the $i_{1},\dots,i_{t}$-th and $j$-th elements of the $j$-th column vector of $\boldsymbol{\Sigma}$. Define
\begin{eqnarray}
	\sigma_{j\cdot(i_{1},\dots,i_{t})} &=& \sigma_{jj} - {\boldsymbol{\sigma}_{(j)}^{(i_{1},\dots,i_{t})}}^{T}{\boldsymbol{\Sigma}^{-1}_{(i_1,\dots,i_t,j)}}{\boldsymbol{\sigma}_{(j)}^{(i_{1},\dots,i_{t})}},\nonumber
\end{eqnarray}
which denotes the conditional variance of $X_{j}$ given $\boldsymbol{X}^{(i_{1},\dots,i_{t})}$.

The MRD procedure is based on a collection of adaptively formed residual
statistics defined as
\begin{eqnarray}\label{MRD_STATISTICS}
	U_{tj}^{(i_1,\dots,i_{t-1})}(\boldsymbol{X}) &=& \frac{X_j - {\boldsymbol{\sigma}_{(j)}^{(i_1,\dots,i_{t-1})}}^{T}{\boldsymbol{\Sigma}^{-1}_{(i_1,\dots,i_{t-1},j)}}{\boldsymbol{X}^{(i_1,\dots,i_{t-1},j)}}}{\sigma^{\frac{1}{2}}_{j\cdot(i_1,\dots,i_{t-1})}},\nonumber
\end{eqnarray}
for $t,j=1,\dots,m$, $1\leqslant i_{1}\neq\dots\neq i_{t-1}\leqslant n$ and $i_{l}\neq j$ for all $l=1,\dots,t-1$.

For $1 \leqslant t \leqslant n$, we define the index $\widetilde{j}_t(\boldsymbol{X})$ as
\begin{align}\label{MRD_INDEX}
	\widetilde{j}_t(\boldsymbol{X})&=\argmax_{j \in \{1,\dots,m\}\setminus\{\widetilde{j}_1(\boldsymbol{X}),\dots,\widetilde{j}_{t-1}(\boldsymbol{X})\}} |U^{(\widetilde{j}_1(\boldsymbol{X}),\dots,\widetilde{j}_{t-1}(\boldsymbol{X}))}_{tj}(\boldsymbol{X})|.
\end{align}

Given a set of positive constants $C_1 \geqslant C_2 \geqslant \dots \geqslant C_{m}$, the MRD method works in a step-down manner as follows:
\begin{enumerate}
	\item At stage 1, consider the statistics $|U_{1j}(\boldsymbol{X})|$, where $j \in \{1,\dots,m\}$. If $|U_{1\widetilde{j}_1(\boldsymbol{X})}(\boldsymbol{X})| \leqslant C_1, $ stop and accept all $H_{0i}$'s. Otherwise reject $H_{0\widetilde{j}_1(\boldsymbol{X})}$ and continue to stage 2.
	
	\item At stage 2, consider the statistics $|U^{(\widetilde{j}_1(\boldsymbol{X}))}_{2\widetilde{j}_1(\boldsymbol{X})}(\boldsymbol{X})|$, where $j\in\{1,\dots,m \}\setminus\{\widetilde{j}_1(\boldsymbol{X})\}$. If $|U^{(\widetilde{j}_1(\boldsymbol{X}))}_{2\widetilde{j}_2(\boldsymbol{X})}(\boldsymbol{X})| \leqslant C_2, $ stop and accept all the remaining $H_{0i}$'s. Otherwise, reject $H_{0\widetilde{j}_2(\boldsymbol{X})}$ and continue to stage 3.
	
	\item In general, at stage $t$, consider the statistics $|U^{(\widetilde{j}_1(\boldsymbol{X}),\dots,\widetilde{j}_{t-1}(\boldsymbol{X}))}_{tj}(\boldsymbol{X})|$, where $j\in\{1,\dots,m\}\setminus\{\widetilde{j}_1(\boldsymbol{X}),\dots,\widetilde{j}_{t-1}(\boldsymbol{X})\}$. If $|U^{(\widetilde{j}_1(\boldsymbol{X}),\dots,\widetilde{j}_{t-1}(\boldsymbol{X}))}_{t\widetilde{j}_t(\boldsymbol{X})}(\boldsymbol{X})| \leqslant C_t$, stop and accept all the remaining $H_{0i}$'s. Otherwise, reject $H_{0\widetilde{j}_t(\boldsymbol{X})}$ and move to stage $(t+1)$.
	
	\item We continue in this fashion until an acceptance occurs or there are no more null hypotheses to be tested, in which case we must stop.
\end{enumerate}

The MRD procedure thus generates a sequence of adaptive residual
statistics that explicitly incorporates the covariance structure of the
observations into the testing process. Through its sequential
elimination mechanism, the procedure updates the evidence associated
with each active hypothesis as earlier rejections occur, thereby
allowing information contained in the dependence structure to influence
subsequent testing decisions. The performance of the procedure,
however, depends crucially on the choice of the stagewise critical
constants \(C_1,\ldots,C_{m}\). In the original formulation of
\citet{COHEN_SACK_XU_2009}, these constants were selected through
numerical considerations tailored to specific dependence settings. This
observation motivates the development of a more systematic calibration
strategy, which is the focus of the next section.

\section{GBS-Calibrated MRD Procedure}
\label{SECTION_MRD_GBS_DEFINITION}

As discussed in the Introduction, practical implementation of the
original MRD procedure requires specification of a sequence of
stagewise critical constants
\[
C_1 \ge C_2 \ge \cdots \ge C_{m} > 0.
\]
In the original work of \citet{COHEN_SACK_XU_2009}, these constants
were selected through numerical studies tailored to the particular
dependence structures under consideration. While such choices yielded
procedures with attractive decision-theoretic properties, they did not
provide a unified calibration principle applicable across arbitrary
covariance structures. Consequently, practical implementation of MRD
remained tied to problem-specific choices of critical constants.

In this paper, we investigate an alternative calibration strategy based
on the adaptive step-down critical constants proposed by
\citet{GBS2009}. The resulting procedure preserves the covariance-aware
residualization mechanism of MRD while replacing the original
model-dependent critical values by a systematic stagewise calibration
rule.

The GBS calibration is particularly well suited to the MRD framework because its stagewise critical values become progressively less conservative as the step-down procedure evolves. Since the MRD algorithm itself proceeds through a sequential elimination mechanism based on covariance-adaptive residual statistics, the GBS calibration provides a natural and principled mechanism for translating the accumulated residual evidence into sequential testing decisions.

\subsection{Definition of the Procedure}

For \(i=1,\ldots,m\), define the GBS critical values
\[
\alpha_i
=
\frac{i\alpha}
{m+1-i(1-\alpha)},
\]
where \(\alpha\in(0,1)\) denotes the nominal significance level.

%Since the MRD procedure is based on two-sided residual statistics,
%we convert these significance levels into critical values through

Since each MRD residual statistic is a standardized conditional residual, having a standard normal distribution under the corresponding null hypothesis, a two-sided significance level \(\alpha_i\) naturally translates into the critical value
\[
C_i^{GBS}
=
\Phi^{-1}
\!\left(
1-\frac{\alpha_i}{2}
\right),
\qquad i=1,\ldots,m,
\]
where \(\Phi\) denotes the standard normal distribution function.

\begin{proposition}
For given $0<\alpha<1$,	the sequence of positive constants $\{C_i^{GBS}\}_{i=1}^{m}$ is strictly decreasing, that is,
\[
C_1^{GBS}>C_2^{GBS}>\ldots>C_{m}^{GBS}>0.
\]
\end{proposition}
\begin{proof}
	The sequence \(\{\alpha_i\}_{i=1}^{m}\) is strictly increasing in \(i\), while
	\(\Phi^{-1}(1-u/2)\) is strictly decreasing in \(u\in(0,1)\).
	Therefore \(\{C_i^{GBS}\}\) is strictly decreasing.
\end{proof}

The GBS-calibrated MRD procedure is obtained by implementing
the MRD algorithm described in Section~\ref{SECTION_PROBLEM_MRD_DEFNITION} with the stagewise
critical constants

\[
C_1,C_2,\ldots,C_{m}
\]

replaced by the GBS critical values

\[
C^{GBS}_1,C^{GBS}_2,\ldots,C^{GBS}_{m}.
\]

Thus, at stage \(t\), the hypothesis corresponding to the largest
active residual statistic is rejected whenever its absolute residual
exceeds \(C_t^{GBS}\); otherwise the procedure terminates and all
remaining hypotheses are accepted.

Observe that
\[
C_1^{GBS}
>
C_2^{GBS}
>
\cdots
>
C_{m}^{GBS},
\]
so that the resulting critical values become progressively less
conservative as the step-down procedure evolves.

It is worth noting that under independence, where $\boldsymbol{\Sigma} = \boldsymbol{I}$, the residual statistics reduce to the original standardized observations and the covariance-adaptive residualization mechanism disappears. In this case, the proposed MRD--GBS procedure reduces exactly to the ordinary GBS step-down procedure of \citet{GBS2009} based on marginal test statistics. Thus, under independence, the proposed methodology coincides with a well-established procedure possessing false discovery rate control properties, while extending naturally to settings involving arbitrary covariance dependence. Thus, the proposed procedure may be viewed as a covariance-adaptive extension of the GBS step-down rule: it agrees with GBS under independence, but replaces marginal statistics by sequential residual statistics when dependence is present.

%At stage \(t\), let
%\[
%\widetilde{j}_t(X)
%=
%\argmax_{j\in \mathcal{A}_{t}}
%|U_{tj}(X)|,
%\]
%where \(\mathcal{A}_{t}\) denotes the set of active hypotheses at stage \(t\).
%
%The GBS-calibrated MRD procedure proceeds as follows.
%
%\begin{enumerate}
%	\item
%	At stage \(t\), compute
%	\[
%	|U_{t\widetilde{j}_t(X)}(X)|.
%	\]
%	
%	\item
%	If
%	\[
%	|U_{t\widetilde{j}_t(X)}(X)|
%	\le
%	C_t^{GBS},
%	\]
%	stop and accept all remaining null hypotheses.
%	
%	\item
%	Otherwise reject
%	\[
%	H_{0\widetilde{j}_t(X)}
%	\]
%	and continue to stage \(t+1\).
%	
%	\item
%	The procedure terminates when an acceptance occurs or when all
%	hypotheses have been rejected.
%\end{enumerate}

\subsection{Admissibility of the GBS-Calibrated MRD}

An important feature of the proposed calibration is that it preserves
the admissibility properties of the original MRD framework. The original
MRD procedure of \citet{COHEN_SACK_XU_2009} was shown to be admissible
with respect to a natural vector-valued loss function. More recently,
\citet{GC_ADMISSIBILITY_2026} established a broader structural result
showing that admissibility is a consequence of the underlying
residual-based step-down geometry rather than of any particular choice
of critical constants.

Recall that any multiple testing procedure
$\Phi(\boldsymbol{x})=(\phi_1(\boldsymbol{x}),\dots,\phi_m(\boldsymbol{x}))$
induces an individual test function $\phi_j(\boldsymbol{x})$ for testing
$H_{0j}$ against $H_{Aj}$, where $\phi_j(\boldsymbol{x})$ denotes the probability
of rejecting the $j$-th null hypothesis when the observation
$\boldsymbol{X}=\boldsymbol{x}$ is realized. We consider the standard $0-1$ loss function
corresponding to $\phi_j$, given by
\begin{equation}\label{INDIVIDUAL_LOSS}
	L_j\big(\phi_j(\boldsymbol{X}),\boldsymbol{\theta}\big)
	=
	I\{\theta_j=0\}\phi_j(\boldsymbol{X})
	+
	I\{\theta_j\neq0\}
	\big(1-\phi_j(\boldsymbol{X})\big),
\end{equation}
while the corresponding risk function is given by
\begin{equation}\label{INDIVIDUAL_RISK}
	R_j\big(\phi_j,\boldsymbol{\theta}\big)
	=
	I\{\theta_j=0\}
	E_{\boldsymbol{\theta}:\theta_j=0}\big(\phi_j(\boldsymbol{X})\big)
	+
	I\{\theta_j\neq0\}
	E_{\boldsymbol{\theta}:\theta_j\neq0}
	\big(1-\phi_j(\boldsymbol{X})\big).
	\nonumber
\end{equation}

We consider the overall loss function for the multiple testing procedure
$\Phi(\boldsymbol{X})$ to be the vector loss
\begin{equation}\label{VECTOR_LOSS}
	L\big(\Phi(\boldsymbol{X}),\boldsymbol{\theta}\big)
	=
	\big(
	L_1\big(\phi_1(\boldsymbol{X}),\boldsymbol{\theta}\big),
	\dots,
	L_{m}\big(\phi_m(\boldsymbol{X}),\boldsymbol{\theta}\big)
	\big),
\end{equation}
with corresponding vector risk function
\begin{equation}\label{VECTOR_RISK}
	R\big(\Phi,\boldsymbol{\theta}\big)
	=
	\big(
	R_1\big(\phi_1,\boldsymbol{\theta}\big),
	\dots,
	R_{n}\big(\phi_m,\boldsymbol{\theta}\big)
	\big).
\end{equation}

The corresponding total loss function is defined as
\begin{equation}\label{TOTAL_LOSS}
	L_T\big(\Phi(\boldsymbol{X}),\boldsymbol{\theta}\big)
	=
	\sum_{j=1}^{m}
	L_j\big(\phi_j(\boldsymbol{X}),\boldsymbol{\theta}\big),
\end{equation}
while the associated total risk function is given by
\begin{equation}\label{TOTAL_RISK}
	R_T\big(\Phi,\boldsymbol{\theta}\big)
	=
	\sum_{j=1}^{m}
	R_j\big(\phi_j,\boldsymbol{\theta}\big).
\end{equation}

Observe that the total loss function (\ref{TOTAL_LOSS}) simply counts
the total number of classification errors, namely the sum of the numbers
of type I and type II errors committed by the multiple testing procedure.
Consequently, the total risk function (\ref{TOTAL_RISK}) corresponds to
the expected total number of classification errors.

A multiple testing procedure $\Phi(\boldsymbol{X})$ is said to be inadmissible
with respect to the vector loss function (\ref{VECTOR_LOSS}) if there exists
another multiple testing procedure $\Phi^{*}(\boldsymbol{X})$ such that
$R_j\big(\phi^{*}_j,\boldsymbol{\theta}\big)
\leq
R_j\big(\phi_j,\boldsymbol{\theta}\big)$
for all $j=1,\dots,m$ and all $\boldsymbol{\theta}\in\mathbb{R}^{m}$,
with strict inequality holding for at least one $j$ and some
$\boldsymbol{\theta}\in\mathbb{R}^{m}$. On the other hand, a multiple testing procedure $\Phi(\boldsymbol{X})$ is said to be inadmissible
with respect to the total loss function (\ref{TOTAL_LOSS}) if there exists
another multiple testing procedure $\Phi^{*}(\boldsymbol{X})$ such that
$R_T\big(\Phi^{*},\boldsymbol{\theta}\big)
\leq
R_T\big(\Phi,\boldsymbol{\theta}\big)$
for all $\boldsymbol{\theta}\in\mathbb{R}^{m}$,
with strict inequality holding for some
$\boldsymbol{\theta}\in\mathbb{R}^{m}$. A multiple testing procedure with respect to a given loss is said to be
admissible if it is not inadmissible in the aforesaid sense with respect to the given loss. Since the total loss function (\ref{TOTAL_LOSS}) is obtained by summing
the component losses in \eqref{VECTOR_LOSS}, inadmissibility with respect to the vector loss
function (\ref{VECTOR_LOSS}) immediately implies inadmissibility with
respect to the total loss function (\ref{TOTAL_LOSS}).

In this context, it is worth recalling that
\citet{COHEN_KOL_SACK_2007},
\citet{COHEN_SACK_2005_B},
\citet{COHEN_SACK_2007} and
\citet{COHEN_SACK_2008}
showed that in many common applications involving dependent test statistics,
typical $p$-value based stepwise testing procedures, including the celebrated
BH method, are inadmissible with respect to the vector loss function
(\ref{VECTOR_LOSS}). Consequently, such procedures also become inadmissible with respect
to the total loss function (\ref{TOTAL_LOSS}). This reveals an undesirable feature of many
traditional stepwise multiple testing procedures under dependence.

The general admissibility theorem of \citet{GC_ADMISSIBILITY_2026}
now yields the following immediate consequence.

\begin{proposition}
Suppose $\boldsymbol{X}\sim N_{m}(\boldsymbol{\theta},\boldsymbol{\Sigma})$, where $\boldsymbol{\theta}\in\mathbb{R}^{m}$ is unknown, but fixed and $\boldsymbol{\Sigma}$ is an $m \times m$ arbitrary but known positive definite covariance matrix. Then, for the two sided multiple testing problem (\ref{TESTING_PROBLEM}), the GBS-calibrated MRD procedure is admissible with respect to the
vector loss function defined in \eqref{VECTOR_LOSS}.
\end{proposition}

\begin{proof}
	The GBS calibration modifies only the stagewise critical constants and
	does not alter the residual statistics \(U_{tj}\), the ordering rule
	used to select hypotheses for rejection, or the underlying step-down
	structure of the procedure. Consequently, the resulting procedure
	remains a monotone residual-based step-down procedure. Admissibility
	therefore follows immediately from Theorem~3.1 of
	\citet{GC_ADMISSIBILITY_2026}.
\end{proof}

The results of this section show that the proposed calibration strategy
provides a systematic and admissible implementation of the MRD
framework under arbitrary covariance dependence. However, practical
implementation of the procedure remains computationally demanding when
the number of hypotheses is large, since the residual statistics
appear to require repeated inversion of a large collection of
submatrices. In the next section, we develop alternative
representations of the MRD residual statistics that substantially
reduce the computational burden of the procedure.

\section{Alternative Representation and Computational Simplification}\label{SECTION_MRD_STAT_ALTERNATIVE}

In the original formulation of the MRD procedure, the residual statistic \eqref{MRD_STATISTICS} at stage
\(t\) is defined through a leave-one-out conditional residual. Consequently, at the
\(t\)-th stage, direct implementation appears to require computing, for each active
coordinate \(j\), the inverse of a separate submatrix
\[
\boldsymbol{\Sigma}_{(i_1,\ldots,i_{t-1},j)}.
\]
Thus, if the procedure continues until the final stage, a direct implementation
requires
\[
m+(m-1)+\cdots+1=\frac{m(m+1)}{2}
\]
sub-matrix inversions in the worst possible scenario. In this section, we show that this
computational burden can be substantially reduced. The key observation is that all
MRD residual statistics at a given stage can be computed from a single inverse of
the active covariance matrix.

Let
\[
\mathcal{A}_{t}=\{1,\ldots,m\}\setminus\{i_1,\ldots,i_{t-1}\}
\]
denote the active set of hypotheses at stage \(t\), and let
\[
\boldsymbol{\Sigma}_{\mathcal{A}_{t}}
=
\boldsymbol{\Sigma}_{(i_1,\ldots,i_{t-1})}
\]
be the corresponding active covariance matrix. Write
\[
\boldsymbol{B}_{\mathcal{A}_{t}}
=
\boldsymbol{\Sigma}_{\mathcal{A}_{t}}^{-1}
=
\big(\big(b^{\mathcal{A}_{t}}_{kl}\big)\big)_{k,l\in \mathcal{A}_{t}}.
\]
For \(j\in \mathcal{A}_{t}\), let \(\boldsymbol{b}^{\mathcal{A}_{t}}_{j}\) denote the column of \(\boldsymbol{B}_{\mathcal{A}_{t}}\)
corresponding to the coordinate \(j\), and let \(b^{\mathcal{A}_{t}}_{jj}\) denote the
corresponding diagonal entry.

We first record a standard block inverse identity that will be used to connect
the original MRD residual statistic with the active precision matrix representation.

\begin{lemma}\label{LEM_BLOCK_INVERSE_MRD}
	Let \(A\subset\{1,\ldots,m\}\), \(j\in A\), and write
	\[
	A_{-j}=A\setminus\{j\}.
	\]
	Partition the covariance matrix \(\boldsymbol{\Sigma}_A\) as
	\[
	\boldsymbol{\Sigma}_A
	=
	\begin{pmatrix}
		\sigma_{jj} & \boldsymbol{\sigma}_{j,A_{-j}}^{T}\\
		\boldsymbol{\sigma}_{j,A_{-j}} & \boldsymbol{\Sigma}_{A_{-j}}
	\end{pmatrix},
	\]
	where the coordinate \(j\) has been placed first. Let
	\[
	\boldsymbol{B}_A=\boldsymbol{\Sigma}_A^{-1}
	=
	\big(\big(b^{A}_{kl}\big)\big)_{k,l\in A}.
	\]
	Then
	\[
	b^{A}_{jj}
	=
	\left(
	\sigma_{jj}
	-
	\boldsymbol{\sigma}_{j,A_{-j}}^{T}
	\boldsymbol{\Sigma}_{A_{-j}}^{-1}
	\boldsymbol{\sigma}_{j,A_{-j}}
	\right)^{-1}.
	\]
	Moreover, the off-diagonal part of the \(j\)-th column of \(\boldsymbol{B}_A\) satisfies
	\[
	\boldsymbol{b}^{A}_{A_{-j},j}
	=
	-
	b^{A}_{jj}
	\boldsymbol{\Sigma}_{A_{-j}}^{-1}
	\boldsymbol{\sigma}_{j,A_{-j}}.
	\]
\end{lemma}

\begin{proof}
See Appendix.
\end{proof}

We now derive the alternative representation of the MRD residual statistics.

\begin{theorem}\label{THM_MRD_PRECISION_REP}
	Let \(\mathcal{A}_{t}=\{1,\ldots,m\}\setminus\{i_1,\ldots,i_{t-1}\}\) be the active set at
	stage \(t\), and let
	\[
	\boldsymbol{B}_{\mathcal{A}_{t}}
	=
	\boldsymbol{\Sigma}_{\mathcal{A}_{t}}^{-1}
	=
	\big(\big(b^{\mathcal{A}_{t}}_{kl}\big)\big)_{k,l\in \mathcal{A}_{t}}.
	\]
	Then, for each \(j\in \mathcal{A}_{t}\), the MRD residual statistic satisfies
	\[
	U^{(i_1,\ldots,i_{t-1})}_{tj}(\boldsymbol{x})
	=
	\frac{
		\displaystyle
		\sum_{k\in \mathcal{A}_{t}}
		b^{\mathcal{A}_{t}}_{kj}x_k
	}{
		\sqrt{b^{\mathcal{A}_{t}}_{jj}}
	}.
	\]
	Consequently,
	\[
	\left|
	U^{(i_1,\ldots,i_{t-1})}_{tj}(\boldsymbol{x})
	\right|
	=
	\frac{
		\left|
		\displaystyle
		\sum_{k\in \mathcal{A}_{t}}
		b^{\mathcal{A}_{t}}_{kj}x_k
		\right|
	}{
		\sqrt{b^{\mathcal{A}_{t}}_{jj}}
	}.
	\]
\end{theorem}

Theorem~\ref{THM_MRD_PRECISION_REP} provides an alternative representation of the MRD residual statistics in terms of the active precision matrix. The result shows that the residual statistic associated with a given active coordinate is determined entirely by the corresponding column of the active precision matrix, its diagonal entry, and the active observation vector. Consequently, all active residual statistics at a given stage may be computed simultaneously from a single inverse of the active covariance matrix. This observation forms the basis of the computational simplification developed below.

\begin{corollary}\label{COR_ONE_INVERSE_PER_STAGE}
	At each stage \(t\), all active MRD residual statistics
	\[
	\left\{
	U^{(i_1,\ldots,i_{t-1})}_{tj}(\boldsymbol{x}): j\in \mathcal{A}_{t}
	\right\}
	\]
	can be computed from the single active precision matrix
	\[
	\boldsymbol{B}_{\mathcal{A}_{t}}=\boldsymbol{\Sigma}_{\mathcal{A}_{t}}^{-1}.
	\]
\end{corollary}

\begin{proof}
	By Theorem~\ref{THM_MRD_PRECISION_REP}, for each active coordinate \(j\in \mathcal{A}_{t}\),
	the statistic \(U^{(i_1,\ldots,i_{t-1})}_{tj}(\boldsymbol{x})\) is obtained from
	the \(j\)-th column of \(\boldsymbol{B}_{\mathcal{A}_{t}}\) and the diagonal element \(b^{\mathcal{A}_{t}}_{jj}\).
	Thus, once \(\boldsymbol{B}_{\mathcal{A}_{t}}\) has been computed, all active residual statistics at stage
	\(t\) are obtained directly from its columns, without computing any additional
	leave-one-out inverses.
\end{proof}

The preceding result leads immediately to the advertised computational
simplification.

\begin{corollary}\label{COR_COMPLEXITY_REDUCTION}
	In the worst possible scenario, implementation of the MRD procedure through the representation
	in Theorem~\ref{THM_MRD_PRECISION_REP} requires at most \(m\) matrix inversions,
	one at each stage of the procedure. In contrast, direct implementation from the
	original leave-one-out definition requires
	\[
	m+(m-1)+\cdots+1
	=
	\frac{m(m+1)}{2}
	\]
	matrix inversions in the worst case.
\end{corollary}

\begin{proof}
	At stage \(t\), the direct leave-one-out implementation computes a separate
	inverse for each active coordinate \(j\in \mathcal{A}_{t}\), namely the inverse of
	\(\boldsymbol{\Sigma}_{\mathcal{A}_{t}\setminus\{j\}}\). Since \(|\mathcal{A}_{t}|=m-t+1\), this requires
	\(m-t+1\) matrix inversions at stage \(t\). If the procedure continues until the
	last stage, the total number of inversions is therefore
	\[
	\sum_{t=1}^{m}(m-t+1)
	=
	m+(m-1)+\cdots+1
	=
	\frac{m(m+1)}{2}.
	\]
	
	On the other hand, by Corollary~\ref{COR_ONE_INVERSE_PER_STAGE}, all active MRD
	statistics at stage \(t\) can be computed from the single inverse
	\(\boldsymbol{\Sigma}_{\mathcal{A}_{t}}^{-1}\). Hence at most one matrix inversion is needed per stage.
	Since the procedure has at most \(m\) stages, the total number of matrix
	inversions is at most \(m\).
\end{proof}

This representation is the one used in our numerical implementation. At stage \(t\), after computing the active precision matrix
\(\boldsymbol{B}_{A_t}=\boldsymbol{\Sigma}_{A_t}^{-1}\), let \(\boldsymbol{b}_j^{A_t}\) denote its
\(j\)-th column vector of \(\boldsymbol{B}_{A_t}\) and \(b_{jj}^{A_t}\) its \(j\)-th diagonal
element. Then the active MRD statistics are obtained as
\[
\left|
U^{(i_1,\ldots,i_{t-1})}_{tj}(\boldsymbol{x})
\right|
=
\frac{
	\left|
	\big(\boldsymbol{b}^{\mathcal{A}_{t}}_{j}\big)^T\boldsymbol{x}_{\mathcal{A}_{t}}
	\right|
}{
	\sqrt{b^{\mathcal{A}_{t}}_{jj}}
},
\qquad j\in \mathcal{A}_{t}.
\]
The hypothesis corresponding to the largest of these quantities is then compared
with the stagewise critical value. This formulation avoids repeated leave-one-out
matrix inversions and is therefore particularly useful when the number of
hypotheses is moderately large or large.

Beyond its computational advantages, the representation in Theorem~\ref{THM_MRD_PRECISION_REP} reveals an interesting connection between residual-based multiple testing and precision-matrix geometry. The contribution of each active coordinate is determined by the corresponding column of the active precision matrix, suggesting that conditional dependence relationships play a direct role in the propagation of information through the sequential testing process. In particular, Theorem 1 shows that the MRD statistic associated
with an active hypothesis can be interpreted as a normalized projection of the active data vector onto a direction determined by the corresponding column of the active precision matrix.

The computational benefits can be even more substantial for particular covariance structures. For example, under equicorrelation dependence, the active covariance matrix retains its equicorrelation form throughout the sequential elimination process. Consequently, the corresponding active precision matrix admits an explicit closed-form representation at every stage, allowing all residual statistics to be computed without performing any numerical matrix inversion. Thus, in certain highly structured dependence settings, the precision-matrix formulation yields not only a reduction in computational complexity but an entirely inversion-free implementation of the MRD procedure.

\section{Simulations}\label{SECTION_SIMULATIONS}
The simulation study serves three complementary objectives. First, we investigate the operating characteristics of the proposed GBS-calibrated MRD procedure under a broad collection of covariance dependence structures. Second, we compare its performance with several widely used multiple testing procedures in terms of normalized misclassification risk, false discovery rate, false non-discovery rate, power, and average number of rejections. Third, we examine the extent to which covariance-adaptive residualization, when combined with stagewise GBS calibration, can facilitate accurate signal recovery under dependence. Particular emphasis is placed on understanding how the geometry of the underlying covariance structure influences the propagation of information through the sequential testing process and the resulting classification performance.

\subsection{Performance Measures}

For a given multiple testing procedure, let \(R\) denote the total number
of rejections, \(V\) the number of false rejections, \(S\) the number of
true rejections, and \(T\) the number of false non-rejections. Let \(m_1\)
denote the number of true alternatives. We evaluate competing procedures
using the following performance measures:
\[
\mathrm{FDR}
=
E\left(\frac{V}{R\vee 1}\right),
\]
\[
\mathrm{Power}
=
E\left(\frac{S}{m_1}\right),
\]
\[
\mathrm{FNR}
=
E\left(\frac{T}{(m-R)\vee 1}\right),
\]
and the normalized misclassification risk
\[
\mathrm{NMR}
=
\frac{E(V+T)}{m}.
\]
Here \(V+T\) is the total number of classification errors. Thus, the normalized misclassification risk is directly related to the total loss function discussed in Section~\ref{SECTION_MRD_GBS_DEFINITION} and may be interpreted as the expected proportion of incorrectly classified hypotheses. Since the proposed methodology is motivated by multiple decision theory and admissibility considerations, normalized misclassification risk serves as the primary performance criterion throughout the simulation study. The remaining measures provide additional insight into the trade-off between false discoveries, missed discoveries, and signal recovery.

\subsection{Data Generation Mechanism}

To assess the finite-sample performance of the competing procedures under a variety of dependence structures, we consider the Gaussian multiple testing model
\[
\boldsymbol{X}\sim N_{m}(\boldsymbol{\theta},\boldsymbol{\Sigma}),
\]
where  $\boldsymbol{\theta}=(\theta_1,\ldots,\theta_m)^T$ denotes the unknown mean vector and $\boldsymbol{\Sigma}$ is a known positive definite covariance matrix.

% The objective is to simultaneously test
%\[
%H_i:\theta_i=0
%\qquad \text{versus} \qquad
%K_i:\theta_i\neq 0,
%\quad i=1,\ldots,m.
%\]
%
% We take
%\(m=100\) and consider sparsity levels
%\[
%p\in\{0.01,0.03,0.05,0.075,0.10,0.15,0.20,0.25,0.30,0.35\}.
%\]

For each covariance structure, we generated data from the Gaussian multiple testing model described in Section~2. Throughout the simulation study, we consider dimensions $m=100$ and $m=200$ and sparsity levels
\[
p\in\{0.01,0.03,0.05,0.075,0.10,0.15,0.20,0.25,0.30,0.35\},
\]
and a nominal significance level \(\alpha=0.10\) for all competing procedures.

For each sparsity level $p$, signal locations are generated independently via
Bernoulli indicators $\nu_i \sim \mathrm{Bernoulli}(p)$, $i=1,\ldots,m$.
Whenever $\nu_i = 1$, the corresponding mean is assigned the value $+\mu$
or $-\mu$ with equal probability, while $\nu_i = 0$ implies $\theta_i = 0$.
Throughout the simulation study, we set
\[
\mu=\sqrt{2\log m}.
\]

This choice corresponds to the classical universal threshold
\(\sqrt{2\log m}\) introduced by \citet{DJ1994},
which plays a fundamental role in sparse high-dimensional
estimation, signal recovery, and multiple testing. Thresholds of this order play a central role in sparse high-dimensional inference and multiple testing, where they often characterize the boundary between detectable and undetectable signals; see, for example, \citet{ABDJ2006}, \citet{JS2004}, \citet{BCFG2011}, and \citet{HJ2010}. Consequently, the resulting simulation setting is neither trivially easy nor excessively difficult, thereby providing a meaningful benchmark for assessing the ability of competing procedures to recover sparse signals under dependence.

For each sparsity level and covariance structure, the performance measures are estimated using \(G=3000\)
Monte Carlo replications.

%For each value of $p$, the expected number of non-null signals is $mp$. For each coordinate \(i=1,\ldots,m\), the signal indicator is generated
%independently as
%\[
%\nu_i\sim \mathrm{Bernoulli}(p).
%\]
%Conditional on \(\nu_i=1\), the corresponding mean is set equal to
%\[
%\theta_i=\pm \mu,
%\qquad
%\mu=\sqrt{2\log m},
%\]
%with the sign chosen independently with equal probability. If
%\(\nu_i=0\), then \(\theta_i=0\). For each sparsity level and covariance
%structure, the performance measures are estimated using \(3000\)
%Monte Carlo replications.

To investigate the effect of problem dimension on the performance of the competing
procedures, simulations were conducted for both $m=100$ and $m=200$. The results
for $m=200$ are presented in the main text, while the corresponding results for
$m=100$ are reported in Appendix C. The larger dimension more clearly reveals
the finite-sample manifestation of the asymptotic behavior suggested by several
of the competing procedures. Consequently, our discussion focuses primarily on
the $m=200$ experiments.

%To investigate the effect of problem dimension on the performance of the competing procedures, simulations were conducted for both $m=100$ and $m=200$. The results for $m=200$ are presented in the main text, while the corresponding results for $m=100$ are reported in Appendix~C. Since the larger dimension more clearly reveals the empirical large-sample operating characteristics of the competing procedures, our discussion focuses primarily on the $m=200$ experiments.

\subsection{Dependence Structures}

To assess the robustness of the proposed methodology under a broad range
of dependence settings, we consider the following covariance models.

\paragraph{Equicorrelation Model.}
The covariance matrix is given by
\[
\Sigma_{ij}
=
\rho +(1-\rho)I(i=j) \ \textrm{ where } \rho =0.7.
\]
%where \(\rho\in\{0.1,0.3,0.5,0.7\}\).

\paragraph{Toeplitz Model.}
The covariance matrix is specified by
\[
\Sigma_{ij}
=
\rho^{|i-j|},  \ \textrm{ where } \rho =0.9.
\]
%where \(\rho\in\{0.3,0.5,0.7\}\).

\paragraph{Factor Model.}
Observations are generated according to
\[
X_i
=
\theta_i+\lambda_iF+\varepsilon_i,
\qquad
i=1,\ldots,m,
\]
where \(F\sim N(0,1)\) is a common factor,
\(\varepsilon_i\stackrel{\mathrm{ind}}{\sim}N(0,1)\),
and \(\lambda_i\) are independently generated from a
\(\mathrm{Uniform}(0.5,1)\) distribution.

\paragraph{Fractional Gaussian Noise Model.}
The covariance matrix is generated from a fractional Gaussian noise
process with Hurst parameter
\[
H = 0.9.
%\in\{0.6,0.7,0.8,0.9\}.
\]
This model exhibits long-range dependence whose strength increases with
\(H\).

\paragraph{Heterogeneous Block Model.}
The coordinates are partitioned into blocks of varying sizes. Within
each block, observations follow an equicorrelated covariance structure,
while observations belonging to different blocks are independent. The
within-block correlations vary across blocks.

\paragraph{Sparse Precision Model.}
The precision matrix
\[
\boldsymbol{\Omega}=\boldsymbol{\Sigma}^{-1}
\]
is generated to be sparse, with each coordinate directly connected to
only a small number of neighboring coordinates. The covariance matrix
is obtained by inverting \(\boldsymbol{\Omega}\).

Collectively, these covariance structures range from highly global forms of dependence to substantially more localized conditional dependence patterns. This diversity allows us to investigate how different covariance geometries influence the performance of residual-based multiple testing procedures and, in particular, the effectiveness of covariance-adaptive information propagation.
\subsection{Competing Procedures}

The proposed GBS-calibrated MRD procedure is compared with the following
multiple testing procedures.

\begin{enumerate}
	
	\item The original MRD procedure of
	\citet{COHEN_SACK_XU_2009}.
	
	\item The Benjamini--Hochberg (BH) procedure
	\citep{BH1995}.
	
	\item Storey's adaptive BH procedure
	\citep{STO_2002,STS_2004}.
	
	\item The Gavrilov-Benjamini--Sarkar (GBS) procedure
	\citep{GBS2009}.
	
\end{enumerate}

The inclusion of the original MRD procedure allows us to isolate the effect of the proposed calibration strategy from that of covariance-adaptive residualization itself. Consequently, comparisons between MRD and MRD-GBS reveal the impact of calibration, whereas comparisons with BH, Storey, and GBS assess the overall benefits of incorporating covariance information directly into the testing statistics. Since the latter procedures are based on marginal test statistics, these comparisons provide insight into the practical advantages of residual-based multiple testing under dependence.

%The inclusion of the original MRD procedure allows us to isolate the effect of the proposed calibration strategy from that of the underlying residual-based testing mechanism. Comparisons with BH, adaptive BH, and GBS provide benchmarks against widely used marginal \(p\)-value based multiple testing procedures. This distinction is important for interpreting the simulation results. The original MRD procedure serves as the natural residual-based benchmark, allowing the effect of calibration to be separated from that of covariance-adaptive residualization. In contrast, BH, Storey-type adaptive BH, and GBS represent the principal marginal testing competitors against which the practical gains of the proposed methodology should be assessed.

%The inclusion of the original MRD procedure allows us to isolate the effect of the proposed calibration strategy from that of the underlying residual-based testing mechanism. Comparisons with BH, adaptive BH, and
%GBS provide benchmarks against widely used marginal \(p\)-value based multiple testing procedures. This distinction is important for interpreting the simulation results. The original MRD procedure is included primarily as an internal residual-based benchmark, whereas BH, Storey-type adaptive BH, and GBS represent the principal marginal testing competitors against which the practical gains of the proposed method should be assessed.

\subsection{Simulation Results}

The simulation study reveals several remarkably consistent patterns across the six covariance structures considered in this paper. Since normalized misclassification risk is directly related to the underlying multiple-decision loss function, our primary focus is on comparing the overall classification accuracy of the competing procedures. Additional insight is obtained through separate analyses of false discovery rates, false non-discovery rates, power, and average numbers of rejections.

Throughout the discussion, particular emphasis is placed on the $m=200$ experiments reported in the main text. While the corresponding results for $m=100$ exhibit broadly similar trends, the larger dimension more clearly reveals the finite-sample manifestation of the large-sample operating characteristics of the competing procedures. Although the quantitative differences become substantially larger when $m=200$, the qualitative ranking of the competing procedures remains largely unchanged between the two dimensions, providing additional evidence for the robustness of the observed patterns. The $m=100$ results are therefore deferred to Appendix~C and serve primarily as supplementary evidence supporting the conclusions drawn from the larger-dimensional setting.

Several findings emerging from the simulation study were considerably stronger than we initially anticipated. Among the most noteworthy is the observation that the strongest overall signal-recovery behavior is often achieved by the proposed GBS-calibrated MRD procedure rather than by the original MRD
procedure of \citet{COHEN_SACK_XU_2009} itself. Since both procedures employ the same covariance-adaptive residualization mechanism, this finding suggests that calibration may play a more substantial role than previously anticipated in translating residualized information into accurate signal recovery.

While the proposed GBS-calibrated MRD procedure was motivated primarily by the goal of developing a systematic calibration strategy within the admissible residual-based testing framework, the resulting procedure consistently exhibited remarkably strong classification and signal-recovery performance across a broad collection of dependence structures. In particular, the simulations reveal that covariance-adaptive
residualization and stagewise GBS calibration interact in a highly favorable manner, often producing substantially improved classification accuracy and signal recovery not only relative to the competing marginal procedures, but also relative to the original MRD procedure itself.

\subsubsection{Normalized Misclassification Rate}

Normalized misclassification rate (NMR) is the primary performance criterion considered in this paper. Recall that the NMR is defined as 

$$\textrm{NMR}=\frac{E(V+T)}{m},$$
where \(V\) and \(T\) denote the numbers of false rejections and false non-rejections, respectively. Consequently, \(\textrm{NMR}\) measures the expected proportion of incorrectly classified hypotheses among all $m$ null hypotheses under consideration and is directly related to the total loss function discussed in Section~\ref{SECTION_MRD_GBS_DEFINITION}. From a practical perspective, a reduction in normalized misclassification risk corresponds directly to a reduction in the total number of incorrectly classified hypotheses and therefore provides a
natural measure of overall decision quality.

Figure~\ref{FIG_NMR_m200} provides a graphical summary of the NMR performance of all competing procedures across the six dependence structures considered in this study.

\begin{figure}[htbp]
	\centering
	
	\begin{subfigure}[b]{0.48\textwidth}
		\centering
		\includegraphics[width=\textwidth]{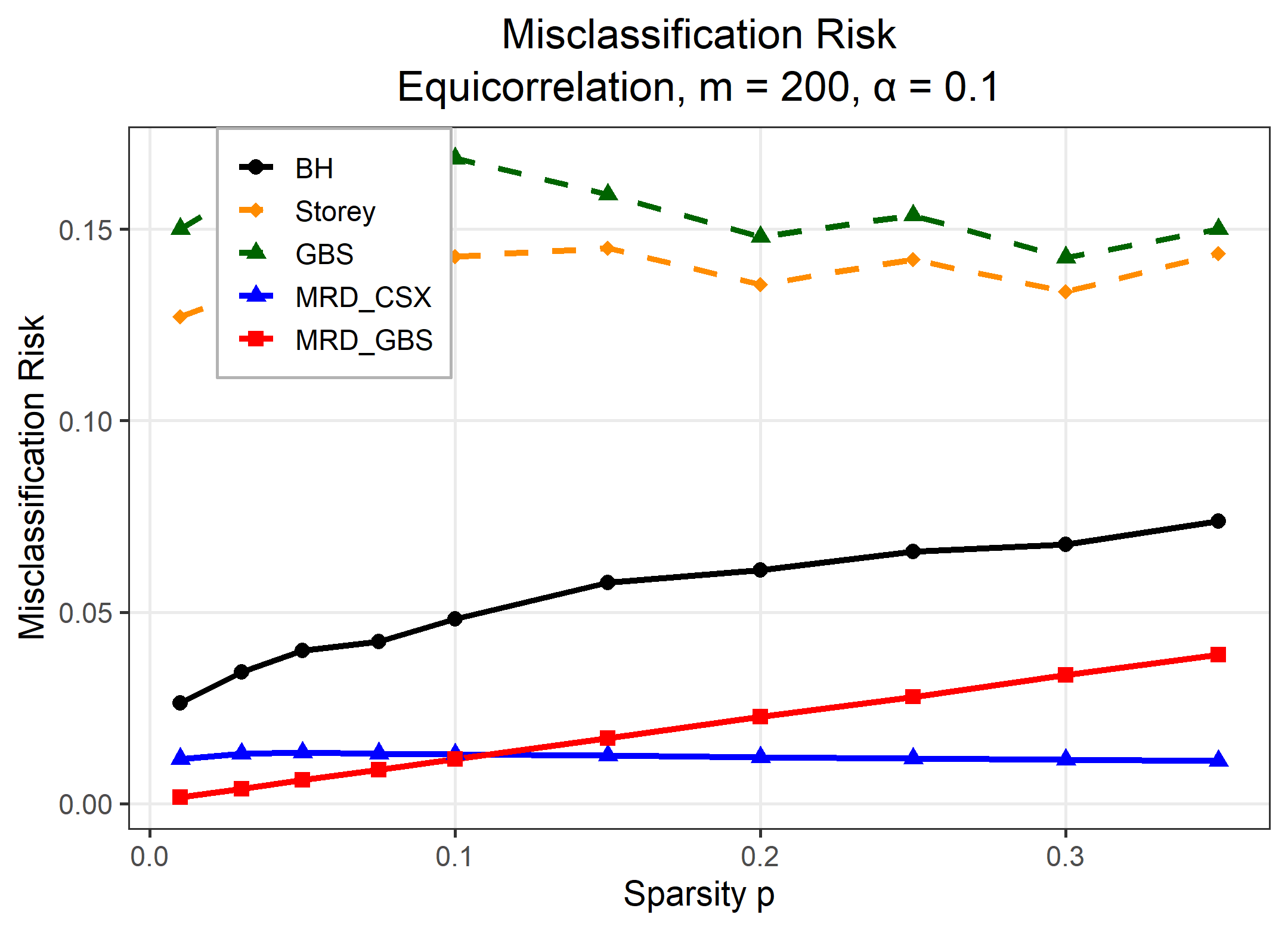}
		\caption{Equicorrelation (\(\rho=0.7\))}
	\end{subfigure}
	\hfill
	\begin{subfigure}[b]{0.48\textwidth}
		\centering
		\includegraphics[width=\textwidth]{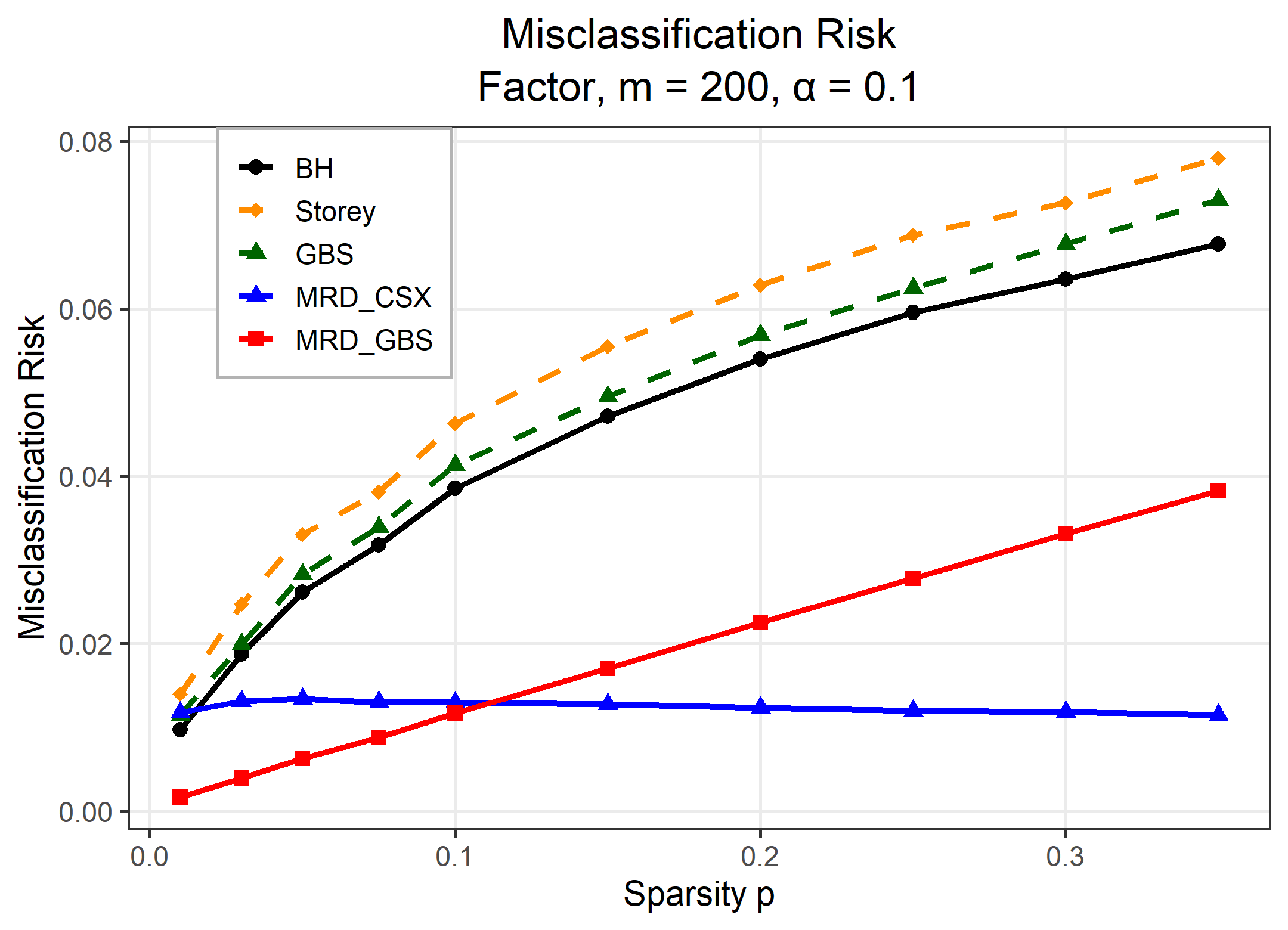}
		\caption{Factor model}
	\end{subfigure}
	
	\vspace{0.3cm}
	
	\begin{subfigure}[b]{0.48\textwidth}
		\centering
		\includegraphics[width=\textwidth]{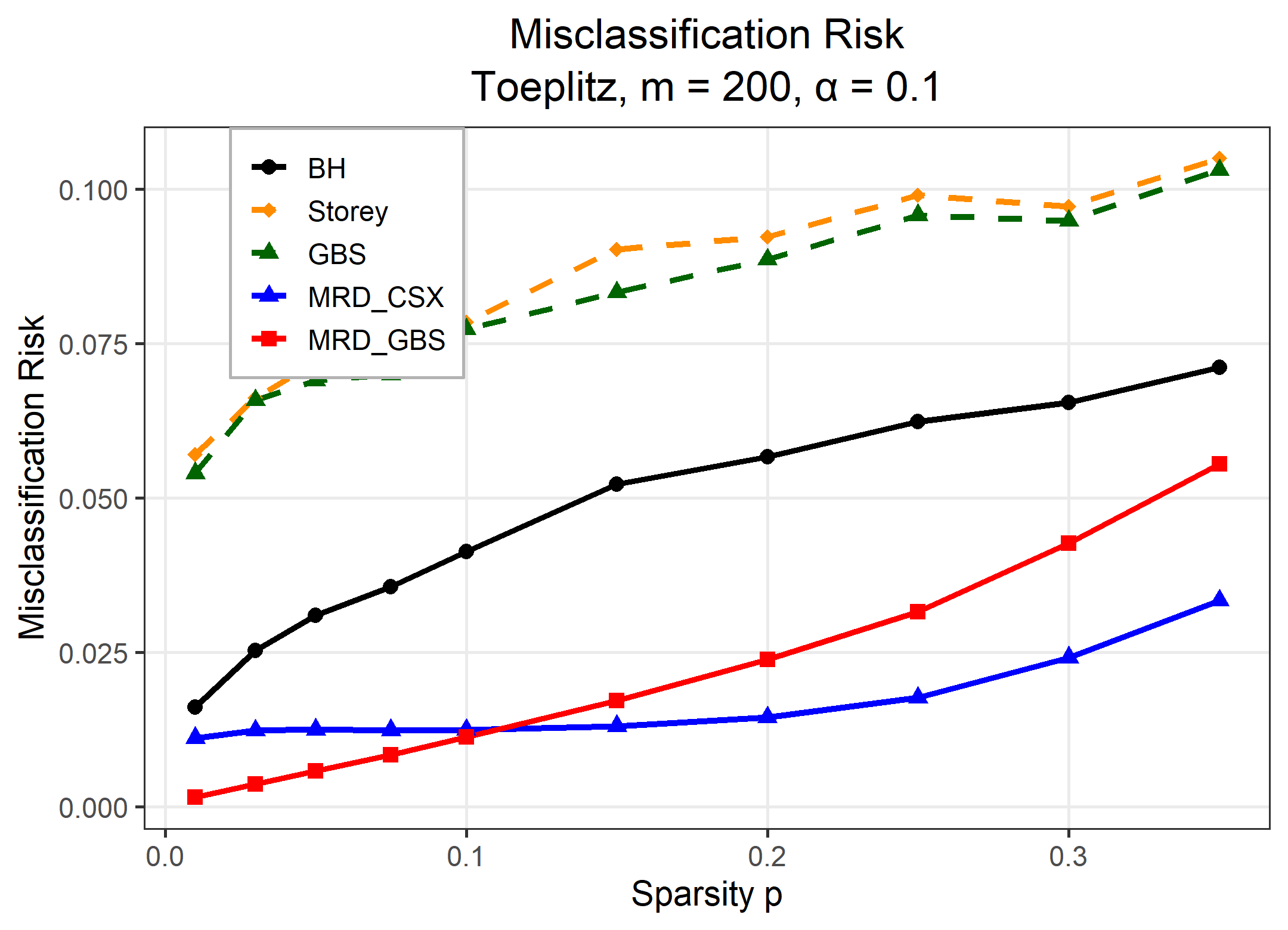}
		\caption{Toeplitz (\(\rho=0.9\))}
	\end{subfigure}
	\hfill
	\begin{subfigure}[b]{0.48\textwidth}
		\centering
		\includegraphics[width=\textwidth]{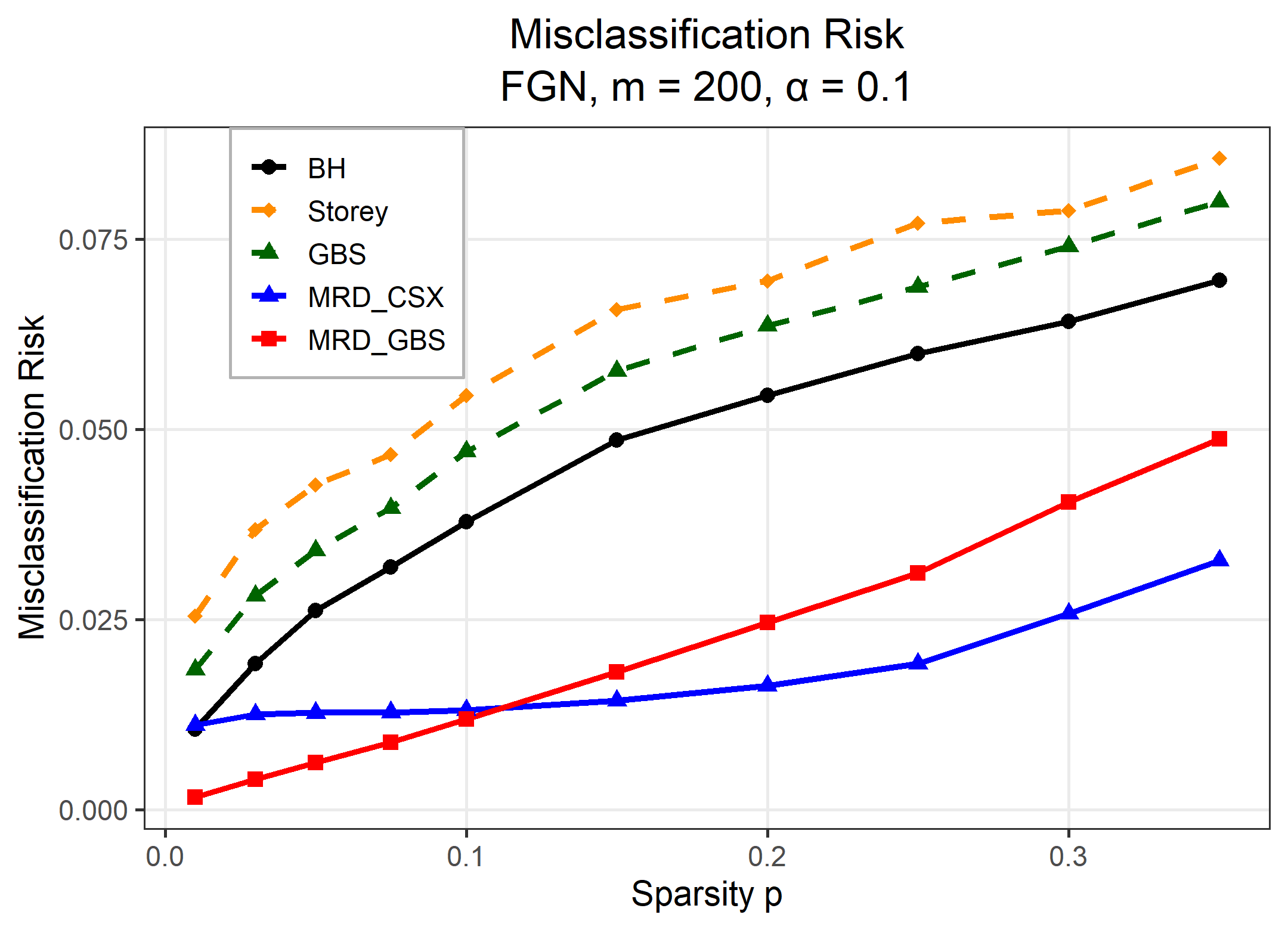}
		\caption{Fractional Gaussian Noise (\(H=0.9\))}
	\end{subfigure}
	
	\vspace{0.3cm}
	
	\begin{subfigure}[b]{0.48\textwidth}
		\centering
		\includegraphics[width=\textwidth]{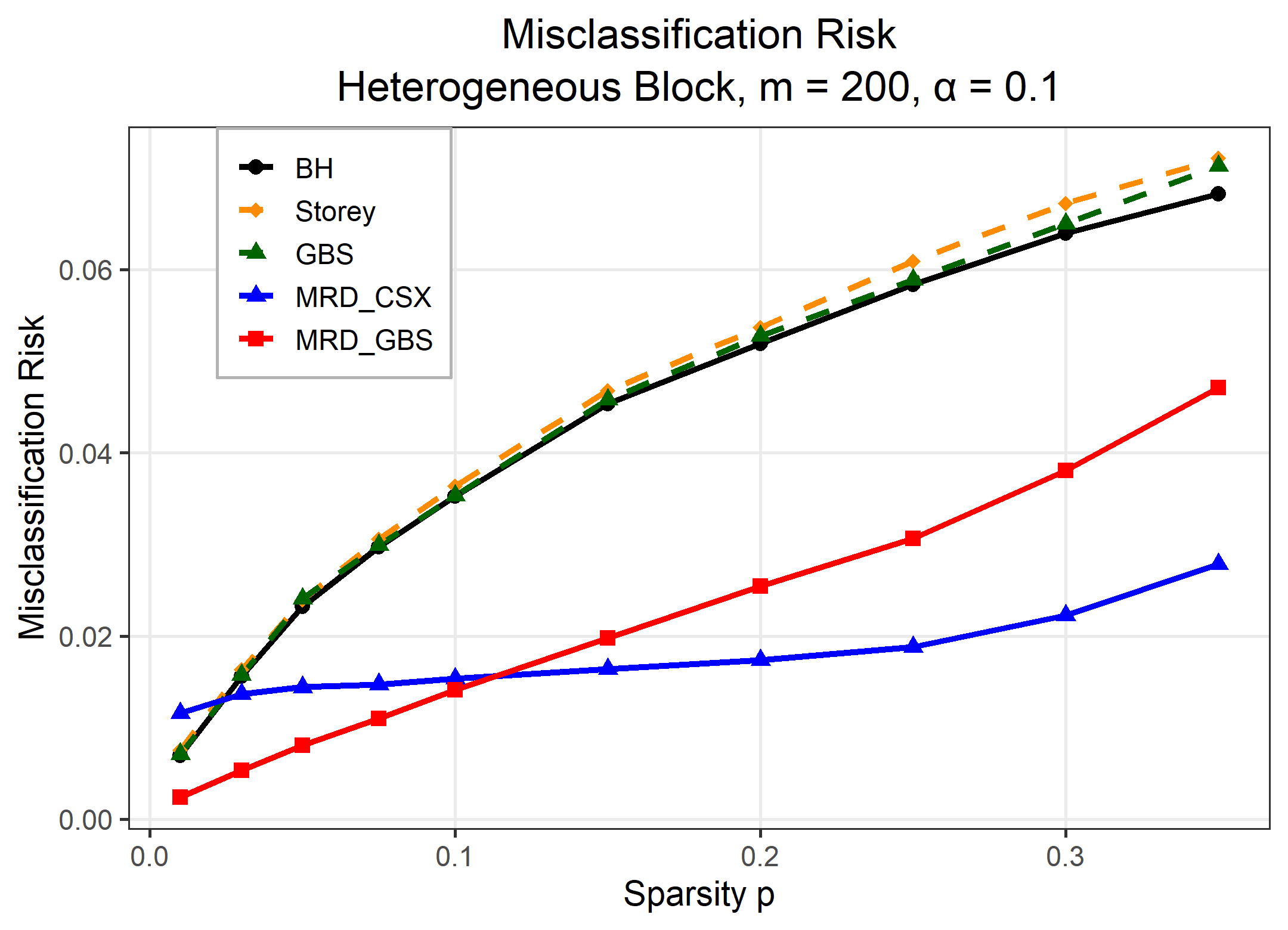}
		\caption{Heterogeneous Block}
	\end{subfigure}
	\hfill
	\begin{subfigure}[b]{0.48\textwidth}
		\centering
		\includegraphics[width=\textwidth]{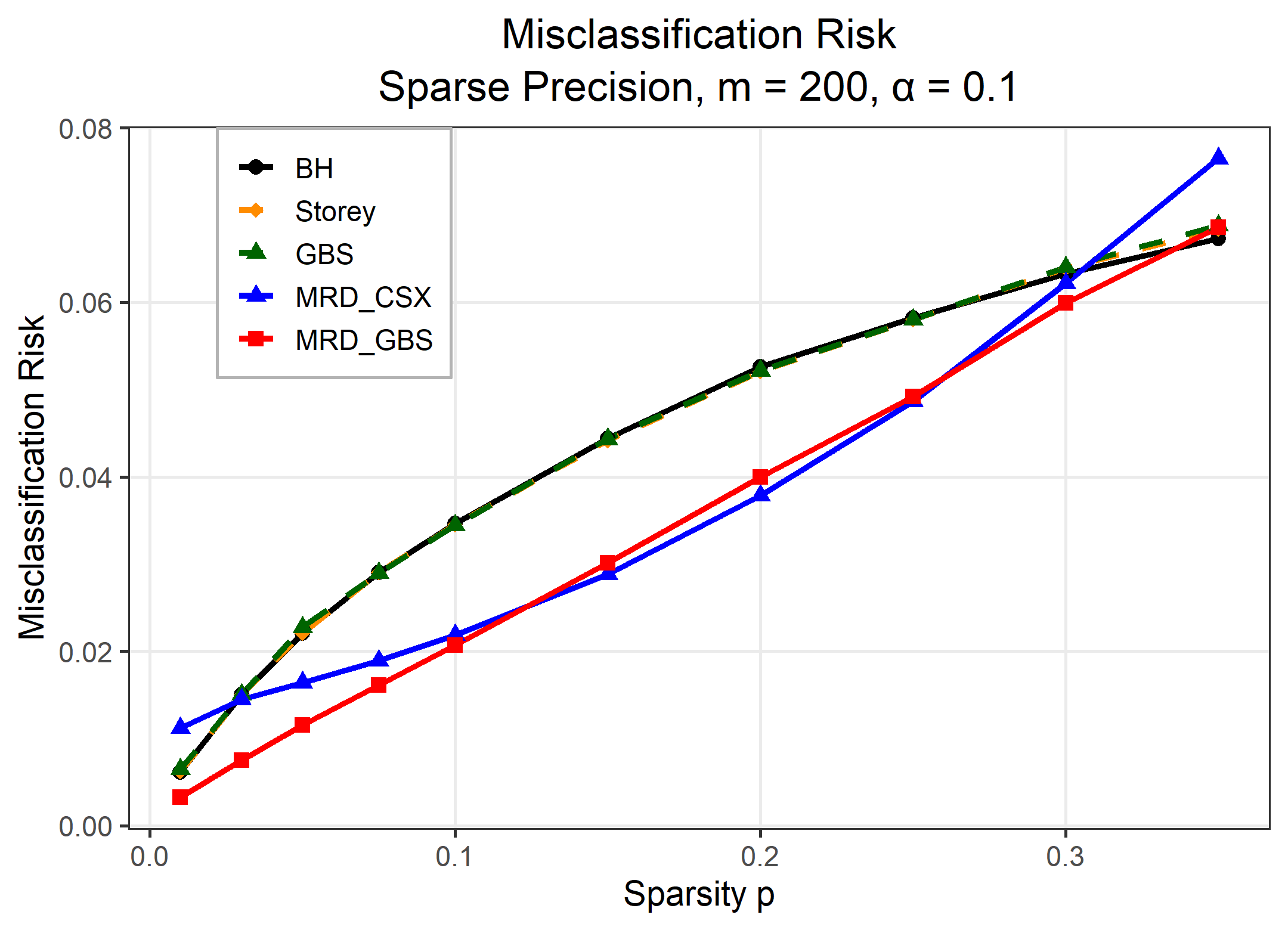}
		\caption{Sparse Precision Matrix}
	\end{subfigure}
	
	\caption{
	Normalized misclassification rates (NMR) for the competing multiple testing
	procedures under six representative dependence structures. The plots reveal
	a remarkably consistent pattern across several covariance models: the
		GBS-calibrated MRD procedure tends to achieve the smallest misclassification
		rates in sparse and moderately sparse regimes, while the original MRD
		procedure often becomes more competitive as the signal proportion increases.
		The sparse precision model exhibits noticeably different behavior,
		highlighting the potential influence of covariance geometry on classification
		performance.
	}
	\label{FIG_NMR_m200}
\end{figure}

Tables~\ref{tab:risk_equicorrelation_m200},
\ref{tab:risk_factor_m200},
\ref{tab:risk_fractional_gaussian_noise_m200},
\ref{tab:risk_toeplitz_m200},
\ref{tab:risk_heterogeneous_block_m200},
and
\ref{tab:risk_sparse_precision_m200}
summarize the normalized misclassification rates (NMR) achieved by the
competing procedures under the various dependence structures considered
in this study. 
%Recall that NMR is defined as
%\[
%\mathrm{NMR}
%=
%\frac{E(V+T)}{n},
%\]
%where \(V\) and \(T\) denote the numbers of false rejections and false
%non-rejections, respectively. Consequently, NMR measures the expected
%proportion of incorrectly classified hypotheses and is directly related
%to the total loss function discussed in Section~\ref{SECTION_MRD_GBS_DEFINITION}.

Several interesting patterns emerge from these results. First, when compared
with the marginal \(p\)-value based procedures, the GBS-calibrated MRD
procedure yields substantial reductions in normalized misclassification risk
over a broad range of sparsity levels. This comparison is particularly relevant
from a practical perspective, since BH, Storey-type adaptive BH, and the
original GBS procedure represent widely used marginal testing benchmarks.
The original MRD procedure of \citet{COHEN_SACK_XU_2009}
serves as an important residual-based benchmark, but its stagewise critical
constants are model-dependent and were selected through numerical
experimentation. Thus, the main empirical message is not merely that the
proposed calibration approximates the behavior of the original MRD procedure,
but that it provides a systematic covariance-aware alternative that can
substantially improve classification accuracy relative to standard marginal
testing methods.

A second and particularly noteworthy phenomenon concerns the relative
performance of the two MRD procedures. Across the extreme-to-moderately
sparse regimes, typically for signal proportions up to approximately
\(p=0.15\)--\(0.20\), the GBS-calibrated MRD procedure frequently
achieves the smallest normalized misclassification rates among all
competing methods. This behavior is remarkably consistent across
several dependence structures, including equicorrelated, factor,
Toeplitz, fractional Gaussian noise, and heterogeneous block covariance
models. These findings suggest that the GBS calibration successfully
balances false discoveries and missed discoveries in the sparse
settings for which large-scale multiple testing procedures are most
commonly intended.

Another noteworthy aspect of the results is that the superiority of the
GBS-calibrated MRD procedure is not confined to a narrow range of sparsity
levels. Across several dependence structures, the reductions in normalized
misclassification risk remain substantial over a broad portion of the sparse
and moderately sparse regimes, suggesting that the observed gains are both
stable and practically meaningful. The consistency of these improvements across several qualitatively distinct
covariance structures suggests that the observed gains are not tied to a
particular dependence model but instead reflect a more general advantage
of combining covariance-adaptive residualization with stagewise GBS
calibration.

%%%% Move to Section~\ref{SECTION_SIMULATIONS}.5.2

%An important feature emerging from the simulation study is the effect of increasing the problem dimension from $m=100$ to $m=200$. Although the qualitative conclusions remain largely unchanged across the two settings, the advantages of the proposed GBS-calibrated MRD procedure become considerably more pronounced in the larger-dimensional experiments. In particular, the reductions in normalized misclassification risk relative to the marginal procedures become more substantial under several dependence structures, while the superiority of the proposed method throughout sparse and moderately sparse regimes becomes increasingly stable. Thus, the $m=200$ experiments provide considerably stronger empirical evidence that the benefits of covariance-adaptive residualization and stagewise GBS calibration become more visible as the dimension of the testing problem increases.
%
%
%The larger-dimensional experiments reveal a phenomenon that is only weakly visible when $m=100$ but becomes strikingly apparent when $m=200$. Across several dependence
%structures, the proposed procedure simultaneously maintains false
%discovery rates close to the nominal level, extremely small false
%non-discovery rates, powers approaching one, and average numbers of
%rejections remarkably close to the expected number of true signals.
%Taken together, these findings provide substantially stronger empirical
%evidence of near-support-recovery behavior in the larger-dimensional
%setting.

%%%% Move to Section~\ref{SECTION_SIMULATIONS}.5.2

As the signal proportion increases, however, the advantage of the
GBS-calibrated procedure gradually diminishes and may eventually
reverse. In denser regimes, the original MRD procedure often becomes
more competitive and frequently attains smaller normalized
misclassification rates than its calibrated counterpart. This behavior
appears to reflect a different balance between type I and type II
errors. The GBS calibration substantially reduces false discoveries,
which is particularly beneficial in sparse settings. When the signal
proportion becomes larger, however, the more aggressive rejection
behavior of the original MRD procedure can lead to fewer missed
discoveries, thereby improving overall classification performance.

The magnitude of the observed reductions in normalized misclassification risk is sufficiently large that it cannot be explained solely through modest improvements in either type I or type II error rates. Rather, the results suggest that the proposed procedure is simultaneously reducing both types of classification errors over substantial portions of the sparsity range. This observation motivates the more detailed investigation of false discovery rates, false non-discovery rates, power, and average numbers of rejections presented in the next subsection.

An additional empirical pattern emerges across several dependence structures. The proposed MRD-GBS procedure often attains the smallest normalized misclassification rates in sparse regimes, whereas the original MRD-CSX procedure frequently becomes competitive and may eventually overtake as the signal proportion increases. This crossover behavior appears consistently across the factor, fractional Gaussian noise, Toeplitz, heterogeneous block, and sparse precision models, suggesting that the relative advantages of the two calibration schemes depend on the underlying sparsity level. In particular, the proposed calibration appears especially effective in highly sparse settings, while the original calibration may become increasingly competitive as the signal configuration becomes denser.

Taken together, the NMR results demonstrate that the proposed GBS-calibrated MRD procedure frequently achieves the smallest classification risk throughout sparse and moderately sparse regimes under a wide variety of dependence structures. The improvements are particularly substantial relative to the marginal procedures BH, adaptive BH, and GBS, highlighting the benefits of incorporating covariance information directly into the testing statistics. Although the advantage becomes less pronounced in denser regimes and under sparse precision dependence, the overall evidence indicates that covariance-adaptive residualization combined with stagewise GBS calibration can substantially improve classification accuracy in dependent multiple testing problems. To better understand the mechanisms underlying these risk reductions, we now examine the corresponding FDR, FNR, power, and average rejection characteristics.

\subsubsection{Signal Recovery Performance}

The normalized misclassification risk results reported in the previous subsection
suggest that the proposed GBS-calibrated MRD procedure achieves a particularly
effective balance between false discoveries and missed discoveries in sparse and
moderately sparse regimes. To understand the mechanism underlying these risk
reductions more directly, we now examine the corresponding signal-recovery
characteristics through the empirical false discovery rates (FDR), false
non-discovery rates (FNR), powers, and average numbers of rejections (ANR)
reported in Figures~\ref{FIG_FDR_m200}--\ref{FIG_ANR_m200} and Tables~\ref{tab:fdr_equicorrelation_m200}--\ref{tab:rej_sparse_precision_m200}.

Taken together, these quantities provide considerably more information than any
single performance measure alone. In particular, simultaneous occurrence of
false discovery rates near the nominal level, extremely small false
non-discovery rates, powers approaching one, and average numbers of rejections
close to the expected number of true signals provides strong empirical evidence
of near-support-recovery behavior. Remarkably, this phenomenon emerges under
several of the dependence structures considered in this study and is most
consistently exhibited by the proposed GBS-calibrated MRD procedure.

A particularly striking feature of the larger-dimensional experiments
($m=200$) is the strength of the observed signal-recovery behavior.
Under the equicorrelation, factor, Toeplitz, fractional Gaussian noise,
and heterogeneous block covariance models, the proposed procedure
frequently maintains false discovery rates close to the nominal level
while simultaneously producing extremely small false non-discovery rates.
These low FNR values translate directly into powers that are often close
to one over substantial portions of the sparsity range. Moreover, the
corresponding average numbers of rejections frequently remain remarkably
close to the expected number of true signals. Viewed collectively, these
findings suggest that the procedure is often able to recover the
underlying support of the signal vector with surprisingly high accuracy.

Perhaps the most noteworthy aspect of these findings is that the strongest
signal-recovery behavior is observed under the proposed GBS calibration
rather than under the original MRD procedure itself. Since both procedures
employ the same covariance-adaptive residualization mechanism, the observed
differences cannot be attributed solely to residualization. Instead, the
results suggest that stagewise calibration plays a crucial role in
determining how effectively information accumulated through sequential
residualization is translated into accurate recovery of the underlying
signal set. Viewed differently, residualization determines how information is constructed and propagated through the active set, whereas calibration determines how that accumulated evidence is converted into testing decisions.

\begin{figure}[p]
	\centering
	
	\begin{tabular}{cc}
		
		\includegraphics[width=0.47\textwidth]
		{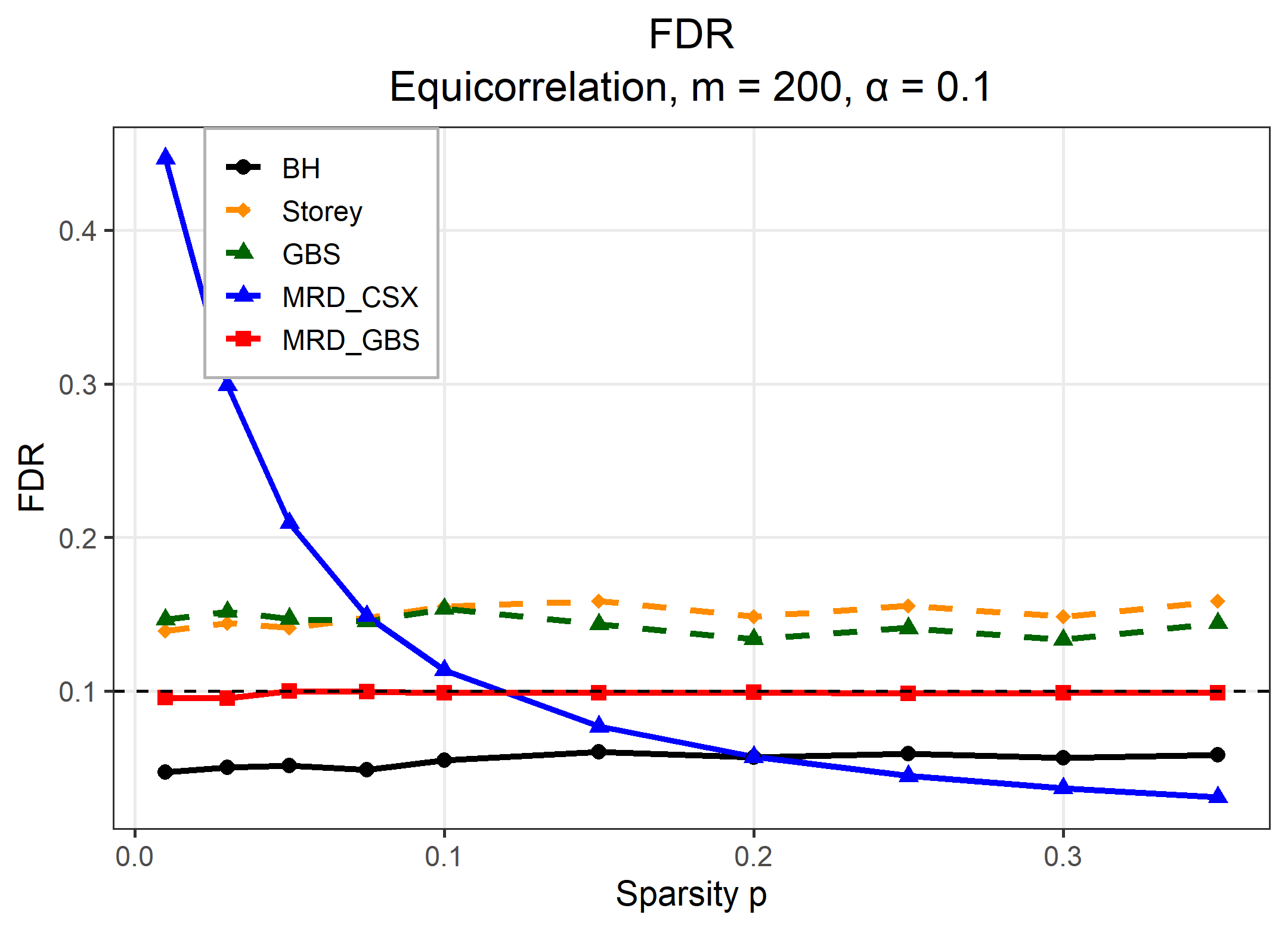}
		&
		\includegraphics[width=0.47\textwidth]
		{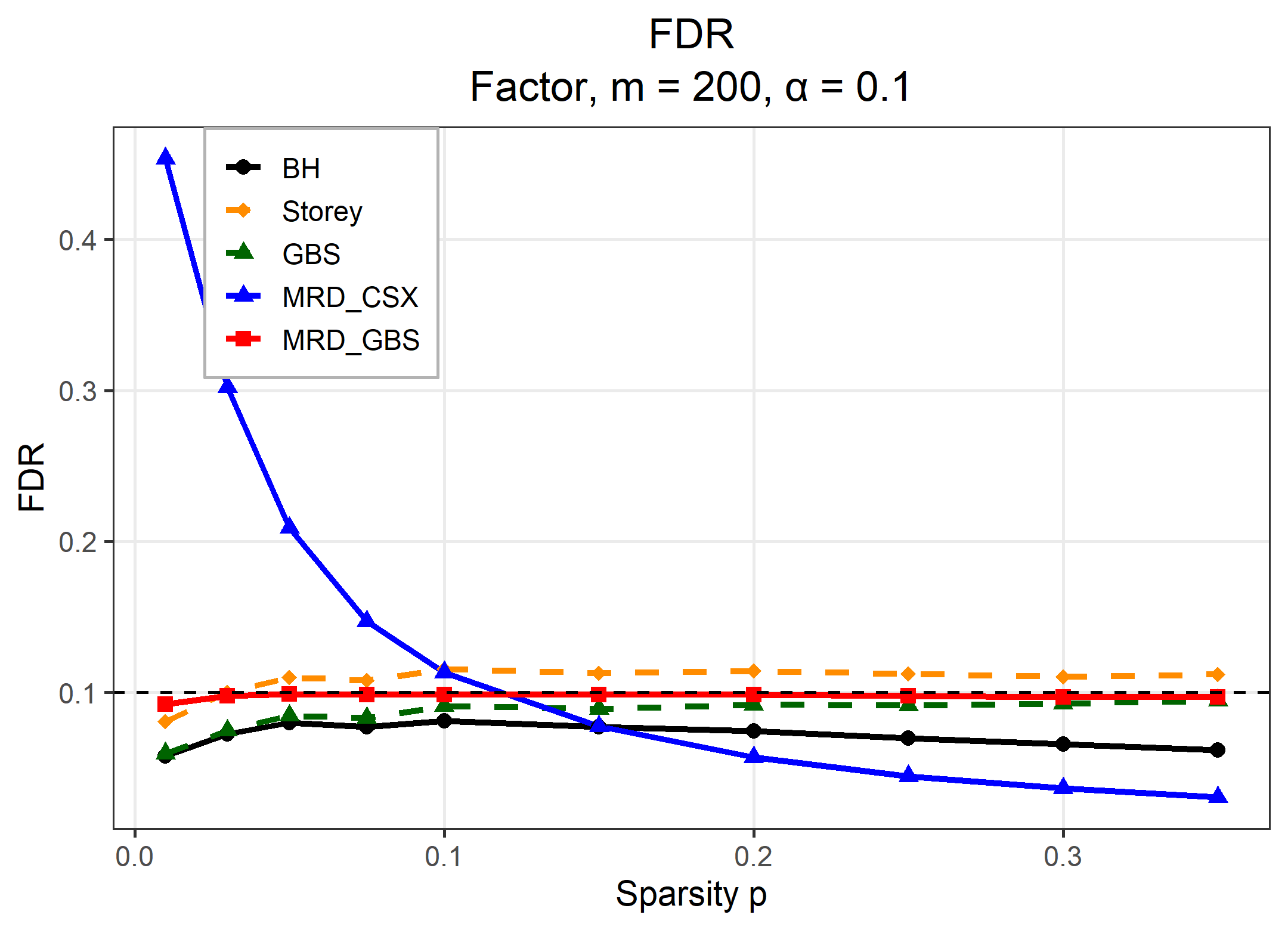}
		\\
		
		(a) Equicorrelation ($\rho=0.7$)
		&
		(b) Factor Model
		\\[0.3cm]
		
		\includegraphics[width=0.47\textwidth]
		{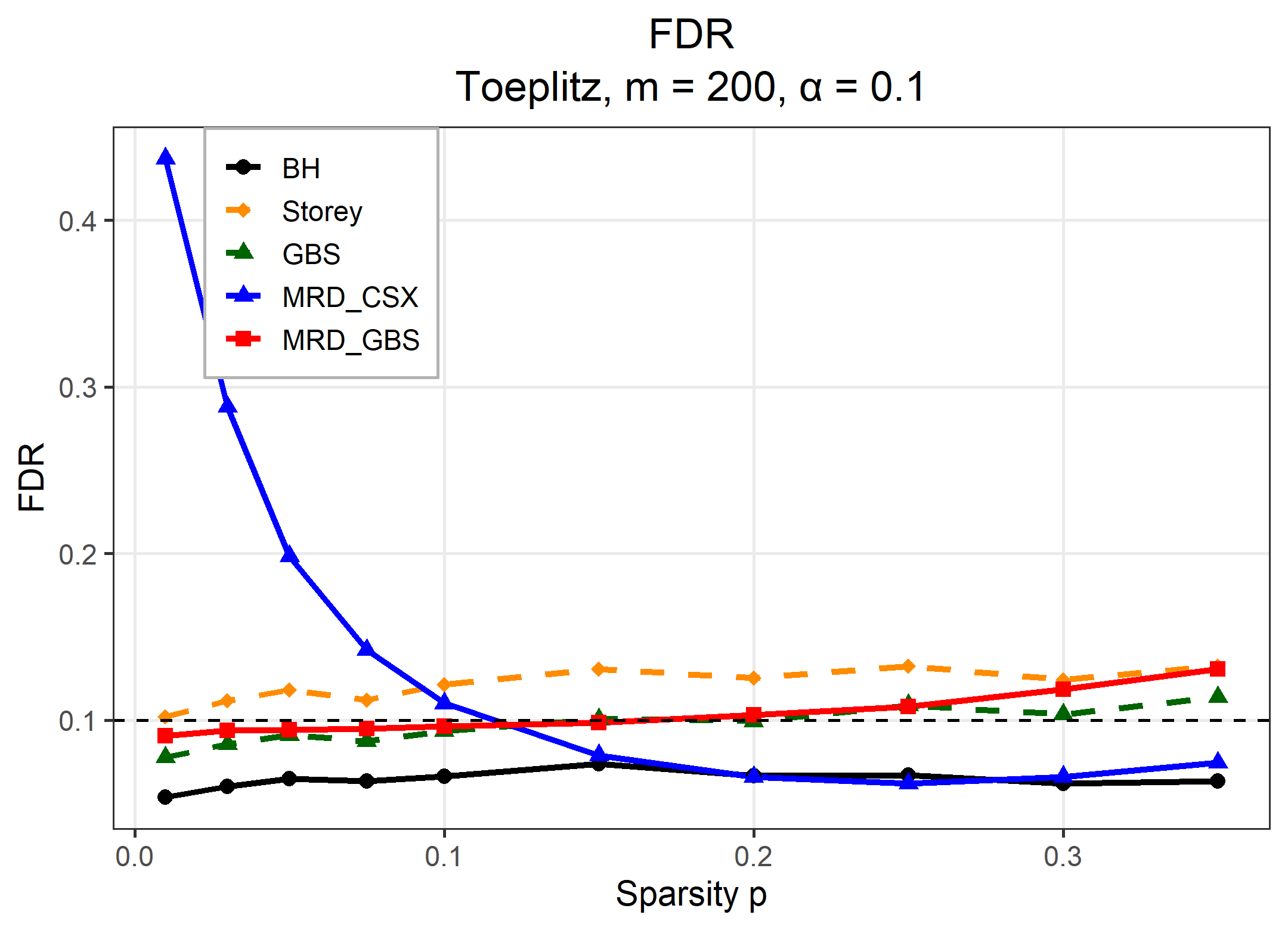}
		&
		\includegraphics[width=0.47\textwidth]
		{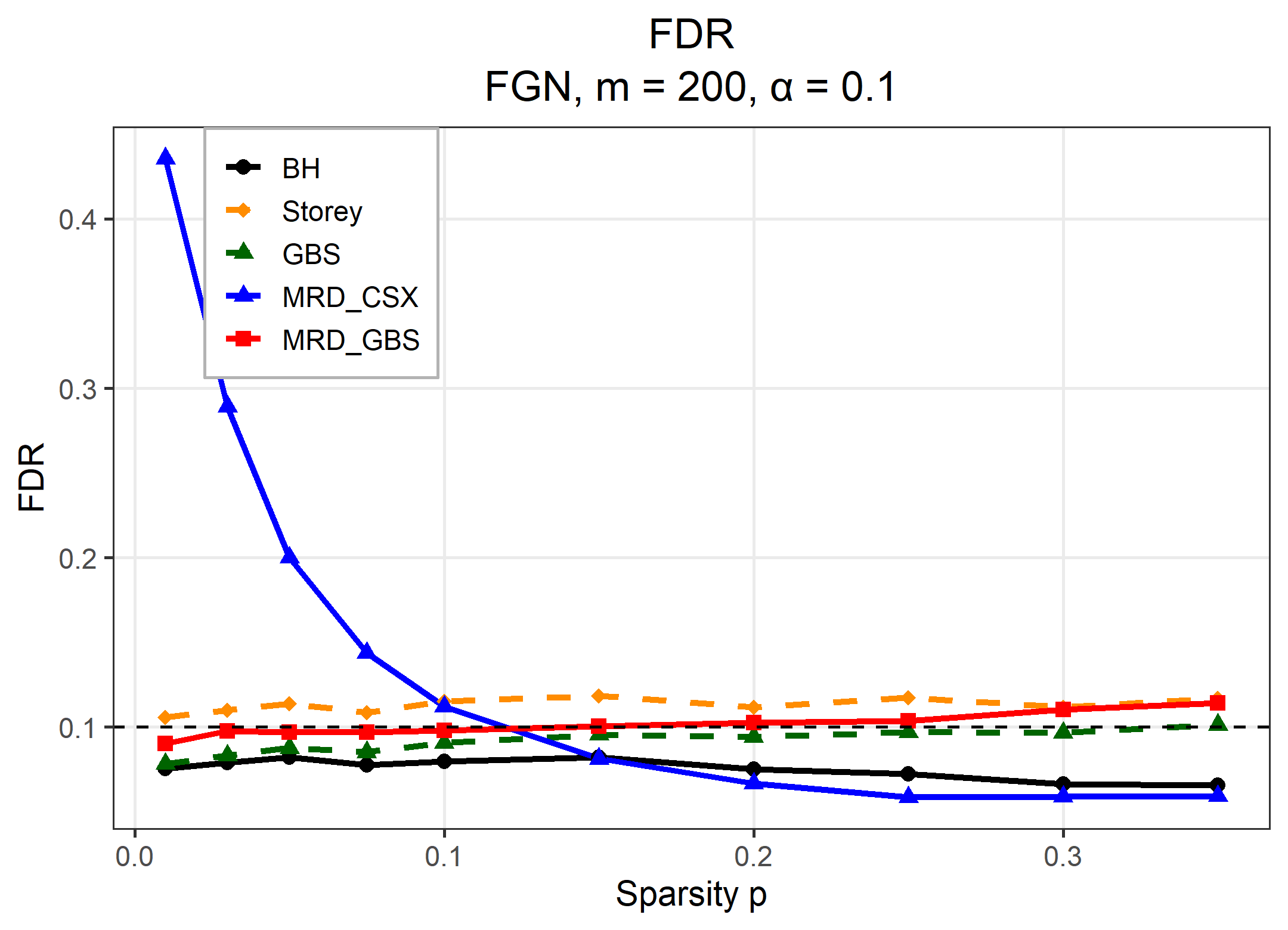}
		\\
		
		(c) Toeplitz ($\rho=0.9$)
		&
		(d) Fractional Gaussian Noise ($H=0.9$)
		\\[0.3cm]
		
		\includegraphics[width=0.47\textwidth]
		{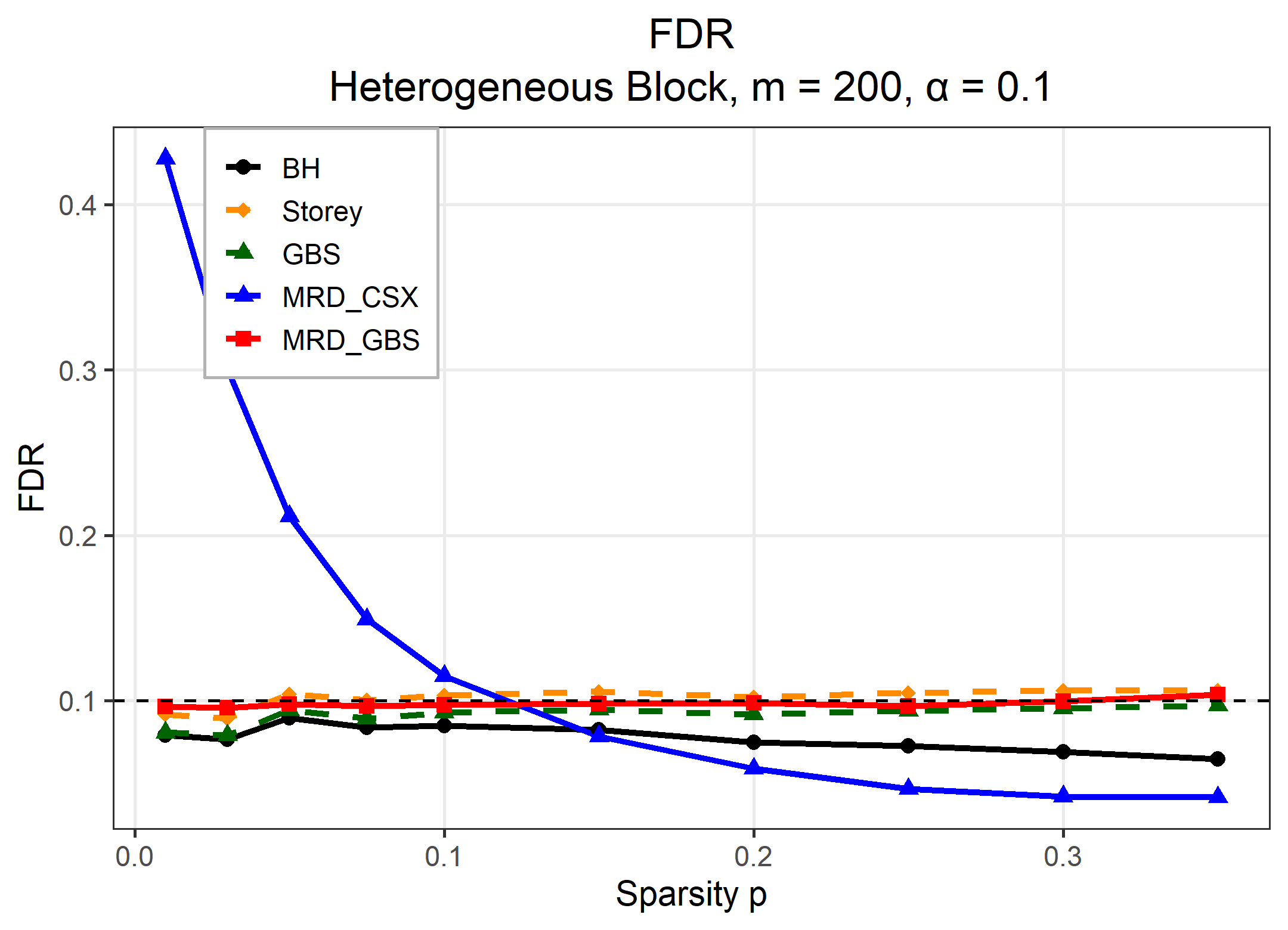}
		&
		\includegraphics[width=0.47\textwidth]
		{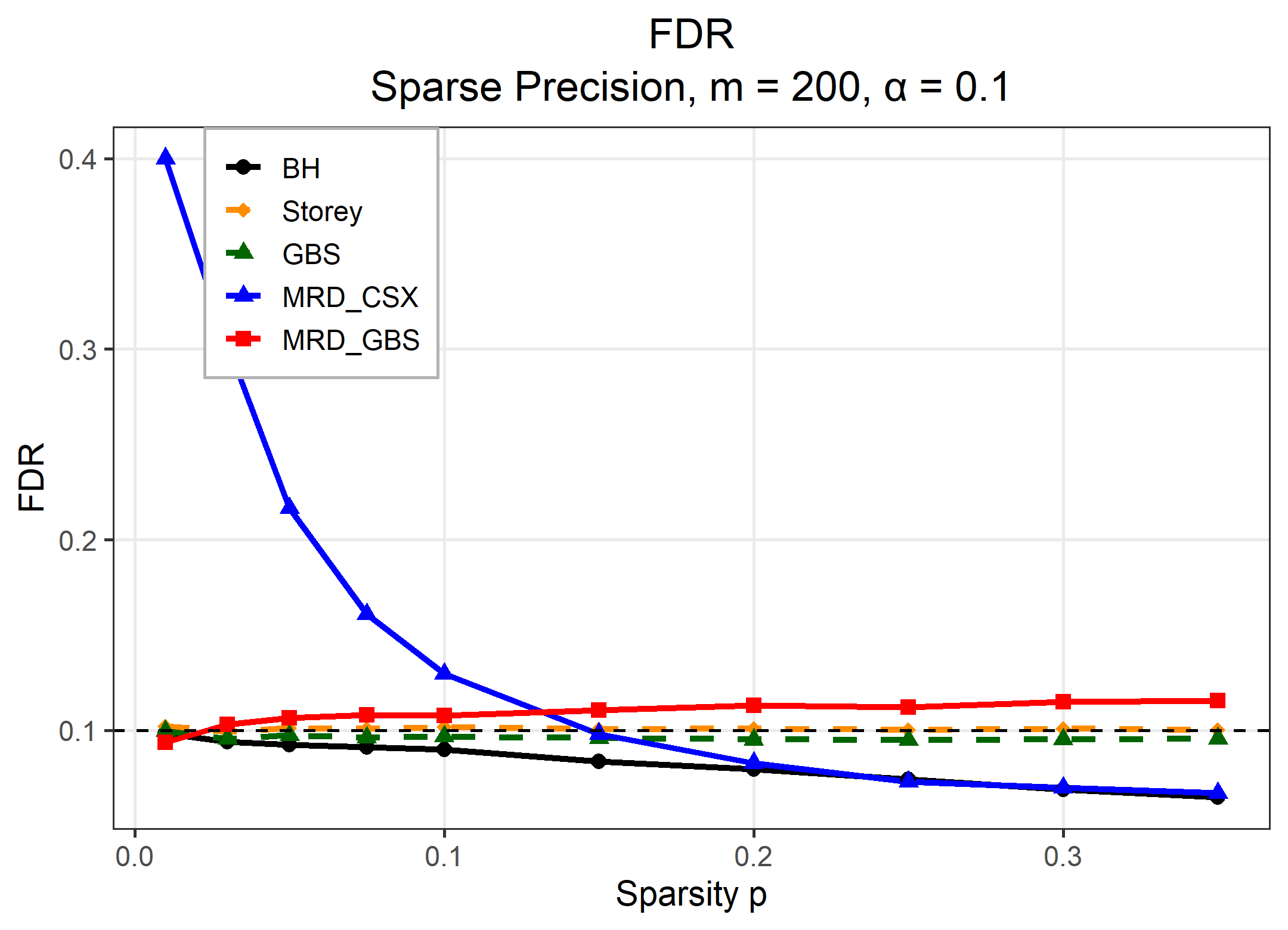}
		\\
		
		(e) Heterogeneous Block
		&
		(f) Sparse Precision Matrix
		
	\end{tabular}
	
	\caption{
		Empirical false discovery rates (FDR) for the competing
		multiple testing procedures under six representative dependence
		structures.
	}
	\label{FIG_FDR_m200}
	
\end{figure}

We begin with the false discovery rate results presented in Figure~\ref{FIG_FDR_m200} and
Tables~\ref{tab:fdr_equicorrelation_m200}--\ref{tab:fdr_sparse_precision_m200}. Across all six dependence structures, the proposed
GBS-calibrated MRD procedure exhibits substantially improved FDR behavior
relative to the original MRD procedure. In particular, the excessive FDR
inflation observed for the original MRD procedure in extremely sparse
settings is greatly reduced under the proposed calibration, with empirical
FDR values remaining close to the nominal level $\alpha=0.10$ across most sparsity levels
and dependence structures.

A particularly noteworthy feature of the results is the remarkable
stability of the empirical FDR values produced by the proposed procedure
across a broad range of covariance structures. Under the equicorrelation,
factor, Toeplitz, fractional Gaussian noise, and heterogeneous block
models, the empirical FDR curves remain close to the nominal level
$\alpha=0.1$ throughout most of the sparsity range. This behavior stands
in sharp contrast to the original MRD procedure, whose FDR values are
substantially inflated in extremely sparse regimes before gradually
decreasing as the proportion of signals increases. The results therefore
suggest that the proposed GBS calibration successfully moderates the
aggressive rejection behavior of the original MRD procedure while
preserving its ability to exploit covariance information.

An important feature of the larger-dimensional experiments is that several
dependence structures for which FDR control was less apparent when
$m=100$ exhibit substantially improved behavior when $m=200$. The
heterogeneous block model provides a particularly notable example. While
the corresponding $m=100$ results suggested mild departures from the
nominal level, the empirical FDR values move considerably closer to
$\alpha=0.1$ in the larger-dimensional setting. Similar improvements are also visible under several other dependence structures, suggesting that the benefits of the proposed calibration become increasingly apparent as the dimensionality of the testing problem grows.

The sparse precision model represents the primary exception to this
pattern. Although the proposed procedure continues to exhibit substantially
better FDR behavior than the original MRD procedure, the empirical FDR
values remain slightly above the nominal level for portions of the
sparsity range. This observation is consistent with the earlier
misclassification-risk results and again suggests that the effectiveness
of covariance-adaptive information propagation depends on the geometry of
the underlying dependence structure.

We next examine the false non-discovery rates reported in Figure~\ref{FIG_FNR_m200} and Tables~\ref{tab:fnr_equicorrelation_m200}–\ref{tab:fnr_sparse_precision_m200}. Across the equicorrelation, factor, Toeplitz, fractional Gaussian noise, and heterogeneous block models, the proposed GBS-calibrated MRD procedure exhibits extraordinarily small FNR values throughout much of the sparsity range. In several settings, the empirical FNR remains very close to zero, indicating that only a negligible proportion of true signals fail to be detected. The reductions are particularly striking when compared with the marginal procedures BH, Storey, and GBS, all of which exhibit substantially larger false non-discovery rates as sparsity increases.

\begin{figure}[p]
	\centering
	
	\begin{tabular}{cc}
		
		\includegraphics[width=0.47\textwidth]
		{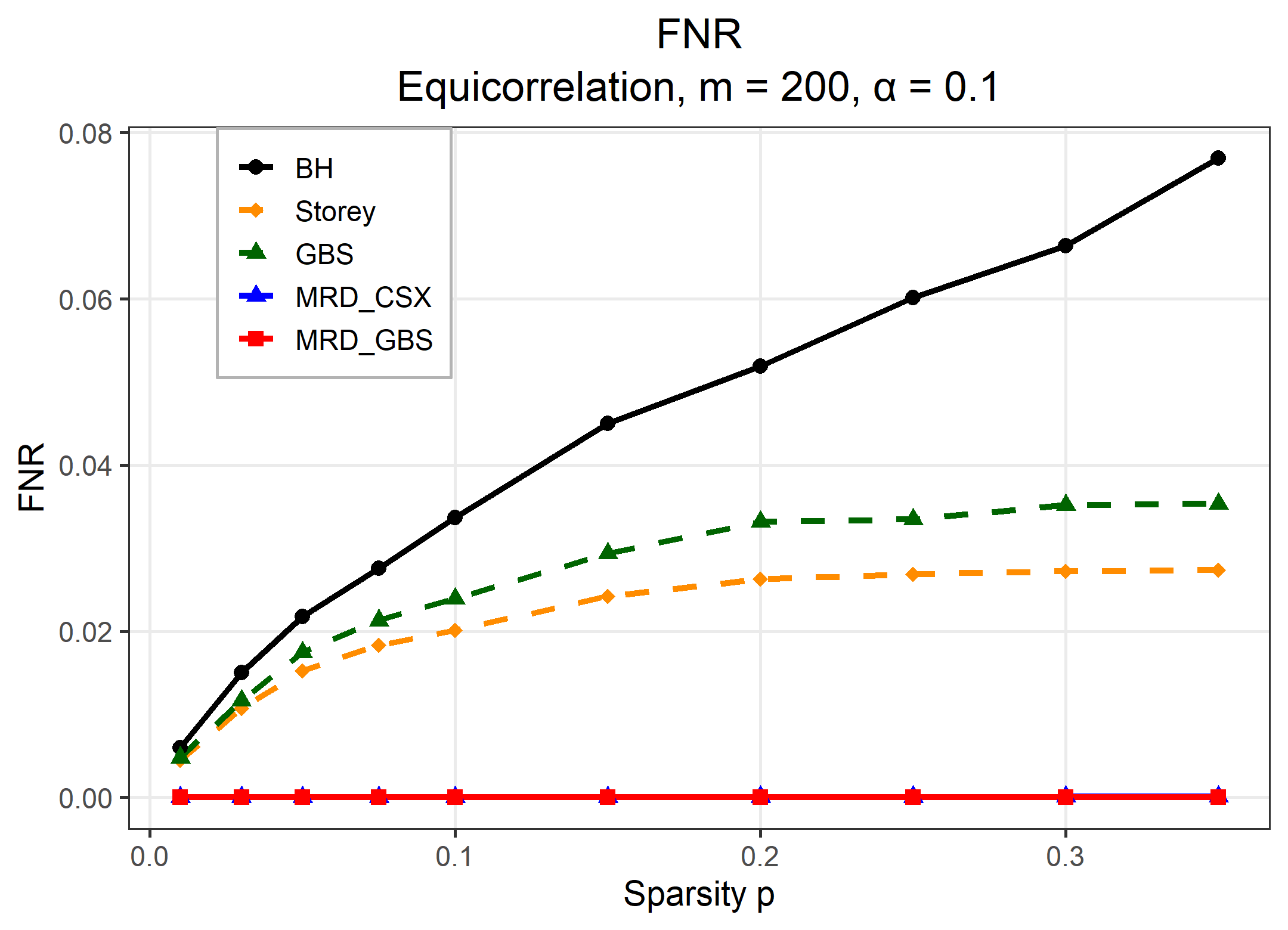}
		&
		\includegraphics[width=0.47\textwidth]
		{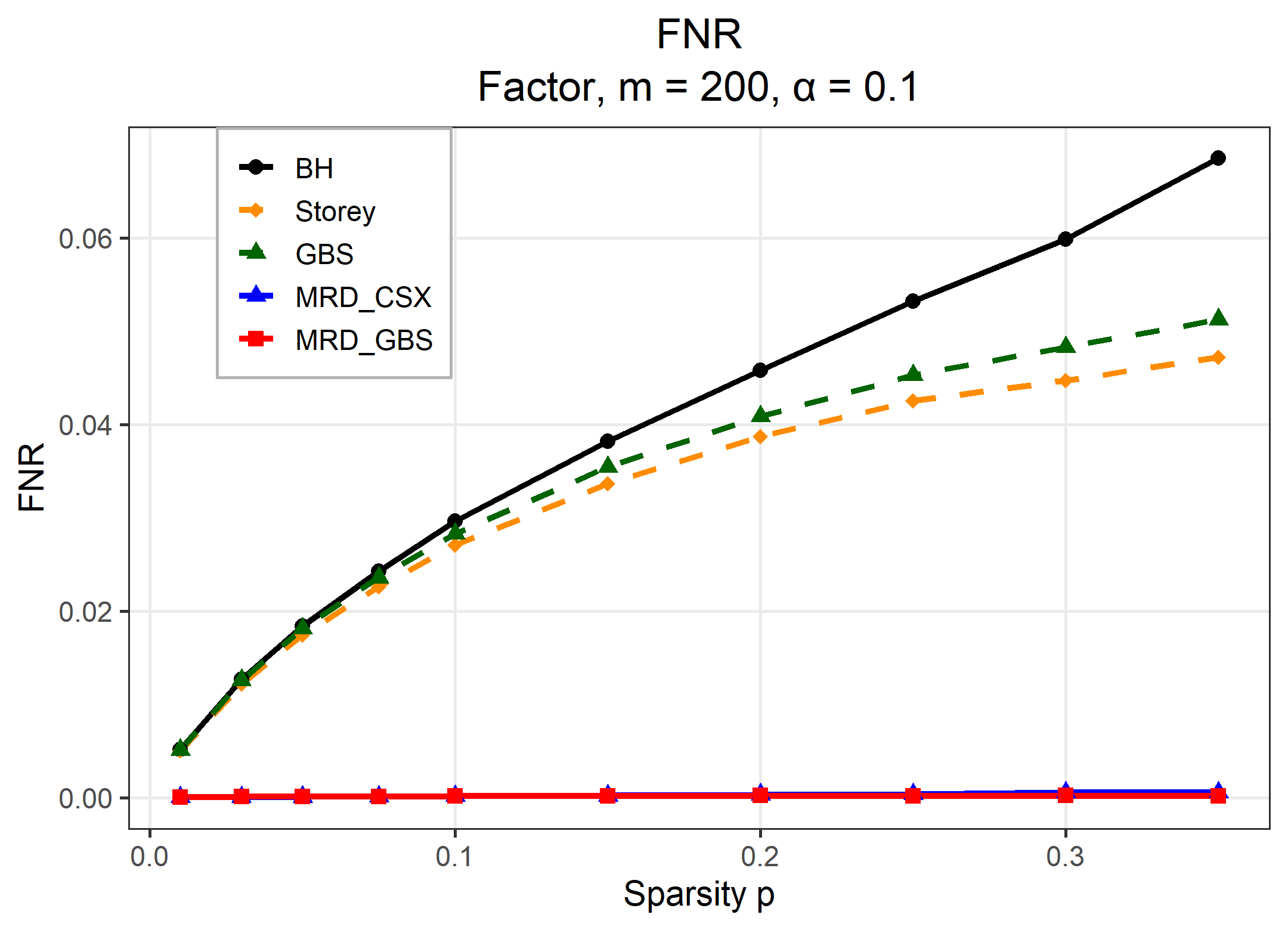}
		\\
		
		(a) Equicorrelation ($\rho=0.7$)
		&
		(b) Factor Model
		\\[0.3cm]
		
		\includegraphics[width=0.47\textwidth]
		{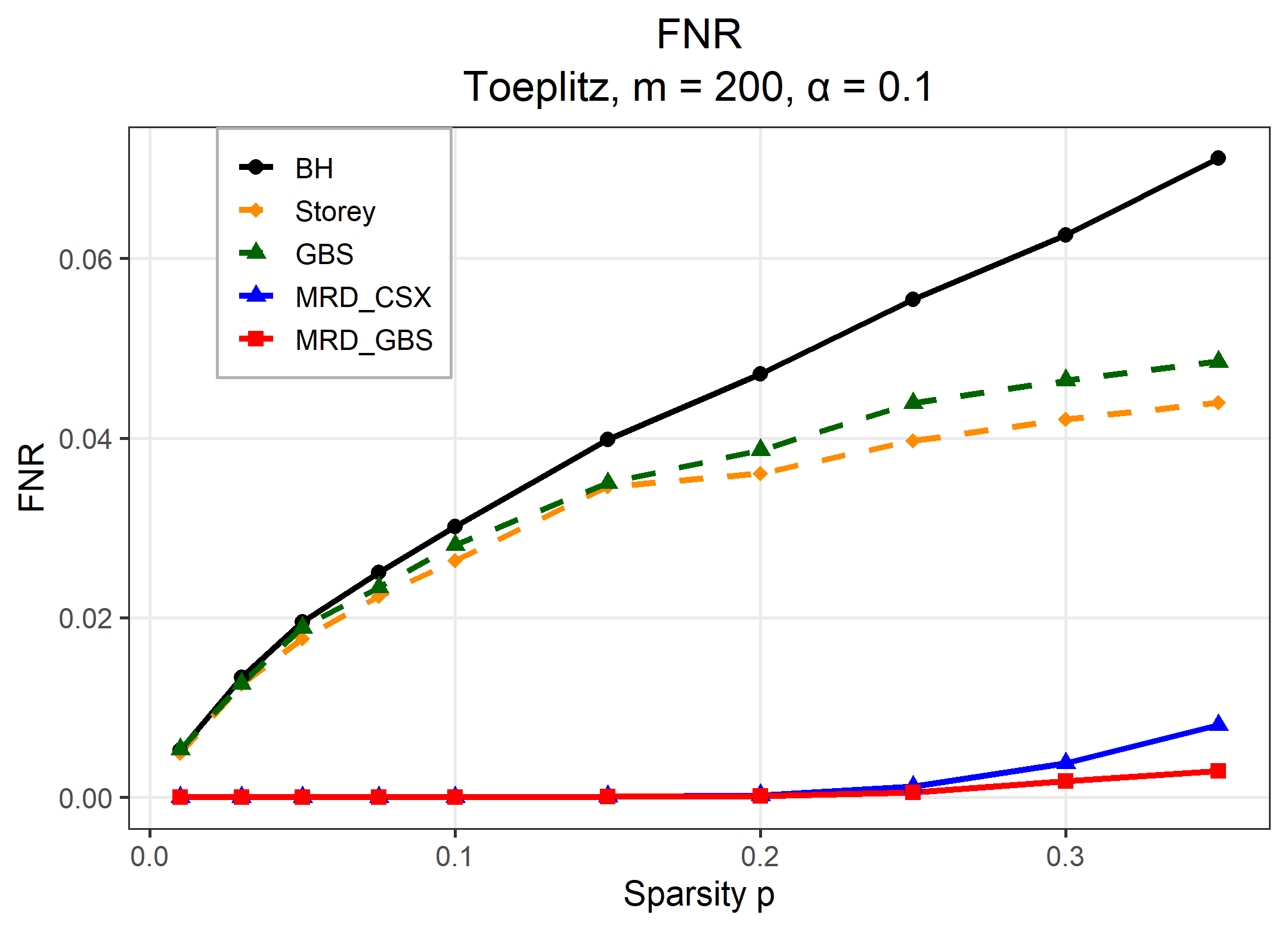}
		&
		\includegraphics[width=0.47\textwidth]
		{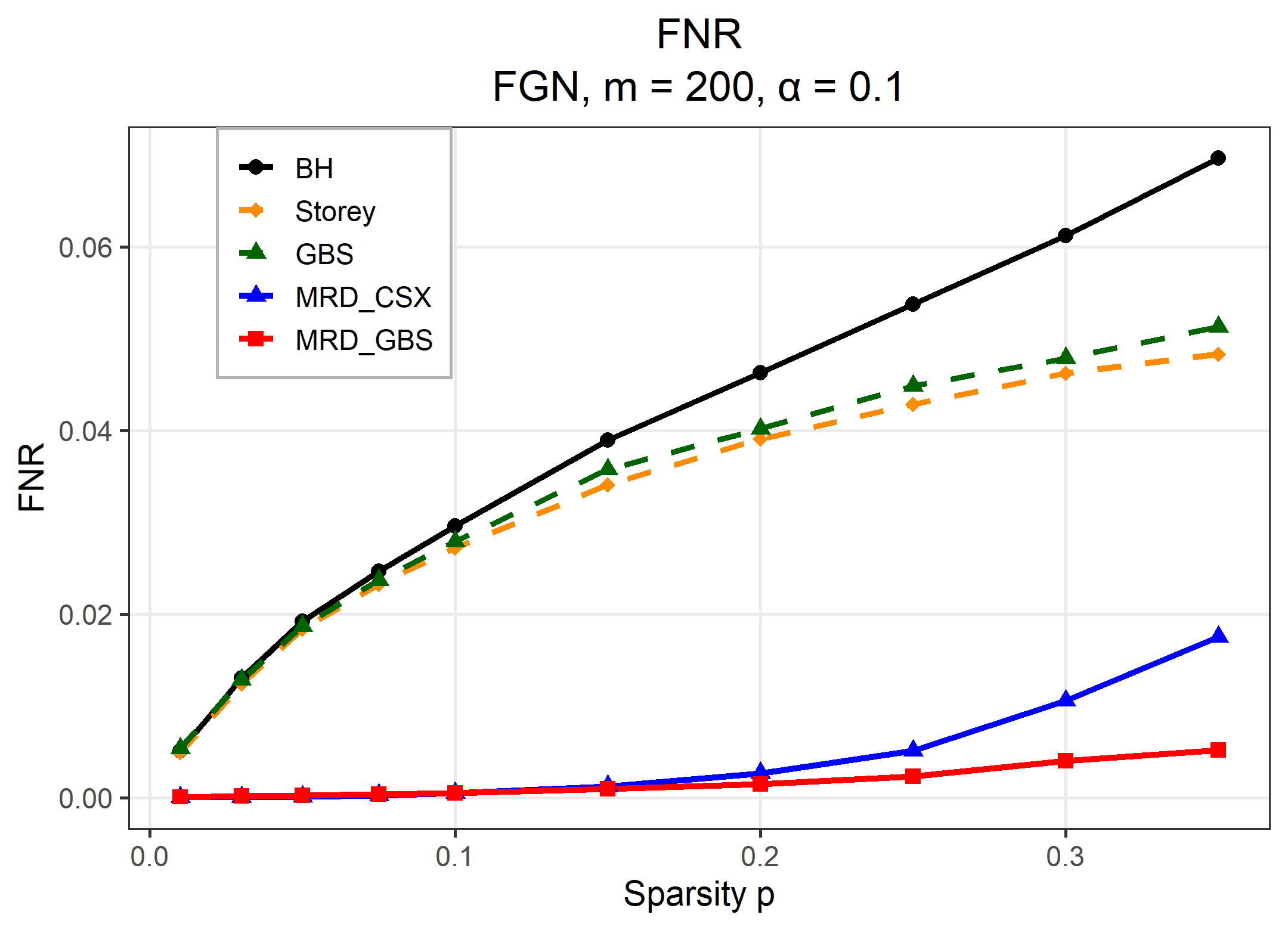}
		\\
		
		(c) Toeplitz ($\rho=0.9$)
		&
		(d) Fractional Gaussian Noise ($H=0.9$)
		\\[0.3cm]
		
		\includegraphics[width=0.47\textwidth]
		{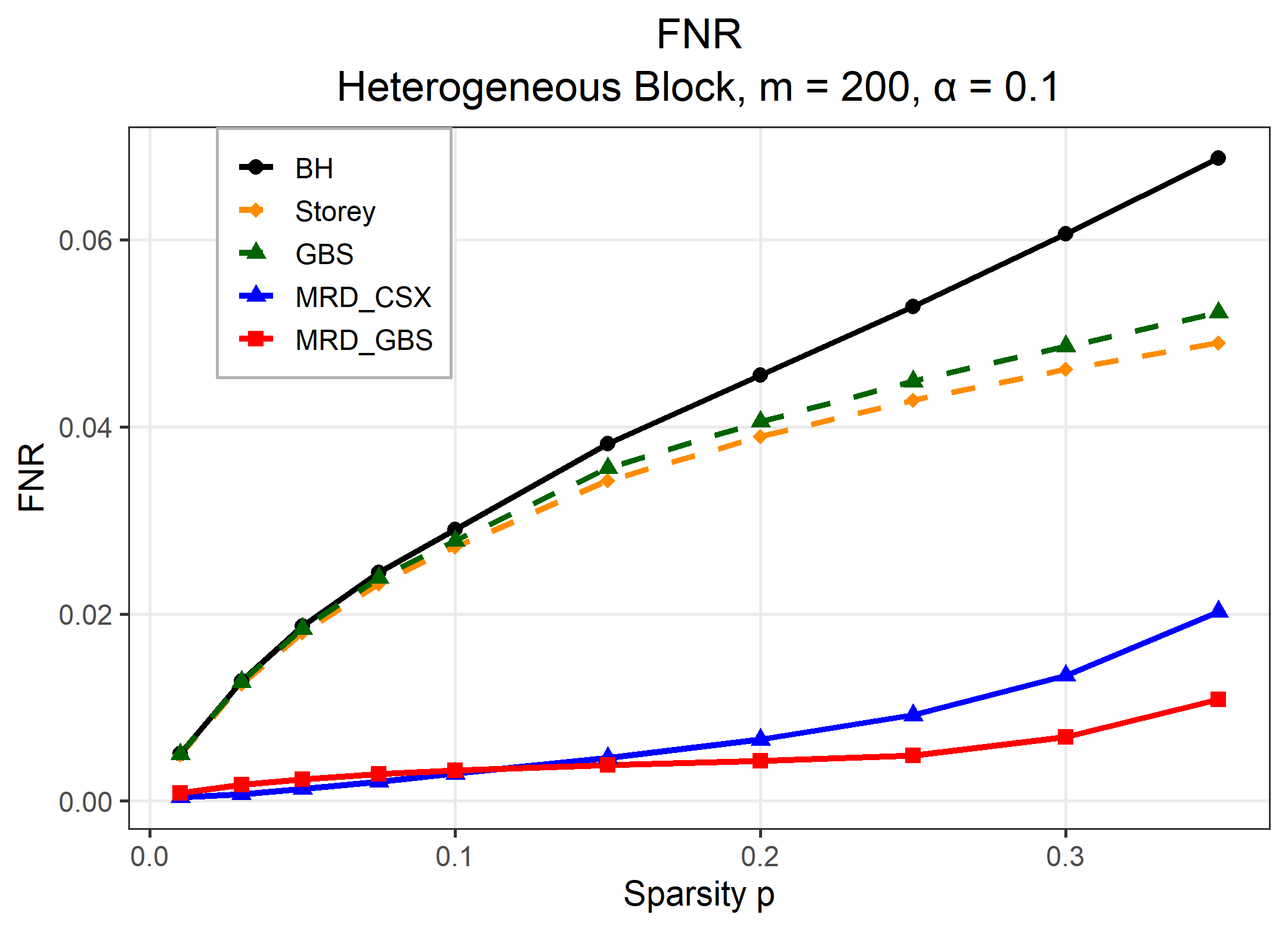}
		&
		\includegraphics[width=0.47\textwidth]
		{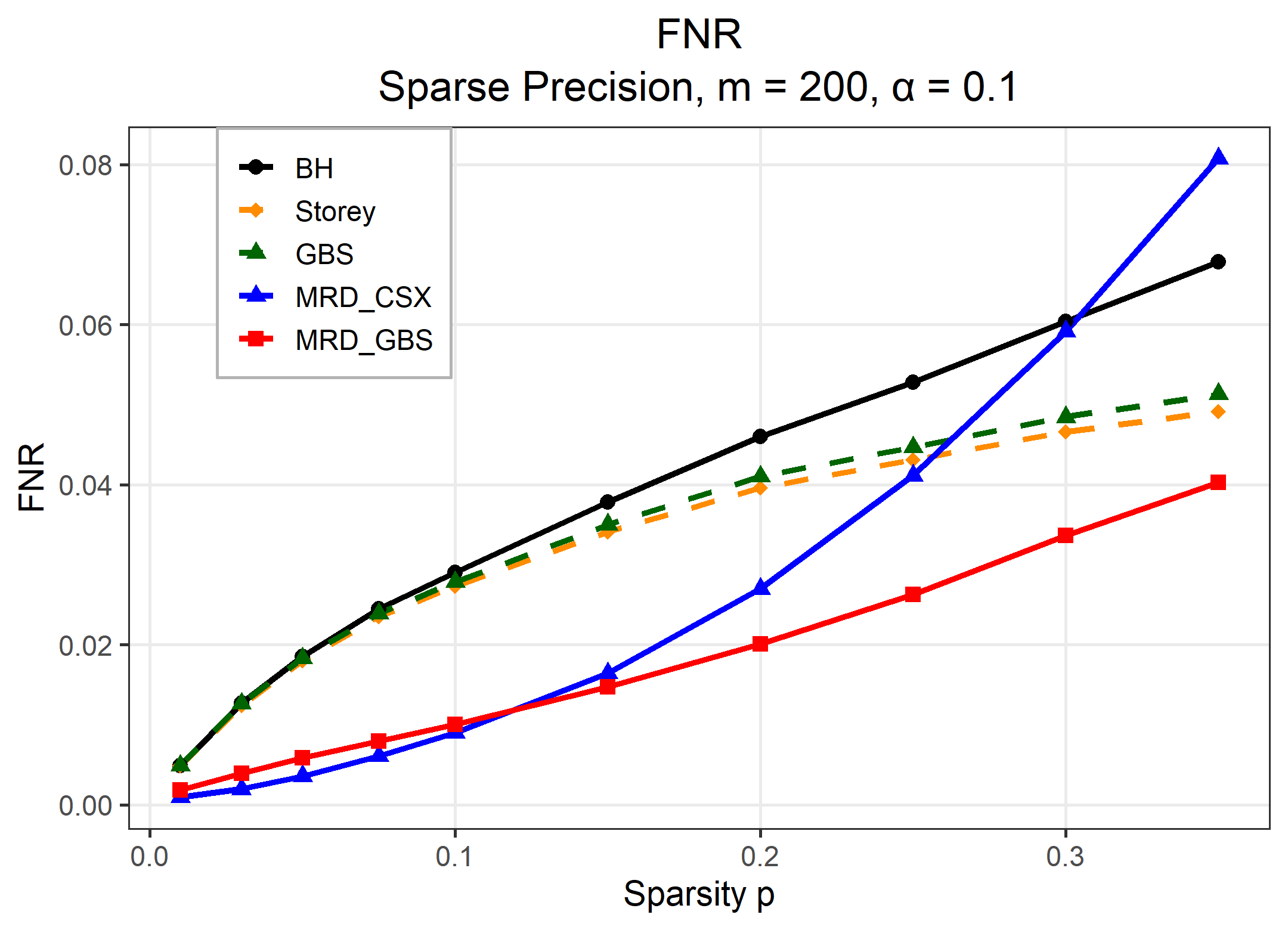}
		\\
		
		(e) Heterogeneous Block
		&
		(f) Sparse Precision Matrix
		
	\end{tabular}
	
	\caption{
		Empirical false non-discovery rates (FNR) for the competing
		multiple testing procedures under six representative dependence
		structures.
	}
	\label{FIG_FNR_m200}
	
\end{figure}

The original MRD procedure often achieves similarly small false non-discovery rates, particularly under the more strongly dependent covariance models. This behavior is expected, since both procedures employ the same covariance-adaptive residualization mechanism. Nevertheless, the proposed calibration frequently preserves these near-zero FNR values while simultaneously achieving substantially improved FDR performance, thereby producing a more favorable overall balance between false discoveries and missed discoveries.

A closer examination of the results reveals a more nuanced pattern. In several covariance settings, the original MRD procedure attains the smallest FNR values in extremely sparse regimes. However, the differences relative to the proposed GBS-calibrated MRD procedure are typically quite small. Moreover, under the fractional Gaussian noise, heterogeneous block, and sparse precision models, the proposed procedure eventually achieves smaller FNR values as the signal proportion increases. Thus, the improved classification performance of the GBS-calibrated procedure cannot be attributed solely to enhanced FDR behavior. Rather, the results suggest that the proposed calibration not only maintains comparable signal-recovery performance in sparse regimes, but may also improve signal recovery in moderately sparse settings while simultaneously providing substantially better false-discovery control.

The crossover phenomenon observed previously in the NMR results reappears in several of the signal-recovery summaries. While the proposed calibration often exhibits superior performance in sparse regimes, the original MRD procedure frequently becomes increasingly competitive as the signal proportion increases. This pattern is particularly visible in the FNR, power, and average-rejection summaries and suggests that the relative advantages of the two calibration schemes depend not only on the covariance structure but also on the underlying sparsity level.

The power curves as shown in Figure~\ref{FIG_POWER_m200} together with Tables~\ref{tab:power_equicorrelation_m200}-\ref{tab:power_sparse_precision_m200} provide further evidence of the strength of the proposed methodology. Across several covariance structures, including the equicorrelation, factor, Toeplitz, fractional Gaussian noise, and heterogeneous block models, the empirical powers are frequently close to one over substantial portions of the sparsity range. The near-perfect power observed under these dependence structures indicates that the overwhelming majority of active signals are successfully identified by the procedure.

\begin{figure}[p]
	\centering
	
	\begin{tabular}{cc}
		
		\includegraphics[width=0.47\textwidth]
		{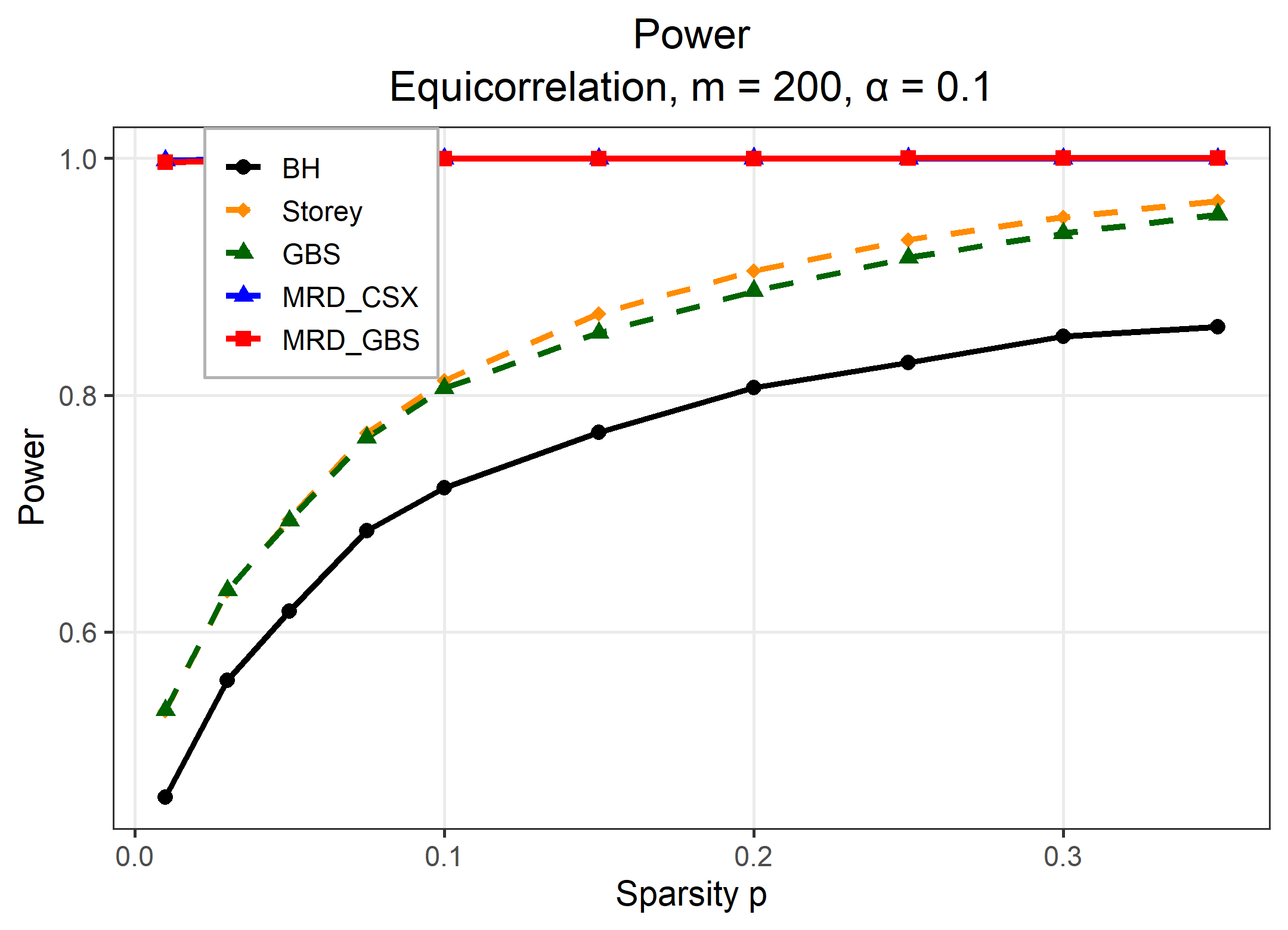}
		&
		\includegraphics[width=0.47\textwidth]
		{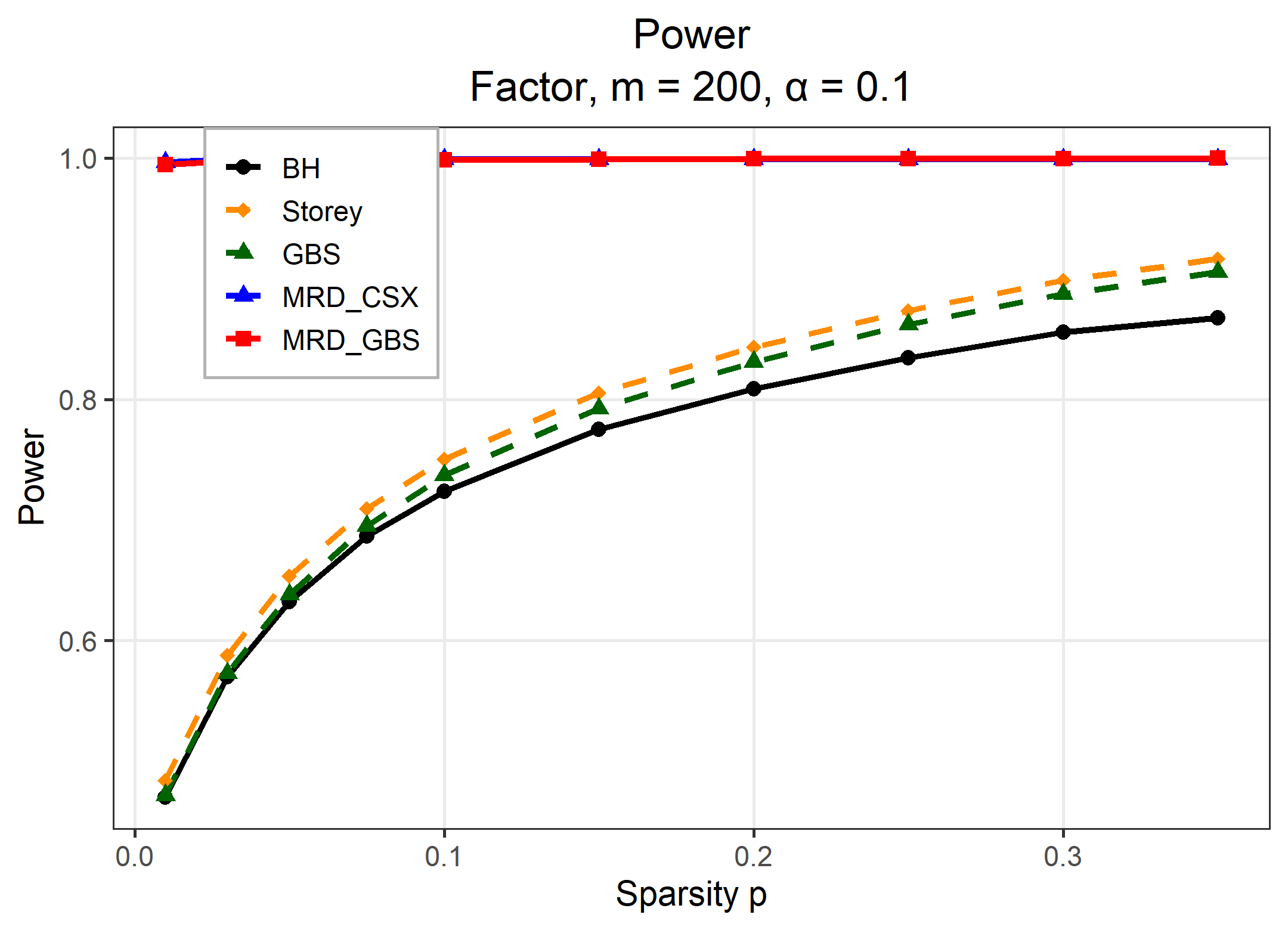}
		\\
		
		(a) Equicorrelation ($\rho=0.7$)
		&
		(b) Factor Model
		\\[0.3cm]
		
		\includegraphics[width=0.47\textwidth]
		{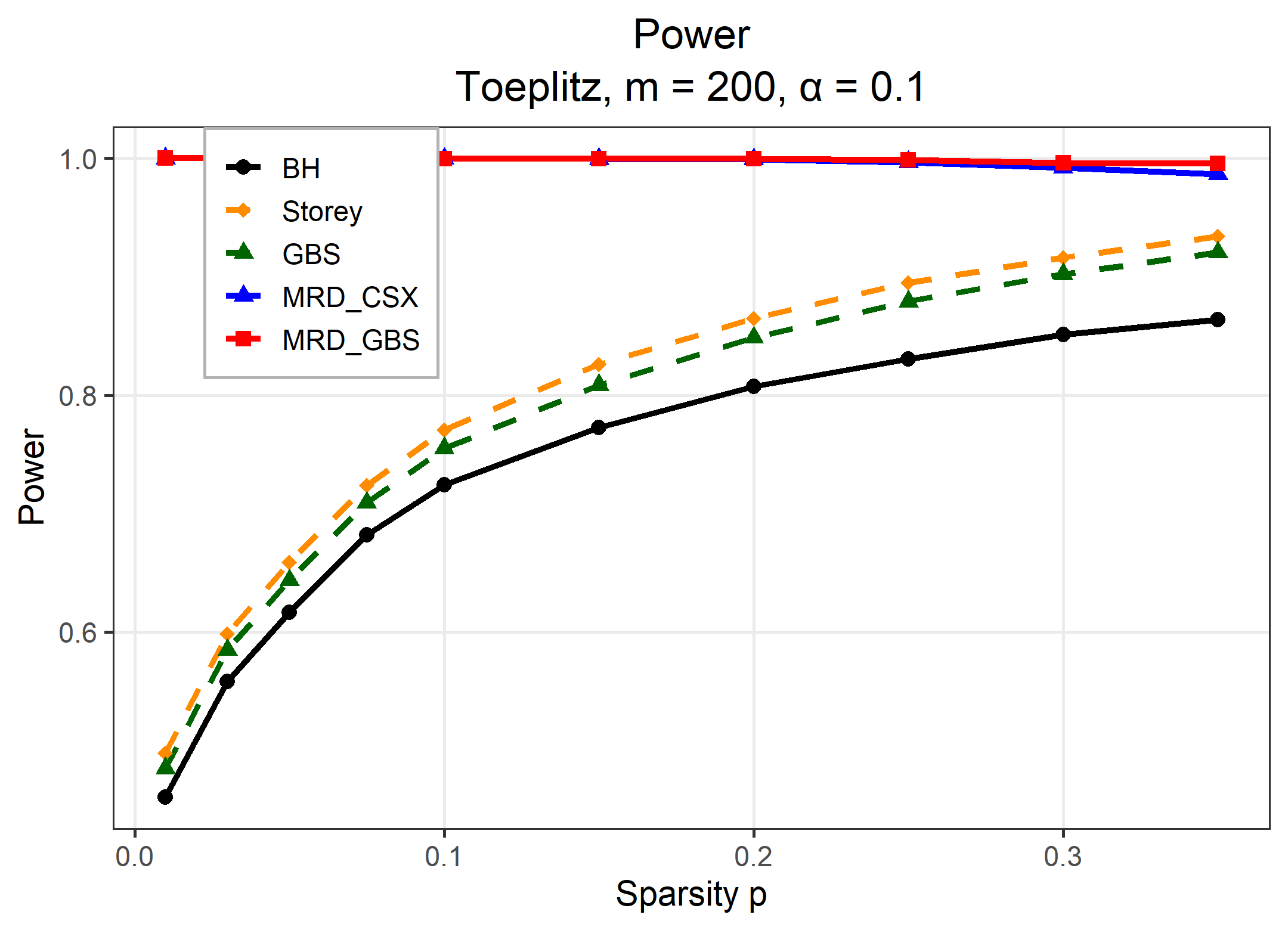}
		&
		\includegraphics[width=0.47\textwidth]
		{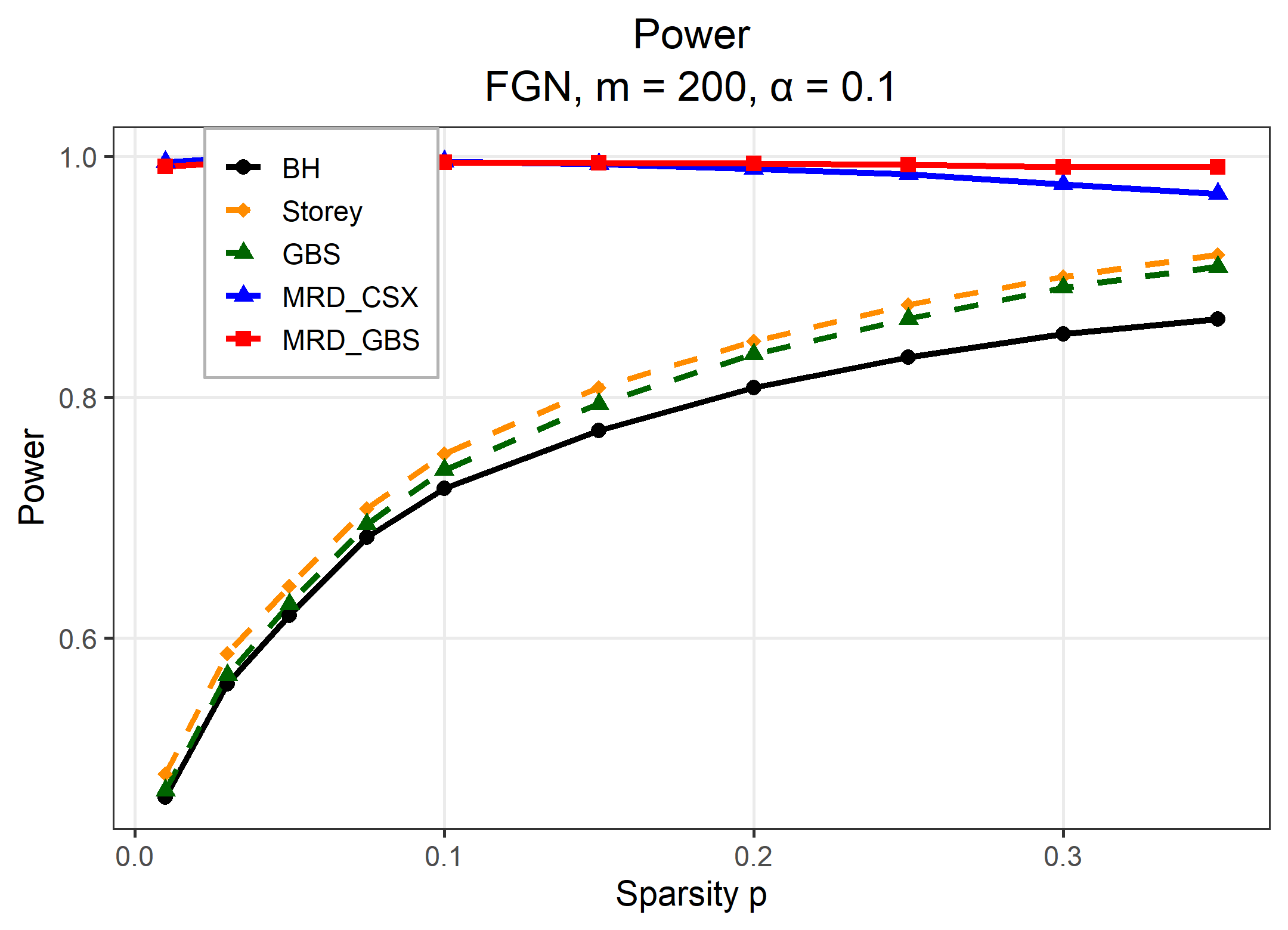}
		\\
		
		(c) Toeplitz ($\rho=0.9$)
		&
		(d) Fractional Gaussian Noise ($H=0.9$)
		\\[0.3cm]
		
		\includegraphics[width=0.47\textwidth]
		{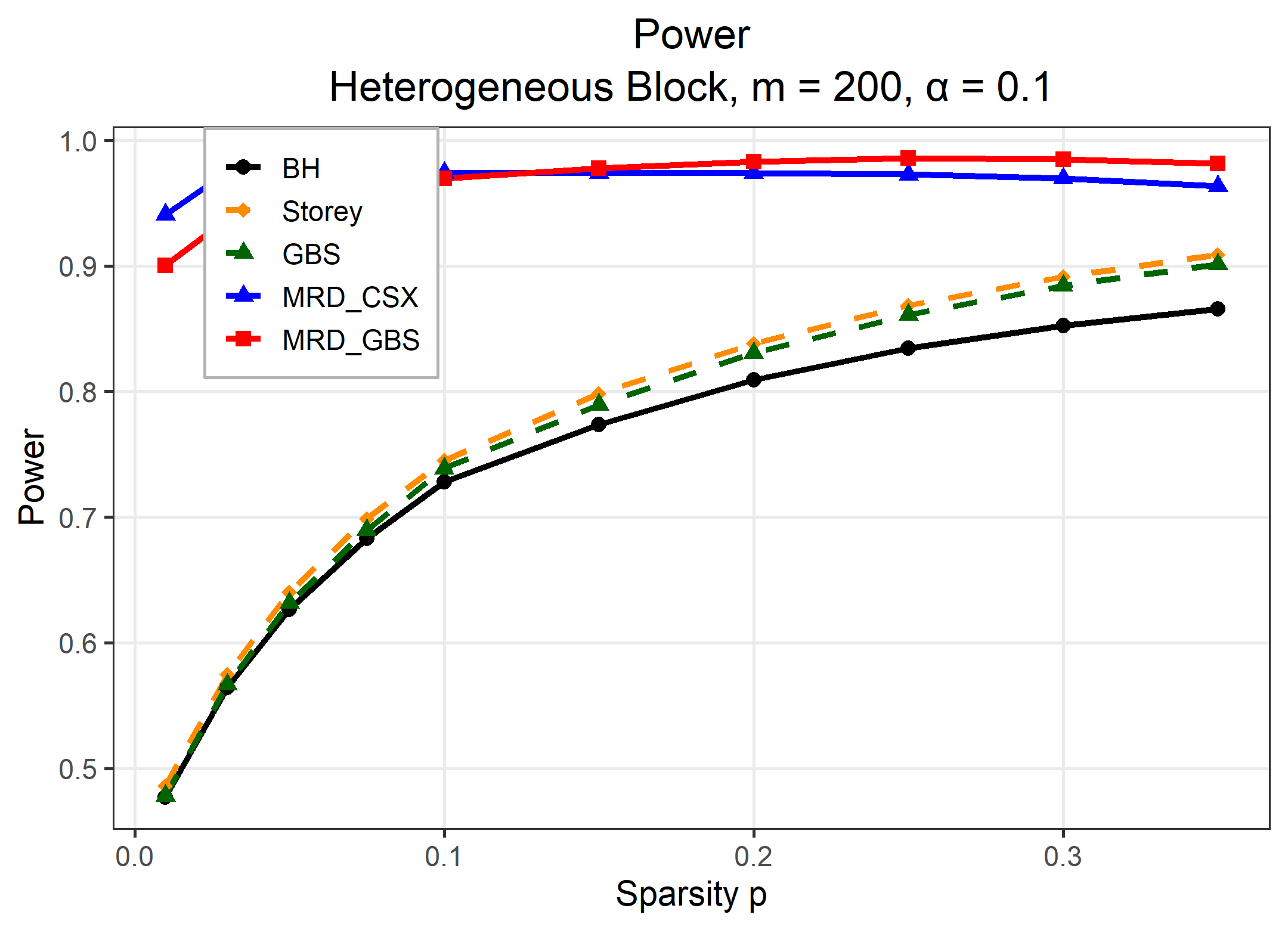}
		&
		\includegraphics[width=0.47\textwidth]
		{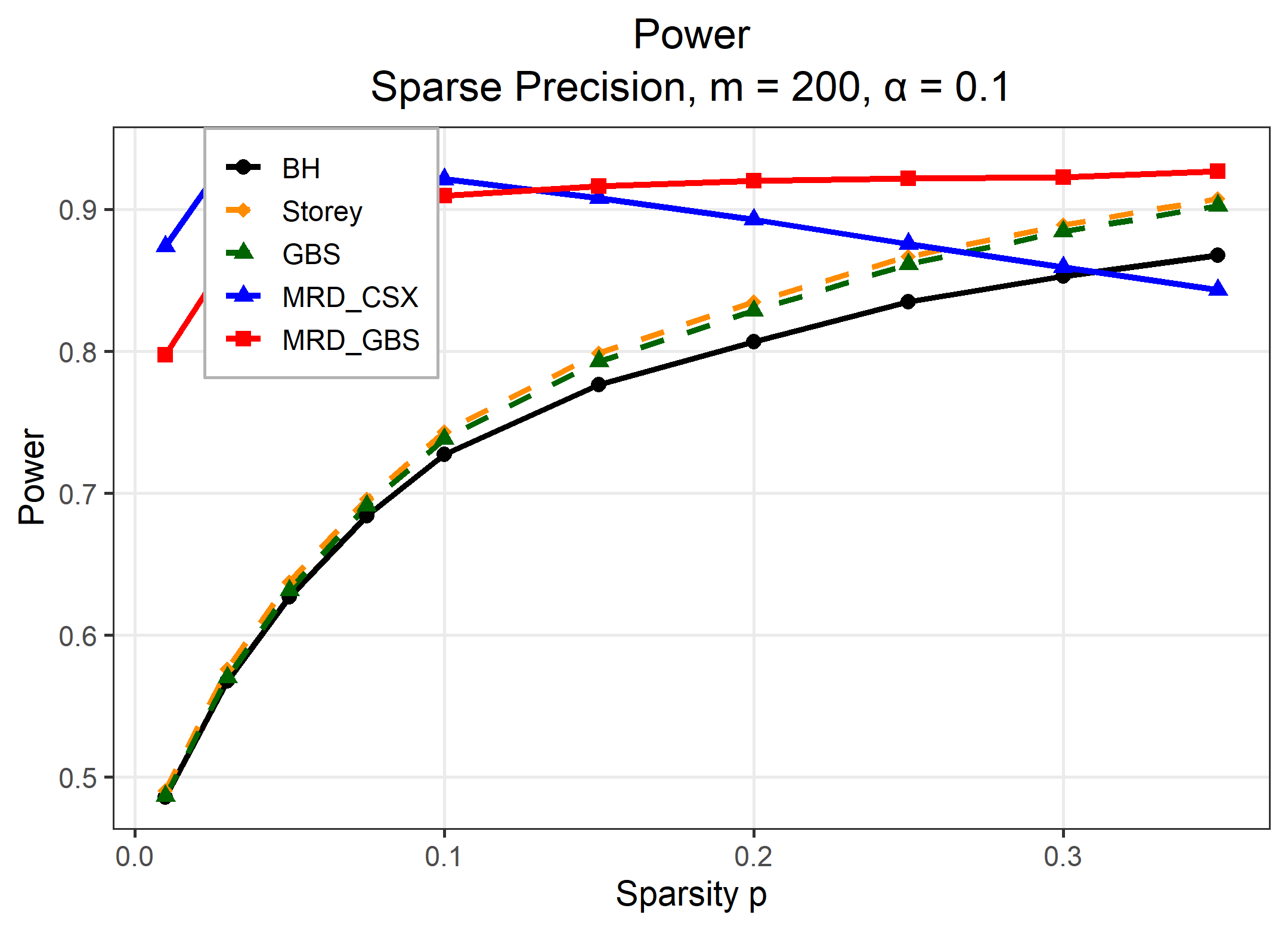}
		\\
		
		(e) Heterogeneous Block
		&
		(f) Sparse Precision Matrix
		
	\end{tabular}
	
	\caption{
		Empirical power of the competing multiple testing procedures under six
		representative dependence structures. Larger values indicate superior
		signal recovery performance. Across several dependence structures,
		the GBS-calibrated MRD procedure exhibits remarkably strong power,
		particularly in sparse and moderately sparse regimes.
	}
	\label{FIG_POWER_m200}
	
\end{figure}

The power results are particularly noteworthy when viewed together with the corresponding FDR behavior. High power alone is not necessarily indicative of effective signal recovery, since a procedure may achieve large power simply by rejecting too aggressively. However, the proposed GBS-calibrated MRD procedure frequently attains powers very close to one while simultaneously maintaining empirical FDR values near the nominal level. This combination is observed repeatedly across the equicorrelation, factor, Toeplitz, fractional Gaussian noise, and heterogeneous block models. The results therefore suggest that the covariance-adaptive residualization mechanism is able to identify a large proportion of the active signals without incurring the substantial false-discovery penalties often associated with highly aggressive testing procedures.

Additional insight is obtained from the average numbers of rejections reported in Figure~\ref{FIG_ANR_m200} and Tables~\ref{tab:rej_equicorrelation_m200}-\ref{tab:rej_sparse_precision_m200}. Across several covariance structures, the average number of rejections produced by the proposed procedure closely tracks the expected number of true signals mp. This behavior is particularly noteworthy when viewed together with the FDR and FNR results. Since false discoveries remain limited while false non-discoveries are nearly absent, the close agreement between the average number of rejections and the expected number of true signals provides strong empirical evidence that the procedure is recovering the underlying support of the signal vector with remarkably high accuracy.

\begin{figure}[p]
	\centering
	
	\begin{tabular}{cc}
		
		\includegraphics[width=0.47\textwidth]
		{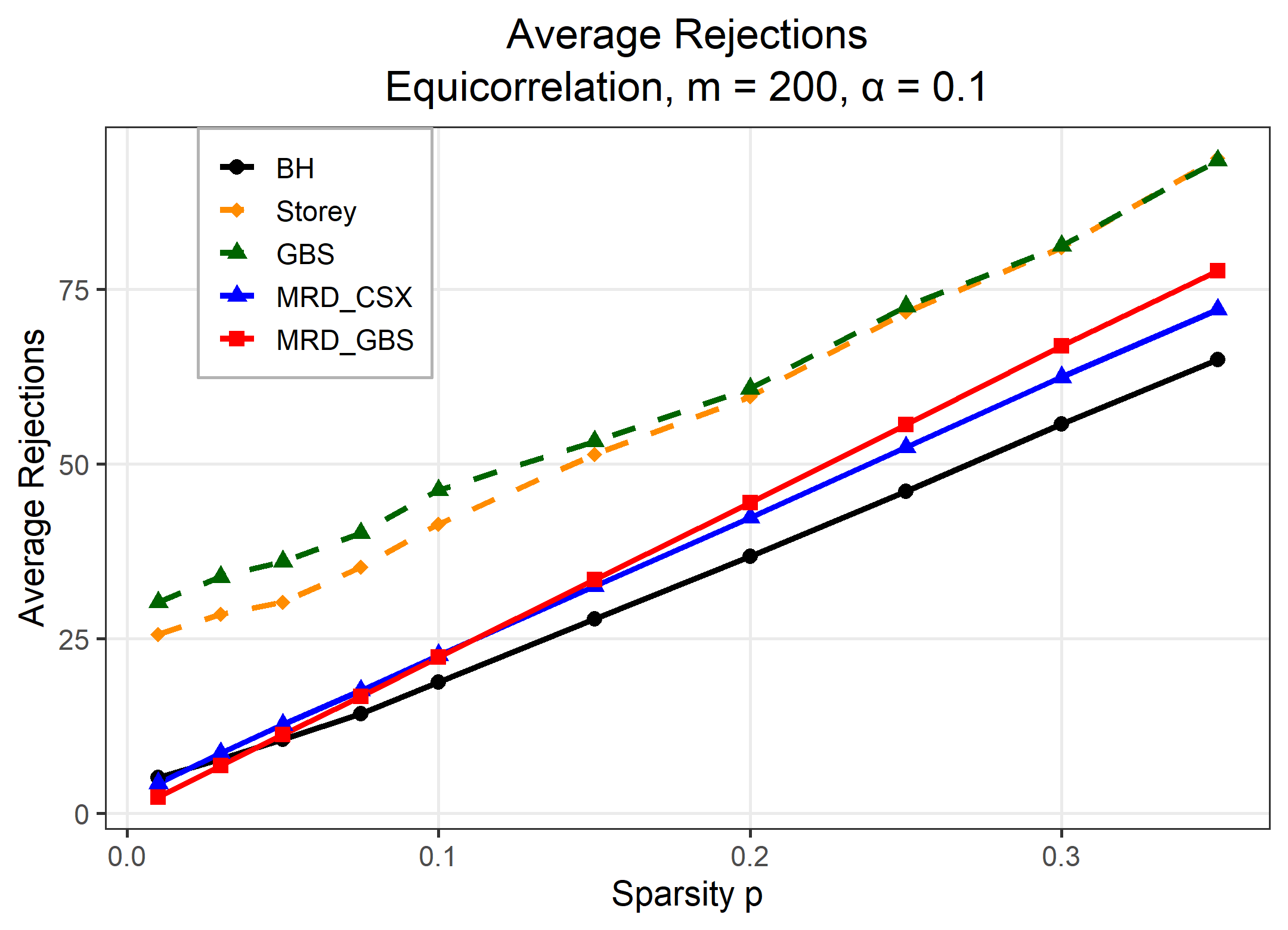}
		&
		\includegraphics[width=0.47\textwidth]
		{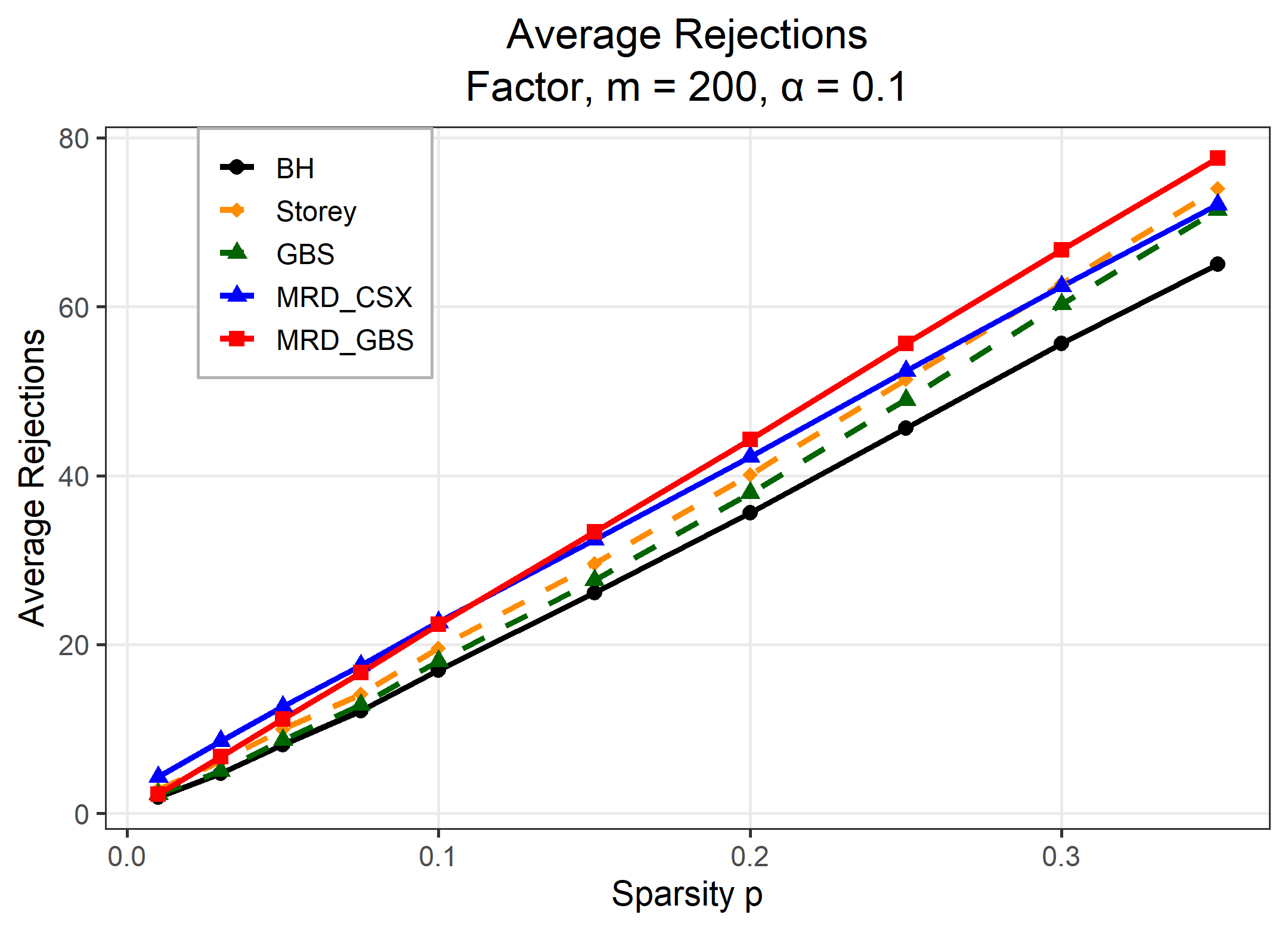}
		\\
		
		(a) Equicorrelation ($\rho=0.7$)
		&
		(b) Factor Model
		\\[0.3cm]
		
		\includegraphics[width=0.47\textwidth]
		{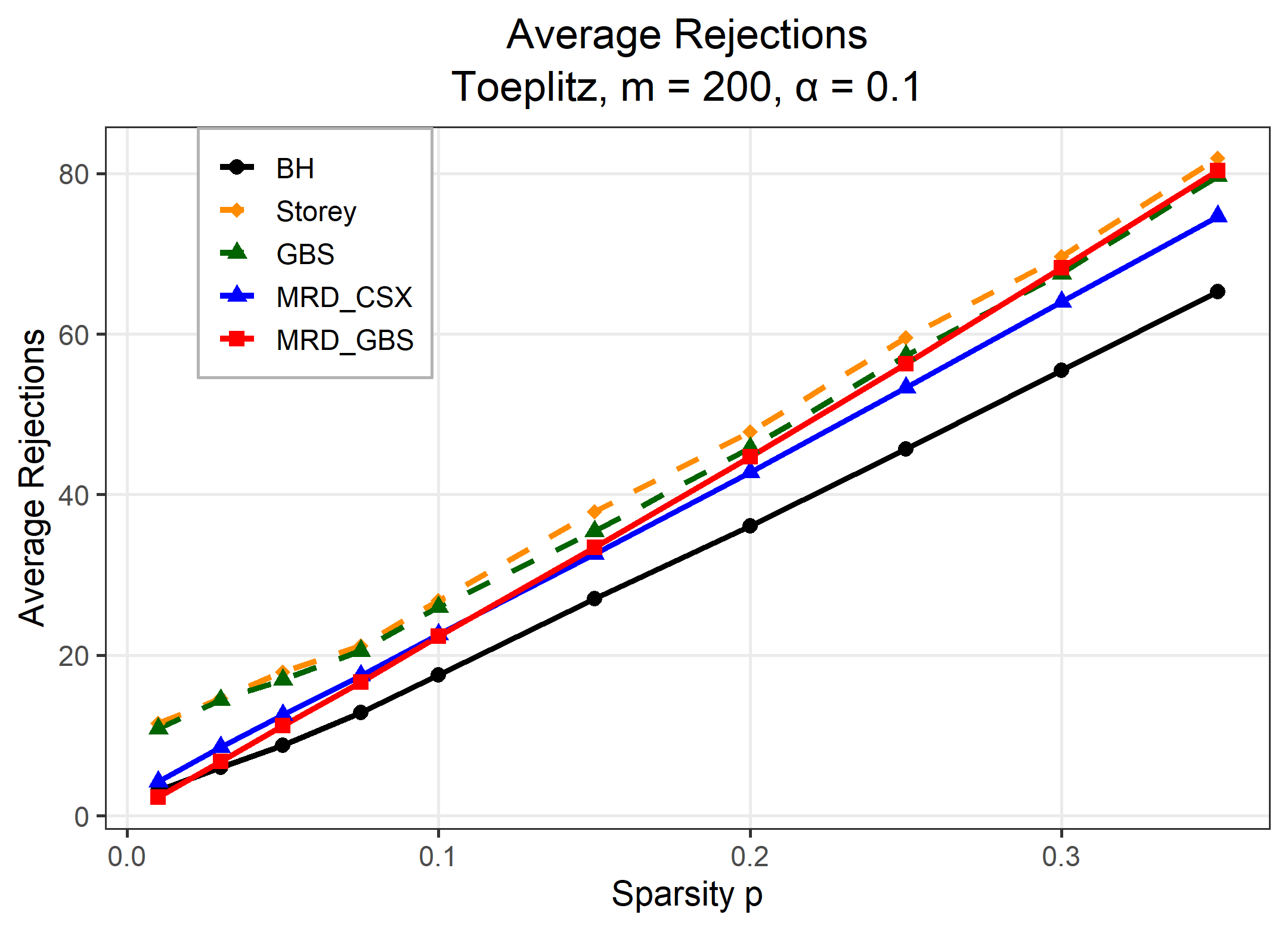}
		&
		\includegraphics[width=0.47\textwidth]
		{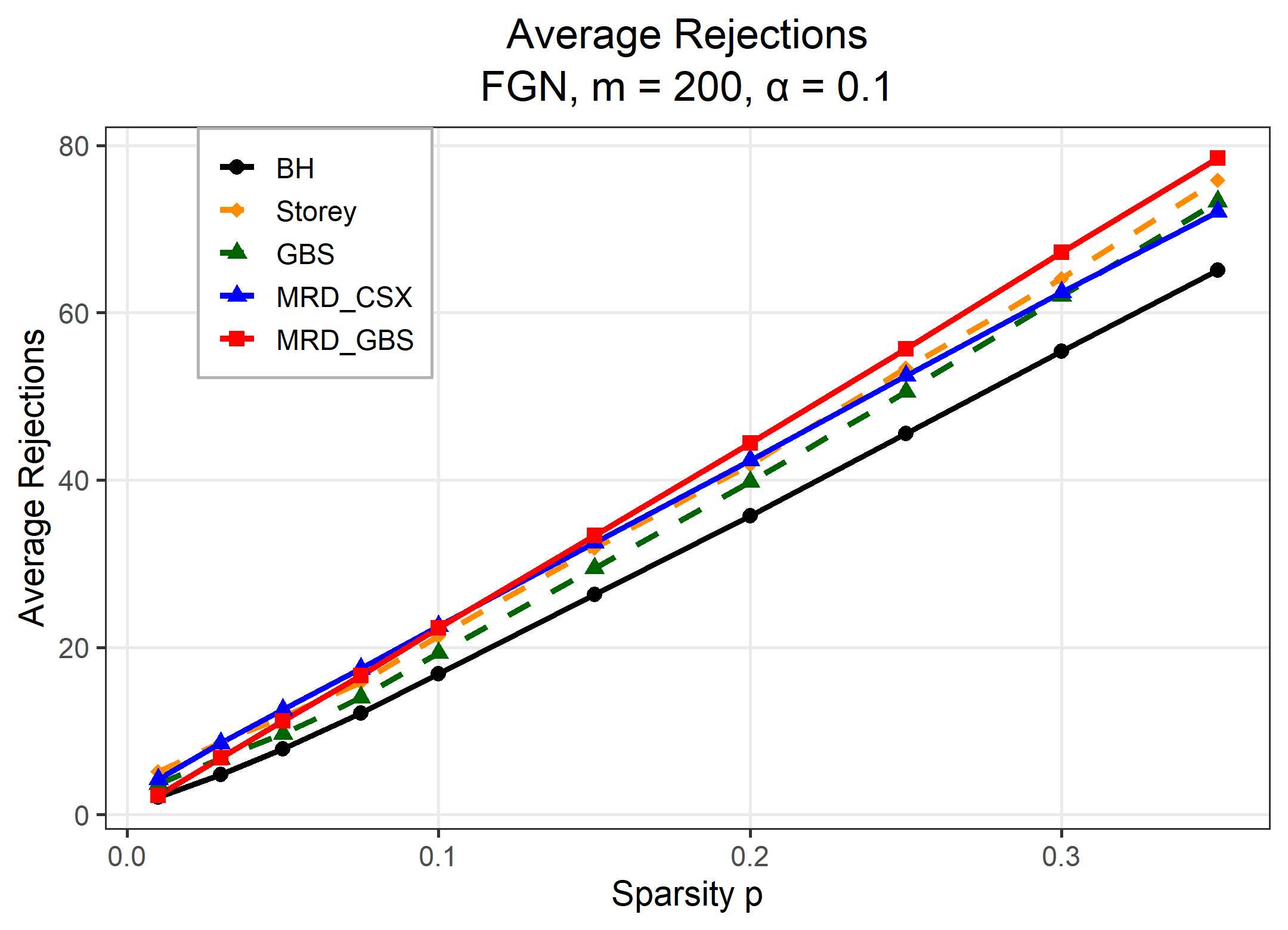}
		\\
		
		(c) Toeplitz ($\rho=0.9$)
		&
		(d) Fractional Gaussian Noise ($H=0.9$)
		\\[0.3cm]
		
		\includegraphics[width=0.47\textwidth]
		{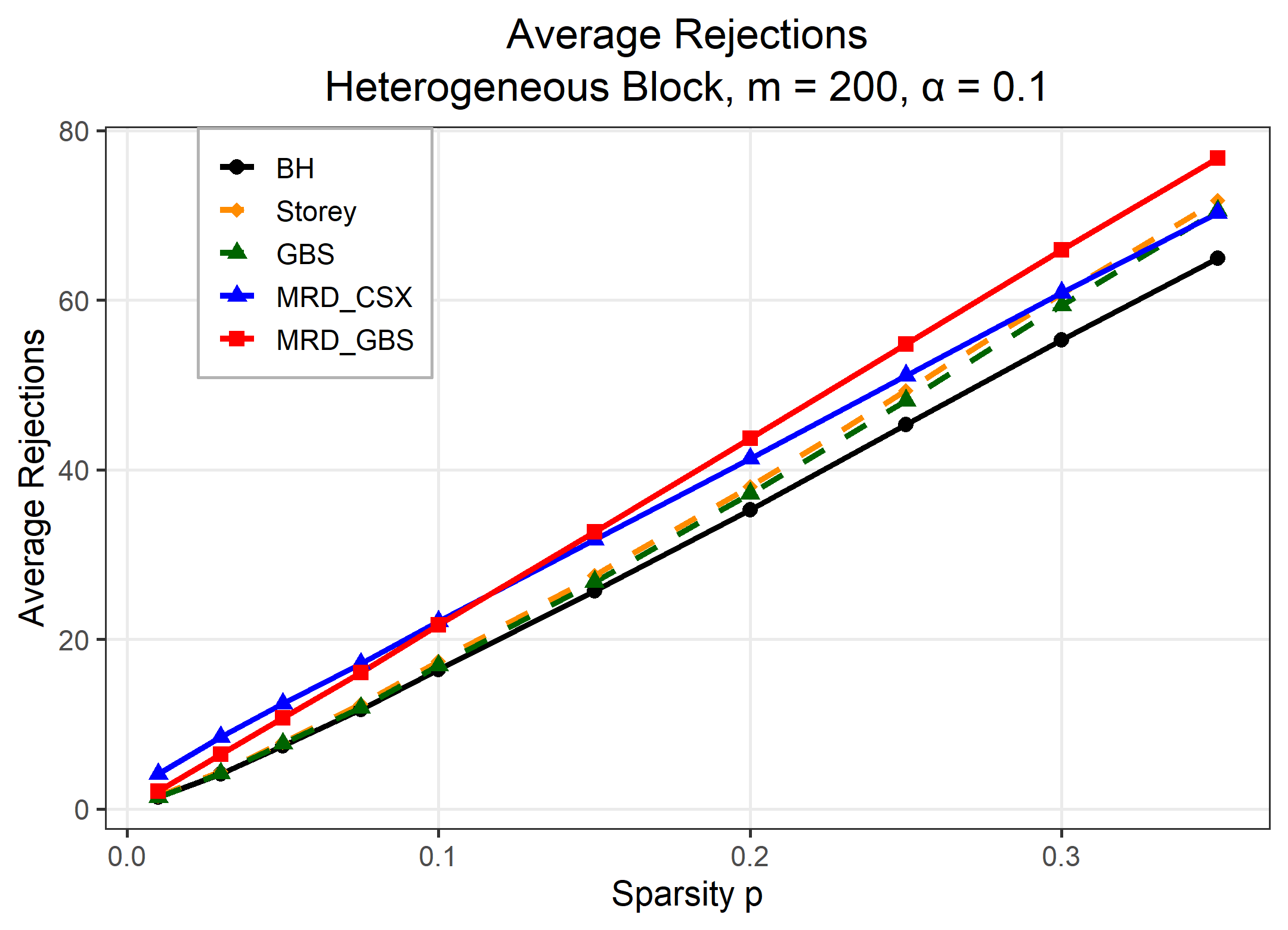}
		&
		\includegraphics[width=0.47\textwidth]
		{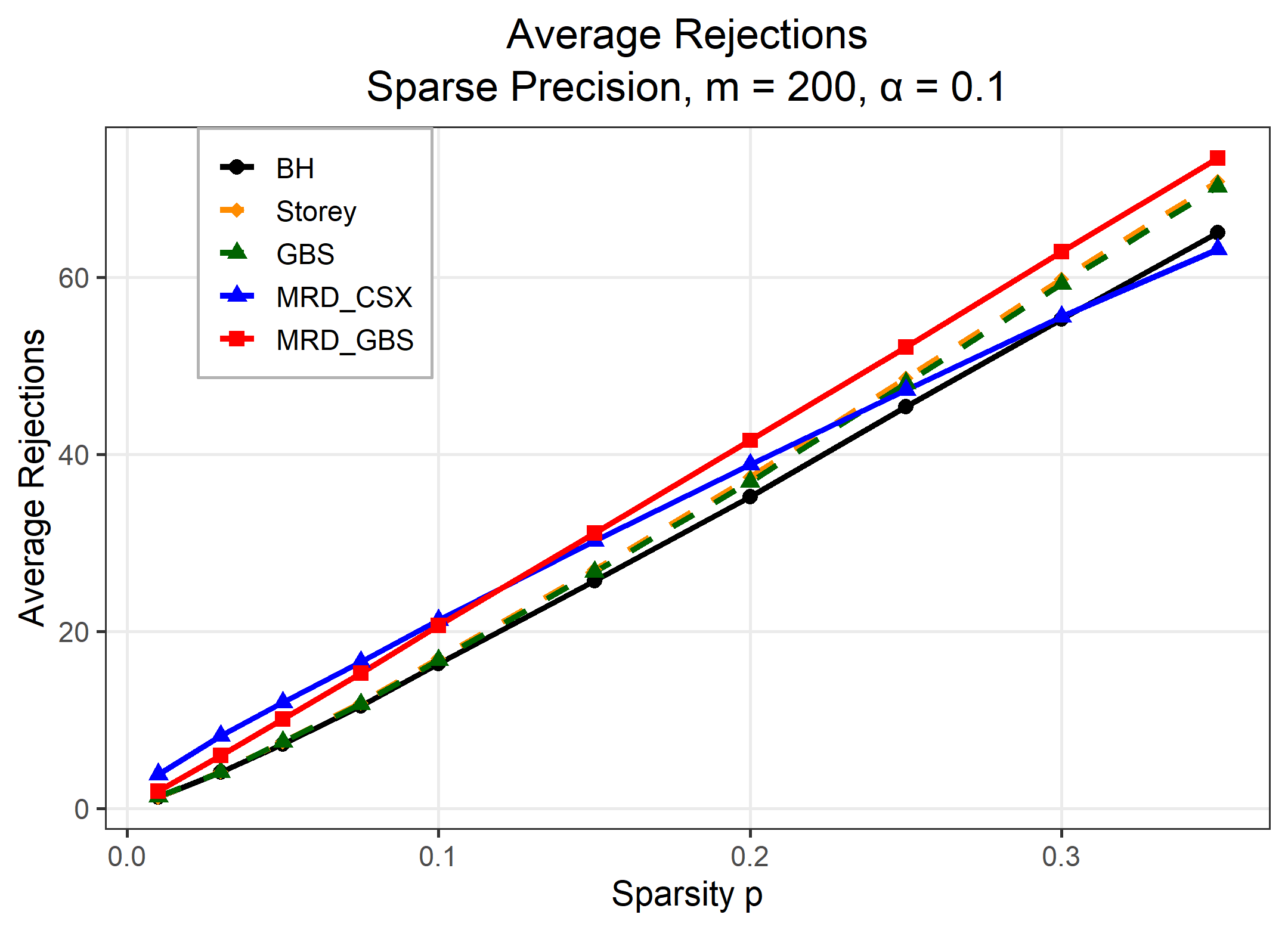}
		\\
		
		(e) Heterogeneous Block
		&
		(f) Sparse Precision Matrix
		
	\end{tabular}
	
	\caption{
		Average numbers of rejections (ANR) produced by the competing multiple
		testing procedures under six representative dependence structures.
		This metric provides additional insight into the aggressiveness of each
		procedure and helps explain the observed trade-offs among false
		discoveries, missed discoveries, and overall misclassification rates.
	}
	\label{FIG_ANR_m200}
	
\end{figure}

The average-rejection results are especially revealing when interpreted together with the known sparsity levels used in the simulations. Since the expected number of non-null hypotheses equals mp, the close agreement between the empirical average numbers of rejections and $mp$ indicates that the procedure is selecting approximately the correct number of signals. This behavior is observed repeatedly across the equicorrelation, factor, Toeplitz, fractional Gaussian noise, and heterogeneous block models. Combined with the simultaneously favorable FDR and FNR characteristics, these findings provide particularly strong empirical evidence that the procedure is recovering the underlying support of the signal vector with high accuracy.

Taken together, the FDR, FNR, power, and average-rejection results provide remarkably consistent evidence of near-support recovery. Across the equicorrelation, factor, Toeplitz, fractional Gaussian noise, and heterogeneous block models, the proposed procedure simultaneously achieves false discovery rates close to the nominal level, false non-discovery rates approaching zero, powers close to one, and average numbers of rejections that closely track the expected number of true signals. This combination of accurate false-discovery control, near-perfect signal recovery, and rejection counts closely matching the expected number of true signals provides a compelling explanation for the substantial reductions in normalized misclassification risk reported in the previous subsection. Such a combination of properties is rarely observed in large-scale multiple testing problems. These findings suggest that covariance-adaptive residualization combined with stagewise GBS calibration can recover the underlying signals with remarkably high accuracy across a broad range of dependent testing environments.

\section{Discussion}\label{SECTION_DISCUSSION}

In this paper, we investigated a new calibration strategy for the Maximum Residual Down (MRD) procedure under arbitrary covariance dependence. The proposed methodology combines two complementary components: covariance-adaptive residualization through the MRD framework of \citet{COHEN_SACK_XU_2009} and stagewise calibration through the critical constants of \citet{GBS2009}. The resulting procedure preserves the dependence-aware residual structure of MRD while providing a simple, systematic, and computationally efficient calibration mechanism that can be applied across a broad range of covariance models. Taken together, these features yield a practically implementable covariance-adaptive multiple-testing procedure that substantially improves empirical operating characteristics under several dependence structures.

%In this paper, we investigated a new calibration strategy for the
%Maximum Residual Down (MRD) procedure under arbitrary covariance
%dependence. The proposed methodology combines two complementary
%ideas: covariance-adaptive residualization through the MRD framework \citet{COHEN_SACK_XU_2009}
%and stagewise calibration through the critical constants of \citet{GBS2009}. The resulting procedure preserves the dependence-aware residual structure of MRD while providing a simple and systematic calibration mechanism that is applicable across a broad range of covariance models. Importantly, the proposed calibration remains within the class of admissible residual-based step-down procedures studied by \citet{GC_ADMISSIBILITY_2026}, so that the resulting methodology retains a formal decision-theoretic justification rather than constituting an ad hoc recalibration of the original MRD framework.

From a theoretical perspective, the proposed procedure inherits the admissibility properties of the residual-based step-down framework under the vector loss function \ref{VECTOR_LOSS}. This follows directly from the general theory developed by \citet{GC_ADMISSIBILITY_2026}, since the proposed calibration modifies only the critical constants and leaves the geometric structure underlying the admissibility of residual-based step-down procedures unchanged.

A second contribution of the paper concerns computation. We derived an alternative representation of the MRD residual statistics that expresses all active residuals at a given stage in terms of a single active precision matrix. This representation reduces the worst-case number of
matrix inversions from $m(m+1)/2$ to at most $m$ and thereby
substantially simplifies practical implementation of the procedure. Importantly, this reduction is achieved without approximation: the resulting statistics are algebraically identical to the original MRD residuals. Thus, the computational simplification preserves the exact
testing procedure while substantially reducing the burden of
implementation. 

%More broadly, the alternative representation reveals that the entire residual-based testing mechanism can be viewed through the geometry of the active precision matrix, thereby providing both a computationally efficient implementation and a more transparent structural interpretation of the MRD framework.

The alternative representation also provides a useful geometric interpretation of the MRD statistics. At each stage, the statistic associated with an active hypothesis may be viewed as a normalized projection of the active data vector onto a direction determined by a column of the active precision matrix. From this perspective, the sequential residualization mechanism operates through a collection of evolving precision-matrix directions, making explicit the role of covariance geometry in both the construction of the testing statistics and the propagation of information across hypotheses.

%The simulation study reveals several interesting empirical findings.
%Across a variety of covariance structures, including equicorrelation,
%factor, Toeplitz, fractional Gaussian noise, sparse precision-matrix and heterogeneous block
%models, the GBS-calibrated MRD procedure frequently achieves the
%smallest normalized misclassification risk in sparse and moderately sparse regimes. The gains arise from a favorable balance between false discoveries and missed discoveries, yielding simultaneously strong power, low false non-discovery rates, and numbers of rejections that often closely track the expected number of true signals.

The simulation study reveals several important empirical findings. Across a variety of covariance structures, including equicorrelation, factor, Toeplitz, fractional Gaussian noise, heterogeneous block, and sparse precision-matrix models, the GBS-calibrated MRD procedure frequently achieves the smallest normalized misclassification risk in sparse and moderately sparse regimes. These gains arise from a favorable balance between false discoveries and missed discoveries, with the procedure often exhibiting strong power, low false non-discovery rates, and average numbers of rejections close to the expected number of true signals.

%An additional observation concerns the distinction between the original MRD procedure and its GBS-calibrated counterpart. Although both procedures employ the same covariance-adaptive residualization mechanism, their empirical behavior differs substantially across sparsity regimes. The proposed calibration often yields stronger overall signal-recovery behavior under structured dependence models, indicating that the empirical near-support-recovery phenomenon cannot be attributed solely to residualization. Rather, the simulation results suggest that calibration plays an important role in determining how information accumulated through sequential residualization is translated into testing decisions. In particular, the findings suggest that near-support recovery may arise from a nontrivial interaction among covariance-adaptive residualization, stagewise calibration, sparsity, and the geometry of the underlying dependence structure. Understanding this interaction theoretically remains an important open problem.

An important conclusion emerging from the simulation study is that the empirical operating characteristics of covariance-adaptive multiple testing procedures cannot be attributed solely to the residualization mechanism. Although the original MRD procedure and its GBS-calibrated counterpart employ identical covariance-adaptive residual statistics, their empirical behavior differs substantially across sparsity regimes and dependence structures. The proposed GBS calibration often yields stronger overall signal-recovery performance under structured dependence models, indicating that stagewise calibration itself plays an essential role in determining the overall balance between false discoveries and missed discoveries. More broadly, the simulation results suggest that the empirical near-support-recovery phenomenon arises from a nontrivial interaction among covariance-adaptive residualization, stagewise calibration, sparsity, and the geometry of the underlying dependence structure. Developing a rigorous theoretical understanding of this interaction remains an important direction for future research.

An interesting development that emerged after completion of the present work is our recent Bayesian investigation of covariance-aware multiple testing under one-factor dependence \citep{GhoshChakrabartiBSD2026}. In that work, the proposed MRD--GBS procedure was found to attain empirical Bayes risks remarkably close to the corresponding Bayes Oracle benchmark across a broad range of problem dimensions and sparsity levels. Since the Bayes Oracle represents the fundamental decision-theoretic benchmark for the underlying multiple testing problem, these findings provide strong independent empirical evidence that covariance-adaptive residualization combined with GBS stagewise calibration can operate remarkably close to the optimal Bayesian procedure under structured dependence.

These empirical observations are particularly intriguing when viewed alongside recent theoretical developments concerning the GBS procedure itself. Under independence, the GBS calibration has recently been shown to attain asymptotic Bayes optimality under sparsity \citep{GhoshChakrabarti2026GBSABOS} and, under suitable sparsity conditions, sharp asymptotic minimaxity in sparse Gaussian multiple testing problems \citep{Ghosh2026GBSMinimaxity}. Although these results do not directly extend to dependent settings, taken together with the near-Bayes-Oracle behavior observed under one-factor dependence, they suggest that covariance-adaptive residualization combined with principled stagewise calibration may possess substantially stronger optimality properties than are presently understood.

%An additional observation concerns the distinction between the original MRD procedure and its GBS-calibrated counterpart. Although both procedures employ the same covariance-adaptive residualization mechanism, the strongest signal-recovery behavior under the structured dependence models is observed under the GBS calibration. This suggests that the empirical support-recovery phenomenon cannot be attributed solely to residualization. Rather, the simulation results indicate that calibration itself may play an important role in determining how effectively information accumulated through the sequential residualization process is translated into accurate recovery of the underlying signal set. In particular, the findings suggest that support recovery may arise from a nontrivial interaction between covariance-adaptive residualization, stagewise calibration, and the geometry of the underlying dependence structure. Understanding this interaction theoretically remains an important open problem.

From a broader multiple-testing perspective, the most important empirical comparison is not between the calibrated MRD procedure and the original MRD procedure, but between covariance-aware residual-based procedures and the widely used class of marginal \(p\)-value based methods represented by BH and its variants. Across several dependence structures, the proposed methodology substantially improves classification accuracy relative to these procedures. This suggests that explicitly incorporating covariance information into the construction of the testing statistics may provide benefits beyond those obtainable through dependence-adjusted calibration of marginal \(p\)-values alone. More broadly, the numerical results indicate that dependence should not merely be viewed as a complication requiring adjustment, but as a potential source of information that can be exploited to improve multiple-testing performance. The observed gains suggest that appropriately incorporating covariance structure into the testing statistics themselves may substantially enhance classification accuracy and signal recovery beyond what can be achieved through marginal procedures alone.

An important practical advantage of the proposed methodology is that it replaces the model-dependent critical values required by the original MRD procedure with a simple and systematic stagewise calibration rule. While the original MRD procedure may occasionally achieve smaller misclassification risks in certain denser signal regimes, its implementation requires threshold constants that depend on the underlying covariance model and were originally selected through numerical experimentation. In contrast, the proposed GBS calibration is fully automatic and requires no model-specific tuning. The strong empirical performance of the calibrated procedure across a broad collection of dependence structures therefore suggests that substantial gains in operating characteristics can be achieved without sacrificing simplicity or generality of implementation.

%An important practical advantage of the proposed calibration is that it replaces the model-dependent critical values required by the original MRD procedure with a simple and systematic stagewise calibration rule. Although the original MRD procedure may occasionally achieve smaller misclassification risks in certain denser signal regimes, its implementation requires specification of threshold constants that depend on the underlying covariance model and were originally selected through numerical experimentation. In contrast, the proposed GBS calibration is fully automatic and requires no model-specific tuning. Viewed from this perspective, the strong empirical performance of the calibrated procedure relative to BH, adaptive BH, and GBS is particularly encouraging, since it is achieved using a unified calibration mechanism applicable across a broad collection of dependence structures.

One of the most intriguing observations concerns the apparent role of covariance geometry. The benefits of residual-based calibration are most pronounced under dependence structures exhibiting global, long-range, or block-wise correlation, whereas substantially weaker gains are observed under sparse precision-matrix models. These findings suggest that dependence may influence multiple testing performance through mechanisms extending beyond traditional error-rate considerations. In particular, the results raise the possibility that
the effectiveness of residual-based procedures depends on the extent to which information can propagate through the underlying dependence structure.

A particularly interesting theoretical question is whether the empirical near-support-recovery behavior observed in the simulations can be justified asymptotically and, if so, to characterize classes of covariance matrices \(\boldsymbol{\Sigma}\) under which the proposed MRD--GBS procedure achieves asymptotically exact support recovery as \(m\to\infty\). The simulation results suggest that such behavior may occur under covariance structures exhibiting substantial global, long-range, or block-wise dependence, whereas considerably weaker signal recovery is observed under sparse precision-matrix models. This raises the possibility that covariance geometry may fundamentally influence the difficulty of support recovery through its effect on information propagation in sequential residual-based testing procedures.

More specifically, one may ask whether there exist structural conditions on \(\boldsymbol{\Sigma}\) together with suitable sparsity and signal-strength
assumptions under which

\[P(\widehat S = S_0)\to 1,\]
where $\widehat S$ denotes the set of rejected hypotheses and $S_0$ denotes the true signal set. Understanding whether such phenomena arise primarily from the residualization mechanism, from the stagewise GBS calibration, or from an interaction between the two remains an important open problem.

%A particularly interesting theoretical question is to characterize
%classes of covariance matrices \(\boldsymbol{\Sigma}\) under which residual-based
%step-down procedures achieve asymptotically exact support recovery
%as $m \to \infty$. More specifically, one may ask whether there exist
%structural conditions on \(\boldsymbol{\Sigma}\) together with suitable sparsity and
%signal-strength assumptions under which the probability of recovering
%the true signal set converges to one. Understanding how covariance
%geometry interacts with sparsity and signal strength in such regimes
%appears to be an important open problem.

%A particularly interesting theoretical question is to characterize classes of
%covariance matrices \(\boldsymbol{\Sigma}\) under which residual-based step-down procedures
%can recover the true signal set with high probability in sparse asymptotic regimes.
%Such a theory would require understanding how the geometry of \(\boldsymbol{\Sigma}\), the
%signal strength, and the sparsity level jointly determine the propagation of
%information through the sequential residualization mechanism.

The influence of covariance geometry may extend beyond statistical performance alone. The alternative precision-matrix representation developed in Section~\ref{SECTION_MRD_STAT_ALTERNATIVE} suggests that dependence structure can also have important computational consequences. Certain covariance models admit particularly simple updates of the active precision matrix throughout the sequential elimination process, leading to substantial reductions in computational complexity. More broadly, these observations suggest that covariance structure may simultaneously govern both the statistical difficulty of a multiple-testing problem and the computational complexity of implementing covariance-adaptive testing procedures.

The present work is primarily methodological, and several important questions remain open. First, although the proposed calibration is motivated by the generalized step-down critical constants originally developed under independence, no formal false discovery rate control result is established here under arbitrary dependence. An important theoretical question is whether the proposed MRD--GBS procedure possesses asymptotic FDR control properties under suitable sparsity and dependence regimes. More broadly, recent work has established that the GBS procedure attains asymptotic Bayes optimality under sparsity \citep{GhoshChakrabarti2026GBSABOS} and, under suitable sparsity conditions, sharp asymptotic minimaxity in sparse Gaussian multiple testing problems \citep{Ghosh2026GBSMinimaxity}. These developments naturally raise the question of whether analogous optimality properties continue to hold after incorporating covariance-adaptive residualization. The encouraging empirical performance observed throughout the present work suggests that such theoretical guarantees may indeed be attainable, but establishing them remains an important challenge.

A second important direction concerns the development of a deeper theoretical understanding of the empirical phenomena observed in this paper. The simulation results suggest a close relationship among covariance structure, stagewise calibration, information propagation, signal recovery, and overall classification accuracy. In particular, the crossover behavior observed between the original MRD procedure and its GBS-calibrated counterpart indicates that the relative merits of different calibration strategies may depend in subtle ways on the underlying sparsity regime and the geometry of the covariance structure. Developing a theoretical framework capable of explaining these interactions would provide valuable insight into the mechanisms through which dependence information influences sequential multiple testing procedures.

The present work also suggests several methodological extensions. It would be worthwhile to investigate alternative covariance-adaptive calibration schemes that incorporate structural characteristics of the covariance or precision matrix directly into the choice of stagewise critical values. Furthermore, the proposed methodology assumes that the covariance structure is known. While this assumption is standard throughout much of the literature on multiple testing under dependence, extending covariance-adaptive residual-based procedures to settings with unknown covariance structures remains an important methodological challenge. A substantial literature now exists on high-dimensional covariance and precision-matrix estimation, including shrinkage estimators, thresholding methods, and sparse graphical modeling approaches \citep{LedoitWolf2004,BickelLevina2008,CaiLiu2011,FriedmanHastieTibshirani2008}. Understanding how covariance-estimation uncertainty propagates through the sequential residualization mechanism and influences calibration, admissibility, and overall operating characteristics constitutes an important direction for future research.

From a practical perspective, it would be of considerable interest to investigate the performance of the proposed methodology on large-scale scientific datasets arising in areas such as genomics, neuroimaging, and other high-dimensional applications where complex dependence structures naturally arise. Such empirical investigations would complement the extensive simulation study presented here while also providing insight into computational scalability, software implementation, and the practical advantages of covariance-adaptive residualization under realistic dependence patterns.

Finally, it would be of considerable interest to extend the present framework beyond the Gaussian setting and investigate whether analogous residual-based representations and covariance-adaptive calibration principles continue to hold under more general classes of dependence models. We hope that the perspective developed in this paper provides a useful foundation for these and related future investigations.

Taken together, the theoretical and empirical findings of the present work suggest that covariance-adaptive residualization and stagewise calibration play fundamentally different but complementary roles in dependent multiple testing. Residualization determines how dependence information is accumulated and propagated through the sequential testing process, whereas stagewise calibration determines how that accumulated information is translated into testing decisions. Beyond providing a practical calibration strategy for the MRD framework, the present work raises broader questions concerning dependence, information propagation, and statistical complexity in large-scale multiple testing problems. The substantial reductions in misclassification risk achieved by covariance-adaptive residualization relative to several widely used marginal testing procedures suggest that the effective complexity of a multiple testing problem may depend not only on its nominal multiplicity but also on the extent to which dependence information can be exploited by the testing procedure itself. A deeper understanding of this relationship may ultimately lead to new covariance-adaptive calibration principles that extend beyond traditional multiplicity adjustments.

\appendix

\section*{Appendix A: Proofs of Technical Results}\label{SECTION_APPENDIX_PROOFS}

This appendix contains the proofs of the technical results presented in Section~\ref{SECTION_MRD_STAT_ALTERNATIVE}. We first establish a useful block-inverse identity connecting the active covariance matrix and its corresponding precision matrix. This identity is then used to derive the alternative precision-matrix representation of the MRD residual statistics developed in Theorem~\ref{THM_MRD_PRECISION_REP}. The resulting representation forms the basis of the computational simplifications described in Corollaries~\ref{COR_ONE_INVERSE_PER_STAGE} and \ref{COR_COMPLEXITY_REDUCTION}.

%Proofs of some of the results of this paper make use of the following important results from theory of matrices. Of them the first one is the celebrated Sherman$-$Morrison$-$Woodbury (SMW) identity, while the other one provides an important formula for obtaining the inverse of a $2\times 2$ partitioned matrix. See, for example, \cite{SANT_WILL_NOTZ} among many other sources.\newline
%% % % of this thesis made use of the following celebrated result from matrix algebra, popularly known as The .\newline
%
%\begin{lem}\label{SWM_IDENTITY}
%	Suppose that B is any $m \times m$ nonsingular matrix, C is a $r \times r$ non-singular matrix, and A is an arbitrary $n \times r$ matrix such that $(A^{T}B^{-1}A+C)^{-1}$ is nonsingular. Then $(B+A^{T}C^{-1}A)$ is $m \times m$ non-singular with inverse given by,
%	\begin{equation}
%		(B+A^{T}C^{-1}A)^{-1} = B^{-1} - B^{-1}A(A^{T}B^{-1}A+C)^{-1}A^{T}B^{-1}.\nonumber
%	\end{equation}   
%\end{lem}
%
%\begin{lem}\label{INV_PART_MAT}
%	Suppose that $B$ is a $m \times m$ non-singular matrix and
%	\begin{eqnarray}
%		T
%		&=&
%		\begin{pmatrix}
%			D & A^{T}\\ 
%			A & B
%		\end{pmatrix},\nonumber
%	\end{eqnarray}
%	where $D$ is $m\times m$ and $A$ is $n\times m$. Then $T$ is non-singular if and only if
%	\begin{equation}
%		Q=D-A^{T}B^{-1}A\nonumber
%	\end{equation}
%	is non-singular. In this case, $T^{-1}$ is given by
%	\begin{eqnarray}
%		T^{-1}
%		&=&
%		\begin{pmatrix}
%			Q^{-1} & -Q^{-1}A^{T}B^{-1}\\ 
%			-B^{-1}AQ^{-1} & B^{-1}+B^{-1}AQ^{-1}A^{T}B^{-1}
%		\end{pmatrix}.\nonumber
%	\end{eqnarray}
%\end{lem} 

\begin{flushleft}
	\textbf{Proof of Lemma \ref{LEM_BLOCK_INVERSE_MRD}:}
\end{flushleft}

\begin{proof}
Without loss of generality, we order the coordinates in the index set \(A\) so that the coordinate \(j\) appears first. Thus, we may write
\[
\boldsymbol{\Sigma}_A =
\begin{pmatrix}
	\sigma_{jj} & \boldsymbol{\sigma}_{j,A_{-j}}^{T} \\
	\boldsymbol{\sigma}_{j,A_{-j}} & \boldsymbol{\Sigma}_{A_{-j}}
\end{pmatrix}.
\]	
	
Let
\[
	S_j=\sigma_{jj}-\boldsymbol{\sigma}_{j,A_{-j}}^{T}\boldsymbol{\Sigma}_{A_{-j}}^{-1}\boldsymbol{\sigma}_{j,A_{-j}}
	\]
	denote the Schur complement of \(\boldsymbol{\Sigma}_{A_{-j}}\) in \(\boldsymbol{\Sigma}_A\). Since
	\(\boldsymbol{\Sigma}_A\) is positive definite, \(S_j>0\). By the standard inverse formula for a partitioned matrix (or equivalently, the Schur complement formula), we obtain
	\[
	\boldsymbol{\Sigma}_A^{-1}
	=
	\begin{pmatrix}
		S_j^{-1}
		&
		- S_j^{-1}\boldsymbol{\sigma}_{j,A_{-j}}^{T}\boldsymbol{\Sigma}_{A_{-j}}^{-1}
		\\[4pt]
		-\boldsymbol{\Sigma}_{A_{-j}}^{-1}\boldsymbol{\sigma}_{j,A_{-j}} S_j^{-1}
		&
		\boldsymbol{\Sigma}_{A_{-j}}^{-1}
		+
		\boldsymbol{\Sigma}_{A_{-j}}^{-1}\boldsymbol{\sigma}_{j,A_{-j}}
		S_j^{-1}
		\boldsymbol{\sigma}_{j,A_{-j}}^{T}\boldsymbol{\Sigma}_{A_{-j}}^{-1}
	\end{pmatrix}.
	\]
	Since \(\boldsymbol{B}_A=\boldsymbol{\Sigma}_A^{-1}\), the upper-left entry of this matrix gives
	\[
	b^A_{jj}=S_j^{-1}
	=
	\left(
	\sigma_{jj}
	-
	\boldsymbol{\sigma}_{j,A_{-j}}^{T}\boldsymbol{\Sigma}_{A_{-j}}^{-1}\boldsymbol{\sigma}_{j,A_{-j}}
	\right)^{-1}.
	\]
	Similarly, the lower-left block gives
	\[
	\boldsymbol{b}^A_{A_{-j},j}
	=
	-\boldsymbol{\Sigma}_{A_{-j}}^{-1}\boldsymbol{\sigma}_{j,A_{-j}}S_j^{-1}
	=
	-b^A_{jj}\boldsymbol{\Sigma}_{A_{-j}}^{-1}\boldsymbol{\sigma}_{j,A_{-j}},
	\]
	which proves the second identity. This completes the proof of Lemma \ref{LEM_BLOCK_INVERSE_MRD}.
\end{proof}

%\begin{proof}
%The result follows from the standard inverse formula for a partitioned matrix.
%	Indeed, the Schur complement of \(\Sigma_{A_{-j}}\) in \(\Sigma_A\) is
%	\[
%	\sigma_{jj}
%	-
%	\sigma_{j,A_{-j}}^{T}
%	\Sigma_{A_{-j}}^{-1}
%	\sigma_{j,A_{-j}}.
%	\]
%	The \((j,j)\)-entry of \(\Sigma_A^{-1}\) is the inverse of this Schur
%	complement, which gives the first identity. The expression for
%	\(b^{A}_{A_{-j},j}\) follows from the corresponding off-diagonal block in the
%	partitioned inverse formula.
%\end{proof}

\begin{flushleft}
	\textbf{Proof of Theorem \ref{THM_MRD_PRECISION_REP}:}
\end{flushleft}

\begin{proof}
	Fix a stage \(t\) and an active coordinate \(j\in \mathcal{A}_{t}\). For notational
	simplicity, write
	\[
	A=\mathcal{A}_{t},
	\qquad
	A_{-j}=A\setminus\{j\}.
	\]
	The original MRD residual statistic is
	\[
	U^{(i_1,\ldots,i_{t-1})}_{tj}(\boldsymbol{x})
	=
	\frac{
		x_j
		-
		\boldsymbol{\sigma}_{j,A_{-j}}^{T}
		\boldsymbol{\Sigma}_{A_{-j}}^{-1}
		\boldsymbol{x}_{A_{-j}}
	}{
		\sigma_{j\cdot A_{-j}}^{1/2}
	},
	\]
	where
	\[
	\sigma_{j\cdot A_{-j}}
	=
	\sigma_{jj}
	-
	\boldsymbol{\sigma}_{j,A_{-j}}^{T}
	\boldsymbol{\Sigma}_{A_{-j}}^{-1}
	\boldsymbol{\sigma}_{j,A_{-j}}
	\]
	is the conditional variance of \(X_j\) given \(\boldsymbol{X}_{A_{-j}}\).
	
	By Lemma~\ref{LEM_BLOCK_INVERSE_MRD},
	\[
	b^A_{jj}
	=
	\frac{1}{\sigma_{j\cdot A_{-j}}},
	\qquad
	\boldsymbol{b}^A_{A_{-j},j}
	=
	-
	b^A_{jj}
	\boldsymbol{\Sigma}_{A_{-j}}^{-1}
	\boldsymbol{\sigma}_{j,A_{-j}}.
	\]
	Now the inner product of the \(j\)-th column of \(\boldsymbol{B}_A=\boldsymbol{\Sigma}_A^{-1}\) with the active
	data vector \(\boldsymbol{x}_A\) can be decomposed as
	\[
	\sum_{k\in A} b^A_{kj}x_k
	=
	b^A_{jj}x_j
	+
	\big(\boldsymbol{b}^A_{A_{-j},j}\big)^T
	\boldsymbol{x}_{A_{-j}}.
	\]
	Substituting the expression for \(\boldsymbol{b}^A_{A_{-j},j}\), we obtain
	\[
	\sum_{k\in A} b^A_{kj}x_k
	=
	b^A_{jj}
	\left(
	x_j
	-
	\boldsymbol{\sigma}_{j,A_{-j}}^{T}
	\boldsymbol{\Sigma}_{A_{-j}}^{-1}
	\boldsymbol{x}_{A_{-j}}
	\right).
	\]
	Since \(b^A_{jj}=1/\sigma_{j\cdot A_{-j}}\), it follows that
	\[
	\frac{
		\sum_{k\in A} b^A_{kj}x_k
	}{
		\sqrt{b^A_{jj}}
	}
	=
	\sqrt{b^A_{jj}}
	\left(
	x_j
	-
	\boldsymbol{\sigma}_{j,A_{-j}}^{T}
	\boldsymbol{\Sigma}_{A_{-j}}^{-1}
	\boldsymbol{x}_{A_{-j}}
	\right)
	=
	\frac{
		x_j
		-
		\boldsymbol{\sigma}_{j,A_{-j}}^{T}
		\boldsymbol{\Sigma}_{A_{-j}}^{-1}
		\boldsymbol{x}_{A_{-j}}
	}{
		\sigma_{j\cdot A_{-j}}^{1/2}
	}.
	\]
	The final expression is precisely
	\[
	U^{(i_1,\ldots,i_{t-1})}_{tj}(\boldsymbol{x}).
	\]
	Therefore,
	\[
	U^{(i_1,\ldots,i_{t-1})}_{tj}(\boldsymbol{x})
	=
	\frac{
		\sum_{k\in A} b^A_{kj}x_k
	}{
		\sqrt{b^A_{jj}}
	},
	\]
	as claimed. This completes the proof of Theorem \ref{THM_MRD_PRECISION_REP}.
\end{proof}

\section*{Appendix B: Detailed Simulation Results for $m=200$}
\label{APPENDIX_SIMULATION_m200}

This Appendix Contains supplementary numerical results from the simulation study described in Section~\ref{SECTION_SIMULATIONS}. The tables provide detailed performance summaries for all competing procedures across the dependence structures considered in this paper and allow a more complete assessment of their classification and signal-recovery behavior. These results complement and reinforce the principal findings reported in the main text.

\subsection*{Normalized misclassification rates ($m=200$)}

The following tables provide the numerical summaries corresponding to the $m=200$ simulations discussed in Section~\ref{SECTION_SIMULATIONS}. For the normalized misclassification-rate tables, smaller values indicate superior overall classification performance, and boldfaced entries mark the smallest value in
each row.

\begin{table}[!htbp]
	\caption{Normalized misclassification rates under the equicorrelation model with $m=200$, $\alpha=0.1$, and 3000 Monte Carlo replications. Smaller values indicate superior classification performance. Boldfaced entries correspond to the smallest normalized misclassification rates.}
	\label{tab:risk_equicorrelation_m200}
	\centering
	\small
	\begin{tabular}{lccccc}
		\toprule
		$p$ & BH & Storey-BH & GBS & MRD-CSX & MRD-GBS \\
		\midrule
		0.01  & 0.0262 & 0.1272 & 0.1501 & 0.0117 & \textbf{0.0016} \\
		0.03  & 0.0343 & 0.1331 & 0.1599 & 0.0131 & \textbf{0.0039} \\
		0.05  & 0.0400 & 0.1301 & 0.1596 & 0.0134 & \textbf{0.0062} \\
		0.075 & 0.0423 & 0.1348 & 0.1597 & 0.0131 & \textbf{0.0089} \\
		0.10  & 0.0483 & 0.1428 & 0.1685 & 0.0129 & \textbf{0.0117} \\
		0.15  & 0.0577 & 0.1451 & 0.1591 & \textbf{0.0126} & 0.0171 \\
		0.20  & 0.0609 & 0.1355 & 0.1480 & \textbf{0.0121} & 0.0226 \\
		0.25  & 0.0658 & 0.1421 & 0.1536 & \textbf{0.0118} & 0.0279 \\
		0.30  & 0.0677 & 0.1337 & 0.1426 & \textbf{0.0115} & 0.0336 \\
		0.35  & 0.0737 & 0.1437 & 0.1501 & \textbf{0.0112} & 0.0389 \\
		\bottomrule
	\end{tabular}
\end{table}

\begin{table}[!htbp]
	\caption{Normalized misclassification rates under the factor model with $m=200$, $\alpha=0.1$, and 3000 Monte Carlo replications. Smaller values indicate superior classification performance. Boldfaced entries correspond to the smallest normalized misclassification rates.}
	\label{tab:risk_factor_m200}
	\centering
	\small
	\begin{tabular}{lccccc}
		\toprule
		$p$ & BH & Storey-BH & GBS & MRD-CSX & MRD-GBS \\
		\midrule
		0.01  & 0.0097 & 0.0140 & 0.0114 & 0.0118 & \textbf{0.0016} \\
		0.03  & 0.0187 & 0.0247 & 0.0199 & 0.0131 & \textbf{0.0039} \\
		0.05  & 0.0261 & 0.0330 & 0.0283 & 0.0134 & \textbf{0.0063} \\
		0.075 & 0.0317 & 0.0381 & 0.0339 & 0.0130 & \textbf{0.0088} \\
		0.10  & 0.0385 & 0.0463 & 0.0413 & 0.0130 & \textbf{0.0117} \\
		0.15  & 0.0472 & 0.0555 & 0.0495 & \textbf{0.0127} & 0.0171 \\
		0.20  & 0.0540 & 0.0628 & 0.0568 & \textbf{0.0124} & 0.0225 \\
		0.25  & 0.0595 & 0.0688 & 0.0625 & \textbf{0.0120} & 0.0278 \\
		0.30  & 0.0635 & 0.0727 & 0.0677 & \textbf{0.0118} & 0.0331 \\
		0.35  & 0.0677 & 0.0780 & 0.0730 & \textbf{0.0114} & 0.0383 \\
		\bottomrule
	\end{tabular}
\end{table}

\begin{table}[!htbp]
	\caption{Normalized misclassification rates under the fractional Gaussian noise model with $m=200$, $\alpha=0.1$, and 3000 Monte Carlo replications. Smaller values indicate superior classification performance. Boldfaced entries correspond to the smallest normalized misclassification rates.}
	\label{tab:risk_fractional_gaussian_noise_m200}
	\centering
	\small
	\begin{tabular}{lccccc}
		\toprule
		$p$ & BH & Storey-BH & GBS & MRD-CSX & MRD-GBS \\
		\midrule
		0.01  & 0.0106 & 0.0255 & 0.0184 & 0.0111 & \textbf{0.0016} \\
		0.03  & 0.0192 & 0.0368 & 0.0281 & 0.0125 & \textbf{0.0040} \\
		0.05  & 0.0262 & 0.0427 & 0.0341 & 0.0127 & \textbf{0.0062} \\
		0.075 & 0.0319 & 0.0467 & 0.0396 & 0.0128 & \textbf{0.0089} \\
		0.10  & 0.0378 & 0.0545 & 0.0471 & 0.0131 & \textbf{0.0119} \\
		0.15  & 0.0486 & 0.0657 & 0.0576 & \textbf{0.0143} & 0.0181 \\
		0.20  & 0.0544 & 0.0695 & 0.0636 & \textbf{0.0163} & 0.0246 \\
		0.25  & 0.0600 & 0.0771 & 0.0687 & \textbf{0.0192} & 0.0311 \\
		0.30  & 0.0642 & 0.0787 & 0.0740 & \textbf{0.0257} & 0.0405 \\
		0.35  & 0.0696 & 0.0856 & 0.0799 & \textbf{0.0327} & 0.0487 \\
		\bottomrule
	\end{tabular}
\end{table}

\begin{table}[!htbp]
	\caption{Normalized misclassification rates under the Toeplitz model with $m=200$, $\alpha=0.1$, and 3000 Monte Carlo replications. Smaller values indicate superior classification performance. Boldfaced entries correspond to the smallest normalized misclassification rates.}
	\label{tab:risk_toeplitz_m200}
	\centering
	\small
	\begin{tabular}{lccccc}
		\toprule
		$p$ & BH & Storey-BH & GBS & MRD-CSX & MRD-GBS \\
		\midrule
		0.01  & 0.0161 & 0.0570 & 0.0540 & 0.0112 & \textbf{0.0016} \\
		0.03  & 0.0253 & 0.0663 & 0.0659 & 0.0124 & \textbf{0.0037} \\
		0.05  & 0.0310 & 0.0720 & 0.0690 & 0.0125 & \textbf{0.0058} \\
		0.075 & 0.0357 & 0.0709 & 0.0700 & 0.0124 & \textbf{0.0084} \\
		0.10  & 0.0414 & 0.0785 & 0.0773 & 0.0125 & \textbf{0.0113} \\
		0.15  & 0.0523 & 0.0903 & 0.0833 & \textbf{0.0131} & 0.0172 \\
		0.20  & 0.0567 & 0.0923 & 0.0886 & \textbf{0.0145} & 0.0239 \\
		0.25  & 0.0624 & 0.0990 & 0.0958 & \textbf{0.0177} & 0.0316 \\
		0.30  & 0.0655 & 0.0972 & 0.0949 & \textbf{0.0242} & 0.0427 \\
		0.35  & 0.0712 & 0.1050 & 0.1031 & \textbf{0.0334} & 0.0555 \\
		\bottomrule
	\end{tabular}
\end{table}

\begin{table}[!htbp]
	\caption{Normalized misclassification rates under the heterogeneous block model with $m=200$, $\alpha=0.1$, and 3000 Monte Carlo replications. Smaller values indicate superior classification performance. Boldfaced entries correspond to the smallest normalized misclassification rates.}
	\label{tab:risk_heterogeneous_block_m200}
	\centering
	\small
	\begin{tabular}{lccccc}
		\toprule
		$p$ & BH & Storey-BH & GBS & MRD-CSX & MRD-GBS \\
		\midrule
		0.01  & 0.0070 & 0.0075 & 0.0071 & 0.0116 & \textbf{0.0024} \\
		0.03  & 0.0156 & 0.0163 & 0.0157 & 0.0137 & \textbf{0.0054} \\
		0.05  & 0.0233 & 0.0239 & 0.0241 & 0.0144 & \textbf{0.0081} \\
		0.075 & 0.0297 & 0.0306 & 0.0299 & 0.0147 & \textbf{0.0110} \\
		0.10  & 0.0353 & 0.0364 & 0.0353 & 0.0153 & \textbf{0.0141} \\
		0.15  & 0.0453 & 0.0468 & 0.0458 & \textbf{0.0164} & 0.0198 \\
		0.20  & 0.0519 & 0.0537 & 0.0528 & \textbf{0.0174} & 0.0255 \\
		0.25  & 0.0584 & 0.0609 & 0.0589 & \textbf{0.0188} & 0.0306 \\
		0.30  & 0.0639 & 0.0672 & 0.0650 & \textbf{0.0222} & 0.0380 \\
		0.35  & 0.0682 & 0.0721 & 0.0713 & \textbf{0.0278} & 0.0471 \\
		\bottomrule
	\end{tabular}
\end{table}

\begin{table}[!htbp]
	\caption{Normalized misclassification rates under the sparse precision model with $m=200$, $\alpha=0.1$, and 3000 Monte Carlo replications. Smaller values indicate superior classification performance. Boldfaced entries correspond to the smallest normalized misclassification rates.}
	\label{tab:risk_sparse_precision_m200}
	\centering
	\small
	\begin{tabular}{lccccc}
		\toprule
		$p$ & BH & Storey-BH & GBS & MRD-CSX & MRD-GBS \\
		\midrule
		0.01  & 0.0062 & 0.0062 & 0.0065 & 0.0112 & \textbf{0.0033} \\
		0.03  & 0.0150 & 0.0150 & 0.0150 & 0.0145 & \textbf{0.0075} \\
		0.05  & 0.0221 & 0.0221 & 0.0227 & 0.0164 & \textbf{0.0115} \\
		0.075 & 0.0290 & 0.0290 & 0.0290 & 0.0189 & \textbf{0.0161} \\
		0.10  & 0.0346 & 0.0345 & 0.0344 & 0.0219 & \textbf{0.0207} \\
		0.15  & 0.0444 & 0.0441 & 0.0443 & \textbf{0.0288} & 0.0301 \\
		0.20  & 0.0526 & 0.0521 & 0.0522 & \textbf{0.0379} & 0.0400 \\
		0.25  & 0.0582 & 0.0580 & 0.0580 & \textbf{0.0487} & 0.0492 \\
		0.30  & 0.0633 & 0.0639 & 0.0640 & 0.0622 & \textbf{0.0599} \\
		0.35  & \textbf{0.0673} & 0.0684 & 0.0688 & 0.0765 & 0.0686 \\
		\bottomrule
	\end{tabular}
\end{table}

\clearpage

\subsection*{False discovery rates ($m=200$)}

The next set of tables reports empirical false discovery rates. Since an
overly conservative procedure may have very small FDR at the expense of missed
signals, boldfaced entries identify the values closest to the nominal level
$\alpha=0.10$, rather than the smallest values.

\begin{table}[!htbp]
	\caption{False discovery rates under the equicorrelation model with $m=200$, $\alpha=0.1$, and 3000 Monte Carlo replications. Values close to the nominal level $\alpha=0.10$ indicate favorable false-discovery behavior. Boldfaced entries correspond to the values closest to the nominal level.}
	\label{tab:fdr_equicorrelation_m200}
	\centering
	\small
	\begin{tabular}{lccccc}
		\toprule
		$p$ & BH & Storey-BH & GBS & MRD-CSX & MRD-GBS \\
		\midrule
		0.01  & 0.0472 & 0.1393 & 0.1465 & 0.4468 & \textbf{0.0956} \\
		0.03  & 0.0502 & 0.1441 & 0.1514 & 0.2990 & \textbf{0.0954} \\
		0.05  & 0.0515 & 0.1408 & 0.1470 & 0.2093 & \textbf{0.0998} \\
		0.075 & 0.0487 & 0.1475 & 0.1455 & 0.1487 & \textbf{0.0993} \\
		0.10  & 0.0548 & 0.1550 & 0.1534 & 0.1135 & \textbf{0.0987} \\
		0.15  & 0.0603 & 0.1584 & 0.1435 & 0.0768 & \textbf{0.0987} \\
		0.20  & 0.0568 & 0.1484 & 0.1335 & 0.0570 & \textbf{0.0991} \\
		0.25  & 0.0592 & 0.1554 & 0.1408 & 0.0445 & \textbf{0.0984} \\
		0.30  & 0.0564 & 0.1485 & 0.1331 & 0.0365 & \textbf{0.0986} \\
		0.35  & 0.0581 & 0.1585 & 0.1442 & 0.0307 & \textbf{0.0985} \\
		\bottomrule
	\end{tabular}
\end{table}

\begin{table}[!htbp]
	\caption{False discovery rates under the factor model with $m=200$, $\alpha=0.1$, and 3000 Monte Carlo replications. Values close to the nominal level $\alpha=0.10$ indicate favorable false-discovery behavior. Boldfaced entries correspond to the values closest to the nominal level.}
	\label{tab:fdr_factor_m200}
	\centering
	\small
	\begin{tabular}{lccccc}
		\toprule
		$p$ & BH & Storey-BH & GBS & MRD-CSX & MRD-GBS \\
		\midrule
		0.01  & 0.0580 & 0.0806 & 0.0597 & 0.4536 & \textbf{0.0923} \\
		0.03  & 0.0724 & \textbf{0.1003} & 0.0748 & 0.3021 & 0.0977 \\
		0.05  & 0.0802 & 0.1098 & 0.0848 & 0.2089 & \textbf{0.0991} \\
		0.075 & 0.0774 & 0.1081 & 0.0831 & 0.1472 & \textbf{0.0984} \\
		0.10  & 0.0814 & 0.1156 & 0.0911 & 0.1133 & \textbf{0.0984} \\
		0.15  & 0.0772 & 0.1126 & 0.0892 & 0.0776 & \textbf{0.0986} \\
		0.20  & 0.0746 & 0.1143 & 0.0917 & 0.0572 & \textbf{0.0984} \\
		0.25  & 0.0697 & 0.1122 & 0.0916 & 0.0446 & \textbf{0.0978} \\
		0.30  & 0.0658 & 0.1103 & 0.0926 & 0.0366 & \textbf{0.0970} \\
		0.35  & 0.0618 & 0.1120 & 0.0948 & 0.0306 & \textbf{0.0969} \\
		\bottomrule
	\end{tabular}
\end{table}

\begin{table}[!htbp]
	\caption{False discovery rates under the fractional Gaussian noise model with $m=200$, $\alpha=0.1$, and 3000 Monte Carlo replications. Values close to the nominal level $\alpha=0.10$ indicate favorable false-discovery behavior. Boldfaced entries correspond to the values closest to the nominal level.}
	\label{tab:fdr_fractional_gaussian_noise_m200}
	\centering
	\small
	\begin{tabular}{lccccc}
		\toprule
		$p$ & BH & Storey-BH & GBS & MRD-CSX & MRD-GBS \\
		\midrule
		0.01 & 0.0754 & \textbf{0.1057} & 0.0784 & 0.4358 & 0.0903 \\
		0.03  & 0.0791 & 0.1101 & 0.0832 & 0.2892 & \textbf{0.0975} \\
		0.05  & 0.0821 & 0.1138 & 0.0876 & 0.2000 & \textbf{0.0968} \\
		0.075 & 0.0777 & 0.1085 & 0.0852 & 0.1437 & \textbf{0.0970} \\
		0.10  & 0.0797 & 0.1152 & 0.0905 & 0.1121 & \textbf{0.0980} \\
		0.15  & 0.0820 & 0.1182 & 0.0955 & 0.0813 & \textbf{0.1005} \\
		0.20  & 0.0750 & 0.1118 & 0.0941 & 0.0666 & \textbf{0.1027} \\
		0.25 & 0.0721 & 0.1172 & \textbf{0.0968} & 0.0585 & 0.1037 \\
		0.30  & 0.0664 & 0.1116 & \textbf{0.0966} & 0.0588 & 0.1104 \\
		0.35  & 0.0657 & 0.1170 & \textbf{0.1013} & 0.0593 & 0.1143 \\
		\bottomrule
	\end{tabular}
\end{table}

\begin{table}[!htbp]
	\caption{False discovery rates under the Toeplitz model with $m=200$, $\alpha=0.1$, and 3000 Monte Carlo replications. Values close to the nominal level $\alpha=0.10$ indicate favorable false-discovery behavior. Boldfaced entries correspond to the values closest to the nominal level.}
	\label{tab:fdr_toeplitz_m200}
	\centering
	\small
	\begin{tabular}{lccccc}
		\toprule
		$p$ & BH & Storey-BH & GBS & MRD-CSX & MRD-GBS \\
		\midrule
		0.01  & 0.0538 & \textbf{0.1023} & 0.0777 & 0.4372 & 0.0907 \\
		0.03  & 0.0602 & 0.1119 & 0.0855 & 0.2880 & \textbf{0.0939} \\
		0.05  & 0.0651 & 0.1182 & 0.0910 & 0.1985 & \textbf{0.0944} \\
		0.075 & 0.0636 & 0.1122 & 0.0875 & 0.1422 & \textbf{0.0949} \\
		0.10  & 0.0665 & 0.1215 & 0.0936 & 0.1102 & \textbf{0.0964} \\
		0.15  & 0.0738 & 0.1307 & \textbf{0.1007} & 0.0790 & 0.0987 \\
		0.20  & 0.0667 & 0.1255 & \textbf{0.0995} & 0.0661 & 0.1031 \\
		0.25 & 0.0672 & 0.1324 & 0.1091 & 0.0620 & \textbf{0.1083} \\
		0.30  & 0.0620 & 0.1242 & \textbf{0.1037} & 0.0661 & 0.1186 \\
		0.35  & 0.0634 & 0.1324 & \textbf{0.1140} & 0.0745 & 0.1306 \\
		\bottomrule
	\end{tabular}
\end{table}

\begin{table}[!htbp]
	\caption{False discovery rates under the heterogeneous block model with $m=200$, $\alpha=0.1$, and 3000 Monte Carlo replications. Values close to the nominal level $\alpha=0.10$ indicate favorable false-discovery behavior. Boldfaced entries correspond to the values closest to the nominal level.}
	\label{tab:fdr_heterogeneous_block_m200}
	\centering
	\small
	\begin{tabular}{lccccc}
		\toprule
		$p$ & BH & Storey-BH & GBS & MRD-CSX & MRD-GBS \\
		\midrule
		0.01  & 0.0792 & 0.0918 & 0.0807 & 0.4279 & \textbf{0.0963} \\
		0.03  & 0.0765 & 0.0890 & 0.0789 & 0.3008 & \textbf{0.0955} \\
		0.05  & 0.0897 & 0.1039 & 0.0941 & 0.2115 & \textbf{0.0978} \\
		0.075 & 0.0837 & \textbf{0.1002} & 0.0893 & 0.1492 & 0.0967 \\
		0.10  & 0.0849 & 0.1030 & 0.0926 & 0.1149 & \textbf{0.0975} \\
		0.15  & 0.0822 & 0.1053 & 0.0947 & 0.0784 & \textbf{0.0981} \\
		0.20  & 0.0747 & 0.1023 & 0.0915 & 0.0588 & \textbf{0.0984} \\
		0.25  & 0.0725 & 0.1046 & 0.0937 & 0.0467 & \textbf{0.0969} \\
		0.30  & 0.0690 & 0.1063 & 0.0954 & 0.0421 & \textbf{0.0995} \\
		0.35  & 0.0648 & 0.1061 & \textbf{0.0970} & 0.0415 & 0.1037 \\
		\bottomrule
	\end{tabular}
\end{table}

\begin{table}[!htbp]
	\caption{False discovery rates under the sparse precision-matrix model with $m=200$, $\alpha=0.1$, and 3000 Monte Carlo replications. Values close to the nominal level $\alpha=0.10$ indicate favorable false-discovery behavior. Boldfaced entries correspond to the values closest to the nominal level.}
	\label{tab:fdr_sparse_precision_m200}
	\centering
	\small
	\begin{tabular}{lccccc}
		\toprule
		$p$ & BH & Storey-BH & GBS & MRD-CSX & MRD-GBS \\
		\midrule
		0.01  & 0.0985 & 0.1021 & \textbf{0.0996} & 0.4001 & 0.0936 \\
		0.03  & 0.0940 & \textbf{0.0994} & 0.0957 & 0.3026 & 0.1031 \\
		0.05  & 0.0923 & \textbf{0.1013} & 0.0975 & 0.2165 & 0.1066 \\
		0.075 & 0.0911 & \textbf{0.1009} & 0.0961 & 0.1608 & 0.1080 \\
		0.10  & 0.0899 & \textbf{0.1017} & 0.0968 & 0.1296 & 0.1079 \\
		0.15  & 0.0838 & \textbf{0.1007} & 0.0958 & 0.0982 & 0.1105 \\
		0.20  & 0.0797 & \textbf{0.1007} & 0.0951 & 0.0829 & 0.1130 \\
		0.25  & 0.0743 & \textbf{0.1004} & 0.0950 & 0.0731 & 0.1121 \\
		0.30  & 0.0689 & \textbf{0.1008} & 0.0953 & 0.0699 & 0.1150 \\
		0.35  & 0.0649 & \textbf{0.1004} & 0.0957 & 0.0671 & 0.1156 \\
		\bottomrule
	\end{tabular}
\end{table}

\clearpage

\subsection*{False non-discovery rates ($m=200$)}

The following tables summarize empirical false non-discovery rates. Smaller
values indicate fewer missed signals, and boldfaced entries mark the smallest
values in each row, with ties retained when they occur at the displayed
precision.

\begin{table}[!htbp]
	\caption{False non-discovery rates under the equicorrelation model with $m=200$, $\alpha=0.1$, and 3000 Monte Carlo replications. Smaller values indicate superior signal-recovery performance. Boldfaced entries correspond to the smallest false non-discovery rates.}
	\label{tab:fnr_equicorrelation_m200}
	\centering
	\small
	\begin{tabular}{lccccc}
		\toprule
		$p$ & BH & Storey-BH & GBS & MRD-CSX & MRD-GBS \\
		\midrule
		0.01  & 0.0060 & 0.0045 & 0.0048 & \textbf{0.0000} & \textbf{0.0000} \\
		0.03  & 0.0150 & 0.0107 & 0.0116 & \textbf{0.0000} & \textbf{0.0000} \\
		0.05  & 0.0217 & 0.0153 & 0.0175 & \textbf{0.0000} & \textbf{0.0000} \\
		0.075 & 0.0276 & 0.0183 & 0.0213 & \textbf{0.0000} & \textbf{0.0000} \\
		0.10  & 0.0337 & 0.0201 & 0.0239 & \textbf{0.0000} & 0.0001 \\
		0.15  & 0.0450 & 0.0242 & 0.0293 & \textbf{0.0000} & \textbf{0.0000} \\
		0.20  & 0.0519 & 0.0263 & 0.0332 & \textbf{0.0000} & \textbf{0.0000} \\
		0.25  & 0.0601 & 0.0269 & 0.0334 & 0.0001 & \textbf{0.0000} \\
		0.30  & 0.0664 & 0.0272 & 0.0352 & 0.0001 & \textbf{0.0000} \\
		0.35  & 0.0769 & 0.0274 & 0.0353 & 0.0001 & \textbf{0.0000} \\
		\bottomrule
	\end{tabular}
\end{table}

\begin{table}[!htbp]
	\caption{False non-discovery rates under the factor model with $m=200$, $\alpha=0.1$, and 3000 Monte Carlo replications. Smaller values indicate superior signal-recovery performance. Boldfaced entries correspond to the smallest false non-discovery rates.}
	\label{tab:fnr_factor_m200}
	\centering
	\small
	\begin{tabular}{lccccc}
		\toprule
		$p$ & BH & Storey-BH & GBS & MRD-CSX & MRD-GBS \\
		\midrule
		0.01  & 0.0051 & 0.0050 & 0.0051 & \textbf{0.0000} & \textbf{0.0000} \\
		0.03  & 0.0127 & 0.0122 & 0.0126 & \textbf{0.0000} & 0.0001 \\
		0.05  & 0.0184 & 0.0174 & 0.0181 & \textbf{0.0001} & \textbf{0.0001} \\
		0.075 & 0.0243 & 0.0226 & 0.0236 & \textbf{0.0001} & \textbf{0.0001} \\
		0.10  & 0.0297 & 0.0271 & 0.0283 & \textbf{0.0001} & 0.0002 \\
		0.15  & 0.0382 & 0.0336 & 0.0355 & \textbf{0.0002} & \textbf{0.0002} \\
		0.20  & 0.0458 & 0.0387 & 0.0409 & 0.0003 & \textbf{0.0002} \\
		0.25  & 0.0532 & 0.0425 & 0.0453 & 0.0004 & \textbf{0.0002} \\
		0.30  & 0.0599 & 0.0447 & 0.0483 & 0.0005 & \textbf{0.0002} \\
		0.35  & 0.0685 & 0.0472 & 0.0512 & 0.0006 & \textbf{0.0002} \\
		\bottomrule
	\end{tabular}
\end{table}

\begin{table}[!htbp]
	\caption{False non-discovery rates under the fractional Gaussian noise model with $m=200$, $\alpha=0.1$, and 3000 Monte Carlo replications. Smaller values indicate superior signal-recovery performance. Boldfaced entries correspond to the smallest false non-discovery rates.}
	\label{tab:fnr_fractional_gaussian_noise_m200}
	\centering
	\small
	\begin{tabular}{lccccc}
		\toprule
		$p$ & BH & Storey-BH & GBS & MRD-CSX & MRD-GBS \\
		\midrule
		0.01  & 0.0051 & 0.0049 & 0.0054 & \textbf{0.0000} & 0.0001 \\
		0.03  & 0.0130 & 0.0124 & 0.0128 & \textbf{0.0000} & 0.0001 \\
		0.05  & 0.0192 & 0.0184 & 0.0187 & \textbf{0.0001} & 0.0002 \\
		0.075 & 0.0246 & 0.0232 & 0.0237 & \textbf{0.0003} & 0.0004 \\
		0.10  & 0.0296 & 0.0272 & 0.0279 & \textbf{0.0005} & \textbf{0.0005} \\
		0.15  & 0.0389 & 0.0341 & 0.0358 & 0.0012 & \textbf{0.0009} \\
		0.20  & 0.0463 & 0.0391 & 0.0402 & 0.0027 & \textbf{0.0015} \\
		0.25  & 0.0537 & 0.0429 & 0.0449 & 0.0051 & \textbf{0.0023} \\
		0.30  & 0.0612 & 0.0462 & 0.0479 & 0.0106 & \textbf{0.0040} \\
		0.35  & 0.0696 & 0.0483 & 0.0513 & 0.0175 & \textbf{0.0051} \\
		\bottomrule
	\end{tabular}
\end{table}

\begin{table}[!htbp]
	\caption{False non-discovery rates under the Toeplitz model with $m=200$, $\alpha=0.1$, and 3000 Monte Carlo replications. Smaller values indicate superior signal-recovery performance. Boldfaced entries correspond to the smallest false non-discovery rates.}
	\label{tab:fnr_toeplitz_m200}
	\centering
	\small
	\begin{tabular}{lccccc}
		\toprule
		$p$ & BH & Storey-BH & GBS & MRD-CSX & MRD-GBS \\
		\midrule
		0.01  & 0.0053 & 0.0049 & 0.0053 & \textbf{0.0000} & \textbf{0.0000} \\
		0.03  & 0.0133 & 0.0126 & 0.0127 & \textbf{0.0000} & \textbf{0.0000} \\
		0.05  & 0.0195 & 0.0177 & 0.0189 & \textbf{0.0000} & \textbf{0.0000} \\
		0.075 & 0.0251 & 0.0224 & 0.0234 & \textbf{0.0000} & \textbf{0.0000} \\
		0.10  & 0.0301 & 0.0264 & 0.0281 & \textbf{0.0000} & \textbf{0.0000} \\
		0.15  & 0.0398 & 0.0346 & 0.0350 & \textbf{0.0001} & \textbf{0.0001} \\
		0.20  & 0.0472 & 0.0360 & 0.0387 & 0.0002 & \textbf{0.0001} \\
		0.25  & 0.0555 & 0.0397 & 0.0439 & 0.0012 & \textbf{0.0005} \\
		0.30  & 0.0626 & 0.0421 & 0.0464 & 0.0038 & \textbf{0.0018} \\
		0.35  & 0.0711 & 0.0440 & 0.0485 & 0.0081 & \textbf{0.0029} \\
		\bottomrule
	\end{tabular}
\end{table}

\begin{table}[!htbp]
	\caption{False non-discovery rates under the heterogeneous block model with $m=200$, $\alpha=0.1$, and 3000 Monte Carlo replications. Smaller values indicate superior signal-recovery performance. Boldfaced entries correspond to the smallest false non-discovery rates.}
	\label{tab:fnr_heterogeneous_block_m200}
	\centering
	\small
	\begin{tabular}{lccccc}
		\toprule
		$p$ & BH & Storey-BH & GBS & MRD-CSX & MRD-GBS \\
		\midrule
		0.01  & 0.0050 & 0.0049 & 0.0050 & \textbf{0.0004} & 0.0009 \\
		0.03  & 0.0128 & 0.0125 & 0.0127 & \textbf{0.0007} & 0.0017 \\
		0.05  & 0.0187 & 0.0180 & 0.0184 & \textbf{0.0013} & 0.0023 \\
		0.075 & 0.0244 & 0.0233 & 0.0239 & \textbf{0.0021} & 0.0029 \\
		0.10  & 0.0290 & 0.0272 & 0.0278 & \textbf{0.0029} & 0.0033 \\
		0.15  & 0.0382 & 0.0342 & 0.0356 & 0.0046 & \textbf{0.0038} \\
		0.20  & 0.0455 & 0.0390 & 0.0406 & 0.0066 & \textbf{0.0043} \\
		0.25  & 0.0528 & 0.0428 & 0.0449 & 0.0092 & \textbf{0.0049} \\
		0.30  & 0.0606 & 0.0461 & 0.0486 & 0.0134 & \textbf{0.0068} \\
		0.35  & 0.0687 & 0.0490 & 0.0522 & 0.0203 & \textbf{0.0108} \\
		\bottomrule
	\end{tabular}
\end{table}

\begin{table}[!htbp]
	\caption{False non-discovery rates under the sparse precision model with $m=200$, $\alpha=0.1$, and 3000 Monte Carlo replications. Smaller values indicate superior signal-recovery performance. Boldfaced entries correspond to the smallest false non-discovery rates.}
	\label{tab:fnr_sparse_precision_m200}
	\centering
	\small
	\begin{tabular}{lccccc}
		\toprule
		$p$ & BH & Storey-BH & GBS & MRD-CSX & MRD-GBS \\
		\midrule
		0.01  & 0.0049 & 0.0049 & 0.0049 & \textbf{0.0009} & 0.0018 \\
		0.03  & 0.0127 & 0.0124 & 0.0126 & \textbf{0.0020} & 0.0039 \\
		0.05  & 0.0185 & 0.0180 & 0.0183 & \textbf{0.0036} & 0.0058 \\
		0.075 & 0.0244 & 0.0235 & 0.0239 & \textbf{0.0061} & 0.0080 \\
		0.10  & 0.0290 & 0.0273 & 0.0278 & \textbf{0.0090} & 0.0100 \\
		0.15  & 0.0378 & 0.0341 & 0.0350 & 0.0165 & \textbf{0.0147} \\
		0.20  & 0.0460 & 0.0396 & 0.0410 & 0.0270 & \textbf{0.0201} \\
		0.25  & 0.0528 & 0.0431 & 0.0447 & 0.0411 & \textbf{0.0263} \\
		0.30  & 0.0604 & 0.0466 & 0.0484 & 0.0591 & \textbf{0.0336} \\
		0.35  & 0.0678 & 0.0491 & 0.0513 & 0.0808 & \textbf{0.0402} \\
		\bottomrule
	\end{tabular}
\end{table}

\clearpage

\subsection*{Powers ($m=200$)}

The following tables report empirical powers. Larger values indicate stronger
signal-recovery performance, and boldfaced entries identify the largest powers
within each sparsity level.

\begin{table}[!htbp]
	\caption{Power under the equicorrelation model with $m=200$, $\alpha=0.1$, and 3000 Monte Carlo replications. Larger values indicate superior signal-recovery performance. Boldfaced entries correspond to the largest empirical powers.}
	\label{tab:power_equicorrelation_m200}
	\centering
	\small
	\begin{tabular}{lccccc}
		\toprule
		$p$ & BH & Storey-BH & GBS & MRD-CSX & MRD-GBS \\
		\midrule
		0.01 & 0.4603 & 0.5333 & 0.5341 & \textbf{0.9980} & 0.9967 \\
		0.03 & 0.5590 & 0.6345 & 0.6351 & \textbf{0.9994} & 0.9980 \\
		0.05 & 0.6174 & 0.6949 & 0.6942 & \textbf{0.9998} & 0.9994 \\
		0.075 & 0.6853 & 0.7682 & 0.7639 & \textbf{0.9997} & 0.9996 \\
		0.10 & 0.7216 & 0.8124 & 0.8058 & \textbf{0.9997} & 0.9995 \\
		0.15 & 0.7687 & 0.8689 & 0.8527 & \textbf{0.9998} & \textbf{0.9998} \\
		0.20 & 0.8063 & 0.9049 & 0.8881 & 0.9998 & \textbf{0.9999} \\
		0.25 & 0.8276 & 0.9314 & 0.9161 & 0.9998 & \textbf{1.0000} \\
		0.30 & 0.8498 & 0.9501 & 0.9368 & 0.9998 & \textbf{0.9999} \\
		0.35 & 0.8576 & 0.9641 & 0.9525 & 0.9998 & \textbf{1.0000} \\
		\bottomrule
	\end{tabular}
\end{table}

\begin{table}[!htbp]
	\caption{Power under the factor model with $m=200$, $\alpha=0.1$, and 3000 Monte Carlo replications. Larger values indicate superior signal-recovery performance. Boldfaced entries correspond to the largest empirical powers.}
	\label{tab:power_factor_m200}
	\centering
	\small
	\begin{tabular}{lccccc}
		\toprule
		$p$ & BH & Storey-BH & GBS & MRD-CSX & MRD-GBS \\
		\midrule
		0.01 & 0.4701 & 0.4842 & 0.4719 & \textbf{0.9965} & 0.9946 \\
		0.03 & 0.5701 & 0.5880 & 0.5728 & \textbf{0.9988} & 0.9972 \\
		0.05 & 0.6324 & 0.6534 & 0.6381 & \textbf{0.9989} & 0.9977 \\
		0.075 & 0.6864 & 0.7095 & 0.6950 & \textbf{0.9991} & 0.9988 \\
		0.10 & 0.7235 & 0.7504 & 0.7368 & \textbf{0.9988} & 0.9985 \\
		0.15 & 0.7751 & 0.8054 & 0.7922 & 0.9989 & \textbf{0.9991} \\
		0.20 & 0.8089 & 0.8430 & 0.8311 & 0.9988 & \textbf{0.9992} \\
		0.25 & 0.8344 & 0.8735 & 0.8621 & 0.9989 & \textbf{0.9995} \\
		0.30 & 0.8554 & 0.8985 & 0.8872 & 0.9987 & \textbf{0.9995} \\
		0.35 & 0.8673 & 0.9164 & 0.9056 & 0.9990 & \textbf{0.9997} \\
		\bottomrule
	\end{tabular}
\end{table}

\begin{table}[!htbp]
	\caption{Power under the fractional Gaussian noise model with $m=200$, $\alpha=0.1$, and 3000 Monte Carlo replications. Larger values indicate superior signal-recovery performance. Boldfaced entries correspond to the largest empirical powers.}
	\label{tab:power_fractional_gaussian_noise_m200}
	\centering
	\small
	\begin{tabular}{lccccc}
		\toprule
		$p$ & BH & Storey-BH & GBS & MRD-CSX & MRD-GBS \\
		\midrule
		0.01 & 0.4678 & 0.4874 & 0.4736 & \textbf{0.9954} & 0.9920 \\
		0.03 & 0.5621 & 0.5875 & 0.5695 & \textbf{0.9985} & 0.9949 \\
		0.05 & 0.6187 & 0.6434 & 0.6281 & \textbf{0.9981} & 0.9961 \\
		0.075 & 0.6836 & 0.7075 & 0.6947 & \textbf{0.9970} & 0.9957 \\
		0.10 & 0.7245 & 0.7533 & 0.7400 & \textbf{0.9960} & 0.9954 \\
		0.15 & 0.7725 & 0.8082 & 0.7945 & 0.9936 & \textbf{0.9950} \\
		0.20 & 0.8082 & 0.8466 & 0.8356 & 0.9900 & \textbf{0.9944} \\
		0.25 & 0.8334 & 0.8770 & 0.8653 & 0.9855 & \textbf{0.9933} \\
		0.30 & 0.8527 & 0.9003 & 0.8911 & 0.9768 & \textbf{0.9914} \\
		0.35 & 0.8650 & 0.9185 & 0.9084 & 0.9692 & \textbf{0.9913} \\
		\bottomrule
	\end{tabular}
\end{table}

\begin{table}[!htbp]
	\caption{Power under the Toeplitz model with $m=200$, $\alpha=0.1$, and 3000 Monte Carlo replications. Larger values indicate superior signal-recovery performance. Boldfaced entries correspond to the largest empirical powers.}
	\label{tab:power_toeplitz_m200}
	\centering
	\small
	\begin{tabular}{lccccc}
		\toprule
		$p$ & BH & Storey-BH & GBS & MRD-CSX & MRD-GBS \\
		\midrule
		0.01 & 0.4604 & 0.4975 & 0.4848 & \textbf{1.0000} & \textbf{1.0000} \\
		0.03 & 0.5583 & 0.5984 & 0.5849 & \textbf{1.0000} & \textbf{1.0000} \\
		0.05 & 0.6164 & 0.6590 & 0.6438 & \textbf{1.0000} & \textbf{1.0000} \\
		0.075 & 0.6822 & 0.7236 & 0.7091 & \textbf{0.9999} & \textbf{0.9999} \\
		0.10 & 0.7245 & 0.7705 & 0.7552 & \textbf{0.9998} & \textbf{0.9998} \\
		0.15 & 0.7727 & 0.8261 & 0.8086 & 0.9995 & \textbf{0.9996} \\
		0.20 & 0.8074 & 0.8648 & 0.8487 & 0.9992 & \textbf{0.9996} \\
		0.25 & 0.8304 & 0.8948 & 0.8792 & 0.9969 & \textbf{0.9986} \\
		0.30 & 0.8512 & 0.9159 & 0.9022 & 0.9921 & \textbf{0.9963} \\
		0.35 & 0.8637 & 0.9341 & 0.9207 & 0.9867 & \textbf{0.9955} \\
		\bottomrule
	\end{tabular}
\end{table}

\begin{table}[!htbp]
	\caption{Power under the heterogeneous block model with $m=200$, $\alpha=0.1$, and 3000 Monte Carlo replications. Larger values indicate superior signal-recovery performance. Boldfaced entries correspond to the largest empirical powers.}
	\label{tab:power_heterogeneous_block_m200}
	\centering
	\small
	\begin{tabular}{lccccc}
		\toprule
		$p$ & BH & Storey-BH & GBS & MRD-CSX & MRD-GBS \\
		\midrule
		0.01  & 0.4769 & 0.4861 & 0.4782 & \textbf{0.9409} & 0.9002 \\
		0.03  & 0.5645 & 0.5746 & 0.5668 & \textbf{0.9751} & 0.9376 \\
		0.05  & 0.6266 & 0.6403 & 0.6317 & \textbf{0.9766} & 0.9535 \\
		0.075 & 0.6831 & 0.6985 & 0.6898 & \textbf{0.9746} & 0.9628 \\
		0.10  & 0.7279 & 0.7451 & 0.7388 & \textbf{0.9741} & 0.9695 \\
		0.15  & 0.7739 & 0.7983 & 0.7895 & 0.9740 & \textbf{0.9778} \\
		0.20  & 0.8090 & 0.8384 & 0.8307 & 0.9740 & \textbf{0.9829} \\
		0.25  & 0.8345 & 0.8686 & 0.8612 & 0.9728 & \textbf{0.9856} \\
		0.30  & 0.8524 & 0.8915 & 0.8842 & 0.9695 & \textbf{0.9851} \\
		0.35  & 0.8659 & 0.9089 & 0.9013 & 0.9634 & \textbf{0.9817} \\
		\bottomrule
	\end{tabular}
\end{table}

\begin{table}[!htbp]
	\caption{Power under the sparse precision model with $m=200$, $\alpha=0.1$, and 3000 Monte Carlo replications. Larger values indicate superior signal-recovery performance. Boldfaced entries correspond to the largest empirical powers.}
	\label{tab:power_sparse_precision_m200}
	\centering
	\small
	\begin{tabular}{lccccc}
		\toprule
		$p$ & BH & Storey-BH & GBS & MRD-CSX & MRD-GBS \\
		\midrule
		0.01  & 0.4857 & 0.4905 & 0.4866 & \textbf{0.8740} & 0.7974 \\
		0.03  & 0.5677 & 0.5760 & 0.5701 & \textbf{0.9360} & 0.8647 \\
		0.05  & 0.6268 & 0.6375 & 0.6315 & \textbf{0.9340} & 0.8846 \\
		0.075 & 0.6843 & 0.6958 & 0.6911 & \textbf{0.9273} & 0.9002 \\
		0.10  & 0.7273 & 0.7432 & 0.7383 & \textbf{0.9212} & 0.9095 \\
		0.15  & 0.7765 & 0.7987 & 0.7930 & 0.9079 & \textbf{0.9164} \\
		0.20  & 0.8068 & 0.8350 & 0.8286 & 0.8928 & \textbf{0.9202} \\
		0.25  & 0.8347 & 0.8670 & 0.8614 & 0.8758 & \textbf{0.9217} \\
		0.30  & 0.8531 & 0.8893 & 0.8843 & 0.8591 & \textbf{0.9228} \\
		0.35  & 0.8677 & 0.9075 & 0.9026 & 0.8432 & \textbf{0.9269} \\
		\bottomrule
	\end{tabular}
\end{table}
\clearpage

\subsection*{Average number of rejections ($m=200$)}

The final set of tables reports average numbers of rejections. These values
should be interpreted relative to the expected number of true signals, $mp$,
and jointly with the corresponding FDR, FNR, and power results; consequently,
no entries are highlighted in boldface.

\begin{table}[!htbp]
	\caption{Average numbers of rejections under the equicorrelation model with $m=200$, $\alpha=0.1$, and 3000 Monte Carlo replications. Larger values indicate more aggressive rejection behavior. Since neither excessively large nor excessively small rejection counts are universally preferable, no entries are highlighted in boldface.}
	\label{tab:rej_equicorrelation_m200}
	\centering
	\small
	\begin{tabular}{lccccc}
		\toprule
		$p$ & BH & Storey-BH & GBS & MRD-CSX & MRD-GBS \\
		\midrule
		0.01 & 5.15 & 25.64 & 30.23 & 4.31 & 2.30 \\
		0.03 & 7.80 & 28.50 & 33.85 & 8.65 & 6.80 \\
		0.05 & 10.60 & 30.22 & 36.06 & 12.69 & 11.26 \\
		0.075 & 14.26 & 35.27 & 40.11 & 17.58 & 16.74 \\
		0.10 & 18.75 & 41.41 & 46.25 & 22.66 & 22.40 \\
		0.15 & 27.87 & 51.41 & 53.23 & 32.53 & 33.43 \\
		0.20 & 36.81 & 59.65 & 60.81 & 42.35 & 44.46 \\
		0.25 & 46.12 & 71.81 & 72.59 & 52.42 & 55.66 \\
		0.30 & 55.77 & 81.03 & 81.23 & 62.44 & 66.88 \\
		0.35 & 64.92 & 93.76 & 93.46 & 72.14 & 77.70 \\
		\bottomrule
	\end{tabular}
\end{table}

\begin{table}[!htbp]
\caption{Average numbers of rejections under the factor model with $m=200$, $\alpha=0.1$, and 3000 Monte Carlo replications. Larger values indicate more aggressive rejection behavior. Since neither excessively large nor excessively small rejection counts are universally preferable, no entries are highlighted in boldface.}
	\label{tab:rej_factor_m200}
	\centering
	\small
	\begin{tabular}{lccccc}
		\toprule
		$p$ & BH & Storey-BH & GBS & MRD-CSX & MRD-GBS \\
		\midrule
		0.01 & 1.92 & 2.83 & 2.27 & 4.35 & 2.32 \\
		0.03 & 4.77 & 6.18 & 5.04 & 8.59 & 6.73 \\
		0.05 & 8.14 & 9.96 & 8.70 & 12.65 & 11.21 \\
		0.075 & 12.18 & 14.15 & 12.88 & 17.53 & 16.68 \\
		0.1 & 16.94 & 19.59 & 18.06 & 22.65 & 22.39 \\
		0.15 & 26.11 & 29.61 & 27.63 & 32.43 & 33.31 \\
		0.2 & 35.57 & 40.09 & 37.95 & 42.20 & 44.26 \\
		0.25 & 45.58 & 51.37 & 48.97 & 52.38 & 55.61 \\
		0.3 & 55.60 & 62.67 & 60.31 & 62.39 & 66.74 \\
		0.35 & 65.01 & 73.98 & 71.48 & 72.08 & 77.55 \\
		\bottomrule
	\end{tabular}
\end{table}

\begin{table}[!htbp]
	\caption{Average numbers of rejections under the fractional gaussian noise model with $m=200$, $\alpha=0.1$, and 3000 Monte Carlo replications. Larger values indicate more aggressive rejection behavior. Since neither excessively large nor excessively small rejection counts are universally preferable, no entries are highlighted in boldface.}
	\label{tab:rej_fractional_gaussian_noise_m200}
	\centering
	\small
	\begin{tabular}{lccccc}
		\toprule
		$p$ & BH & Storey-BH & GBS & MRD-CSX & MRD-GBS \\
		\midrule
		0.01 & 2.07 & 5.13 & 3.66 & 4.20 & 2.28 \\
		0.03 & 4.81 & 8.63 & 6.68 & 8.53 & 6.78 \\
		0.05 & 7.88 & 11.70 & 9.66 & 12.53 & 11.18 \\
		0.075 & 12.10 & 15.81 & 14.01 & 17.44 & 16.62 \\
		0.1 & 16.82 & 21.34 & 19.33 & 22.54 & 22.29 \\
		0.15 & 26.29 & 31.90 & 29.45 & 32.48 & 33.33 \\
		0.2 & 35.70 & 41.82 & 39.78 & 42.37 & 44.41 \\
		0.25 & 45.55 & 53.39 & 50.55 & 52.43 & 55.63 \\
		0.3 & 55.41 & 64.08 & 62.03 & 62.44 & 67.20 \\
		0.35 & 65.10 & 75.83 & 73.29 & 72.04 & 78.44 \\
		\bottomrule
	\end{tabular}
\end{table}

\begin{table}[!htbp]
	\caption{Average numbers of rejections under the toeplitz model with $m=200$, $\alpha=0.1$, and 3000 Monte Carlo replications. Larger values indicate more aggressive rejection behavior. Since neither excessively large nor excessively small rejection counts are universally preferable, no entries are highlighted in boldface.}
\label{tab:rej_toeplitz_m200}	
	\centering
	\small
	\begin{tabular}{lccccc}
		\toprule
		$p$ & BH & Storey-BH & GBS & MRD-CSX & MRD-GBS \\
		\midrule
		0.01 & 3.15 & 11.48 & 10.82 & 4.22 & 2.29 \\
		0.03 & 5.99 & 14.66 & 14.41 & 8.52 & 6.79 \\
		0.05 & 8.79 & 17.87 & 16.95 & 12.53 & 11.19 \\
		0.075 & 12.84 & 21.15 & 20.54 & 17.46 & 16.65 \\
		0.1 & 17.51 & 26.84 & 25.98 & 22.58 & 22.35 \\
		0.15 & 27.00 & 37.87 & 35.42 & 32.61 & 33.44 \\
		0.2 & 36.10 & 47.84 & 45.82 & 42.77 & 44.68 \\
		0.25 & 45.71 & 59.54 & 57.34 & 53.29 & 56.26 \\
		0.3 & 55.49 & 69.65 & 67.54 & 64.00 & 68.24 \\
		0.35 & 65.24 & 81.88 & 79.65 & 74.66 & 80.37 \\
		\bottomrule
	\end{tabular}
\end{table}

\begin{table}[!htbp]
	\caption{Average numbers of rejections under the heterogeneous block model with $m=200$, $\alpha=0.1$, and 3000 Monte Carlo replications. Larger values indicate more aggressive rejection behavior. Since neither excessively large nor excessively small rejection counts are universally preferable, no entries are highlighted in boldface.}
\label{tab:rej_heterogeneous_block_m200}	
	\centering
	\small
	\begin{tabular}{lccccc}
		\toprule
		$p$ & BH & Storey-BH & GBS & MRD-CSX & MRD-GBS \\
		\midrule
		0.01 & 1.38 & 1.52 & 1.41 & 4.13 & 2.12 \\
		0.03 & 4.12 & 4.39 & 4.19 & 8.50 & 6.45 \\
		0.05 & 7.45 & 7.87 & 7.72 & 12.42 & 10.75 \\
		0.075 & 11.69 & 12.34 & 11.95 & 17.16 & 16.11 \\
		0.1 & 16.45 & 17.39 & 16.92 & 22.11 & 21.72 \\
		0.15 & 25.74 & 27.53 & 26.80 & 31.75 & 32.70 \\
		0.2 & 35.29 & 38.00 & 37.21 & 41.34 & 43.68 \\
		0.25 & 45.36 & 49.32 & 48.17 & 51.12 & 54.79 \\
		0.3 & 55.34 & 60.72 & 59.41 & 60.90 & 65.95 \\
		0.35 & 64.96 & 71.78 & 70.56 & 70.29 & 76.73 \\
		\bottomrule
	\end{tabular}
\end{table}

\begin{table}[!htbp]
	\caption{Average numbers of rejections under the sparse precision model with $m=200$, $\alpha=0.1$, and 3000 Monte Carlo replications. Larger values indicate more aggressive rejection behavior. Since neither excessively large nor excessively small rejection counts are universally preferable, no entries are highlighted in boldface.}
\label{tab:rej_sparse_precision_m200}	
	\centering
	\small
	\begin{tabular}{lccccc}
		\toprule
		$p$ & BH & Storey-BH & GBS & MRD-CSX & MRD-GBS \\
		\midrule
		0.01 & 1.25 & 1.28 & 1.32 & 3.85 & 1.91 \\
		0.03 & 4.06 & 4.17 & 4.10 & 8.19 & 6.01 \\
		0.05 & 7.26 & 7.49 & 7.50 & 11.96 & 10.10 \\
		0.075 & 11.54 & 11.91 & 11.76 & 16.53 & 15.25 \\
		0.1 & 16.32 & 16.94 & 16.74 & 21.24 & 20.63 \\
		0.15 & 25.69 & 26.99 & 26.69 & 30.21 & 31.07 \\
		0.2 & 35.23 & 37.40 & 36.91 & 38.83 & 41.56 \\
		0.25 & 45.34 & 48.57 & 48.02 & 47.27 & 52.12 \\
		0.3 & 55.27 & 59.78 & 59.21 & 55.51 & 62.88 \\
		0.35 & 65.01 & 70.83 & 70.24 & 63.13 & 73.42 \\
		\bottomrule
	\end{tabular}
\end{table}

\clearpage
%\section{Additional Simulation Results for $m=100$}

\section*{Appendix C: Additional Simulation Results for $m=100$}
\label{APPENDIX_SIMULATION_m100}

This appendix reports the simulation results corresponding to the lower-dimensional setting $m=100$. As discussed in Section~\ref{SECTION_SIMULATIONS}, our primary analysis focuses on the larger-dimensional experiments with $m=200$, since several of the key empirical phenomena become considerably more pronounced in that setting. Nevertheless, the $m=100$ results provide important supplementary evidence regarding the robustness of the conclusions across problem dimensions. Although the quantitative differences become substantially larger when $m=200$, the qualitative ranking of the competing procedures remains largely unchanged between the two dimensions.

Overall, the qualitative behavior observed for $m=100$ is broadly consistent with that reported in the main text. In particular, the proposed GBS-calibrated MRD procedure continues to exhibit favorable normalized misclassification risk, false discovery rate, false non-discovery rate, power, and average-rejection characteristics across a wide variety of dependence structures. The overall conclusions remain largely unchanged and provide additional evidence for the effectiveness of covariance-adaptive residualization combined with stagewise GBS calibration. In particular, several of the near-support-recovery phenomena observed for $m=200$—including FDR values close to the nominal level, extremely small FNR values, powers approaching one, and average numbers of rejections closely tracking the expected number of true signals—remain visible for $m=100$, although generally in a less pronounced form.

An additional feature of the $m=100$ experiments is that they reinforce one
of the most noteworthy findings of the simulation study. Although the
original MRD procedure and the proposed GBS-calibrated MRD procedure are
both built upon the same covariance-adaptive residualization mechanism,
their empirical operating characteristics differ substantially across a
number of dependence structures. As in the $m=200$ experiments, the
proposed calibration frequently exhibits stronger signal-recovery characteristics than the original MRD
procedure in sparse and moderately sparse regimes. This persistent difference suggests that covariance-adaptive
residualization alone cannot fully explain the observed signal-recovery
behavior. Rather, the results provide further evidence that stagewise
calibration plays a crucial role in determining how effectively
residualized information is translated into accurate recovery of the
underlying signal set.

The appendix figures and tables are organized in the same manner as the main-text simulations. Figure~\ref{FIG_NMR_m100} reports the normalized misclassification risk curves, Figures \ref{FIG_FDR_m100}–\ref{FIG_ANR_m100} report the FDR, FNR, power, and average-rejection characteristics, respectively, while Tables~\ref{TAB_NMR_EQ_m100}–\ref{TAB_REJ_SPARSE_m100} provide the corresponding numerical summaries. Together, these results allow a direct comparison with the $m=200$ experiments presented in Section~\ref{SECTION_SIMULATIONS}.

%This Appendix Contains additional simulation results supporting the
%findings reported in Section~\ref{SECTION_SIMULATIONS}. The tables presented below summarize
%the false discovery rates (FDR), false non-discovery rates (FNR),
%power, and average numbers of rejections obtained under the various
%dependence structures considered in the simulation study. The
%corresponding normalized misclassification rate (NMR) results are
%reported in the main text.

%\clearpage

\clearpage

\subsection*{Normalized misclassification rates ($m=100$)}

Figure~\ref{FIG_NMR_m100} summarizes the normalized misclassification
rates of the competing procedures under the six dependence structures
considered in this study. Consistent with the corresponding $m=200$
experiments, the proposed GBS-calibrated MRD procedure frequently
achieves the smallest misclassification rates across sparse and
moderately sparse regimes, providing further evidence of its strong
classification performance under dependence.

\begin{figure}[htbp]
	\centering
	
	\begin{subfigure}[b]{0.48\textwidth}
		\centering
		\includegraphics[width=\textwidth]{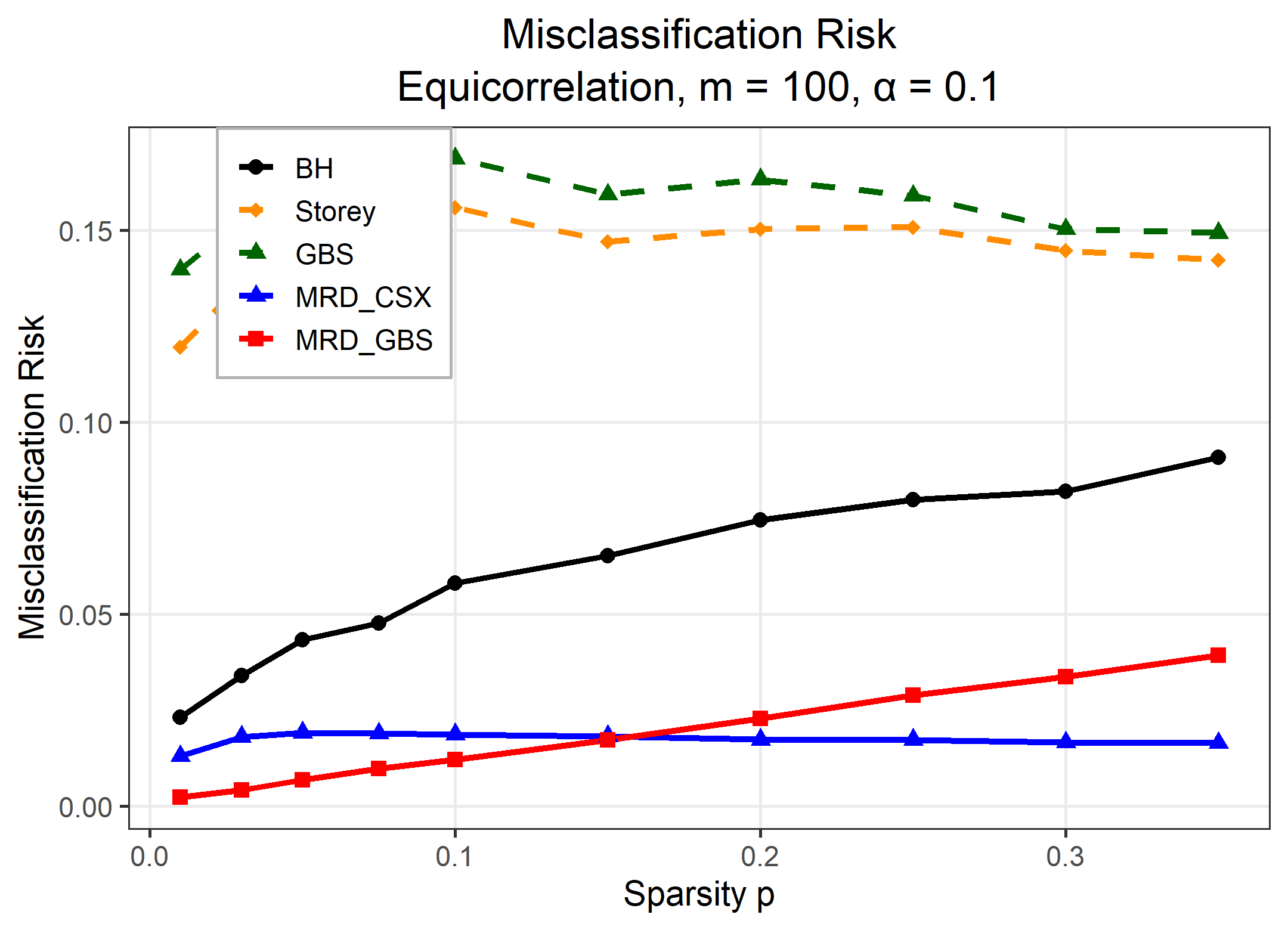}
		\caption{Equicorrelation (\(\rho=0.7\))}
	\end{subfigure}
	\hfill
	\begin{subfigure}[b]{0.48\textwidth}
		\centering
		\includegraphics[width=\textwidth]{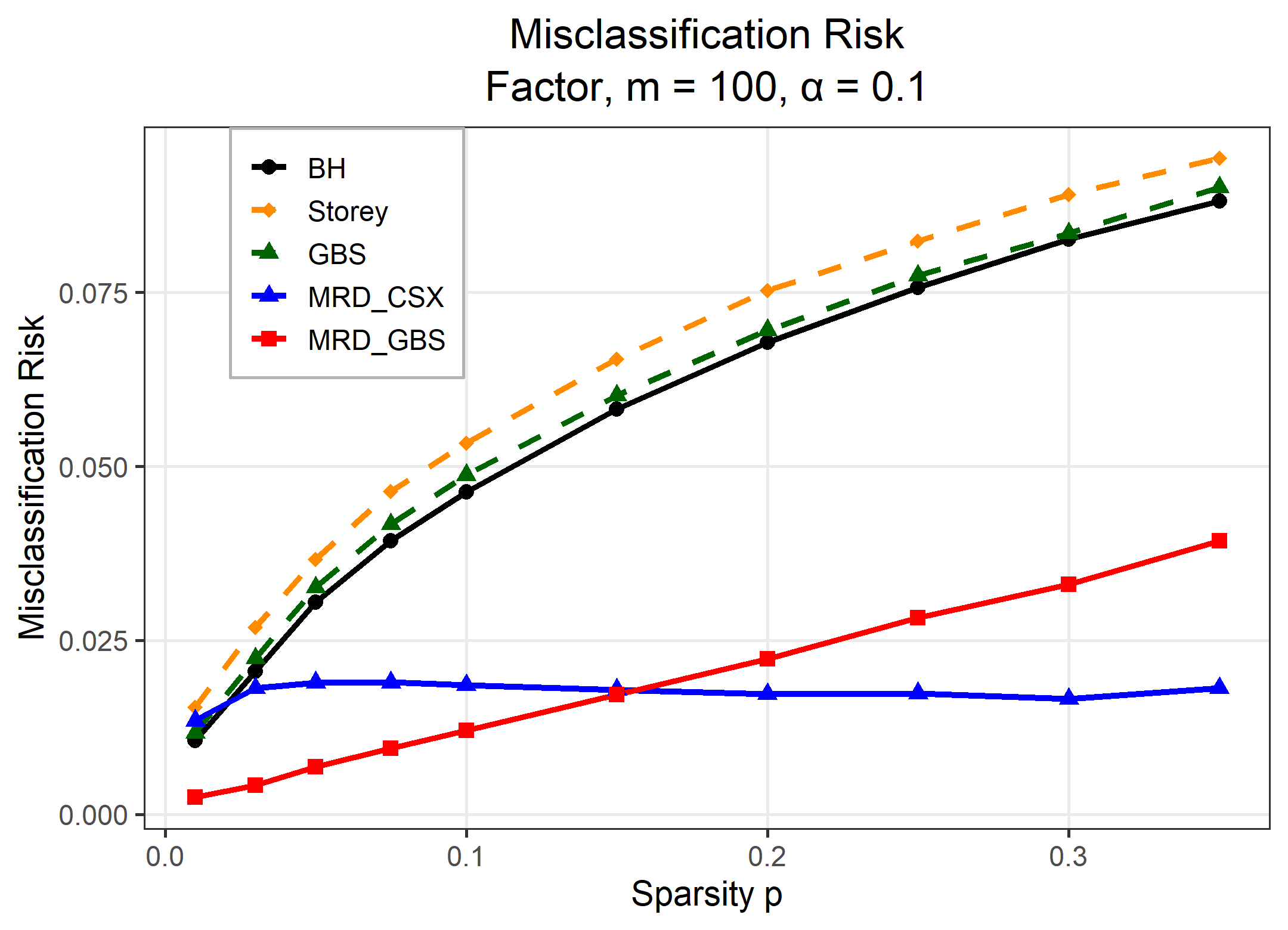}
		\caption{Factor model}
	\end{subfigure}
	
	\vspace{0.3cm}
	
	\begin{subfigure}[b]{0.48\textwidth}
		\centering
		\includegraphics[width=\textwidth]{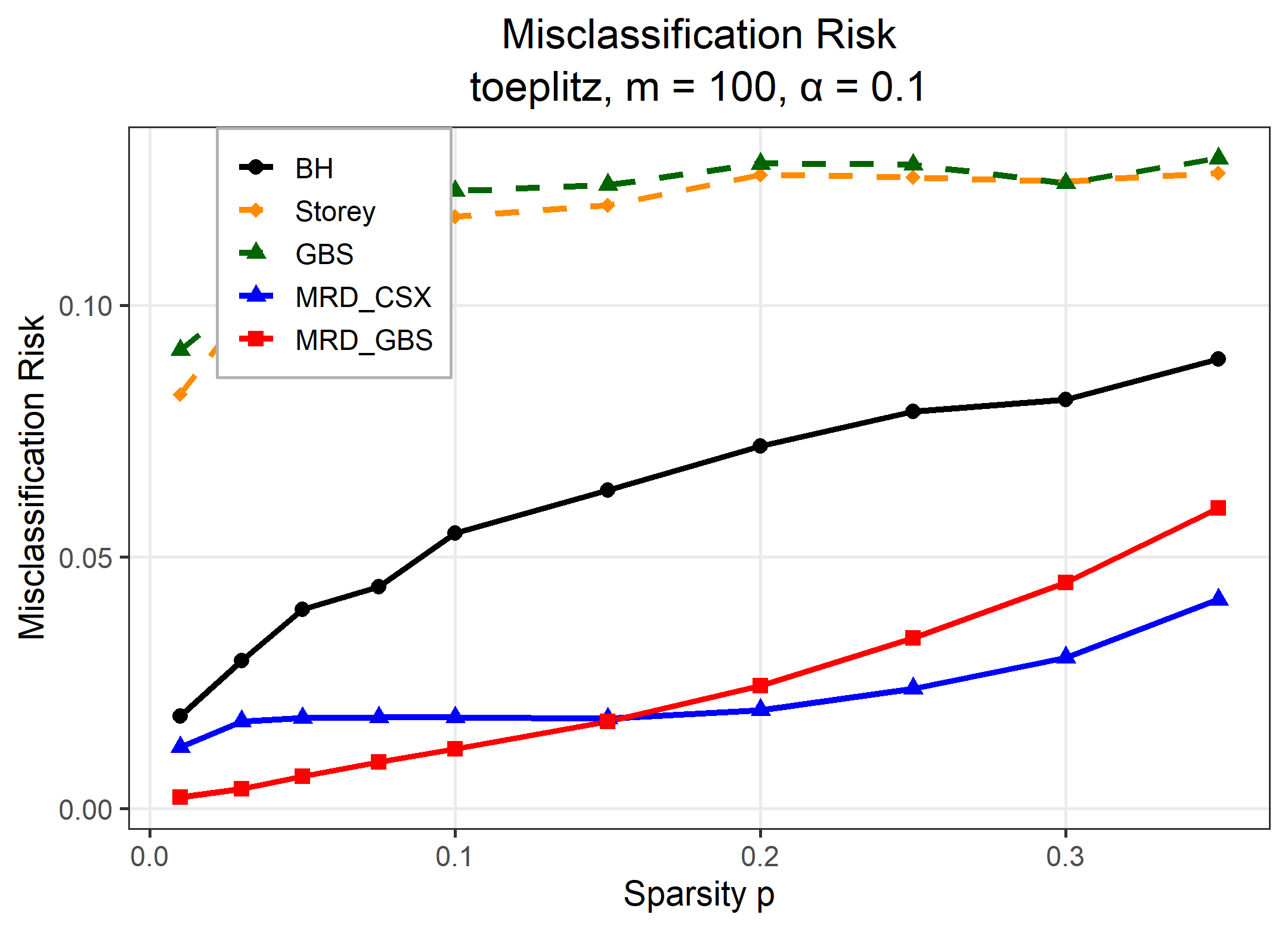}
		\caption{Toeplitz (\(\rho=0.9\))}
	\end{subfigure}
	\hfill
	\begin{subfigure}[b]{0.48\textwidth}
		\centering
		\includegraphics[width=\textwidth]{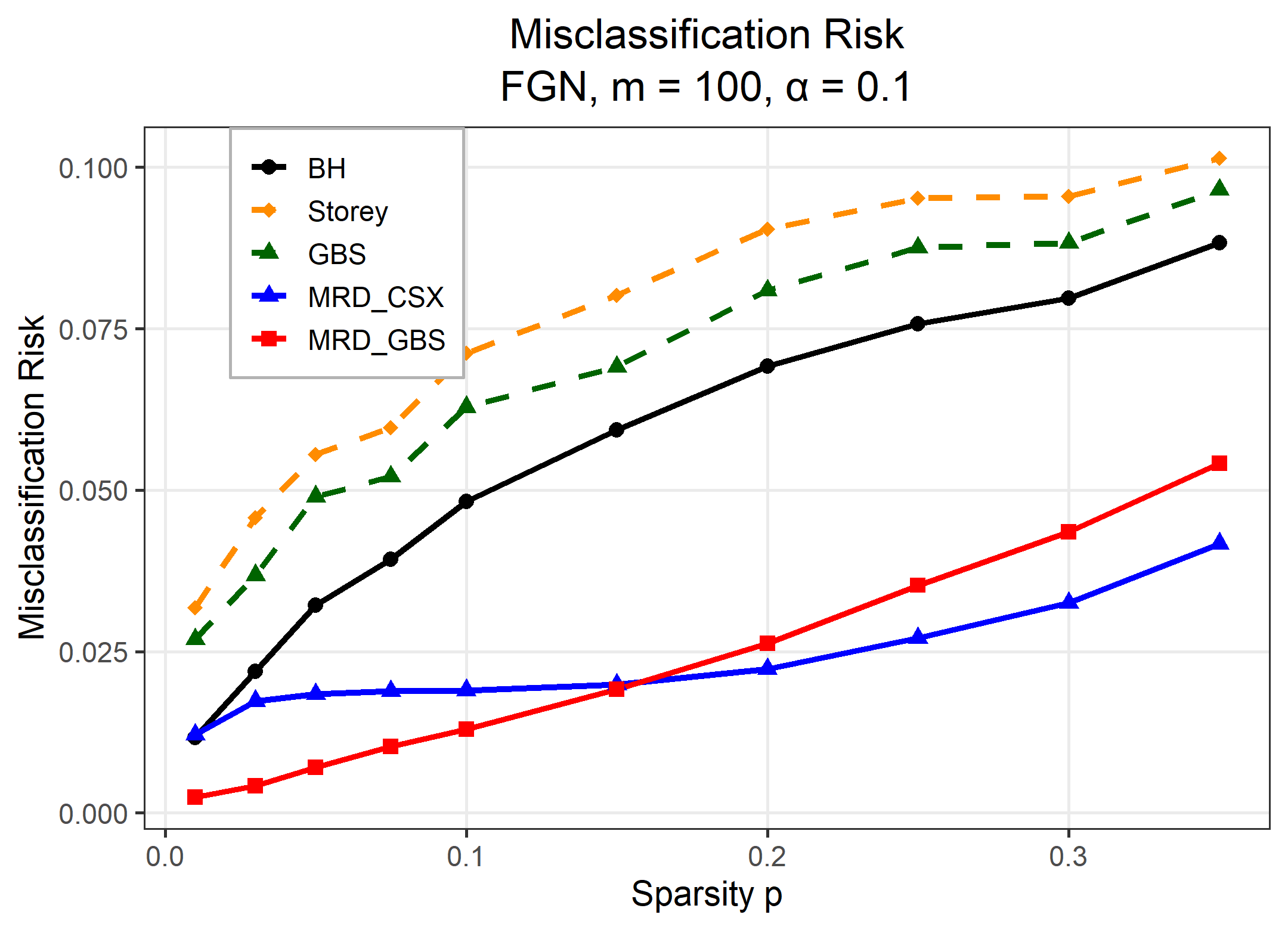}
		\caption{Fractional Gaussian Noise (\(H=0.9\))}
	\end{subfigure}
	
	\vspace{0.3cm}
	
	\begin{subfigure}[b]{0.48\textwidth}
		\centering
		\includegraphics[width=\textwidth]{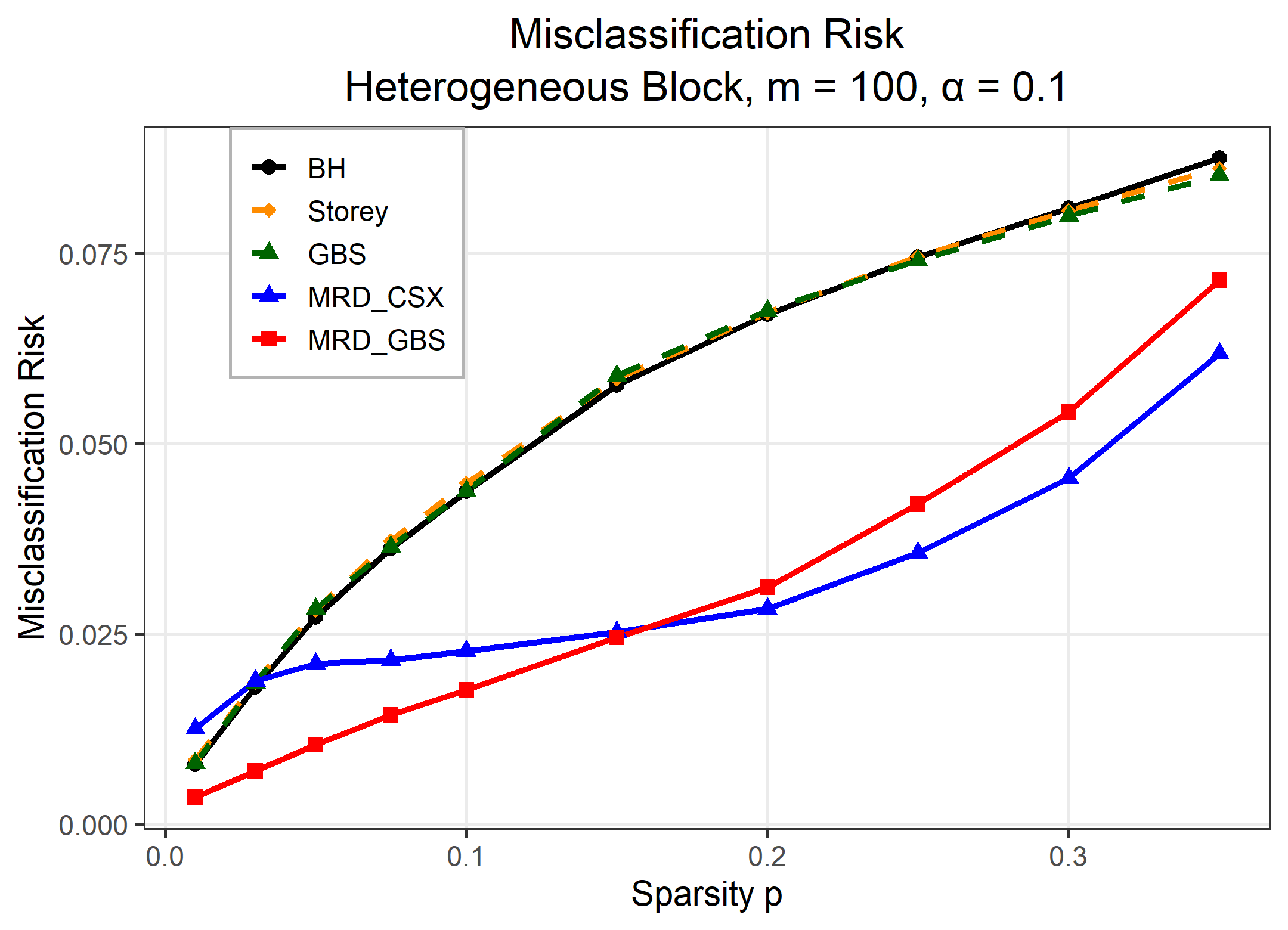}
		\caption{Heterogeneous Block}
	\end{subfigure}
	\hfill
	\begin{subfigure}[b]{0.48\textwidth}
		\centering
		\includegraphics[width=\textwidth]{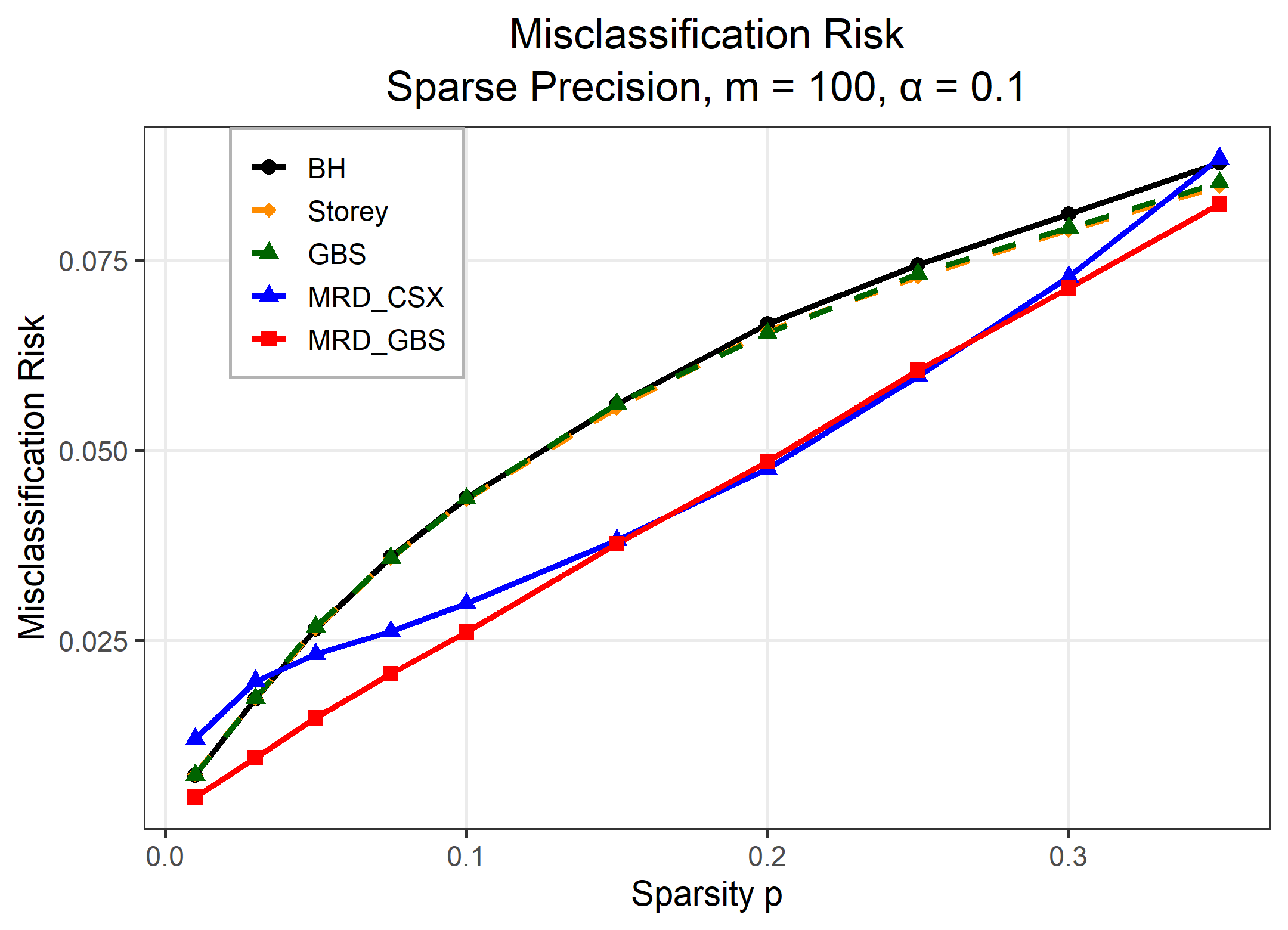}
		\caption{Sparse Precision Matrix}
	\end{subfigure}
	
\caption{
	Normalized misclassification rates (NMR) for the competing multiple testing
	procedures under six representative dependence structures when $m=100$.
	The overall qualitative behavior is broadly consistent with that observed
	for $m=200$, with the GBS-calibrated MRD procedure frequently achieving
	the smallest misclassification rates throughout sparse and moderately
	sparse regimes. However, several of the signal-recovery and classification
	advantages become substantially more pronounced in the larger-dimensional
	experiments reported in the main text.
}

	\label{FIG_NMR_m100}
\end{figure}

\begin{table}[htbp]
	\centering
	\caption{Normalized misclassification rates under the equicorrelation model
		(\(\rho=0.7\)). Smaller values indicate better performance.}
	\label{TAB_NMR_EQ_m100}
	\begin{tabular}{lccccc}
		\toprule
		$p$ & BH & Storey-BH & GBS & MRD-CSX & MRD-GBS \\
		\midrule
		0.01  & 0.0231 & 0.1196 & 0.1398 & 0.0130 & \textbf{0.0024} \\
		0.03  & 0.0341 & 0.1360 & 0.1528 & 0.0181 & \textbf{0.0042} \\
		0.05  & 0.0433 & 0.1456 & 0.1627 & 0.0192 & \textbf{0.0069} \\
		0.075 & 0.0478 & 0.1423 & 0.1583 & 0.0190 & \textbf{0.0098} \\
		0.10  & 0.0581 & 0.1559 & 0.1687 & 0.0187 & \textbf{0.0122} \\
		0.15  & 0.0652 & 0.1471 & 0.1593 & 0.0183 & \textbf{0.0173} \\
		0.20  & 0.0745 & 0.1503 & 0.1633 & \textbf{0.0174} & 0.0228 \\
		0.25  & 0.0799 & 0.1508 & 0.1590 & \textbf{0.0173} & 0.0289 \\
		0.30  & 0.0820 & 0.1447 & 0.1503 & \textbf{0.0166} & 0.0338 \\
		0.35  & 0.0909 & 0.1422 & 0.1493 & \textbf{0.0165} & 0.0393 \\
		\bottomrule
	\end{tabular}
\end{table}

\begin{table}[htbp]
	\centering
	\caption{Normalized misclassification rates under the factor model.
		Smaller values indicate better performance.}
	\label{TAB_NMR_FACTOR_m100}
	\begin{tabular}{lccccc}
		\toprule
		$p$ & BH & Storey-BH & GBS & MRD-CSX & MRD-GBS \\
		\midrule
		0.01  & 0.0106 & 0.0154 & 0.0117 & 0.0135 & \textbf{0.0024} \\
		0.03  & 0.0205 & 0.0268 & 0.0224 & 0.0182 & \textbf{0.0041} \\
		0.05  & 0.0305 & 0.0366 & 0.0326 & 0.0189 & \textbf{0.0068} \\
		0.075 & 0.0393 & 0.0464 & 0.0417 & 0.0190 & \textbf{0.0095} \\
		0.10  & 0.0464 & 0.0534 & 0.0488 & 0.0186 & \textbf{0.0120} \\
		0.15  & 0.0582 & 0.0654 & 0.0602 & 0.0179 & \textbf{0.0172} \\
		0.20  & 0.0679 & 0.0753 & 0.0696 & \textbf{0.0173} & 0.0223 \\
		0.25  & 0.0758 & 0.0824 & 0.0774 & \textbf{0.0174} & 0.0283 \\
		0.30  & 0.0827 & 0.0891 & 0.0835 & \textbf{0.0166} & 0.0331 \\
		0.35  & 0.0882 & 0.0943 & 0.0901 & \textbf{0.0182} & 0.0393 \\
		\bottomrule
	\end{tabular}
\end{table}

\begin{table}[htbp]
	\centering
	\caption{Normalized misclassification rates under the fractional Gaussian noise
		model (\(H=0.9\)). Smaller values indicate better performance.}
	\label{TAB_NMR_FGN_m100}
	\begin{tabular}{lccccc}
		\toprule
		$p$ & BH & Storey-BH & GBS & MRD-CSX & MRD-GBS \\
		\midrule
		0.01  & 0.0117 & 0.0318 & 0.0269 & 0.0121 & \textbf{0.0024} \\
		0.03  & 0.0219 & 0.0457 & 0.0368 & 0.0174 & \textbf{0.0042} \\
		0.05  & 0.0322 & 0.0556 & 0.0490 & 0.0184 & \textbf{0.0070} \\
		0.075 & 0.0393 & 0.0597 & 0.0521 & 0.0188 & \textbf{0.0103} \\
		0.10  & 0.0482 & 0.0712 & 0.0629 & 0.0190 & \textbf{0.0130} \\
		0.15  & 0.0594 & 0.0802 & 0.0691 & 0.0199 & \textbf{0.0191} \\
		0.20  & 0.0692 & 0.0905 & 0.0809 & \textbf{0.0223} & 0.0263 \\
		0.25  & 0.0757 & 0.0952 & 0.0875 & \textbf{0.0271} & 0.0353 \\
		0.30  & 0.0798 & 0.0955 & 0.0883 & \textbf{0.0326} & 0.0436 \\
		0.35  & 0.0883 & 0.1014 & 0.0965 & \textbf{0.0417} & 0.0541 \\
		\bottomrule
	\end{tabular}
\end{table}

\begin{table}[htbp]
	\centering
	\caption{Normalized misclassification rates under the Toeplitz model
		(\(\rho=0.9\)). Smaller values indicate better performance.}
	\label{TAB_NMR_TOE_m100}
	\begin{tabular}{lccccc}
		\toprule
		$p$ & BH & Storey-BH & GBS & MRD-CSX & MRD-GBS \\
		\midrule
		0.01  & 0.0183 & 0.0823 & 0.0911 & 0.0122 & \textbf{0.0023} \\
		0.03  & 0.0294 & 0.0979 & 0.1011 & 0.0174 & \textbf{0.0039} \\
		0.05  & 0.0396 & 0.1070 & 0.1118 & 0.0180 & \textbf{0.0065} \\
		0.075 & 0.0441 & 0.1069 & 0.1139 & 0.0181 & \textbf{0.0093} \\
		0.10  & 0.0548 & 0.1176 & 0.1227 & 0.0181 & \textbf{0.0119} \\
		0.15  & 0.0632 & 0.1199 & 0.1240 & 0.0179 & \textbf{0.0173} \\
		0.20  & 0.0720 & 0.1260 & 0.1281 & \textbf{0.0196} & 0.0244 \\
		0.25  & 0.0790 & 0.1255 & 0.1279 & \textbf{0.0239} & 0.0340 \\
		0.30  & 0.0813 & 0.1246 & 0.1243 & \textbf{0.0301} & 0.0449 \\
		0.35  & 0.0893 & 0.1262 & 0.1292 & \textbf{0.0416} & 0.0598 \\
		\bottomrule
	\end{tabular}
\end{table}

\begin{table}[htbp]
	\centering
	\caption{Normalized misclassification rates under the heterogeneous block
		covariance model. Smaller values indicate better performance.}
	\label{TAB_NMR_BLOCK_m100}
	\begin{tabular}{lccccc}
		\toprule
		$p$ & BH & Storey-BH & GBS & MRD-CSX & MRD-GBS \\
		\midrule
		0.01  & 0.0079 & 0.0085 & 0.0080 & 0.0126 & \textbf{0.0036} \\
		0.03  & 0.0181 & 0.0187 & 0.0186 & 0.0189 & \textbf{0.0070} \\
		0.05  & 0.0272 & 0.0282 & 0.0283 & 0.0211 & \textbf{0.0104} \\
		0.075 & 0.0362 & 0.0373 & 0.0365 & 0.0216 & \textbf{0.0144} \\
		0.10  & 0.0438 & 0.0448 & 0.0438 & 0.0227 & \textbf{0.0177} \\
		0.15  & 0.0577 & 0.0586 & 0.0590 & 0.0253 & \textbf{0.0245} \\
		0.20  & 0.0670 & 0.0673 & 0.0675 & \textbf{0.0283} & 0.0312 \\
		0.25  & 0.0746 & 0.0746 & 0.0741 & \textbf{0.0357} & 0.0421 \\
		0.30  & 0.0810 & 0.0807 & 0.0800 & \textbf{0.0455} & 0.0542 \\
		0.35  & 0.0875 & 0.0862 & 0.0853 & \textbf{0.0619} & 0.0715 \\
		\bottomrule
	\end{tabular}
\end{table}

\begin{table}[htbp]
	\centering
	\caption{Normalized misclassification rates under the sparse precision-matrix
		model. Smaller values indicate better performance.}
	\label{TAB_NMR_SPARSE_m100}
	\begin{tabular}{lccccc}
		\toprule
		$p$ & BH & Storey-BH & GBS & MRD-CSX & MRD-GBS \\
		\midrule
		0.01  & 0.0073 & 0.0074 & 0.0073 & 0.0121 & \textbf{0.0044} \\
		0.03  & 0.0174 & 0.0174 & 0.0174 & 0.0196 & \textbf{0.0096} \\
		0.05  & 0.0265 & 0.0266 & 0.0268 & 0.0232 & \textbf{0.0148} \\
		0.075 & 0.0360 & 0.0359 & 0.0359 & 0.0262 & \textbf{0.0206} \\
		0.10  & 0.0437 & 0.0436 & 0.0437 & 0.0299 & \textbf{0.0261} \\
		0.15  & 0.0560 & 0.0556 & 0.0561 & 0.0382 & \textbf{0.0377} \\
		0.20  & 0.0666 & 0.0657 & 0.0654 & \textbf{0.0476} & 0.0486 \\
		0.25  & 0.0745 & 0.0729 & 0.0733 & \textbf{0.0597} & 0.0605 \\
		0.30  & 0.0811 & 0.0790 & 0.0793 & 0.0728 & \textbf{0.0714} \\
		0.35  & 0.0878 & 0.0848 & 0.0853 & 0.0884 & \textbf{0.0824} \\
		\bottomrule
	\end{tabular}
\end{table}

\clearpage

\subsection*{False discovery rates ($m=100$)}

%Figure~\ref{FIG_FDR_m100} summarizes the empirical false discovery rates
%of the competing procedures under the six dependence structures
%considered in this study. As in the $m=200$ experiments, the proposed
%GBS-calibrated MRD procedure generally maintains FDR values close to the
%nominal level and substantially improves upon the original MRD procedure
%across a broad range of sparsity levels.

Figure~\ref{FIG_FDR_m100} summarizes the empirical false discovery rates
of the competing procedures under the six dependence structures
considered in this study. Consistent with the $m=200$ results, the
proposed GBS-calibrated MRD procedure generally maintains FDR values
close to the nominal level while substantially reducing the excessive
FDR inflation often exhibited by the original MRD procedure.

\begin{figure}[p]
	\centering
	
	\begin{tabular}{cc}
		
		\includegraphics[width=0.47\textwidth]
		{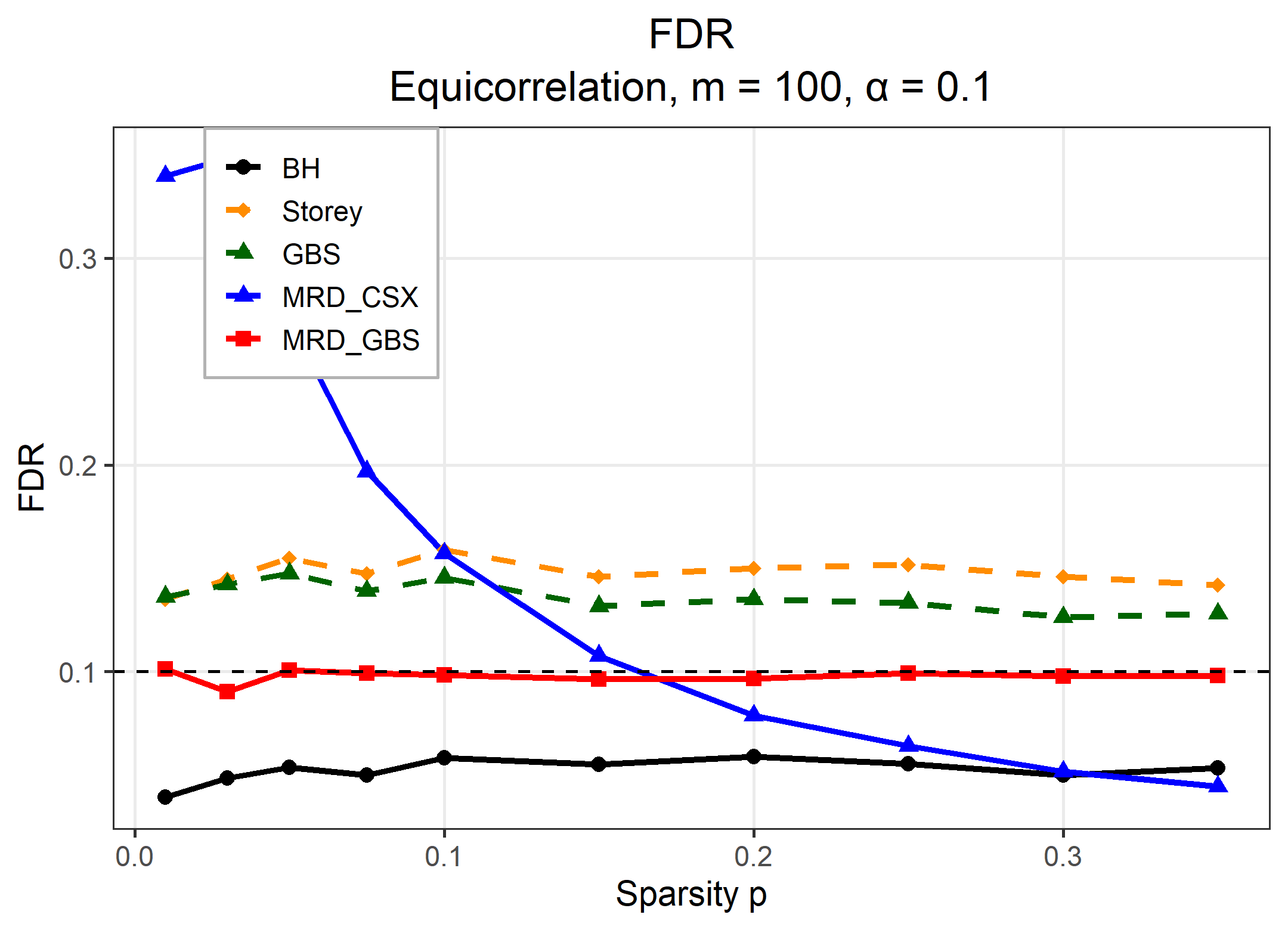}
		&
		\includegraphics[width=0.47\textwidth]
		{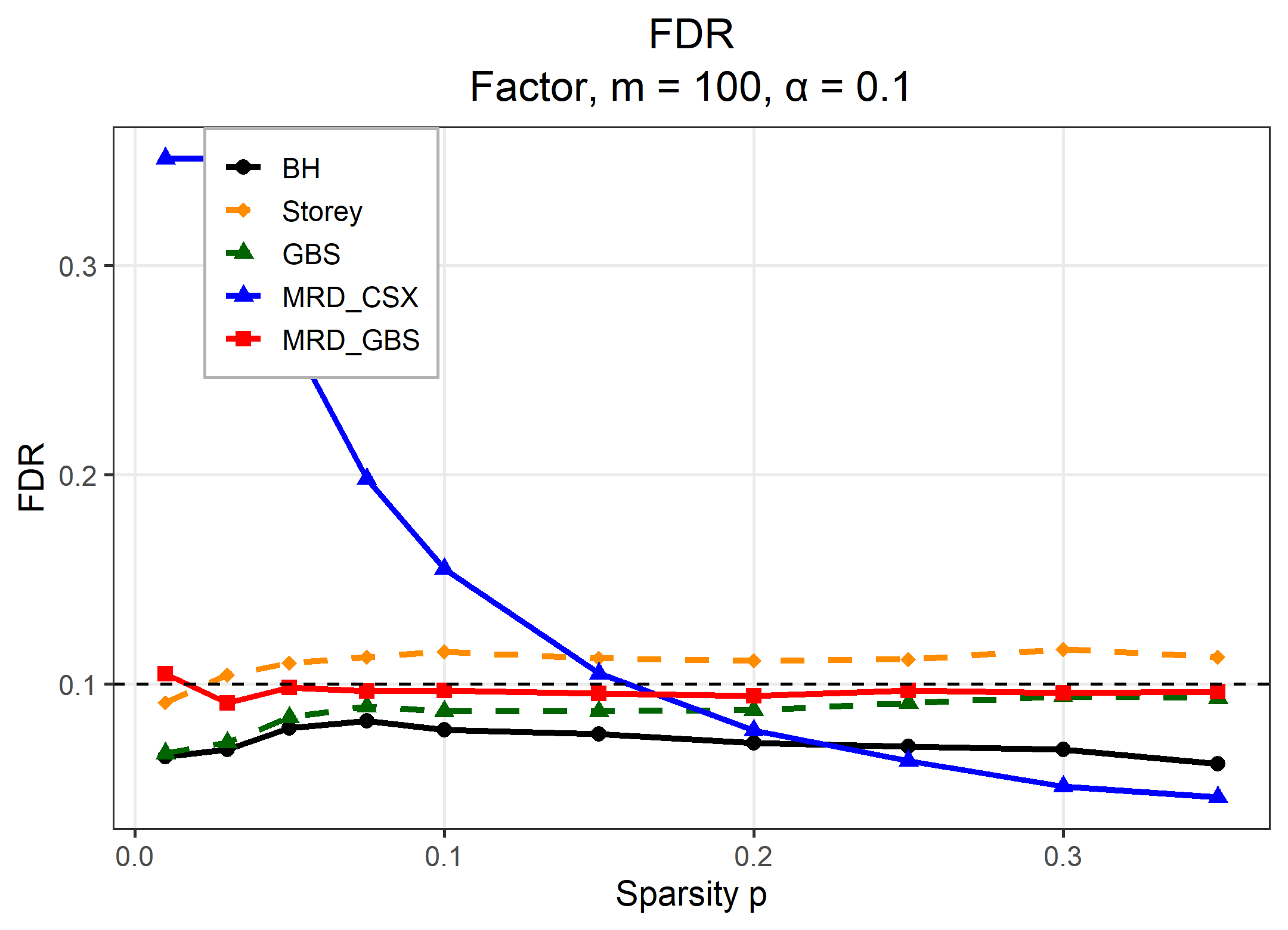}
		\\
		
		(a) Equicorrelation ($\rho=0.7$)
		&
		(b) Factor Model
		\\[0.3cm]
		
		\includegraphics[width=0.47\textwidth]
		{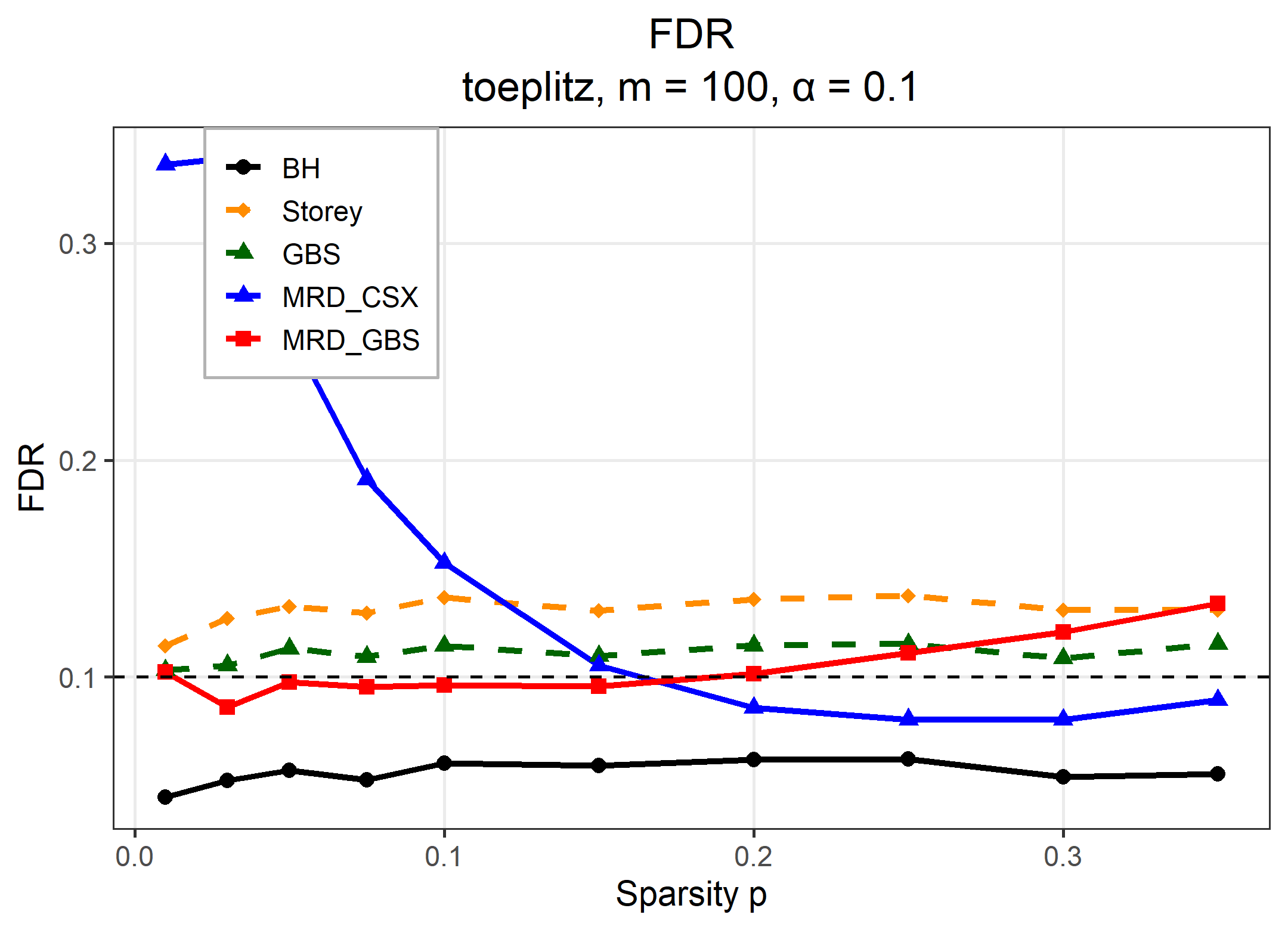}
		&
		\includegraphics[width=0.47\textwidth]
		{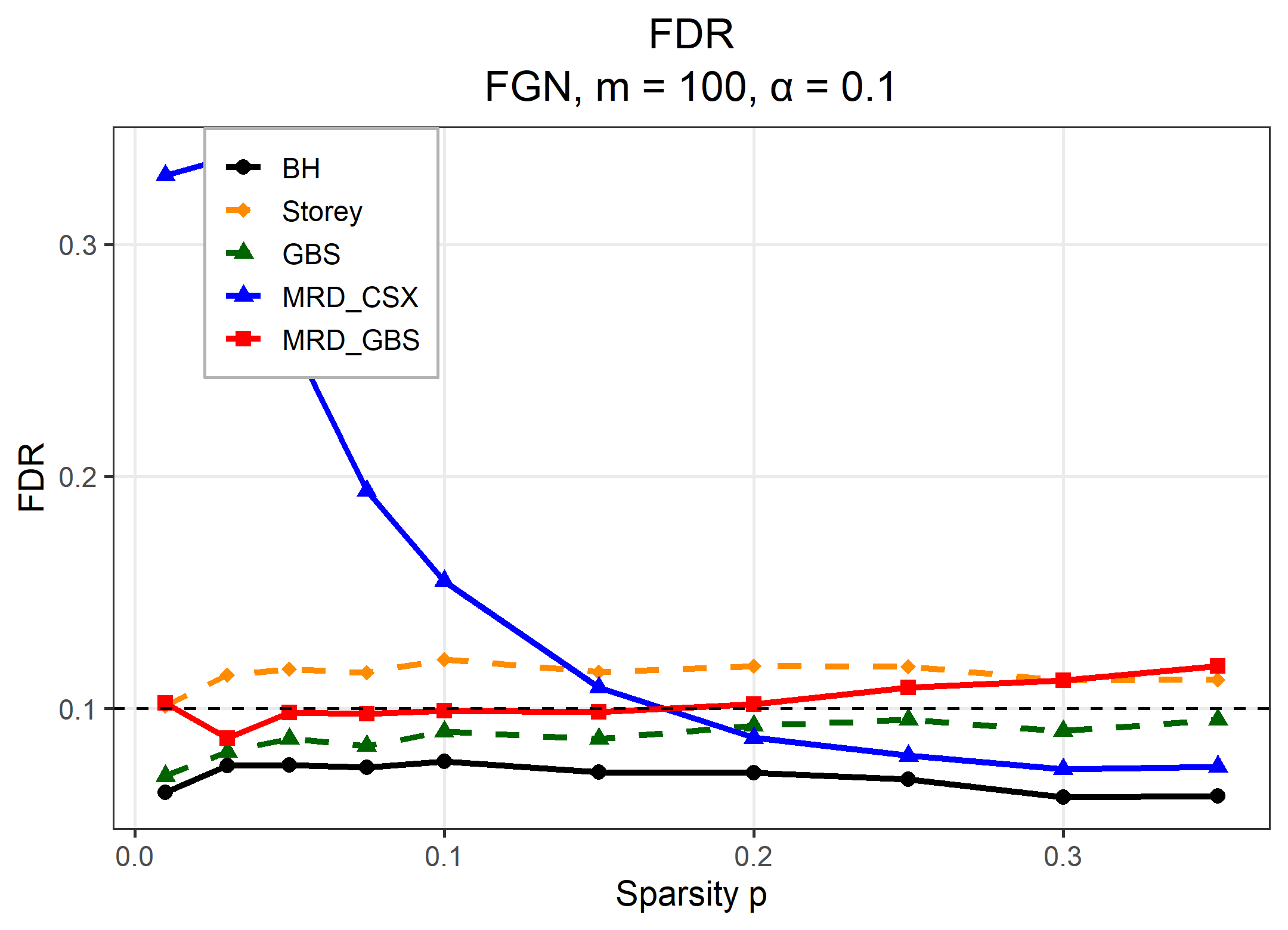}
		\\
		
		(c) Toeplitz ($\rho=0.9$)
		&
		(d) Fractional Gaussian Noise ($H=0.9$)
		\\[0.3cm]
		
		\includegraphics[width=0.47\textwidth]
		{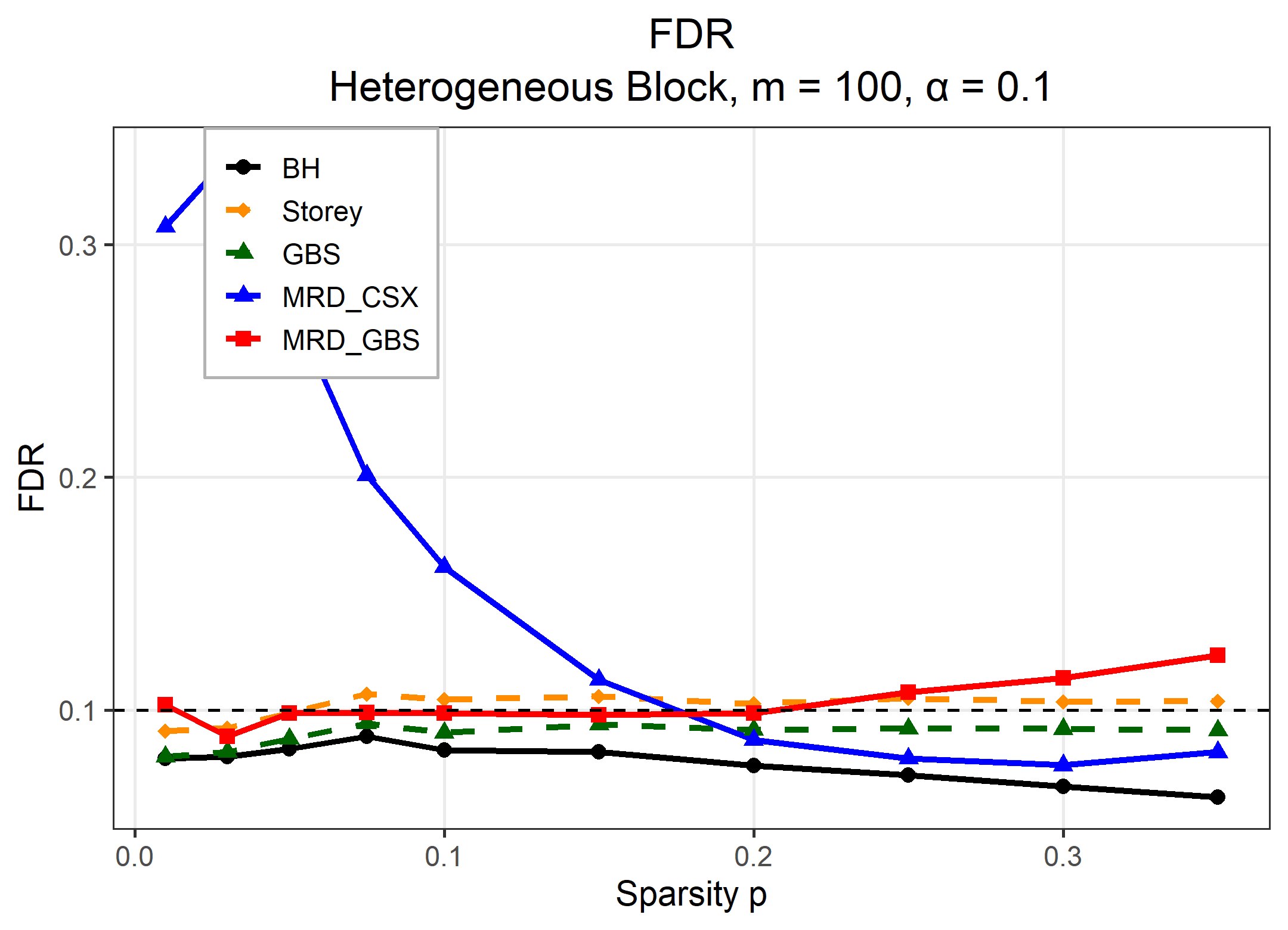}
		&
		\includegraphics[width=0.47\textwidth]
		{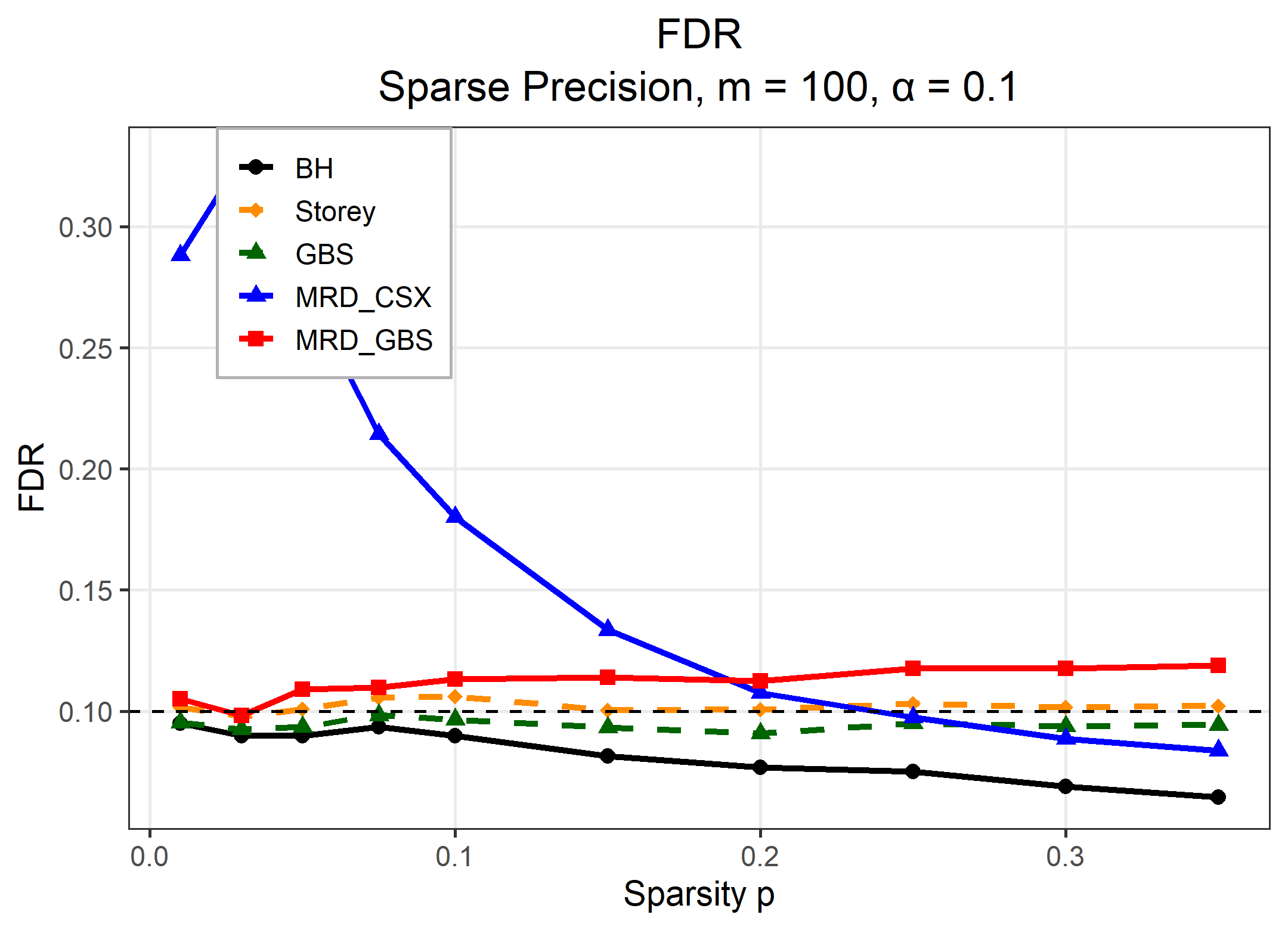}
		\\
		
		(e) Heterogeneous Block
		&
		(f) Sparse Precision Matrix
		
	\end{tabular}
	
\caption{
	Empirical false discovery rates (FDR) for the competing multiple testing
	procedures under six representative dependence structures when $m=100$.
	The overall trends are broadly consistent with those observed for
	$m=200$, with the proposed GBS-calibrated MRD procedure generally
	exhibiting substantially improved FDR behavior relative to the original
	MRD procedure. However, the benefits of the proposed calibration become
	considerably more pronounced in the larger-dimensional experiments
	reported in the main text.
}

	\label{FIG_FDR_m100}
	
\end{figure}

\begin{table}[htbp]
	\centering
	\caption{False discovery rates under the Equicorrelation model ($\rho=0.7$). Values close to the nominal level $\alpha=0.10$ indicate favorable false-discovery behavior. Boldfaced entries correspond to the values closest to the nominal level.}
	\label{TAB_FDR_EQ}
	\begin{tabular}{lccccc}
		\toprule
		$p$ & BH & Storey-BH & GBS & MRD-CSX & MRD-GBS \\
		\midrule
0.01 & 0.0392 & 0.1349 & 0.1363 & 0.3398 & \textbf{0.1013} \\
0.03 & 0.0485 & 0.1448 & 0.1422 & 0.3483 & \textbf{0.0902} \\
0.05 & 0.0538 & 0.1549 & 0.1475 & 0.2707 & \textbf{0.1006} \\
0.075 & 0.0498 & 0.1475 & 0.1390 & 0.1969 & \textbf{0.0992} \\
0.10 & 0.0582 & 0.1589 & 0.1455 & 0.1572 & \textbf{0.0984} \\
0.15 & 0.0550 & 0.1459 & 0.1317 & 0.1077 & \textbf{0.0962} \\
0.20 & 0.0588 & 0.1499 & 0.1351 & 0.0788 & \textbf{0.0965} \\
0.25 & 0.0554 & 0.1518 & 0.1334 & 0.0640 & \textbf{0.0992} \\
0.30 & 0.0500 & 0.1459 & 0.1262 & 0.0517 & \textbf{0.0979} \\
0.35 & 0.0533 & 0.1420 & 0.1280 & 0.0444 & \textbf{0.0982} \\
		\bottomrule
	\end{tabular}
\end{table}

\begin{table}[htbp]
	\centering
	\caption{False discovery rates under the factor model. Values close to the nominal level $\alpha=0.10$ indicate favorable
false-discovery behavior. Boldfaced entries correspond to the values
closest to the nominal level.}
	\label{TAB_FDR_FACTOR_m100}
	\begin{tabular}{lccccc}
		\toprule
		$p$ & BH & Storey-BH & GBS & MRD-CSX & MRD-GBS \\
		\midrule
0.01 & 0.0655 & 0.0912 & 0.0669 & 0.3513 & \textbf{0.1050} \\
0.03 & 0.0687 & \textbf{0.1045} & 0.0723 & 0.3510 & 0.0910 \\
0.05 & 0.0791 & 0.1100 & 0.0842 & 0.2675 & \textbf{0.0985} \\
0.075 & 0.0825 & 0.1129 & 0.0894 & 0.1981 & \textbf{0.0968} \\
0.10 & 0.0783 & 0.1155 & 0.0870 & 0.1550 & \textbf{0.0970} \\
0.15 & 0.0762 & 0.1125 & 0.0871 & 0.1051 & \textbf{0.0955} \\
0.20 & 0.0720 & 0.1113 & 0.0876 & 0.0779 & \textbf{0.0944} \\
0.25 & 0.0701 & 0.1117 & 0.0909 & 0.0635 & \textbf{0.0970} \\
0.30 & 0.0688 & 0.1165 & 0.0942 & 0.0511 & \textbf{0.0957} \\
0.35 & 0.0620 & 0.1129 & 0.0934 & 0.0460 & \textbf{0.0965} \\
		\bottomrule
	\end{tabular}
\end{table}

\begin{table}[htbp]
	\centering
	\caption{False discovery rates under the fractional Gaussian noise model ($H=0.9$). Values close to the nominal level $\alpha=0.10$ indicate favorable
false-discovery behavior. Boldfaced entries correspond to the values
closest to the nominal level.}
	\label{TAB_FDR_FGN_m100}
	\begin{tabular}{lccccc}
		\toprule
		$p$ & BH & Storey-BH & GBS & MRD-CSX & MRD-GBS \\
		\midrule
0.01 & 0.0640 & \textbf{0.1011} & 0.0710 & 0.3297 & 0.1025 \\
0.03 & 0.0755 & 0.1144 & 0.0813 & 0.3370 & \textbf{0.0874} \\
0.05 & 0.0758 & 0.1171 & 0.0871 & 0.2613 & \textbf{0.0984} \\
0.075 & 0.0747 & 0.1154 & 0.0840 & 0.1939 & \textbf{0.0977} \\
0.10 & 0.0773 & 0.1212 & 0.0900 & 0.1549 & \textbf{0.0991} \\
0.15 & 0.0727 & 0.1159 & 0.0869 & 0.1091 & \textbf{0.0985} \\
0.20 & 0.0724 & 0.1184 & 0.0926 & 0.0874 & \textbf{0.1020} \\
0.25 & 0.0696 & 0.1180 & \textbf{0.0952} & 0.0797 & 0.1090 \\
0.30 & 0.0618 & 0.1121 & \textbf{0.0903} & 0.0739 & 0.1121 \\
0.35 & 0.0624 & 0.1124 & \textbf{0.0951} & 0.0751 & 0.1184 \\
		\bottomrule
	\end{tabular}
\end{table}

\begin{table}[htbp]
	\centering
	\caption{False discovery rates under the Toeplitz model ($\rho=0.9$). Values close to the nominal level $\alpha=0.10$ indicate favorable
		false-discovery behavior. Boldfaced entries correspond to the values
		closest to the nominal level.}
	\label{TAB_FDR_TOE_m100}
	\begin{tabular}{lccccc}
		\toprule
		$p$ & BH & Storey-BH & GBS & MRD-CSX & MRD-GBS \\
		\midrule
		0.01 & 0.0446 & 0.1144 & 0.1031 & 0.3364 & \textbf{0.1023} \\
		0.03 & 0.0523 & 0.1269 & \textbf{0.1054} & 0.3392 & 0.0862 \\
		0.05 & 0.0570 & 0.1324 & 0.1132 & 0.2595 & \textbf{0.0977} \\
		0.075 & 0.0527 & 0.1296 & 0.1093 & 0.1912 & \textbf{0.0954} \\
		0.10 & 0.0603 & 0.1366 & 0.1144 & 0.1526 & \textbf{0.0964} \\
		0.15 & 0.0590 & 0.1307 & 0.1097 & 0.1053 & \textbf{0.0959} \\
		0.20 & 0.0618 & 0.1359 & 0.1146 & 0.0857 & \textbf{0.1015} \\
		0.25 & 0.0622 & 0.1375 & 0.1152 & 0.0804 & \textbf{0.1110} \\
		0.30 & 0.0539 & 0.1309 & \textbf{0.1086} & 0.0804 & 0.1208 \\
		0.35 & 0.0554 & 0.1309 & 0.1153 & \textbf{0.0893} & 0.1339 \\
		\bottomrule
	\end{tabular}
\end{table}

\begin{table}[htbp]
	\centering
	\caption{False discovery rates under the heterogeneous block covariance model. Values close to the nominal level $\alpha=0.10$ indicate favorable
		false-discovery behavior. Boldfaced entries correspond to the values
		closest to the nominal level.}
	\label{TAB_FDR_BLOCK_m100}
	\begin{tabular}{lccccc}
		\toprule
		$p$ & BH & Storey-BH & GBS & MRD-CSX & MRD-GBS \\
		\midrule
		0.01 & 0.0792 & 0.0908 & 0.0799 & 0.3077 & \textbf{0.1021} \\
		0.03 & 0.0799 & \textbf{0.0921} & 0.0821 & 0.3371 & 0.0888 \\
		0.05 & 0.0834 & \textbf{0.0990} & 0.0876 & 0.2757 & 0.0986 \\
		0.075 & 0.0888 & 0.1070 & 0.0942 & 0.2008 & \textbf{0.0989} \\
		0.10 & 0.0827 & 0.1046 & 0.0902 & 0.1615 & \textbf{0.0987} \\
		0.15 & 0.0820 & 0.1059 & 0.0939 & 0.1130 & \textbf{0.0978} \\
		0.20 & 0.0760 & 0.1027 & 0.0914 & 0.0871 & \textbf{0.0986} \\
		0.25 & 0.0720 & \textbf{0.1051} & 0.0919 & 0.0793 & 0.1075 \\
		0.30 & 0.0672 & \textbf{0.1036} & 0.0919 & 0.0763 & 0.1137 \\
		0.35 & 0.0625 & \textbf{0.1038} & 0.0912 & 0.0819 & 0.1234 \\
		\bottomrule
	\end{tabular}
\end{table}

\begin{table}[htbp]
	\centering
	\caption{False discovery rates under the sparse precision-matrix model. Values close to the nominal level $\alpha=0.10$ indicate favorable
		false-discovery behavior. Boldfaced entries correspond to the values
		closest to the nominal level.}
	\label{TAB_FDR_SPARSE_m100}
	\begin{tabular}{lccccc}
		\toprule
		$p$ & BH & Storey-BH & GBS & MRD-CSX & MRD-GBS \\
		\midrule
		0.01 & 0.0949 & \textbf{0.1015} & 0.0953 & 0.2881 & 0.1051 \\
		0.03 & 0.0897 & 0.0978 & 0.0925 & 0.3282 & \textbf{0.0982} \\
		0.05 & 0.0898 & \textbf{0.1009} & 0.0935 & 0.2764 & 0.1091 \\
		0.075 & 0.0935 & 0.1055 & \textbf{0.0983} & 0.2143 & 0.1098 \\
		0.10 & 0.0899 & 0.1060 & \textbf{0.0966} & 0.1801 & 0.1132 \\
		0.15 & 0.0814 & \textbf{0.1005} & 0.0933 & 0.1335 & 0.1139 \\
		0.20 & 0.0768 & \textbf{0.1007} & 0.0909 & 0.1076 & 0.1125 \\
		0.25 & 0.0750 & 0.1032 & 0.0951 & \textbf{0.0974} & 0.1177 \\
		0.30 & 0.0689 & \textbf{0.1017} & 0.0937 & 0.0886 & 0.1175 \\
		0.35 & 0.0644 & \textbf{0.1022} & 0.0942 & 0.0838 & 0.1188 \\
		\bottomrule
	\end{tabular}
\end{table}

\clearpage

\subsection*{False non-discovery rates ($m=100$)}

Figure~\ref{FIG_FNR_m100} summarizes the empirical false non-discovery rates
of the competing procedures under the six dependence structures
considered in this study. Consistent with the corresponding $m=200$
experiments, the proposed GBS-calibrated MRD procedure frequently
maintains very small FNR values across sparse and moderately sparse
regimes, providing further evidence of its strong signal-recovery
performance.

\begin{figure}[p]
	\centering
	
	\begin{tabular}{cc}
		
		\includegraphics[width=0.47\textwidth]
		{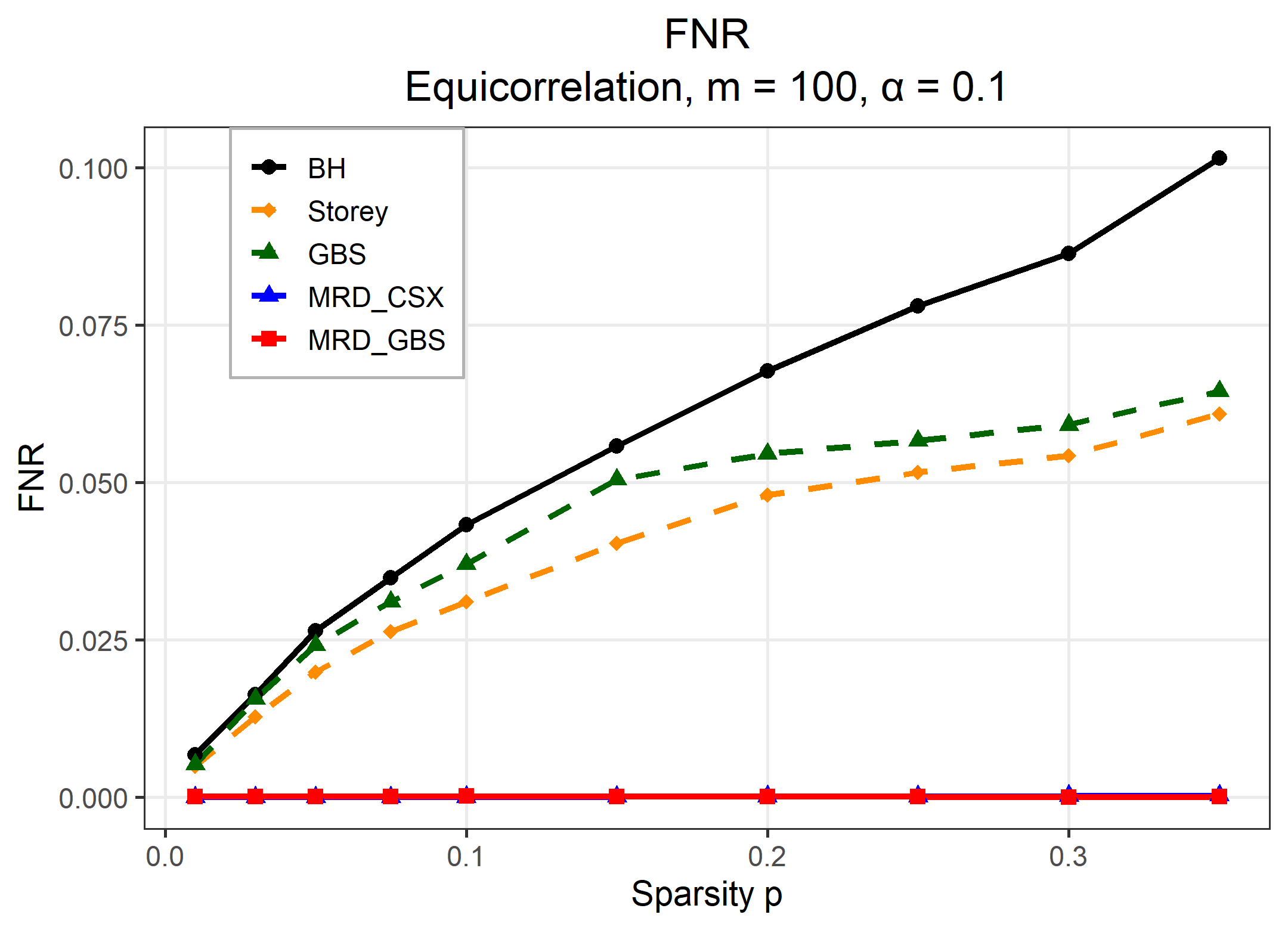}
		&
		\includegraphics[width=0.47\textwidth]
		{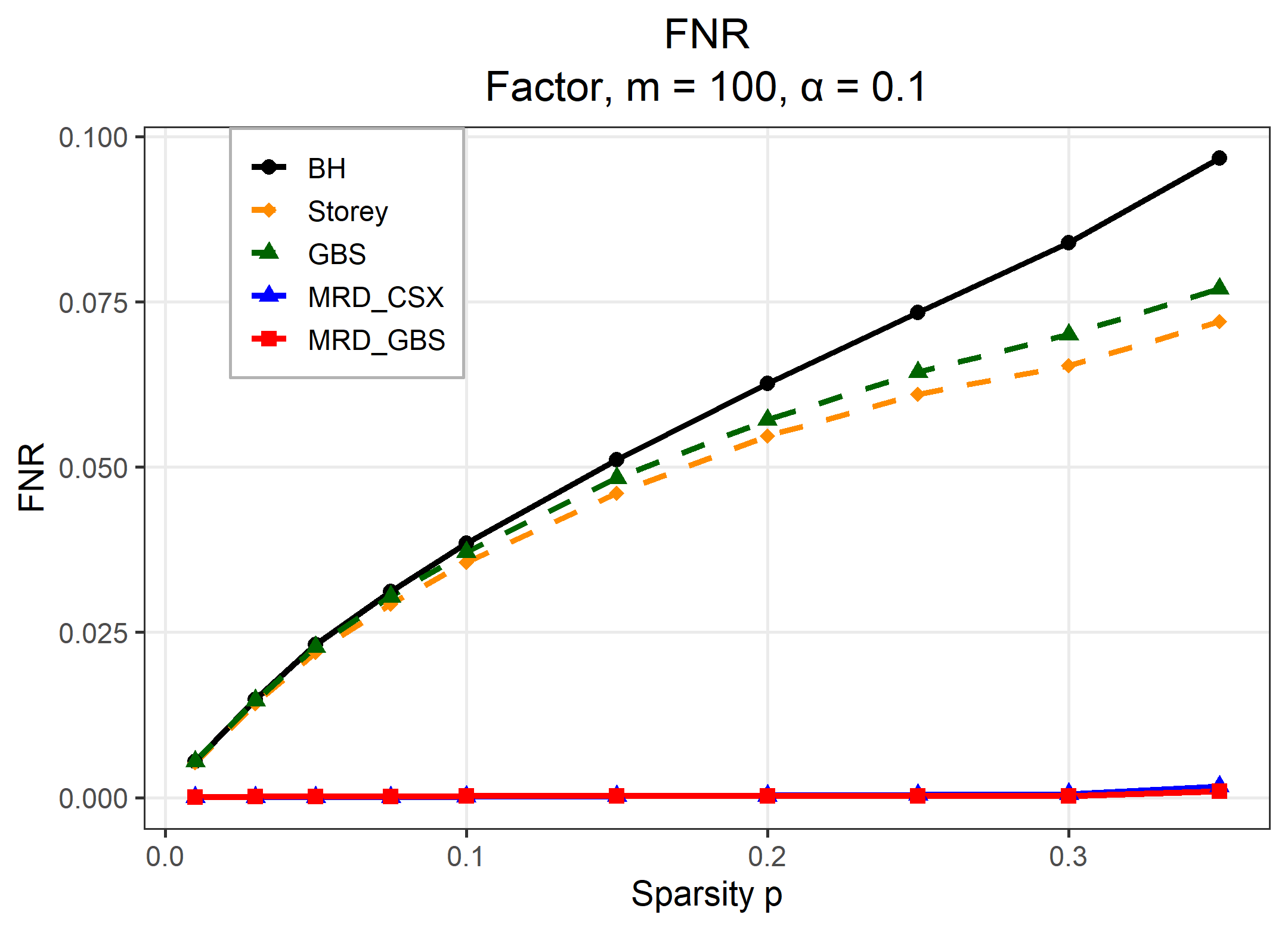}
		\\
		
		(a) Equicorrelation ($\rho=0.7$)
		&
		(b) Factor Model
		\\[0.3cm]
		
		\includegraphics[width=0.47\textwidth]
		{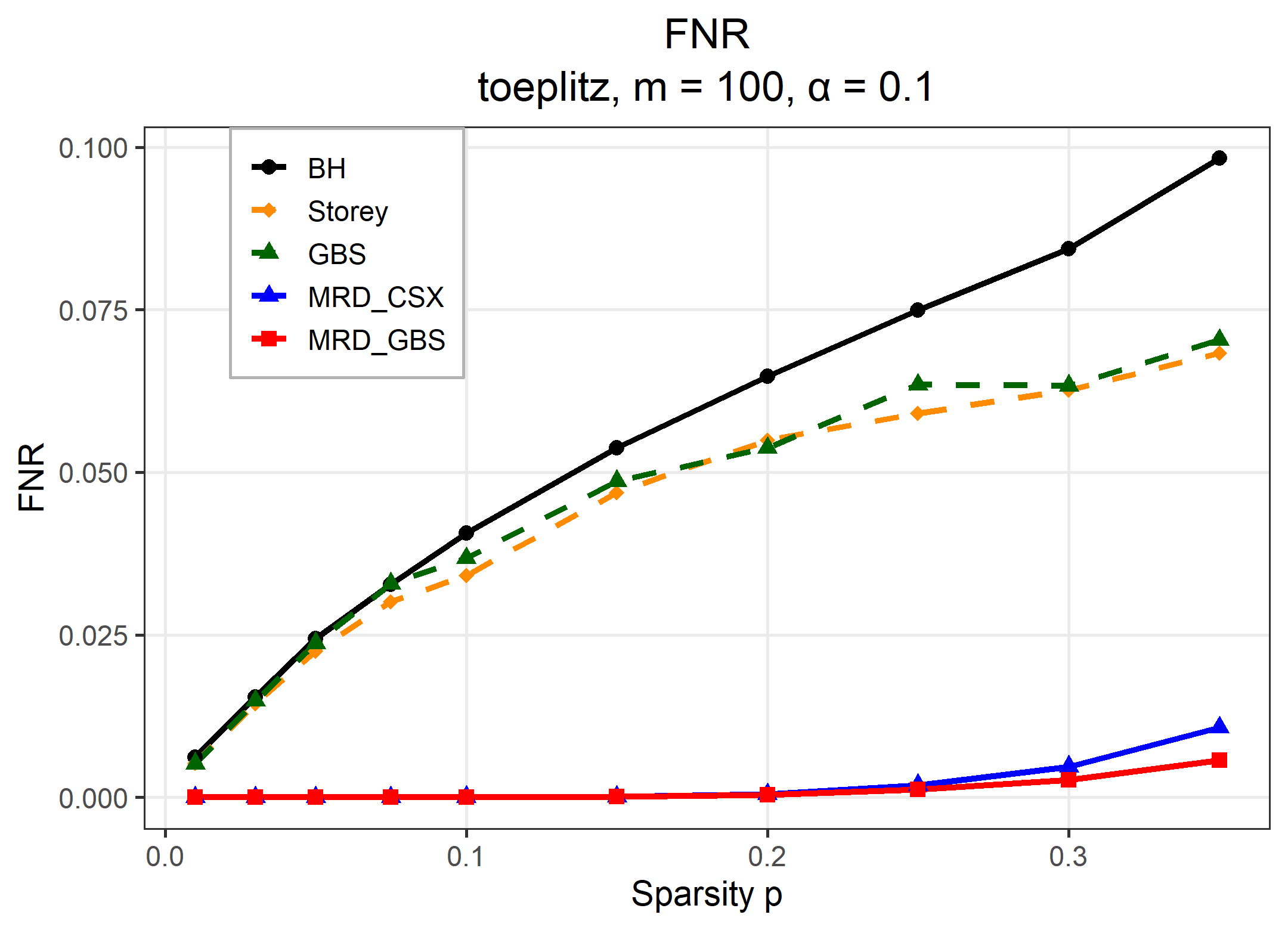}
		&
		\includegraphics[width=0.47\textwidth]
		{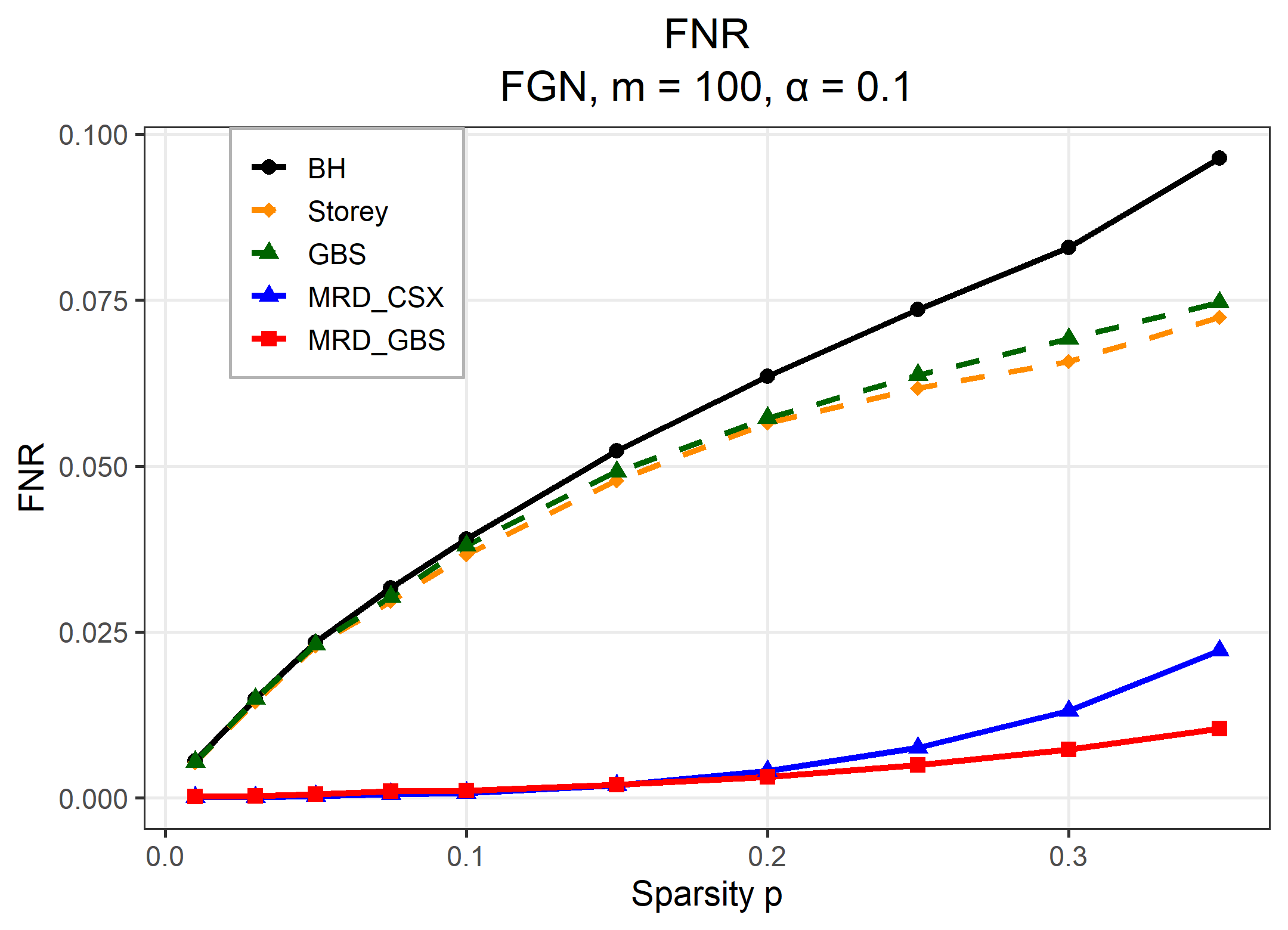}
		\\
		
		(c) Toeplitz ($\rho=0.9$)
		&
		(d) Fractional Gaussian Noise ($H=0.9$)
		\\[0.3cm]
		
		\includegraphics[width=0.47\textwidth]
		{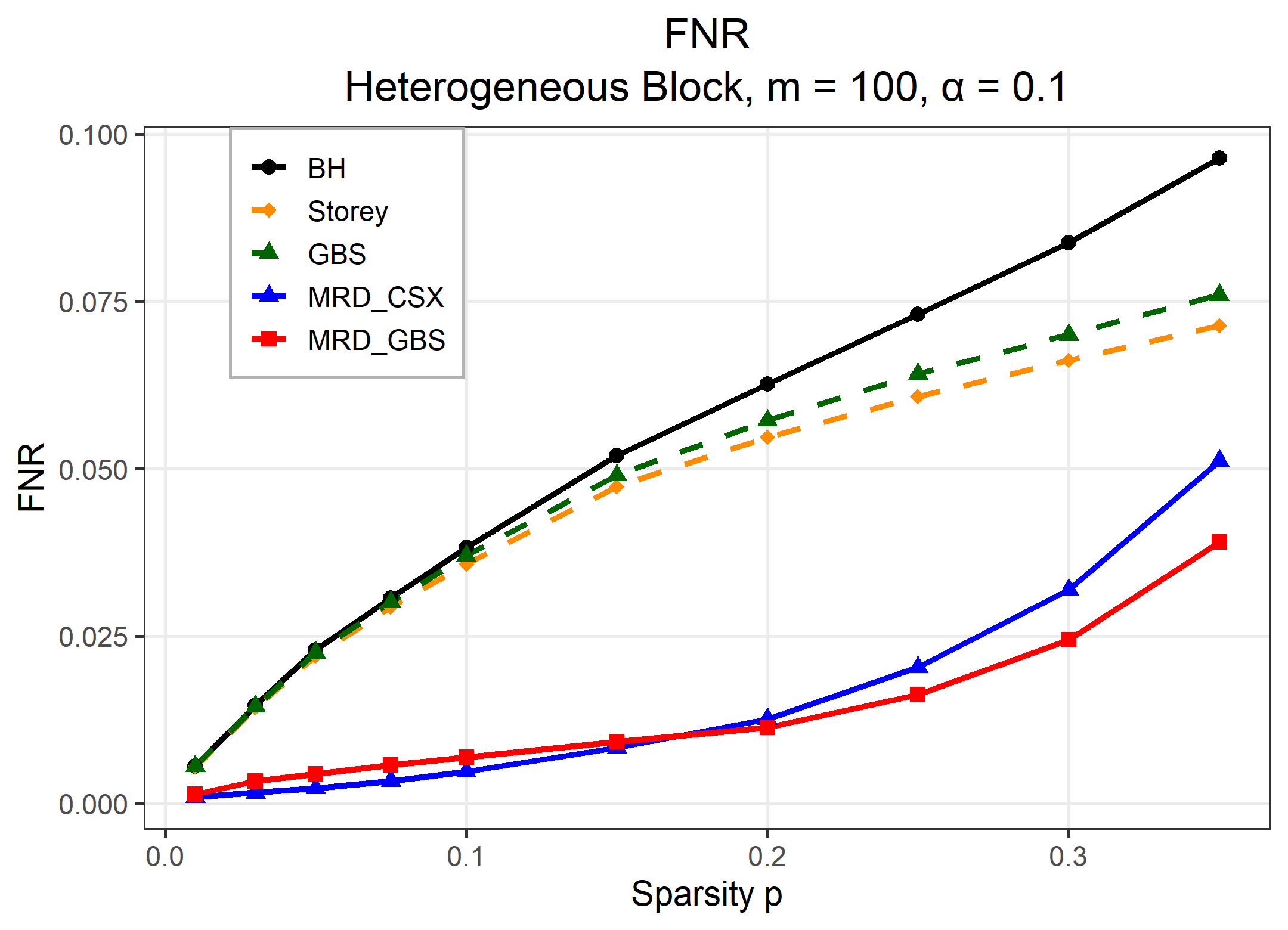}
		&
		\includegraphics[width=0.47\textwidth]
		{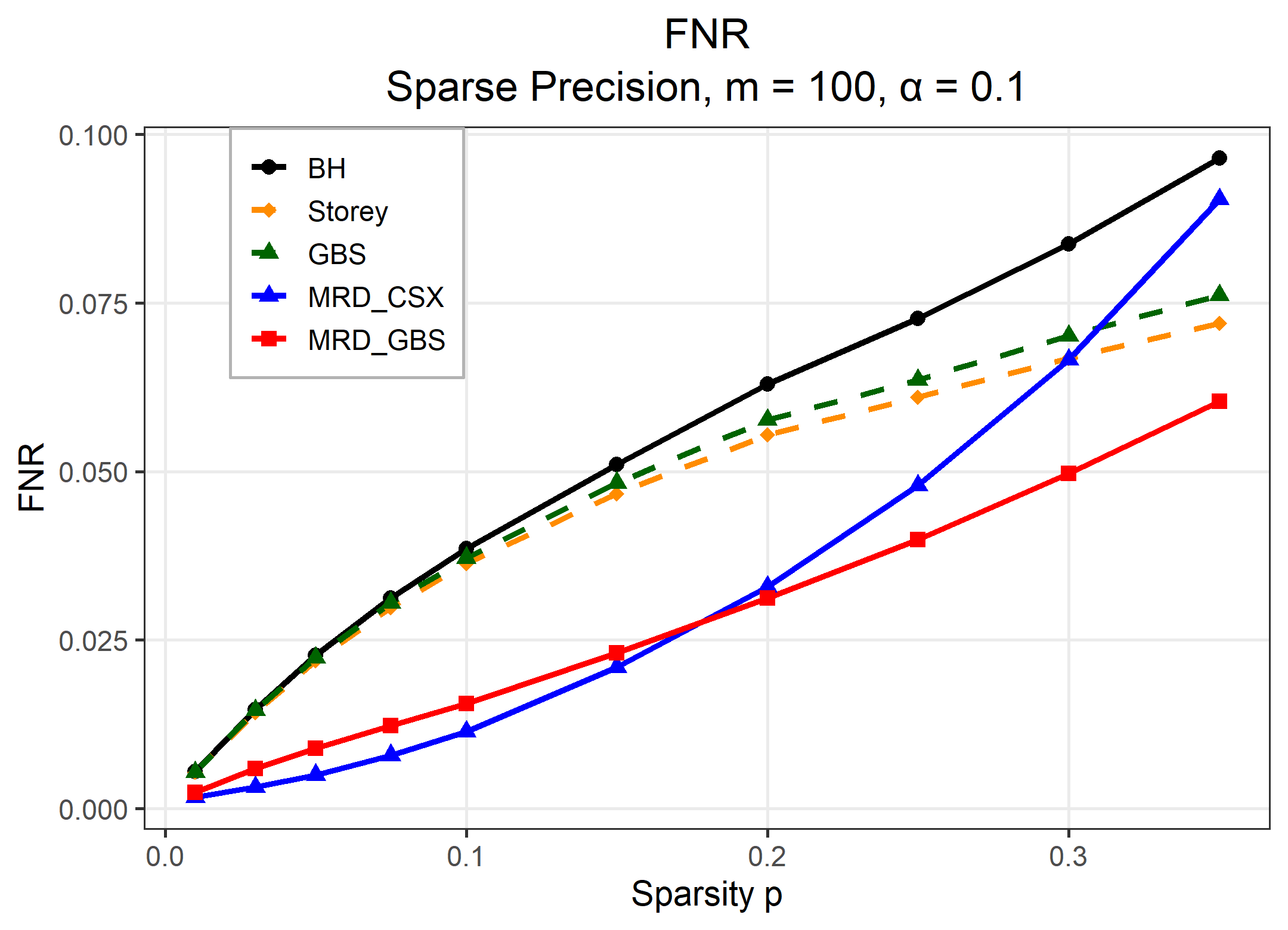}
		\\
		
		(e) Heterogeneous Block
		&
		(f) Sparse Precision Matrix
		
	\end{tabular}
	
\caption{
	Empirical false non-discovery rates (FNR) for the competing multiple
	testing procedures under six representative dependence structures when
	$m=100$. The overall behavior is qualitatively similar to that observed
	for $m=200$, with the proposed GBS-calibrated MRD procedure frequently
	exhibiting very small FNR values across a broad range of sparsity levels.
	As in the larger-dimensional experiments, the results suggest that the
	procedure is highly effective at identifying true signals while
	maintaining favorable false-discovery characteristics.
}

	\label{FIG_FNR_m100}
	
\end{figure}

\begin{table}[htbp]
	\centering
	\caption{False non-discovery rates under the Equicorrelation model ($\rho=0.7$). Smaller values indicate better performance.}
	\label{TAB_FNR_EQ_m100}
	\begin{tabular}{lccccc}
		\toprule
		$p$ & BH & Storey-BH & GBS & MRD-CSX & MRD-GBS \\
		\midrule
		0.01 & 0.0068 & 0.0050 & 0.0053 & \textbf{0.0001} & \textbf{0.0001} \\
		0.03 & 0.0164 & 0.0128 & 0.0156 & \textbf{0.0000} & 0.0001 \\
		0.05 & 0.0265 & 0.0200 & 0.0242 & \textbf{0.0000} & 0.0001 \\
		0.075 & 0.0349 & 0.0264 & 0.0311 & \textbf{0.0001} & 0.0002 \\
		0.10 & 0.0433 & 0.0311 & 0.0370 & \textbf{0.0001} & 0.0002 \\
		0.15 & 0.0558 & 0.0404 & 0.0505 & \textbf{0.0001} & 0.0002 \\
		0.20 & 0.0677 & 0.0481 & 0.0546 & \textbf{0.0002} & \textbf{0.0002} \\
		0.25 & 0.0781 & 0.0516 & 0.0567 & 0.0002 & \textbf{0.0001} \\
		0.30 & 0.0864 & 0.0543 & 0.0592 & 0.0003 & \textbf{0.0001} \\
		0.35 & 0.1015 & 0.0609 & 0.0645 & 0.0003 & \textbf{0.0001} \\
		\bottomrule
	\end{tabular}
\end{table}

\begin{table}[htbp]
	\centering
	\caption{False non-discovery rates under the factor model. Smaller values indicate better performance.}
	\label{TAB_FNR_FACTOR_m100}
	\begin{tabular}{lccccc}
		\toprule
		$p$ & BH & Storey-BH & GBS & MRD-CSX & MRD-GBS \\
		\midrule
		0.01 & 0.0055 & 0.0053 & 0.0055 & \textbf{0.0000} & 0.0001 \\
		0.03 & 0.0148 & 0.0142 & 0.0147 & \textbf{0.0000} & 0.0001 \\
		0.05 & 0.0231 & 0.0220 & 0.0228 & \textbf{0.0001} & 0.0002 \\
		0.075 & 0.0312 & 0.0293 & 0.0304 & \textbf{0.0001} & \textbf{0.0001} \\
		0.10 & 0.0385 & 0.0356 & 0.0371 & \textbf{0.0002} & 0.0003 \\
		0.15 & 0.0511 & 0.0461 & 0.0483 & \textbf{0.0002} & \textbf{0.0002} \\
		0.20 & 0.0627 & 0.0547 & 0.0572 & \textbf{0.0003} & \textbf{0.0003} \\
		0.25 & 0.0734 & 0.0610 & 0.0643 & 0.0004 & \textbf{0.0002} \\
		0.30 & 0.0839 & 0.0653 & 0.0700 & 0.0005 & \textbf{0.0002} \\
		0.35 & 0.0967 & 0.0720 & 0.0770 & 0.0017 & \textbf{0.0010} \\
		\bottomrule
	\end{tabular}
\end{table}

\begin{table}[htbp]
	\centering
	\caption{False non-discovery rates under the fractional Gaussian noise model ($H=0.9$). Smaller values indicate better performance.}
	\label{TAB_FNR_FGN_m100}
	\begin{tabular}{lccccc}
		\toprule
		$p$ & BH & Storey-BH & GBS & MRD-CSX & MRD-GBS \\
		\midrule
		0.01 & 0.0056 & 0.0053 & 0.0055 & \textbf{0.0001} & 0.0002 \\
		0.03 & 0.0149 & 0.0145 & 0.0149 & \textbf{0.0001} & 0.0003 \\
		0.05 & 0.0235 & 0.0229 & 0.0231 & \textbf{0.0003} & 0.0006 \\
		0.075 & 0.0316 & 0.0297 & 0.0303 & \textbf{0.0006} & 0.0010 \\
		0.10 & 0.0390 & 0.0367 & 0.0380 & \textbf{0.0008} & 0.0011 \\
		0.15 & 0.0523 & 0.0478 & 0.0491 & \textbf{0.0019} & 0.0020 \\
		0.20 & 0.0636 & 0.0565 & 0.0573 & 0.0040 & \textbf{0.0032} \\
		0.25 & 0.0736 & 0.0617 & 0.0637 & 0.0076 & \textbf{0.0050} \\
		0.30 & 0.0829 & 0.0658 & 0.0692 & 0.0132 & \textbf{0.0073} \\
		0.35 & 0.0963 & 0.0724 & 0.0746 & 0.0222 & \textbf{0.0104} \\
		\bottomrule
	\end{tabular}
\end{table}

\begin{table}[htbp]
	\centering
	\caption{False non-discovery rates under the Toeplitz model ($\rho=0.9$). Smaller values indicate better performance.}
	\label{TAB_FNR_TOE_m100}
	\begin{tabular}{lccccc}
		\toprule
		$p$ & BH & Storey-BH & GBS & MRD-CSX & MRD-GBS \\
		\midrule
		0.01 & 0.0062 & 0.0052 & 0.0052 & \textbf{0.0000} & \textbf{0.0000} \\
		0.03 & 0.0154 & 0.0144 & 0.0148 & \textbf{0.0000} & \textbf{0.0000} \\
		0.05 & 0.0244 & 0.0224 & 0.0237 & \textbf{0.0000} & \textbf{0.0000} \\
		0.075 & 0.0328 & 0.0301 & 0.0329 & \textbf{0.0000} & \textbf{0.0000} \\
		0.10 & 0.0406 & 0.0341 & 0.0367 & \textbf{0.0001} & \textbf{0.0001} \\
		0.15 & 0.0537 & 0.0469 & 0.0486 & \textbf{0.0001} & \textbf{0.0001} \\
		0.20 & 0.0648 & 0.0549 & 0.0537 & 0.0005 & \textbf{0.0003} \\
		0.25 & 0.0750 & 0.0590 & 0.0635 & 0.0019 & \textbf{0.0012} \\
		0.30 & 0.0844 & 0.0626 & 0.0633 & 0.0047 & \textbf{0.0026} \\
		0.35 & 0.0983 & 0.0684 & 0.0704 & 0.0108 & \textbf{0.0057} \\
		\bottomrule
	\end{tabular}
\end{table}

\begin{table}[htbp]
	\centering
	\caption{False non-discovery rates under the heterogeneous block covariance model. Smaller values indicate better performance.}
	\label{TAB_FNR_BLOCK_m100}
	\begin{tabular}{lccccc}
		\toprule
		$p$ & BH & Storey-BH & GBS & MRD-CSX & MRD-GBS \\
		\midrule
		0.01 & 0.0056 & 0.0054 & 0.0056 & \textbf{0.0009} & 0.0014 \\
		0.03 & 0.0146 & 0.0143 & 0.0145 & \textbf{0.0016} & 0.0033 \\
		0.05 & 0.0229 & 0.0221 & 0.0225 & \textbf{0.0022} & 0.0044 \\
		0.075 & 0.0307 & 0.0293 & 0.0300 & \textbf{0.0034} & 0.0058 \\
		0.10 & 0.0382 & 0.0358 & 0.0370 & \textbf{0.0048} & 0.0069 \\
		0.15 & 0.0519 & 0.0473 & 0.0490 & \textbf{0.0083} & 0.0092 \\
		0.20 & 0.0627 & 0.0547 & 0.0572 & 0.0126 & \textbf{0.0113} \\
		0.25 & 0.0731 & 0.0608 & 0.0642 & 0.0204 & \textbf{0.0162} \\
		0.30 & 0.0837 & 0.0663 & 0.0700 & 0.0319 & \textbf{0.0244} \\
		0.35 & 0.0964 & 0.0713 & 0.0760 & 0.0512 & \textbf{0.0390} \\
		\bottomrule
	\end{tabular}
\end{table}

\begin{table}[htbp]
	\centering
	\caption{False non-discovery rates under the sparse precision-matrix model. Smaller values indicate better performance.}
	\label{TAB_FNR_SPARSE_m100}
	\begin{tabular}{lccccc}
		\toprule
		$p$ & BH & Storey-BH & GBS & MRD-CSX & MRD-GBS \\
		\midrule
		0.01 & 0.0055 & 0.0054 & 0.0054 & \textbf{0.0017} & 0.0024 \\
		0.03 & 0.0147 & 0.0143 & 0.0146 & \textbf{0.0032} & 0.0059 \\
		0.05 & 0.0227 & 0.0220 & 0.0224 & \textbf{0.0050} & 0.0090 \\
		0.075 & 0.0312 & 0.0298 & 0.0305 & \textbf{0.0079} & 0.0123 \\
		0.10 & 0.0386 & 0.0363 & 0.0372 & \textbf{0.0114} & 0.0156 \\
		0.15 & 0.0511 & 0.0467 & 0.0483 & \textbf{0.0210} & 0.0231 \\
		0.20 & 0.0630 & 0.0555 & 0.0577 & 0.0329 & \textbf{0.0312} \\
		0.25 & 0.0727 & 0.0610 & 0.0636 & 0.0479 & \textbf{0.0399} \\
		0.30 & 0.0838 & 0.0667 & 0.0702 & 0.0666 & \textbf{0.0497} \\
		0.35 & 0.0965 & 0.0720 & 0.0761 & 0.0904 & \textbf{0.0604} \\
		\bottomrule
	\end{tabular}
\end{table}

\clearpage

\subsection*{Powers ($m=100$)}

Figure~\ref{FIG_POWER_m100} summarizes the empirical power of the competing
procedures under the six dependence structures considered in this study.
The qualitative behavior closely mirrors that observed for $m=200$,
with the proposed GBS-calibrated MRD procedure frequently attaining
powers very close to one under several dependence structures while
simultaneously maintaining favorable false-discovery characteristics.

\begin{figure}[p]
	\centering
	
	\begin{tabular}{cc}
		
		\includegraphics[width=0.47\textwidth]
		{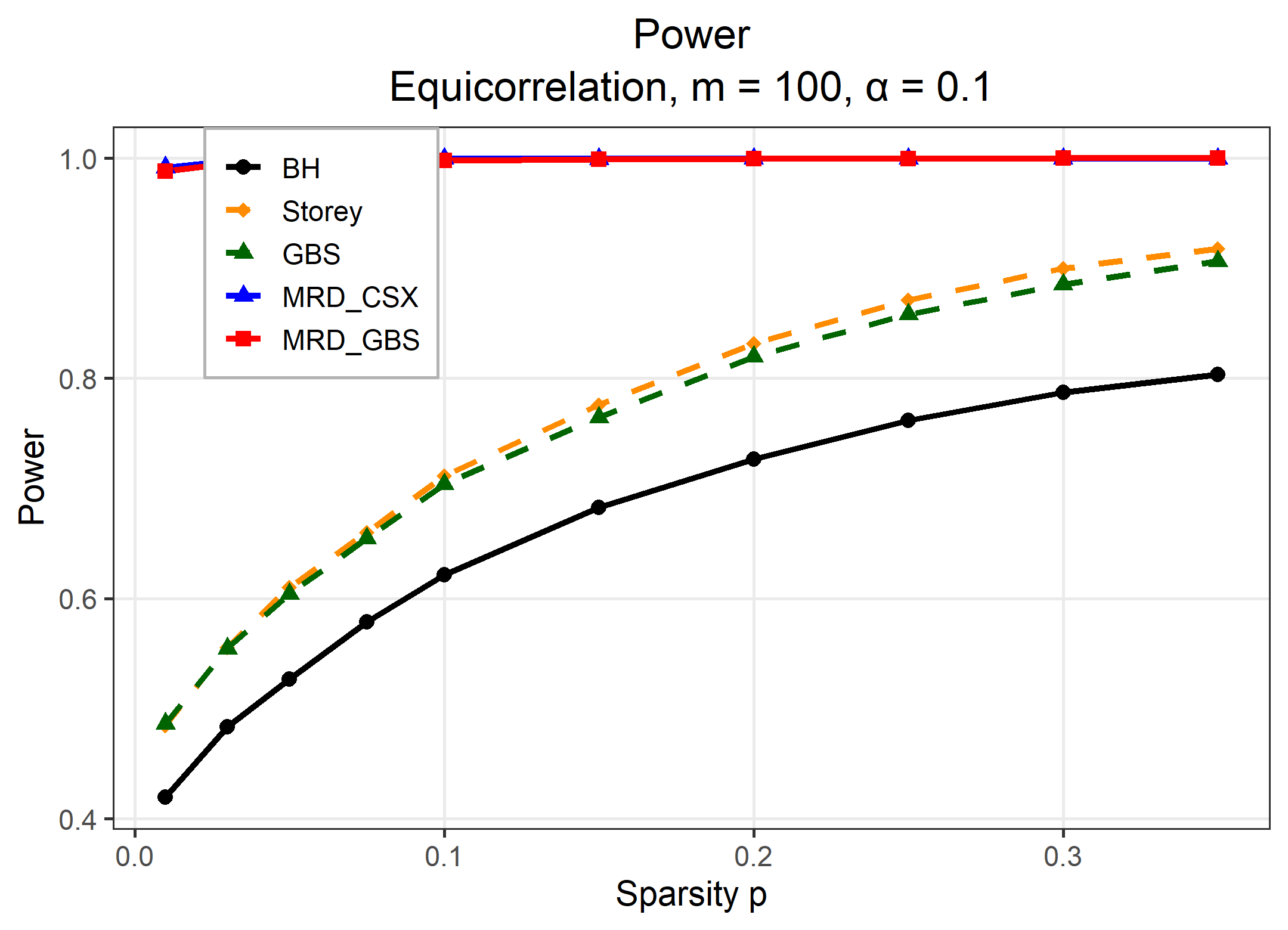}
		&
		\includegraphics[width=0.47\textwidth]
		{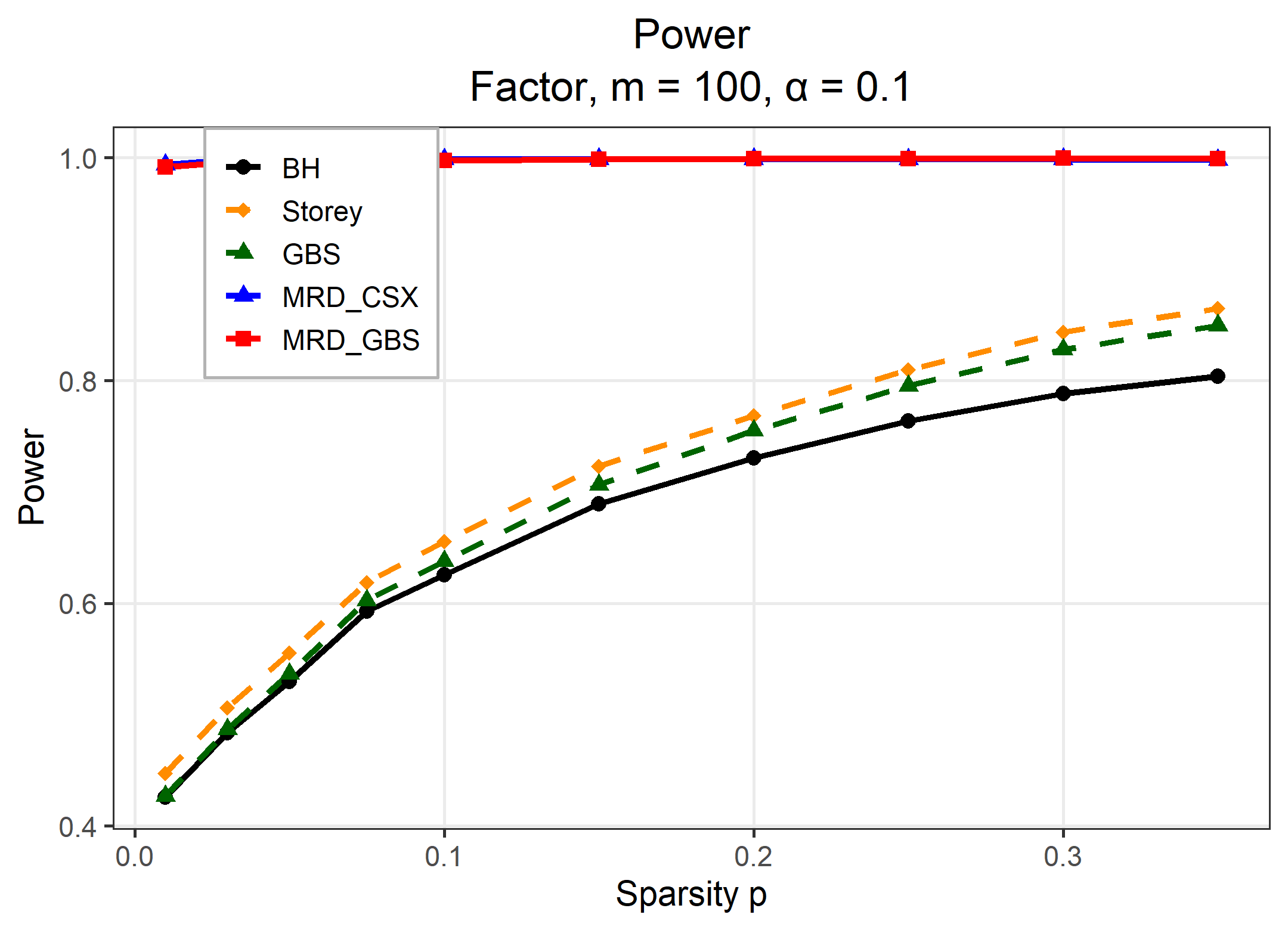}
		\\
		
		(a) Equicorrelation ($\rho=0.7$)
		&
		(b) Factor Model
		\\[0.3cm]
		
		\includegraphics[width=0.47\textwidth]
		{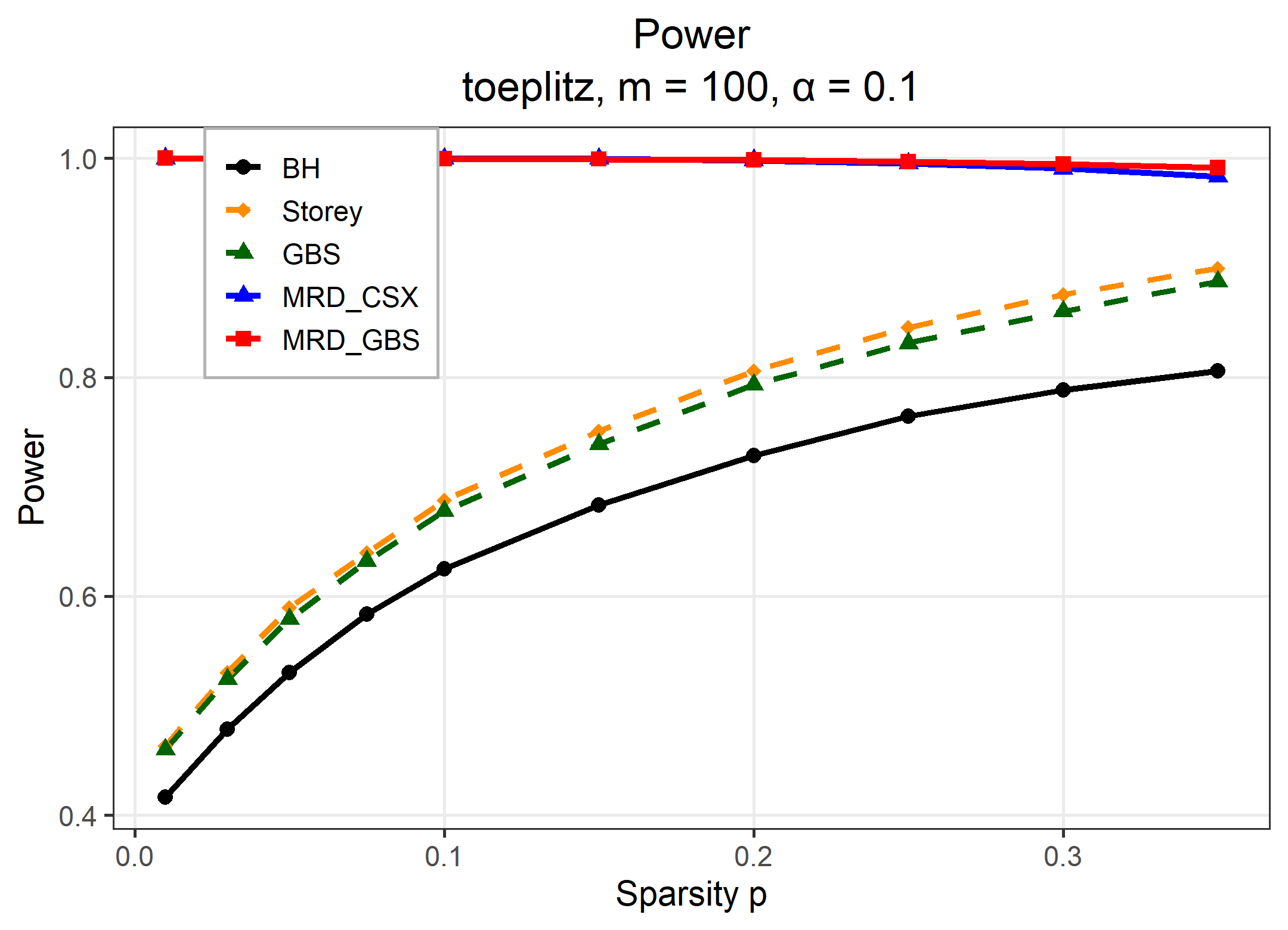}
		&
		\includegraphics[width=0.47\textwidth]
		{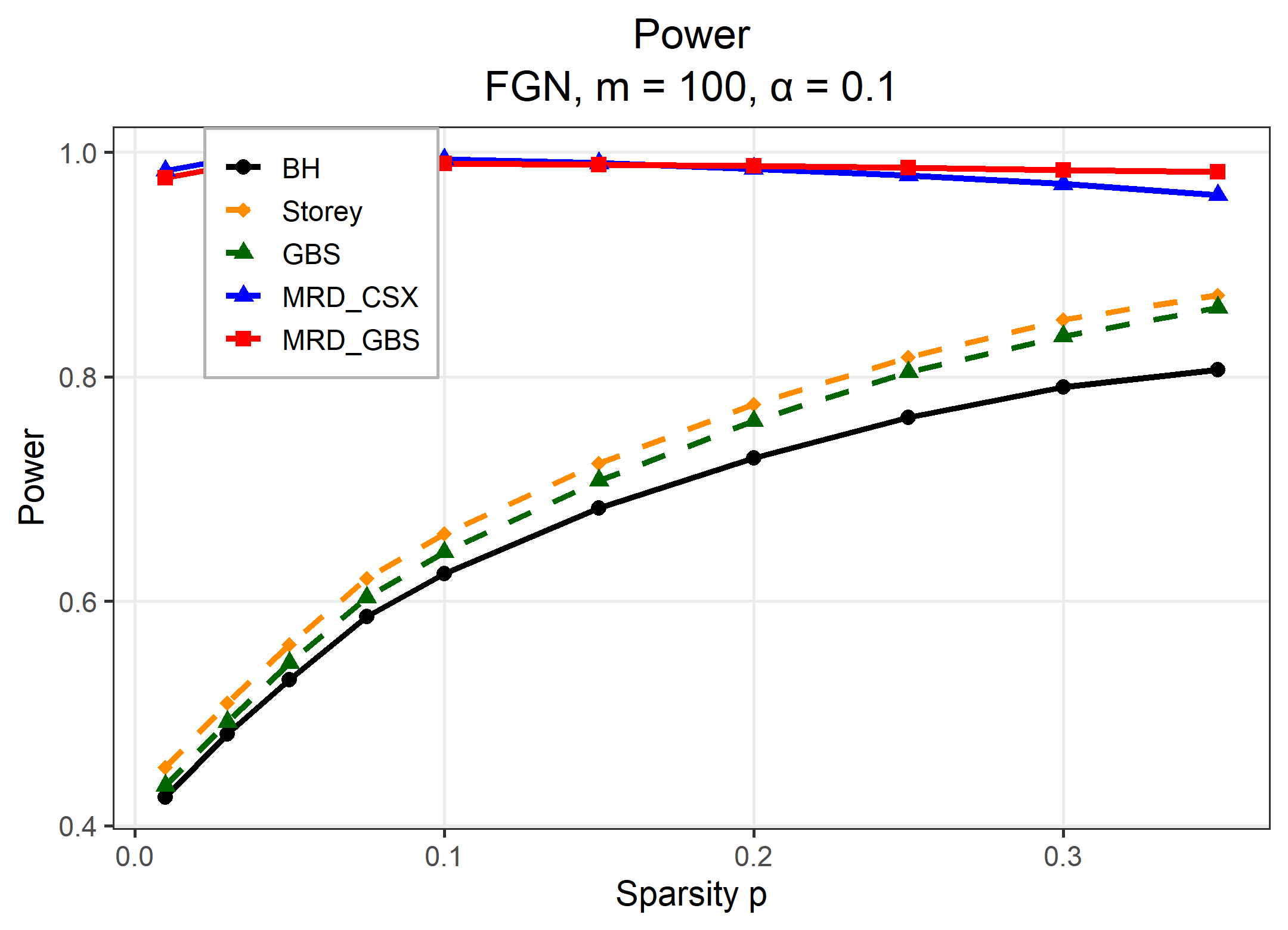}
		\\
		
		(c) Toeplitz ($\rho=0.9$)
		&
		(d) Fractional Gaussian Noise ($H=0.9$)
		\\[0.3cm]
		
		\includegraphics[width=0.47\textwidth]
		{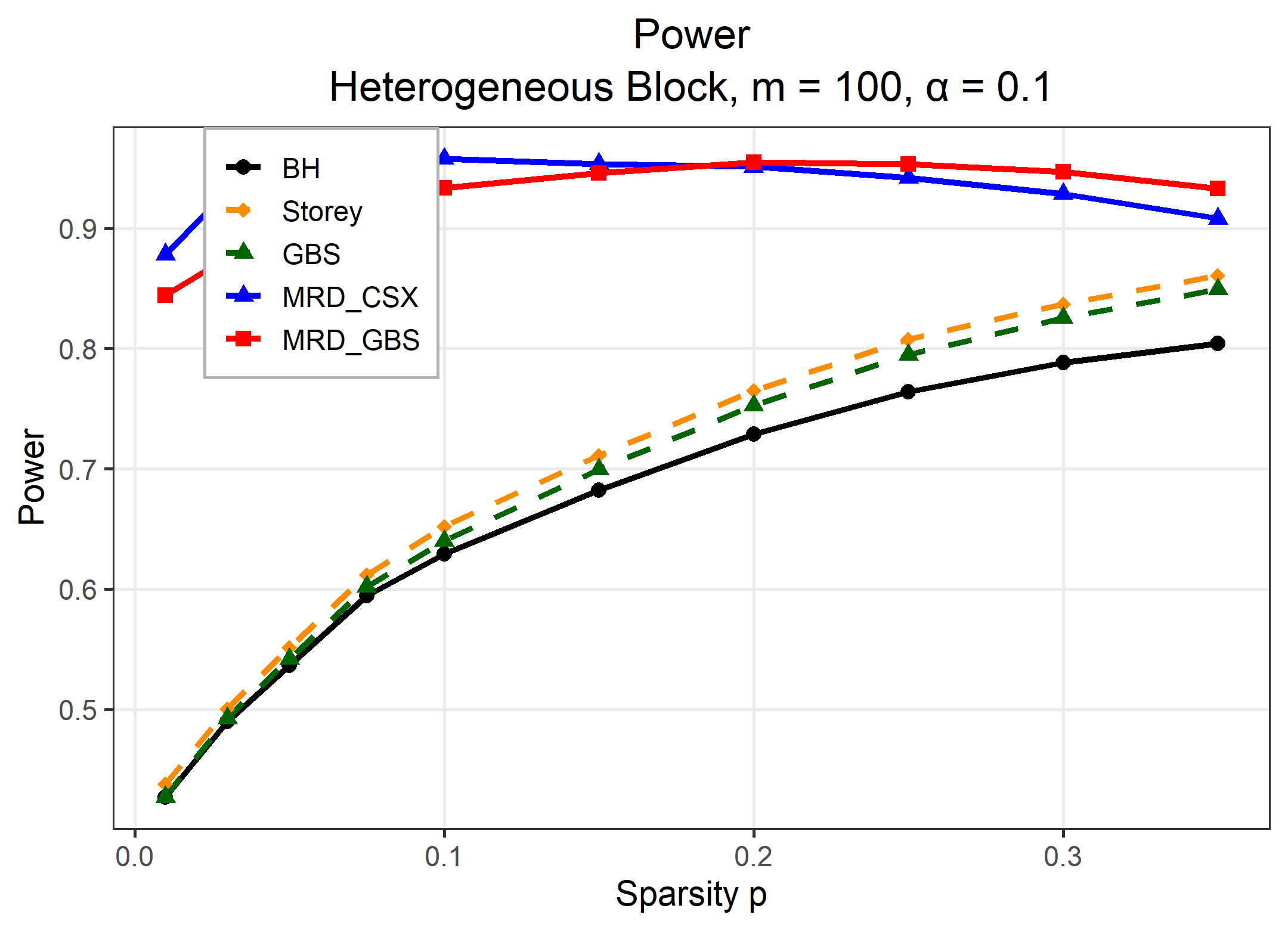}
		&
		\includegraphics[width=0.47\textwidth]
		{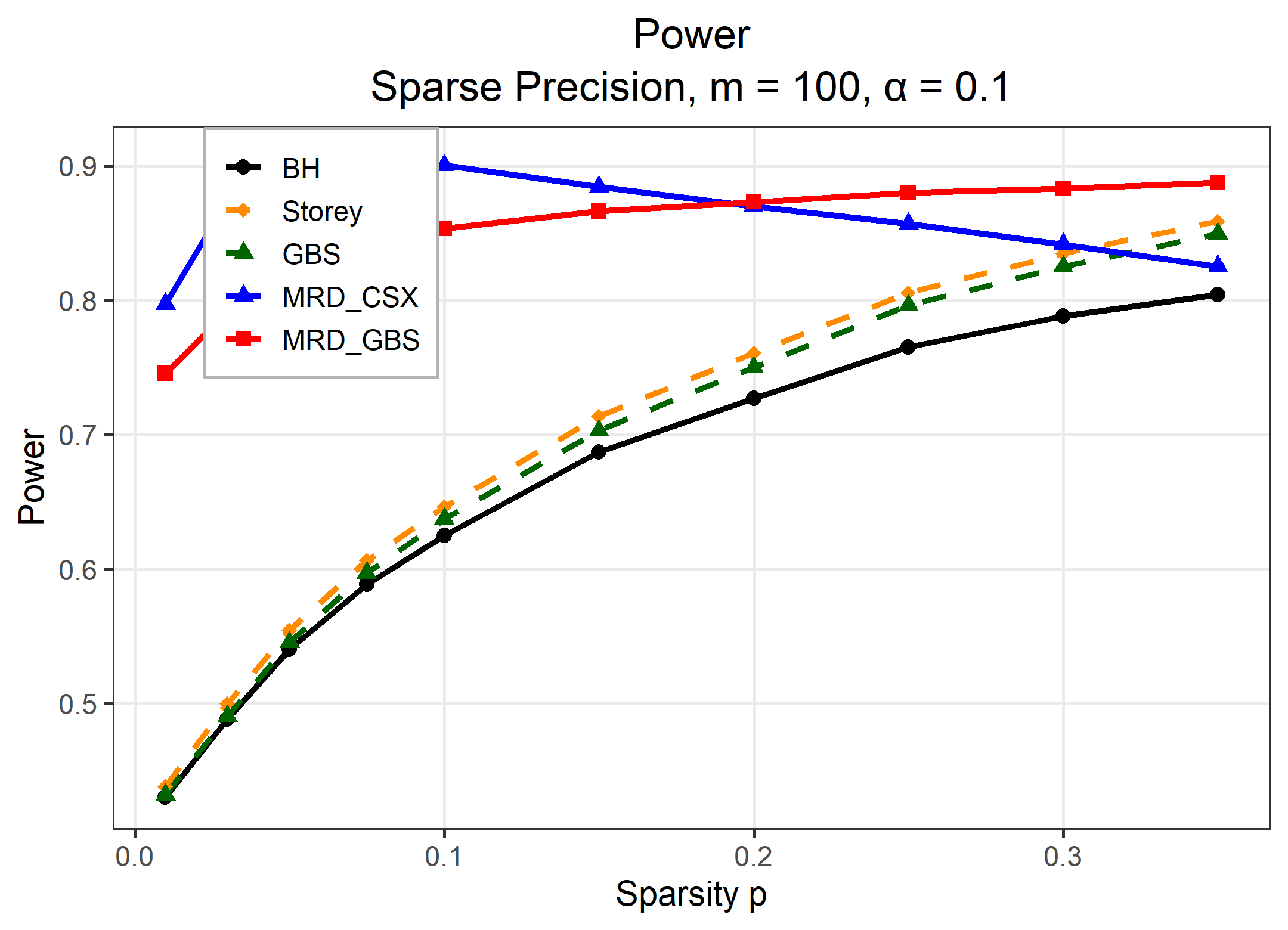}
		\\
		
		(e) Heterogeneous Block
		&
		(f) Sparse Precision Matrix
		
	\end{tabular}
	
\caption{
	Empirical power of the competing multiple testing procedures under six
	representative dependence structures when $m=100$. Larger values indicate
	superior signal-recovery performance. The qualitative behavior is broadly
	consistent with that observed for $m=200$, with the proposed
	GBS-calibrated MRD procedure frequently exhibiting remarkably strong
	power across sparse and moderately sparse regimes. The near-perfect power
	observed under several dependence structures provides further evidence of
	the effectiveness of covariance-adaptive residualization combined with
	stagewise GBS calibration in identifying active signals.
}

	\label{FIG_POWER_m100}
	
\end{figure}

\begin{table}[htbp]
	\centering
	\caption{Power under the Equicorrelation model ($\rho=0.7$). Larger values indicate better performance.}
	\label{TAB_POWER_EQ_m100}
	\begin{tabular}{lccccc}
		\toprule
		$p$ & BH & Storey-BH & GBS & MRD-CSX & MRD-GBS \\
		\midrule
		0.01 & 0.4196 & 0.4847 & 0.4864 & \textbf{0.9915} & 0.9880 \\
		0.03 & 0.4836 & 0.5557 & 0.5548 & \textbf{0.9975} & 0.9949 \\
		0.05 & 0.5267 & 0.6103 & 0.6044 & \textbf{0.9988} & 0.9967 \\
		0.075 & 0.5786 & 0.6598 & 0.6547 & \textbf{0.9993} & 0.9977 \\
		0.10 & 0.6218 & 0.7112 & 0.7038 & \textbf{0.9993} & 0.9980 \\
		0.15 & 0.6826 & 0.7758 & 0.7643 & \textbf{0.9992} & 0.9990 \\
		0.20 & 0.7266 & 0.8319 & 0.8198 & \textbf{0.9993} & \textbf{0.9993} \\
		0.25 & 0.7617 & 0.8713 & 0.8582 & 0.9994 & \textbf{0.9996} \\
		0.30 & 0.7873 & 0.9000 & 0.8853 & 0.9994 & \textbf{0.9998} \\
		0.35 & 0.8035 & 0.9179 & 0.9063 & 0.9994 & \textbf{0.9998} \\
		\bottomrule
	\end{tabular}
\end{table}

\begin{table}[htbp]
	\centering
	\caption{Power under the factor model. Larger values indicate better performance.}
	\label{TAB_POWER_FACTOR_m100}
	\begin{tabular}{lccccc}
		\toprule
		$p$ & BH & Storey-BH & GBS & MRD-CSX & MRD-GBS \\
		\midrule
		0.01 & 0.4259 & 0.4475 & 0.4270 & \textbf{0.9939} & 0.9917 \\
		0.03 & 0.4839 & 0.5060 & 0.4871 & \textbf{0.9978} & 0.9958 \\
		0.05 & 0.5299 & 0.5554 & 0.5367 & \textbf{0.9983} & 0.9955 \\
		0.075 & 0.5931 & 0.6184 & 0.6027 & \textbf{0.9993} & 0.9979 \\
		0.10 & 0.6256 & 0.6557 & 0.6379 & \textbf{0.9984} & 0.9974 \\
		0.15 & 0.6893 & 0.7229 & 0.7064 & \textbf{0.9988} & 0.9986 \\
		0.20 & 0.7303 & 0.7687 & 0.7550 & \textbf{0.9989} & \textbf{0.9989} \\
		0.25 & 0.7636 & 0.8097 & 0.7952 & 0.9987 & \textbf{0.9992} \\
		0.30 & 0.7885 & 0.8435 & 0.8277 & 0.9988 & \textbf{0.9994} \\
		0.35 & 0.8040 & 0.8648 & 0.8495 & 0.9979 & \textbf{0.9990} \\
		\bottomrule
	\end{tabular}
\end{table}

\begin{table}[htbp]
	\centering
	\caption{Power under the fractional Gaussian noise model ($H=0.9$). Larger values indicate better performance.}
	\label{TAB_POWER_FGN_m100}
	\begin{tabular}{lccccc}
		\toprule
		$p$ & BH & Storey-BH & GBS & MRD-CSX & MRD-GBS \\
		\midrule
		0.01 & 0.4254 & 0.4519 & 0.4358 & \textbf{0.9834} & 0.9773 \\
		0.03 & 0.4815 & 0.5092 & 0.4926 & \textbf{0.9938} & 0.9874 \\
		0.05 & 0.5302 & 0.5614 & 0.5451 & \textbf{0.9945} & 0.9888 \\
		0.075 & 0.5866 & 0.6201 & 0.6032 & \textbf{0.9938} & 0.9876 \\
		0.10 & 0.6244 & 0.6601 & 0.6435 & \textbf{0.9935} & 0.9901 \\
		0.15 & 0.6829 & 0.7230 & 0.7076 & \textbf{0.9902} & 0.9889 \\
		0.20 & 0.7278 & 0.7753 & 0.7605 & 0.9854 & \textbf{0.9880} \\
		0.25 & 0.7636 & 0.8175 & 0.8039 & 0.9794 & \textbf{0.9860} \\
		0.30 & 0.7906 & 0.8508 & 0.8362 & 0.9717 & \textbf{0.9843} \\
		0.35 & 0.8064 & 0.8724 & 0.8616 & 0.9619 & \textbf{0.9825} \\
		\bottomrule
	\end{tabular}
\end{table}

\begin{table}[htbp]
	\centering
	\caption{Power under the Toeplitz model ($\rho=0.9$). Larger values indicate better performance.}
	\label{TAB_POWER_TOE_m100}
	\begin{tabular}{lccccc}
		\toprule
		$p$ & BH & Storey-BH & GBS & MRD-CSX & MRD-GBS \\
		\midrule
		0.01 & 0.4164 & 0.4631 & 0.4599 & \textbf{1.0000} & \textbf{1.0000} \\
		0.03 & 0.4786 & 0.5304 & 0.5241 & \textbf{0.9999} & \textbf{0.9999} \\
		0.05 & 0.5303 & 0.5895 & 0.5791 & \textbf{1.0000} & 0.9997 \\
		0.075 & 0.5838 & 0.6400 & 0.6322 & \textbf{0.9998} & 0.9997 \\
		0.10 & 0.6252 & 0.6879 & 0.6778 & \textbf{0.9996} & \textbf{0.9996} \\
		0.15 & 0.6831 & 0.7510 & 0.7389 & \textbf{0.9995} & 0.9994 \\
		0.20 & 0.7283 & 0.8059 & 0.7934 & 0.9983 & \textbf{0.9988} \\
		0.25 & 0.7643 & 0.8458 & 0.8312 & 0.9954 & \textbf{0.9971} \\
		0.30 & 0.7883 & 0.8754 & 0.8601 & 0.9907 & \textbf{0.9950} \\
		0.35 & 0.8059 & 0.8995 & 0.8876 & 0.9832 & \textbf{0.9916} \\
		\bottomrule
	\end{tabular}
\end{table}

\begin{table}[htbp]
	\centering
	\caption{Power under the heterogeneous block covariance model. Larger values indicate better performance.}
	\label{TAB_POWER_BLOCK_m100}
	\begin{tabular}{lccccc}
		\toprule
		$p$ & BH & Storey-BH & GBS & MRD-CSX & MRD-GBS \\
		\midrule
		0.01 & 0.4269 & 0.4381 & 0.4272 & \textbf{0.8783} & 0.8442 \\
		0.03 & 0.4897 & 0.5005 & 0.4927 & \textbf{0.9311} & 0.8757 \\
		0.05 & 0.5364 & 0.5515 & 0.5420 & \textbf{0.9551} & 0.9080 \\
		0.075 & 0.5946 & 0.6121 & 0.6022 & \textbf{0.9582} & 0.9234 \\
		0.10 & 0.6293 & 0.6519 & 0.6401 & \textbf{0.9579} & 0.9338 \\
		0.15 & 0.6823 & 0.7111 & 0.6997 & \textbf{0.9536} & 0.9462 \\
		0.20 & 0.7288 & 0.7650 & 0.7528 & 0.9512 & \textbf{0.9547} \\
		0.25 & 0.7639 & 0.8075 & 0.7947 & 0.9418 & \textbf{0.9535} \\
		0.30 & 0.7881 & 0.8367 & 0.8254 & 0.9285 & \textbf{0.9468} \\
		0.35 & 0.8040 & 0.8609 & 0.8492 & 0.9081 & \textbf{0.9330} \\
		\bottomrule
	\end{tabular}
\end{table}

\begin{table}[htbp]
	\centering
	\caption{Power under the sparse precision-matrix model. Larger values indicate better performance.}
	\label{TAB_POWER_SPARSE_m100}
	\begin{tabular}{lccccc}
		\toprule
		$p$ & BH & Storey-BH & GBS & MRD-CSX & MRD-GBS \\
		\midrule
		0.01 & 0.4301 & 0.4383 & 0.4320 & \textbf{0.7972} & 0.7456 \\
		0.03 & 0.4881 & 0.4997 & 0.4906 & \textbf{0.8743} & 0.7915 \\
		0.05 & 0.5404 & 0.5543 & 0.5455 & \textbf{0.9016} & 0.8183 \\
		0.075 & 0.5883 & 0.6064 & 0.5971 & \textbf{0.9055} & 0.8414 \\
		0.10 & 0.6249 & 0.6463 & 0.6373 & \textbf{0.9002} & 0.8533 \\
		0.15 & 0.6869 & 0.7134 & 0.7030 & \textbf{0.8845} & 0.8663 \\
		0.20 & 0.7269 & 0.7605 & 0.7501 & 0.8697 & \textbf{0.8731} \\
		0.25 & 0.7651 & 0.8053 & 0.7960 & 0.8568 & \textbf{0.8801} \\
		0.30 & 0.7880 & 0.8348 & 0.8250 & 0.8416 & \textbf{0.8832} \\
		0.35 & 0.8043 & 0.8588 & 0.8493 & 0.8251 & \textbf{0.8875} \\
		\bottomrule
	\end{tabular}
\end{table}

\clearpage

\subsection*{Average number of rejections ($m=100$)}

Figure~\ref{FIG_ANR_m100} summarizes the average numbers of rejections
produced by the competing procedures under the six dependence
structures considered in this study. Consistent with the corresponding
$m=200$ experiments, the proposed GBS-calibrated MRD procedure frequently
produces average numbers of rejections that remain close to the expected
numbers of true signals across a broad range of sparsity levels. When
viewed together with the corresponding FDR, FNR, and power results,
this behavior provides additional evidence that the procedure is often
able to identify the underlying support of the signal vector with a
high degree of accuracy.

\begin{figure}[p]
	\centering
	
	\begin{tabular}{cc}
		
		\includegraphics[width=0.47\textwidth]
		{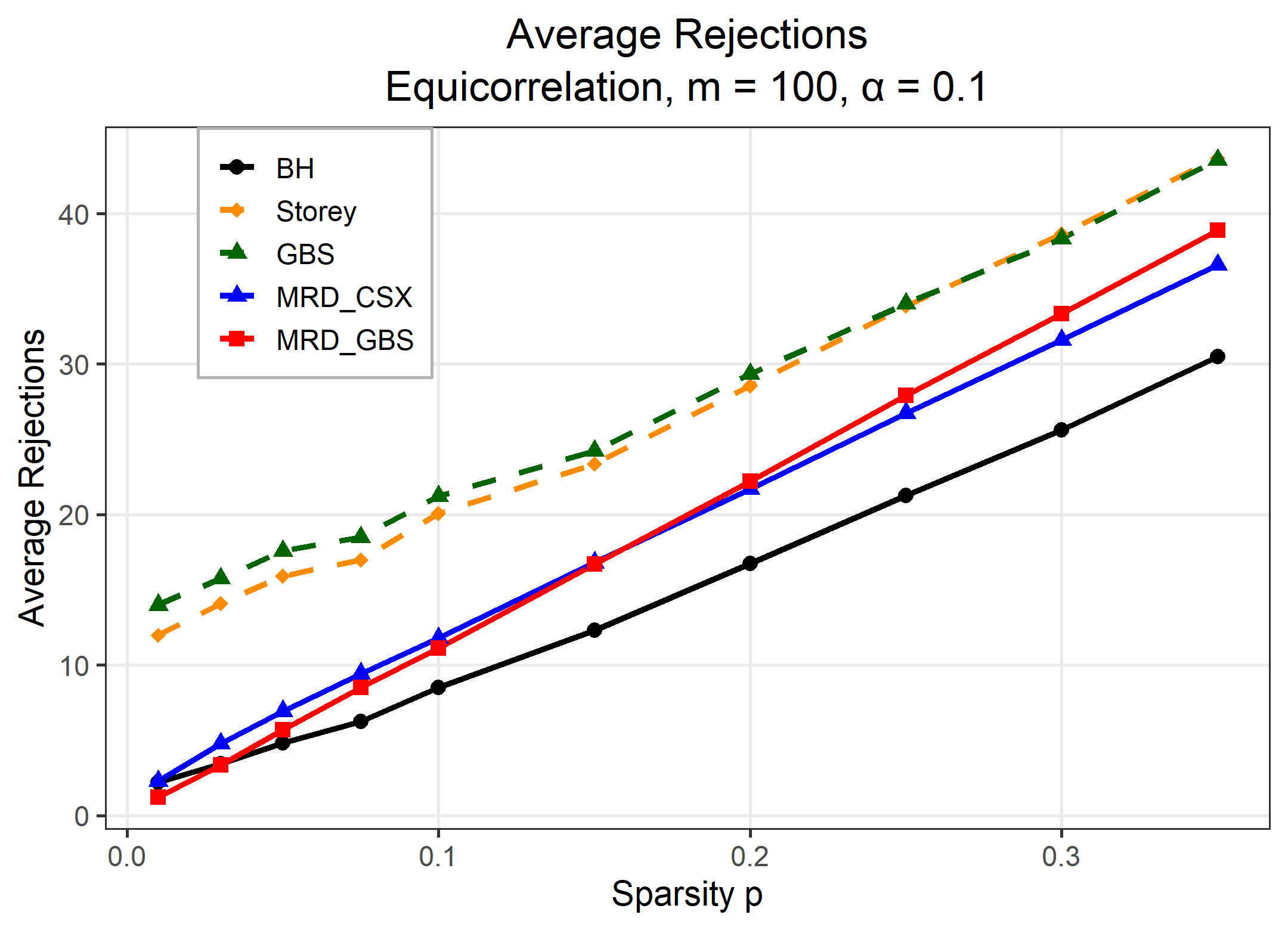}
		&
		\includegraphics[width=0.47\textwidth]
		{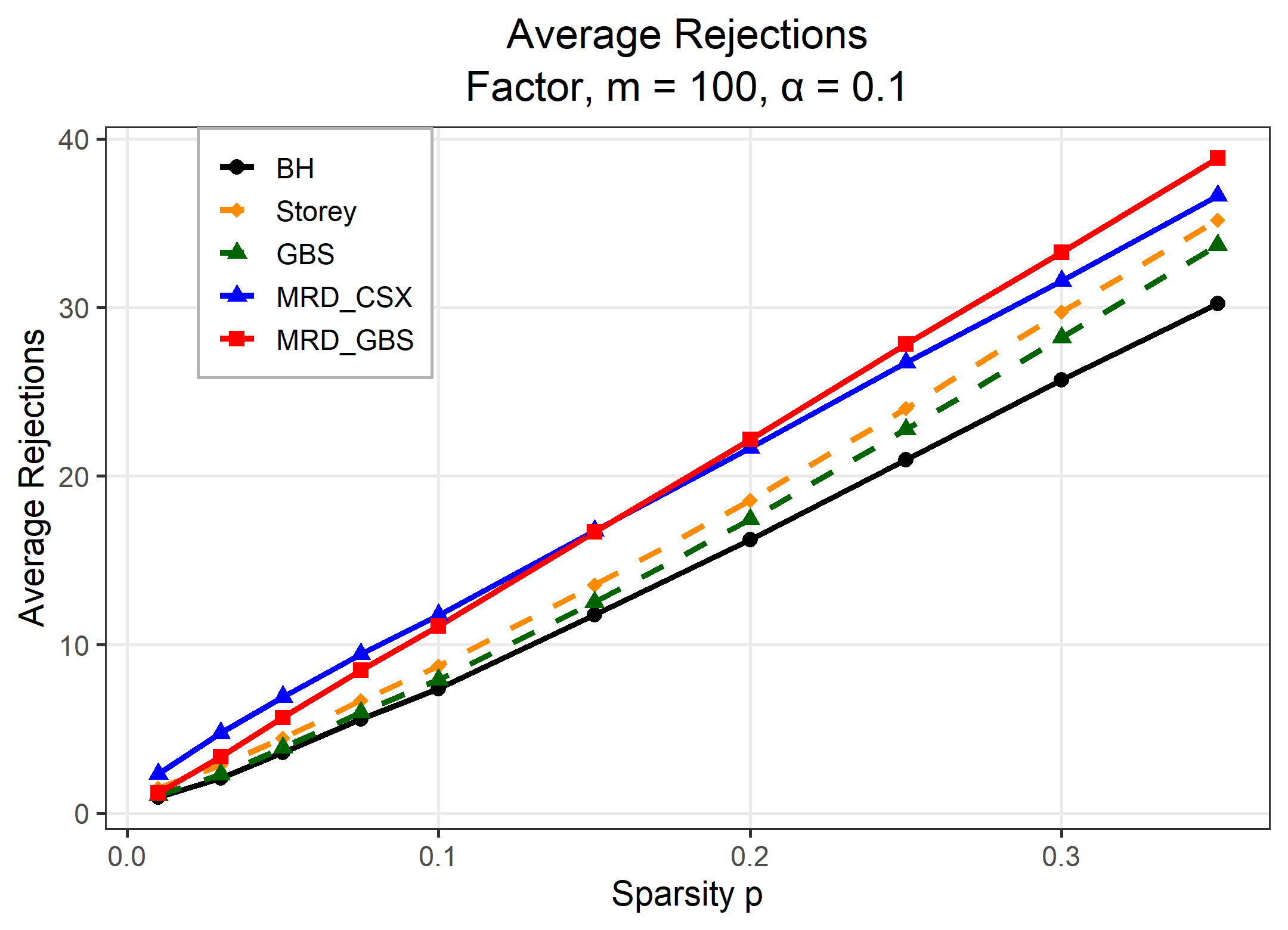}
		\\
		
		(a) Equicorrelation ($\rho=0.7$)
		&
		(b) Factor Model
		\\[0.3cm]
		
		\includegraphics[width=0.47\textwidth]
		{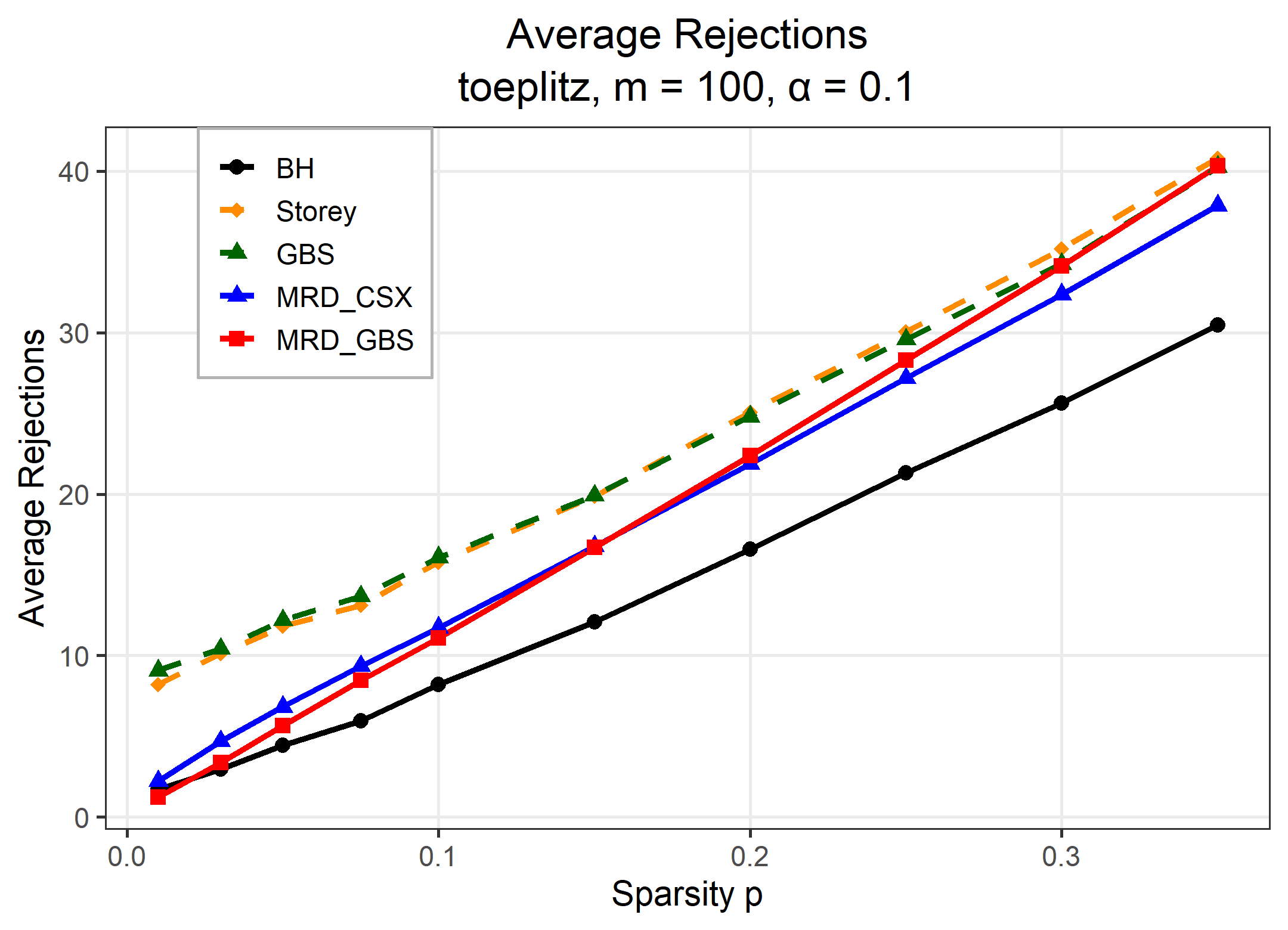}
		&
		\includegraphics[width=0.47\textwidth]
		{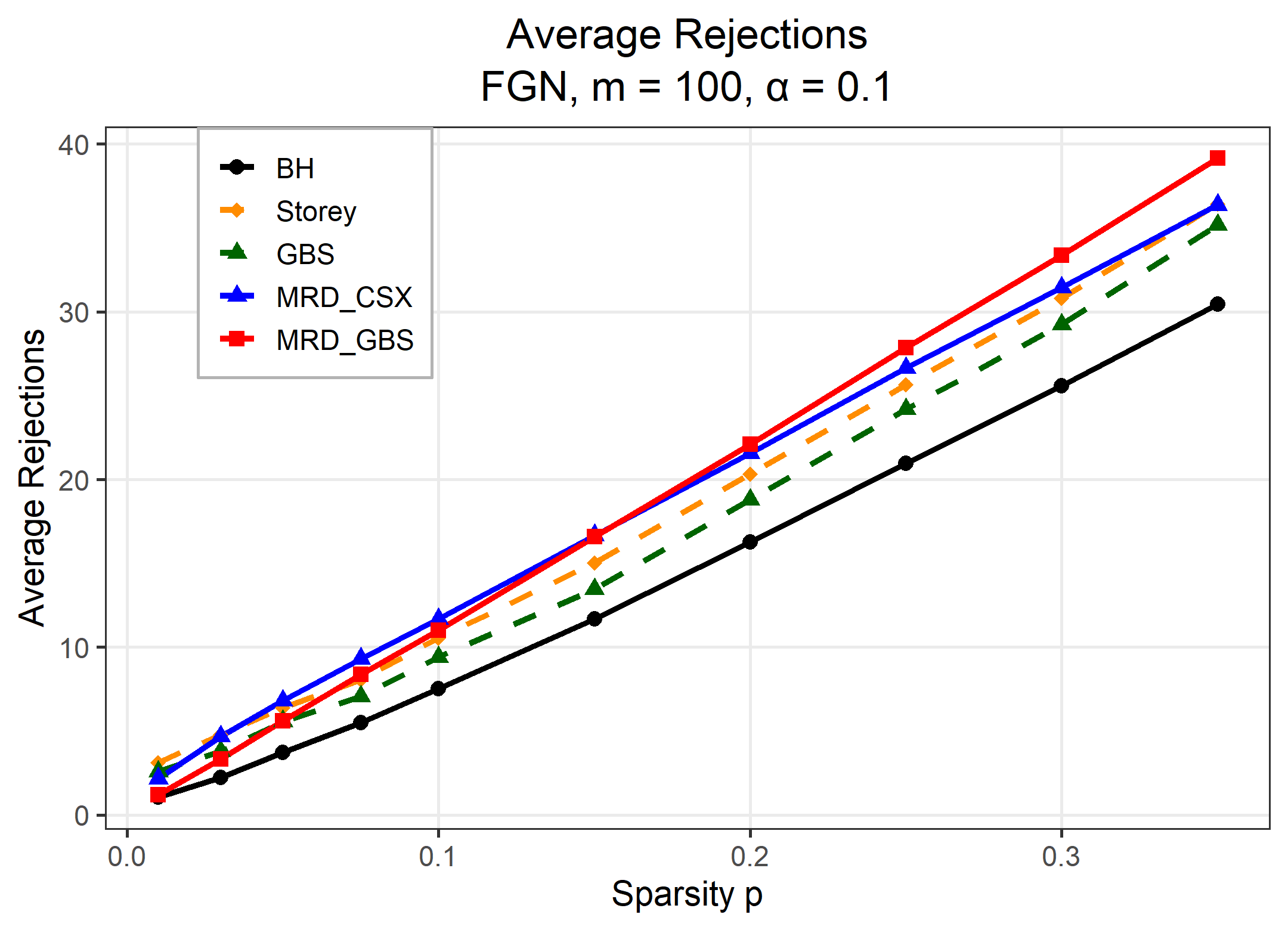}
		\\
		
		(c) Toeplitz ($\rho=0.9$)
		&
		(d) Fractional Gaussian Noise ($H=0.9$)
		\\[0.3cm]
		
		\includegraphics[width=0.47\textwidth]
		{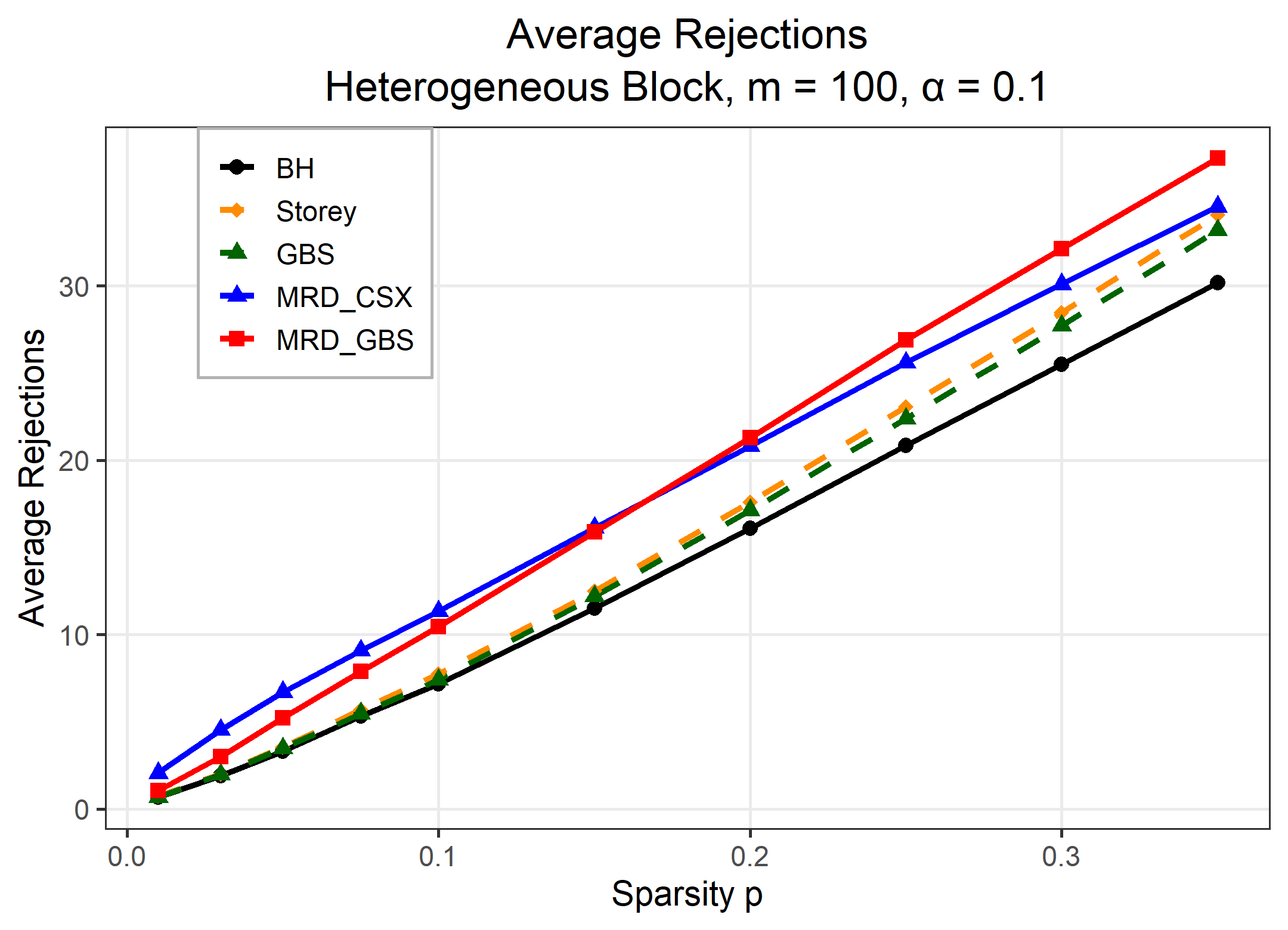}
		&
		\includegraphics[width=0.47\textwidth]
		{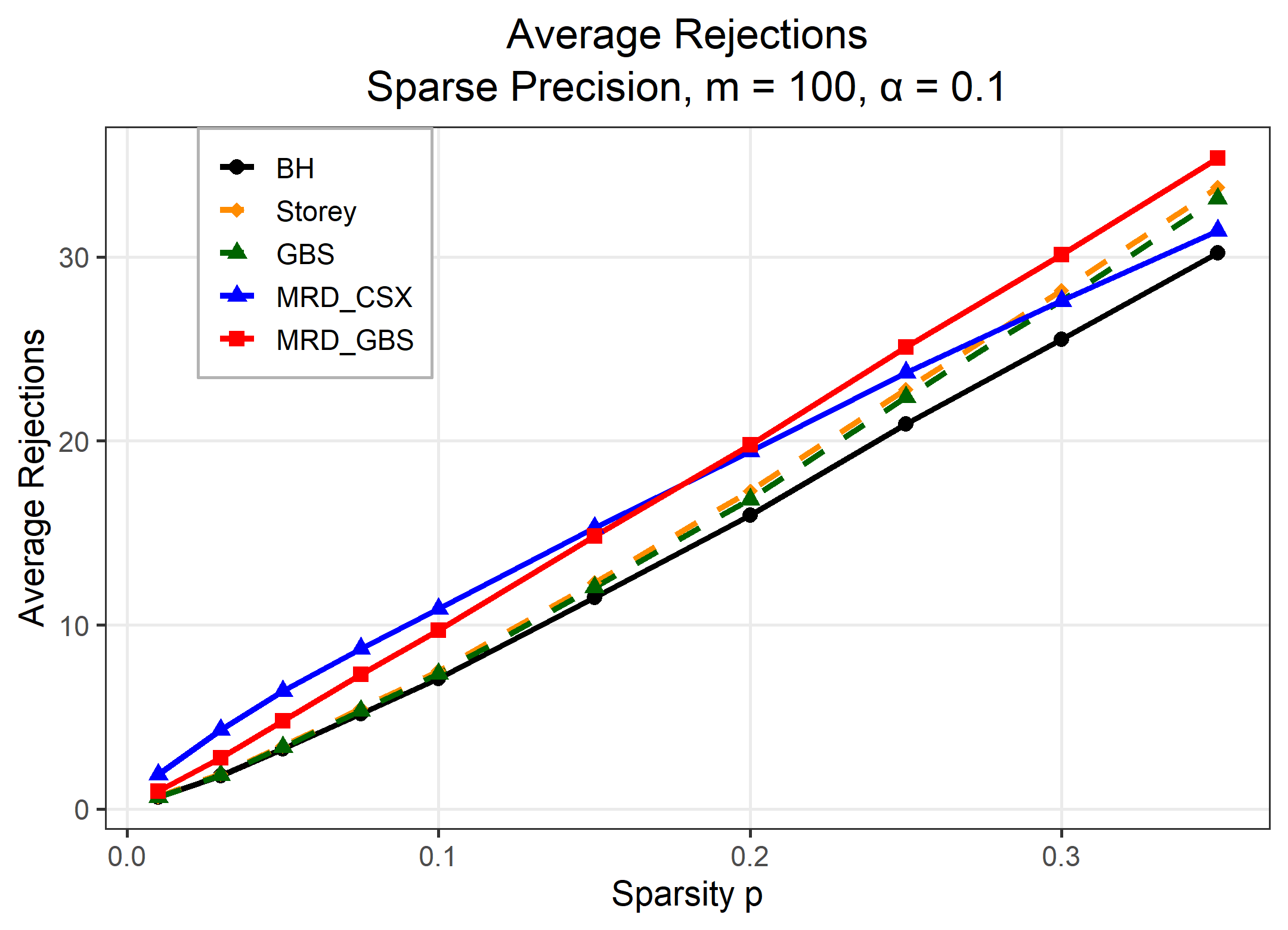}
		\\
		
		(e) Heterogeneous Block
		&
		(f) Sparse Precision Matrix
		
	\end{tabular}
	
\caption{
	Average numbers of rejections (ANR) produced by the competing multiple
	testing procedures under six representative dependence structures when
	$m=100$. This metric provides additional insight into the aggressiveness
	of the competing procedures and helps explain the observed trade-offs
	among false discoveries, missed discoveries, and overall classification
	accuracy. Consistent with the larger-dimensional experiments reported in
	the main text, the proposed GBS-calibrated MRD procedure frequently
	produces average numbers of rejections that closely track the expected
	numbers of true signals, providing further empirical evidence of its
	strong signal-recovery behavior.
}

	\label{FIG_ANR_m100}
	
\end{figure}

\begin{table}[htbp]
	\centering
	\caption{Average numbers of rejections under the Equicorrelation model ($\rho=0.7$). Larger values indicate more aggressive rejection behavior. Since neither excessively large nor excessively small rejection counts are universally preferable, no entries are highlighted in boldface.}
	\label{TAB_REJ_EQ_m100}
	\begin{tabular}{lccccc}
		\toprule
		$p$ & BH & Storey-BH & GBS & MRD-CSX & MRD-GBS \\
		\midrule
		0.01 & 2.19 & 11.97 & 13.99 & 2.28 & 1.21 \\
		0.03 & 3.45 & 14.09 & 15.75 & 4.77 & 3.36 \\
		0.05 & 4.81 & 15.93 & 17.57 & 6.93 & 5.68 \\
		0.075 & 6.24 & 16.98 & 18.48 & 9.43 & 8.49 \\
		0.10 & 8.51 & 20.08 & 21.22 & 11.78 & 11.10 \\
		0.15 & 12.29 & 23.34 & 24.22 & 16.81 & 16.71 \\
		0.20 & 16.76 & 28.57 & 29.35 & 21.68 & 22.23 \\
		0.25 & 21.28 & 33.88 & 34.03 & 26.75 & 27.91 \\
		0.30 & 25.61 & 38.64 & 38.34 & 31.59 & 33.34 \\
		0.35 & 30.50 & 43.67 & 43.58 & 36.59 & 38.90 \\
		\bottomrule
	\end{tabular}
\end{table}

\begin{table}[htbp]
	\centering
	\caption{Average numbers of rejections under the factor model. Larger values indicate more aggressive rejection behavior. Since neither excessively large nor excessively small rejection counts are universally preferable, no entries are highlighted in boldface.}
	\label{TAB_REJ_FACTOR_m100}
	\begin{tabular}{lccccc}
		\toprule
		$p$ & BH & Storey-BH & GBS & MRD-CSX & MRD-GBS \\
		\midrule
		0.01 & 0.95 & 1.48 & 1.06 & 2.32 & 1.21 \\
		0.03 & 2.09 & 2.86 & 2.31 & 4.75 & 3.34 \\
		0.05 & 3.60 & 4.47 & 3.88 & 6.91 & 5.67 \\
		0.075 & 5.57 & 6.68 & 5.97 & 9.43 & 8.47 \\
		0.10 & 7.40 & 8.72 & 7.91 & 11.75 & 11.08 \\
		0.15 & 11.78 & 13.55 & 12.52 & 16.77 & 16.69 \\
		0.20 & 16.22 & 18.55 & 17.43 & 21.66 & 22.17 \\
		0.25 & 20.98 & 24.00 & 22.77 & 26.72 & 27.84 \\
		0.30 & 25.71 & 29.71 & 28.20 & 31.56 & 33.25 \\
		0.35 & 30.24 & 35.16 & 33.68 & 36.63 & 38.84 \\
		\bottomrule
	\end{tabular}
\end{table}

\begin{table}[htbp]
	\centering
	\caption{Average numbers of rejections under the fractional Gaussian noise model ($H=0.9$). Larger values indicate more aggressive rejection behavior. Since neither excessively large nor excessively small rejection counts are universally preferable, no entries are highlighted in boldface.}
	\label{TAB_REJ_FGN_m100}
	\begin{tabular}{lccccc}
		\toprule
		$p$ & BH & Storey-BH & GBS & MRD-CSX & MRD-GBS \\
		\midrule
		0.01 & 1.06 & 3.12 & 2.60 & 2.18 & 1.20 \\
		0.03 & 2.24 & 4.79 & 3.81 & 4.67 & 3.32 \\
		0.05 & 3.74 & 6.41 & 5.59 & 6.81 & 5.61 \\
		0.075 & 5.50 & 8.07 & 7.07 & 9.32 & 8.38 \\
		0.10 & 7.54 & 10.55 & 9.41 & 11.68 & 11.02 \\
		0.15 & 11.70 & 15.03 & 13.47 & 16.68 & 16.58 \\
		0.20 & 16.27 & 20.34 & 18.80 & 21.58 & 22.11 \\
		0.25 & 20.96 & 25.64 & 24.20 & 26.65 & 27.85 \\
		0.30 & 25.58 & 30.80 & 29.23 & 31.44 & 33.36 \\
		0.35 & 30.43 & 36.43 & 35.18 & 36.37 & 39.14 \\
		\bottomrule
	\end{tabular}
\end{table}

\begin{table}[htbp]
	\centering
	\caption{Average numbers of rejections under the Toeplitz model ($\rho=0.9$). Larger values indicate more aggressive rejection behavior. Since neither excessively large nor excessively small rejection counts are universally preferable, no entries are highlighted in boldface.}
	\label{TAB_REJ_TOE_m100}
	\begin{tabular}{lccccc}
		\toprule
		$p$ & BH & Storey-BH & GBS & MRD-CSX & MRD-GBS \\
		\midrule
		0.01 & 1.71 & 8.20 & 9.07 & 2.22 & 1.22 \\
		0.03 & 2.97 & 10.14 & 10.41 & 4.70 & 3.36 \\
		0.05 & 4.45 & 11.82 & 12.20 & 6.82 & 5.66 \\
		0.075 & 5.94 & 13.10 & 13.68 & 9.35 & 8.47 \\
		0.10 & 8.20 & 15.76 & 16.08 & 11.72 & 11.10 \\
		0.15 & 12.08 & 19.86 & 19.90 & 16.78 & 16.71 \\
		0.20 & 16.59 & 25.09 & 24.80 & 21.86 & 22.37 \\
		0.25 & 21.32 & 30.08 & 29.57 & 27.17 & 28.28 \\
		0.30 & 25.62 & 35.19 & 34.25 & 32.37 & 34.14 \\
		0.35 & 30.49 & 40.79 & 40.27 & 37.87 & 40.34 \\
		\bottomrule
	\end{tabular}
\end{table}

\begin{table}[htbp]
	\centering
	\caption{Average numbers of rejections under the heterogeneous block covariance model. Larger values indicate more aggressive rejection behavior. Since neither excessively large nor excessively small rejection counts are universally preferable, no entries are highlighted in boldface.}
	\label{TAB_REJ_BLOCK_m100}
	\begin{tabular}{lccccc}
		\toprule
		$p$ & BH & Storey-BH & GBS & MRD-CSX & MRD-GBS \\
		\midrule
		0.01 & 0.68 & 0.76 & 0.69 & 2.07 & 1.07 \\
		0.03 & 1.90 & 2.03 & 1.97 & 4.54 & 3.02 \\
		0.05 & 3.30 & 3.56 & 3.48 & 6.71 & 5.22 \\
		0.075 & 5.33 & 5.71 & 5.48 & 9.09 & 7.91 \\
		0.10 & 7.16 & 7.75 & 7.41 & 11.34 & 10.44 \\
		0.15 & 11.53 & 12.51 & 12.21 & 16.13 & 15.90 \\
		0.20 & 16.09 & 17.61 & 17.14 & 20.82 & 21.30 \\
		0.25 & 20.85 & 23.07 & 22.38 & 25.60 & 26.89 \\
		0.30 & 25.51 & 28.46 & 27.71 & 30.09 & 32.12 \\
		0.35 & 30.18 & 34.09 & 33.18 & 34.54 & 37.31 \\
		\bottomrule
	\end{tabular}
\end{table}

\begin{table}[htbp]
	\centering
	\caption{Average numbers of rejections under the sparse precision-matrix model. Larger values indicate more aggressive rejection behavior. Since neither excessively large nor excessively small rejection counts are universally preferable, no entries are highlighted in boldface.}
	\label{TAB_REJ_SPARSE_m100}
	\begin{tabular}{lccccc}
		\toprule
		$p$ & BH & Storey-BH & GBS & MRD-CSX & MRD-GBS \\
		\midrule
		0.01 & 0.64 & 0.66 & 0.64 & 1.87 & 0.96 \\
		0.03 & 1.81 & 1.90 & 1.83 & 4.31 & 2.77 \\
		0.05 & 3.26 & 3.42 & 3.36 & 6.41 & 4.80 \\
		0.075 & 5.19 & 5.47 & 5.33 & 8.72 & 7.32 \\
		0.10 & 7.09 & 7.52 & 7.35 & 10.88 & 9.70 \\
		0.15 & 11.51 & 12.30 & 12.04 & 15.27 & 14.82 \\
		0.20 & 15.97 & 17.27 & 16.82 & 19.44 & 19.80 \\
		0.25 & 20.91 & 22.81 & 22.37 & 23.71 & 25.10 \\
		0.30 & 25.52 & 28.17 & 27.61 & 27.61 & 30.11 \\
		0.35 & 30.21 & 33.79 & 33.17 & 31.43 & 35.35 \\
		\bottomrule
	\end{tabular}
\end{table}

\clearpage

%\begin{enumerate}
%\item MRD is decision-theoretically attractive but difficult to calibrate.
%\item GBS provides a systematic calibration.
%\item Admissibility is preserved.
%\item Computation is dramatically simplified.
%\item Covariance geometry becomes visible through the precision-matrix representation.
%\item Relative to BH/Storey/GBS, substantial gains in NMR are observed.
%\item Strong signal-recovery behavior emerges empirically.
%\item Dependence may be a source of information, not merely a nuisance.
%
%\end{enumerate}

%\printbibliography
%\bibliographystyle{abbrvnat}

%\bibliographystyle{apalike}
\bibliography{MRD_GBS_References.bib}

\end{document}